\documentclass[american,letterpaper]{article}
\usepackage{amsmath}
\usepackage{graphicx}
\usepackage{enumerate}
\usepackage{paralist}
\usepackage{natbib}
\usepackage{fullpage}
\usepackage{url} % not crucial - just used below for the URL 

\usepackage{algorithm,algorithmic}
\usepackage{amssymb,amsmath,amsthm,fancyhdr}
\usepackage{bbm}
\usepackage{mathtools}
\usepackage{caption}
\usepackage{subcaption}
\usepackage{placeins}
\usepackage{setspace}

\usepackage{multirow}
\usepackage{footmisc}
\usepackage{paralist}

\usepackage{fancyhdr,amsmath,amsthm,amssymb,bm,url,enumerate,float, bbm} % Required for custom headers
\usepackage{lastpage} % Required to determine the last page for the footer
\usepackage{extramarks} % Required for headers and footers
\usepackage{graphicx,caption} % Required to
\usepackage{pgfplots}
\usepackage{pgfplotstable}
\usepackage{xcolor}
\usepackage{hyperref}
\usepackage{newfloat}
\usepackage[utf8]{inputenc} % allow utf-8 input
\usepackage[T1]{fontenc}    % use 8-bit T1 fonts
\usepackage{booktabs}       % professional-quality tables
\usepackage{amsfonts}       % blackboard 
\usepackage{natbib}
\usepackage{doi}
\usepackage{caption} 
\newcommand\norm[1]{\left\lVert#1\right\rVert} %define norm
 
\theoremstyle{plain}

\newtheorem{lemma}{Lemma}
\newtheorem{theorem}{Theorem}
\newtheorem{proposition}{Proposition}
\newtheorem{corollary}{Corollary}
\theoremstyle{definition}
\newtheorem{definition}{Definition}

\newtheorem*{asu*}{Assumption}
\newtheorem{remark}{Remark}

\usepackage{etoolbox} 

\makeatletter
\newif\if@gather@prefix 
\preto\place@tag@gather{% 
  \if@gather@prefix\iftagsleft@ 
    \kern-\gdisplaywidth@ 
    \rlap{\gather@prefix}% 
    \kern\gdisplaywidth@ 
  \fi\fi 
} 
\appto\place@tag@gather{% 
  \if@gather@prefix\iftagsleft@\else 
    \kern-\displaywidth 
    \rlap{\gather@prefix}% 
    \kern\displaywidth 
  \fi\fi 
  \global\@gather@prefixfalse 
} 
\preto\place@tag{% 
  \if@gather@prefix\iftagsleft@ 
    \kern-\gdisplaywidth@ 
    \rlap{\gather@prefix}% 
    \kern\displaywidth@ 
  \fi\fi 
} 
\appto\place@tag{% 
  \if@gather@prefix\iftagsleft@\else 
    \kern-\displaywidth 
    \rlap{\gather@prefix}% 
    \kern\displaywidth 
  \fi\fi 
  \global\@gather@prefixfalse 
} 
\newcommand*{\beforetext}[1]{% 
  \ifmeasuring@\else
  \gdef\gather@prefix{#1}% 
  \global\@gather@prefixtrue 
  \fi
} 
\makeatother

\DeclareMathOperator*{\argmin}{arg\,min}

\DeclareFloatingEnvironment[name={Supplementary Figure},fileext=lof]{suppfigure}

\pgfplotsset{compat=1.17}

\newcommand{\cN}{\mathcal{N}}

\newcommand{\cJ}{\mathcal{J}}

\newcommand{\bX}{\mathbb{X}}

\newcommand{\bB}{\mathbb{B}}

\newcommand{\R}{\mathbb{R}}
\newcommand{\CS}{S}
\newcommand{\CR}{\mathcal{R}}

\newcommand{\p}{{\rm I}\kern-0.18em{\rm P}}
\newcommand{\E}{{\rm I}\kern-0.18em{\rm E}}
\newcommand{\B}{\boldsymbol}
\newcommand{\by}{\mathbf{y}}
\newcommand{\Bb}{\mathbf}
\newcommand{\schur}[3]{{#1}/ [{#2},{#3}]}
\newcommand{\1}{{\rm 1}\kern-0.24em{\rm I}}
\usepackage{stackengine}
\newcommand\barbelow[1]{\stackunder[1.2pt]{$#1$}{\rule{.8ex}{.075ex}}}

\newcommand{\ind}{\mathrel{\perp\!\!\!\perp}} 
\let\OLDthebibliography\thebibliography
\renewcommand\thebibliography[1]{
  \OLDthebibliography{#1}
  \setlength{\parskip}{0pt}
  \setlength{\itemsep}{0pt plus 0.3ex}
}

\begin{document}
\def\spacingset#1{\renewcommand{\baselinestretch}%
{#1}\small\normalsize} \spacingset{1}

%%%%%%%%%%%%%%%%%%%%%%%%%%%%%%%%%%%%%%%%%%%
% Title of paper
\title{
Multi-Task Learning for Sparsity Pattern Heterogeneity:\\ Statistical and Computational Perspectives
}
\date{}

% List of authors, with corresponding author marked by asterisk
\author{Kayhan Behdin,$^{\dagger,1}$ Gabriel Loewinger,$^{\ast,1,2}$ Kenneth T. Kishida,$^{\dagger\dagger}$
	\\ Giovanni Parmigiani,$^{\ast\ast, \mathsection}$  Rahul Mazumder$^{\dagger,\ddag}$ 
	\\[4pt]
	% Author addresses
	$^{\dagger}$\textit{MIT Operations Research Center, Cambridge, MA}
	\\[2pt]
	$^{\ast}$\textit{Machine Learning Team, National Institute on Mental Health, Bethesda, MD} \\[2pt]
	$^{\dagger\dagger}$\textit{Department of Physiology and Pharmacology,} \textit{Department of Neurosurgery,} \\ \textit{Wake Forest School of Medicine, Winston Salem, NC} \\[2pt]
	$^{\ast\ast}$\textit{Department of Biostatistics, Harvard School of Public Health, Boston, MA} \\[2pt]
	$^{\mathsection}$\textit{Department of Data Science,
		Dana Farber Cancer Institute, Boston, MA}\\ [2pt]
	$^{\ddag}$\textit{MIT Sloan Schools of Management, Cambridge, MA}}

\footnotetext[1]{KB and GL contributed equally to this work.}
\footnotetext[2]{To whom correspondence should be addressed. \textit{gloewinger@gmail.com}}
\maketitle
%%
% Add a footnote for the corresponding author if one has been
% identified in the author list
% \footnotetext[1]{To whom correspondence should be addressed. \textit{gloewinger@gmail.com}}
% \footnotetext[2]{These authors contributed equally to this work.}

\hypersetup{
	pdfkeywords={Multi-Task Learning, Sparse Linear Regression, Discrete Optimization, Integer Programming, Variable Selection}
}
\begin{abstract}{
We consider a problem in Multi-Task Learning (MTL) where multiple linear models are jointly trained on a collection of datasets (``tasks’’). A key novelty of our framework is that it allows the sparsity pattern of regression coefficients and the values of non-zero coefficients to differ across tasks while still leveraging partially shared structure. Our methods encourage models to share information across tasks through \textit{separately} encouraging 1) coefficient \textit{supports}, and/or 2) nonzero coefficient \textit{values} to be similar. This allows models to borrow strength during variable selection even when non-zero coefficient values differ across tasks.  We propose a novel mixed-integer programming formulation for our estimator. We develop custom scalable algorithms based on block coordinate descent and combinatorial local search to obtain high-quality (approximate) solutions for our estimator. Additionally, we propose a novel exact optimization algorithm to obtain globally optimal solutions.
We investigate the theoretical properties of our estimators.
We formally show how our estimators leverage the shared support information across tasks to achieve better variable selection performance. We evaluate the performance of our methods in simulations and two biomedical applications. Our proposed approaches appear to outperform other sparse MTL methods in variable selection and prediction accuracy. 
We provide the $\texttt{sMTL}$ package on \texttt{CRAN}.
}
\end{abstract}

\section{Introduction}
\label{intro}
Multi-task learning (MTL) seeks to leverage structure shared across datasets to improve prediction performance of each model on its respective task. Implementation of MTL has been successful in a variety of biomedical settings, such as neuroscience \citep{Moran, neuroImg_multiTask}, oncology \citep{ivan} and the synthesis of microarray datasets \citep{kim}. Medical applications of machine learning models increasingly rely on model interpretability as a means of drawing scientific conclusions, avoiding biases, engendering trust in model predictions, and encouraging widespread adoption in clinical settings \citep{modInterpret}. As such, considerable methodological research continues to explore linear models, where parameter interpretation is more transparent than more flexible approaches, such as kernel-based or neural network-based methods. A major challenge in biomedical settings is that covariates are often high-dimensional and sample sizes are low, thereby requiring regularization or variable selection. ``Multi-task feature learning'' is a sub-field of MTL that has a rich literature on methods for these settings. Our focus here is on feature selection where we seek to select a subset of relevant features~\citep{Zhang}. 
In the MTL setting, we are given $K\geq 2$ tasks, each identified with a design matrix $\bX_k\in\R^{n_k\times p}$ and an outcome vector $\by_k\in\R^{n_k}$ for $k\in[K]$, where $n_k$ is the number of observations of task $k$, $p$ is the dimension of the covariates, and $[K]$ denotes the set $\{1,2,...,K\}$. In this paper, we seek a linear approximation to each task's model. In particular, our goal is to present ``interpretable'' linear estimators $\hat{\B{\beta}}_k, k\in[K]$ such that $\by_k\approx \bX_k\hat{\B{\beta}}_k$. Before discussing our contributions in this paper, we briefly review relevant sparse regression and MTL literature that provide the methodological basis for our contribution.

\subsection{Best Subset Selection} We motivate our methods through the well-known Sparse Linear Regression (SLR) problem. Given a model matrix $\bX\in\R^{n\times p}$ and outcome vector $\by\in\R^n$, we seek to estimate the vector $\B{\beta}\in\R^p$ such that $\B{\beta}$ is sparse (i.e., has only a few nonzero coefficients) and the least squares error $\|\by-\bX\B{\beta}\|_2^2$ is minimized. Sparsity in the model is desirable from statistical and interpretability perspectives, especially in high-dimensional settings where $p\gg n$.

In general, estimators proposed for the SLR problem seek to minimize a penalized (or constrained) version of the least squares error, where the penalty (or constraint) promotes sparsity. Common choices for this setup include the $\ell_2$ penalty, or Ridge~\citep{ridge}, $\ell_1$ penalties such as the Lasso~\citep{Lasso} and non-convex penalties like SCAD~\citep{scad}. In this paper, we build on the Best Subset Selection (BSS) estimator~\citep{miller1990subset}. Formally, the BSS estimator is defined as
\begin{align} \label{l0}
	&\underset{\boldsymbol{\beta} }{\mbox{min }} \norm{  \mathbf{y} - \mathbb{X} \boldsymbol{\beta} }_2^2 ~~~
	\mbox{s.t. }~\norm{\boldsymbol{\beta}}_0 \leq s 
\end{align}
where $\|\cdot\|_0$ denotes the number of nonzero elements of a vector. Problem~\eqref{l0} is computationally challenging~\citep{natarajan1995sparse}. However, recent advances in mixed-integer programming have led to the development of algorithms that can obtain good or optimal solutions to problem~\eqref{l0} for moderate to large problem instances. Some effective approaches to obtaining good solutions to problem~\eqref{l0} include approximate methods~\citep[see also, references therein]{l0learn}   
and integer programming based global optimization methods~\citep{bestSubset,bart,l0BnB}. 
Recent work has also explored the performance of combined penalties such as $\ell_0\ell_2$~\citep{lowSNRL0}, as combining a $\ell_0$ penalty with coefficient regularization can improve prediction performance. 
Importantly, these methods also have connections to Bayesian models using priors to achieve desirable shrinkage of regression coefficients and regularization. 
Certain sparsity-inducing priors for linear regression yield maximum a posterior (MAP) estimates that are the point estimates arising from the SLR problem (e.g., see \cite{lowSNRL0} and references therein). For example, the Bernoulli-Gaussian mixture models yield MAPs corresponding to solutions to the $\ell_2$ regularized $\ell_0$-constrained SLR problem \citep{bayesianL0, berGauss}. Similarly, imposing a Bernoulli-Laplace prior on the regression coefficients yields MAPs that also solve the $\ell_1$ regularized $\ell_0$-constrained SLR problem \citep{Amini2012TheAF, bayes_L0L1}. These and other related Bayesian methods (e.g., see \cite{Decoupling_Bayes} and references therein) achieve sparse coefficient estimates, and decouple variable selection and shrinkage of regression coefficients, an approach that bears conceptual similarities to the strategy we apply in the MTL setting.

\subsection{Multi-task Learning} A common linear MTL strategy is to jointly fit $K$ linear models with a shared rank restriction \citep{rrr} or penalty on the matrix of model coefficients, $\mathbb{B} = [\B{\beta}_1,\ldots, \B{\beta}_K]_{p \times K}$. For example, penalties such as trace \citep{traceNorm}, graph Laplacian \citep{laplacian} or spectral \citep{spectralNorm} norms encourage models to borrow strength across tasks. These methods generally do not
result in sparse $\B{\beta}_k$ estimates
thereby limiting interpretability. One alternative is to use sparsity-inducing penalties on $\mathbb{B}$. For example, the Group Lasso regularizer $\sum_{j=1}^p
\norm{ \boldsymbol{\beta}^{(j)}}_2$, where $\B{\beta}^{(j)} \in \mathbb{R}^{K}$ is the vector of coefficients for covariate $j$, is a commonly used convex regularizer to induce sparse solutions~\citep{groupLasso,groupLasso2}. This method also encourages tasks to share a common support (i.e., $\hat{\B{\beta}}_k$ for all $k\in[K]$ share the same location of nonzeros). However, this general \emph{all-in} or \emph{all-out} variable selection approach  can degrade the prediction performance, or result in misleading variable selection, if the ``true'' supports are not identical across tasks. A rich literature on variations of $\ell_{p,q}$ norms exist to induce sparsity and borrow strength across tasks in linear MTL (see~\cite{Zhang} and references therein). 

An alternative approach is to allow each $\B{\beta}_k$ to have a different sparsity pattern, but still encourage them to share information across tasks. For example, regularization of the model coefficients with the non-convex penalty, $\alpha \sum_{k=1}^K \left ( \norm{\boldsymbol{\beta}_k}_1 - \norm{\boldsymbol{\beta}_k}_2 \right )$, can allow for support heterogeneity \citep{sparseMultiTask}. The multi-level Lasso proposed by \cite{het1} writes $\B{\beta}_k$ as the element-wise product of two vectors, one common to all tasks and one common to task $k$. This \textit{product based decomposition} of task coefficients allows for heterogeneous sparsity patterns and has been extended to more general regularizers \citep{multiplicativeMT}.
Bayesian approaches to MTL can also encourage sparse solutions with potentially differing sparsity patterns by, for example, using sparsity-inducing priors like the matrix-variate generalized normal prior \citep{prob1} or generalized horseshoe prior \citep{prob3}. Some priors yield MAPs that are the same as penalized maximum likelihood point estimates, and have been used to motivate specific penalties in optimization-based MTL formulations (e.g., see \cite{trippa} and references therein). 
It can be helpful to apply methods that provide the flexibility to separately select variables and shrink nonzero coefficients.
\subsection{Outline of our approach and contributions}
We propose a flexible framework in linear MTL that: 1) 
estimates a sparse solution by directly controlling the number of nonzero regression coefficients, 
2) allows for differing sparsity patterns (supports) across tasks, and 3) shares information through the supports and values of the task regression coefficients. 
To achieve such flexibility, we propose a novel estimator given by the solution to an integer program. 
Our estimator allows the regression coefficient estimates of each task to have their own separate and sparse support. By using binary variables to model the supports of the regression coefficients ($\boldsymbol{z}_k = \mathbbm{1}(\B{\beta}_k\neq \boldsymbol{0}) $), we introduce a penalty that encourages the supports of the tasks $\{\B{z}_k\}_1^{K}$ to be similar to each other. This encourages tasks to share information during variable selection while still allowing the coefficient estimates of the tasks to have different sparsity patterns. Importantly, this allows for borrowing strength across the $\B{\beta}_k$'s even when the \textit{values} of the nonzero $\beta_{k,j}$'s differ widely across tasks for a given covariate, $j$. As a result, our framework allows for sharing information across tasks through two \textit{separate} mechanisms: 1) penalties that shrink the $\B{\beta}_k$ \textit{values} together, and/or 2) penalties that encourage  $\boldsymbol{z}_k$'s, the \textit{supports} of $\B{\beta}_k$'s,
to be similar across tasks.

To provide further insight into our modeling framework, we present numerical experiments on synthetic data\footnote{Please see Section \ref{coef_supp} for more information on the design and results of these sets of experiments.}. In Figure~\ref{fig:introSupp}, we compare coefficient estimates from a 
Group Lasso estimator and our method, ``Zbar+L2,'' against the true simulated regression coefficients. In Figure~\ref{fig:introSupp} [Left panel], we show the true $\B{\beta}_k^*$ for $k\in\{1,2\}$. Next, we show the output of our model and the Group Lasso-based estimator. The latter selects a variable for either both tasks or neither task and fails to recover the true supports (note that the true supports differ across tasks). Since the Group Lasso penalty shrinks all $\boldsymbol{\beta}_k$ values, the nonzero coefficient estimates are noticeably shrunk. Finally, the number of nonzero coefficients in the 
estimates are far too large. Our proposed method recovers the full support by allowing for heterogeneity of supports while encouraging the supports of different tasks to be similar. In Figure~\ref{fig:introSupp} [Right panel], we show the results of our method for different values of $\delta$, which is the regularization parameter (aka hyperparameter) that controls the degree to which the $\B{z}_k$ are encouraged to be similar to each other. For $\delta=0$, no support heterogeneity shrinkage is applied and thus no information is shared across tasks. Given the low sample size, borrowing information across tasks is critical and thus the solutions fail to recover the true support.
For $\delta=0.05$, our method encourages the tasks to have similar supports, which allows the models to leverage the shared information across tasks to recover the supports of both tasks perfectly. Finally, increasing $\delta$ to one makes the supports equal, though the coefficient values can differ. This common-support solution misses some locations of the correct support. This shows the strength of our \textit{support heterogeneity regularization} approach compared to common support methods, independent models ($\delta = 0$), and methods like the Group Lasso.

\begin{figure*}[t] 
	\centering
	\begin{subfigure}[t]{0.4\textwidth} 
		\centering
		\includegraphics[width=0.80\linewidth]{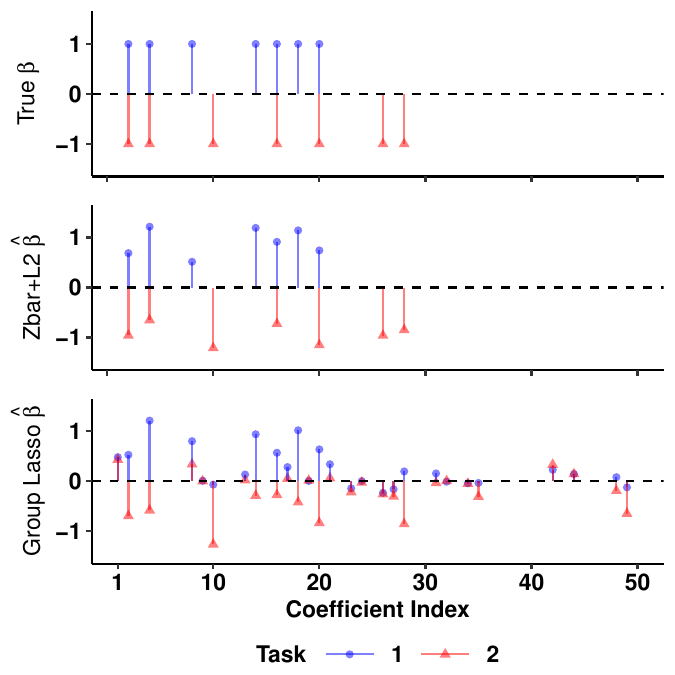}
		\label{fig:SHR_c1}
	\end{subfigure}
	\begin{subfigure}[t]{0.4\textwidth}
		\centering
		\includegraphics[width=0.80\linewidth]{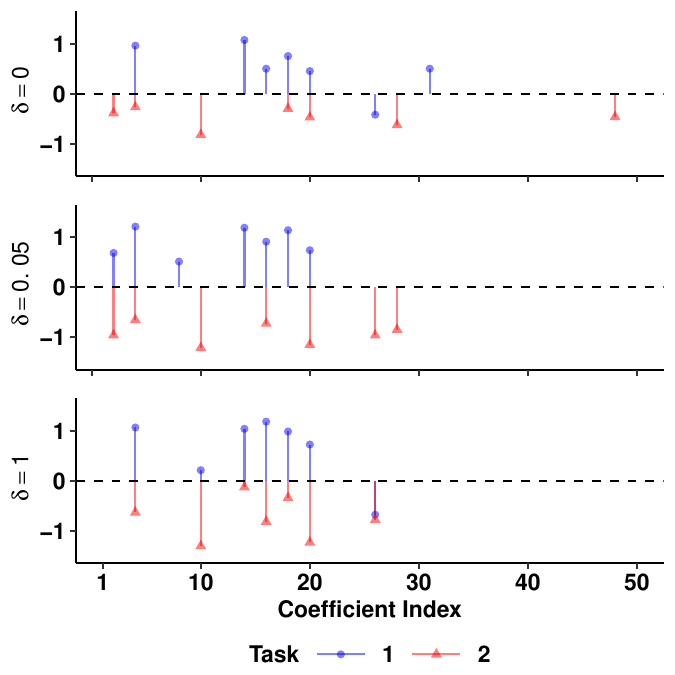}
		\label{fig:SHR_d1}
	\end{subfigure}
	\caption{\footnotesize [Left] The top panel shows the true simulated $\boldsymbol{\beta}$, and the bottom two panels show coefficient estimates from tuned models. The color and symbol indicate task index. Our proposed method, Zbar+L2, recovers the support and improves coefficient estimate accuracy over the group Lasso estimates made with $\texttt{glmnet}$. [Right] Our penalty reduces support heterogeneity across tasks as $\delta$ increases.} 
	\label{fig:introSupp}
\end{figure*}

Motivated by our empirical findings, we theoretically analyze our proposed estimator under suitable regularity assumptions.
We derive an upper bound on the sum of prediction errors across all $K$ tasks that scales as $O\left((Ks/n)\log(p/s)\right)$ where $s$ is an upper bound on the sparsity level of different tasks (cf. Section~\ref{sec:theory-1} for a more detailed discussion). Moreover, our theory formalizes how the proposed estimator limits the support heterogeneity of the solution and clarifies that the method encourages tasks to share information through the coefficient \textit{supports} rather than through coefficient \textit{values}. In particular, if the true model coefficients have (near) common supports, the estimated regression coefficients will also have (near) common supports. We also investigate the variable selection performance of our framework and show that by sharing information through task \textit{supports}, our framework produces good variable selection performance even when some tasks' observations have low signal. Importantly, this holds even if the exact common support assumption does not hold (cf. Section~\ref{sec:theory-2} for a more general discussion).

Since our estimator is given by the solution to a mixed integer program (MIP), it can be computed using off-the-shelf solvers for small/moderate-scale instances. 
To enhance the scalability of our estimator, we propose novel algorithms. 
In particular, we propose (i) approximate algorithms that can obtain good solutions quickly, and 
(ii) exact algorithms that deliver optimal solutions. The approximate algorithms do not have (global) optimality guarantees but are useful from a practical standpoint. The exact algorithms serve to quantify the quality of the solutions available from (i) and possibly improve the quality of the approximate solutions with associated optimality certificates via dual bounds.  
Our experiments on synthetic and real data show that 
our approximate algorithms lead to high-quality solutions (which are often optimal), result in good statistical performance, and yield interpretable solutions. In terms of statistical properties, we study estimators corresponding to the global optimum as well as approximate solutions available from our computational framework.

Our contributions in this paper can be summarized as follows. \textbf{(1)} We propose a flexible family of sparse MTL estimators which borrow strength across tasks through \textit{separately} shrinking coefficient values towards each other, and/or encouraging tasks to have similar supports. Our estimator is given by a MIP. 
\textbf{(2)} We develop scalable approximate algorithms based on first-order optimization and local combinatorial search that provide high-quality solutions and allow for quickly fitting paths of solutions for tuning hyperparameters. \textbf{(3)} 
We also propose custom exact algorithms that can certify quality of solutions available from approximate algorithms and possibly improve them. Our exact algorithms are much more scalable than off-the-shelf commercial solvers.
\textbf{(4)} We establish statistical guarantees for our method that show our estimator generally leads to good prediction and variable selection performance while encouraging task-specific models to have similar supports. \textbf{(5)} We compare the performance of our methods with other sparse MTL approaches in two applications and on synthetic datasets. \textbf{(6)} An R package implementing our methods is available on $\texttt{CRAN}$.
\section{Methods}\label{sec:method}

\noindent {\bf Notation:} We discuss some notation that will be used in this paper. 
A single observation of the outcome and covariates are denoted as ${y}_{k,i} \in \mathbb{R}$ and $\mathbf{x}_{k,i} \in \mathbb{R}^p$, respectively. 
We denote the support of the regression coefficients as $\boldsymbol{z}_k = \1({\boldsymbol{\beta}_k \neq \boldsymbol{0}})$. We write the matrix of $\boldsymbol{\beta}_k$ as $\mathbb{B} \in \mathbb{R}^{p \times K}$. Similarly, we define the matrix of $\B{z}_k$ as $\B{Z}\in\{0,1\}^{p\times K}$. $\beta_{k,j}$ is the regression coefficient associated with task $k$ and covariate $j$ and $z_{k,j} = \1({\beta_{k,j} \neq 0})$. For a binary vector $\B{z}\in\{0,1\}^p$ we let $S(\B{z})$ denote the support of $\B{z}$, $S(\B{z})=\{j:z_j=1\}$. We also denote the smallest eigenvalue of a symmetric $\B{\Sigma}\in\R^{p\times p}$ as $\lambda_{\min}(\B{\Sigma})$. We use $\|\cdot\|_{\text{op}}$ to denote the operator norm of a matrix. For $S_1\subseteq[n_1],S_2\subseteq[n_2]$, we use $\B{A}_{S_1,S_2}$ to denote the submatrix of $\B{A}\in\R^{n_1\times n_2}$ with rows and columns indexed by $S_1$ and $S_2$, respectively. We use $\odot$ to indicate element-wise multiplication. We use the notations $\lesssim,\gtrsim$ to show the inequality holds up to a universal constant.

\subsection{Proposed Estimators}\label{section:proposed-estimator}
We now introduce our proposed estimator. First, we describe a special case of our general framework corresponding to the basic ``common support'' model.

In the common support case, we assume the $K$ task-specific models are parameterized by separate $\boldsymbol{\beta}_k$, but that the $\boldsymbol{\beta}_k$ share the same support across tasks (i.e., $\boldsymbol{z}_k = \boldsymbol{z} ~\forall~ k \in [K]$). In addition, we use a penalty that allows models to borrow strength across tasks when estimating the task-specific $\boldsymbol{\beta}_k$'s. The common support (CS) estimator can be written as:
\begin{align} \label{suppCommon}    
	\beforetext{(CS)}\underset{\boldsymbol{z}, \mathbb{B}, \bar{\boldsymbol{\beta}} }{\mbox{min }}~~& \sum_{k=1}^K \frac{1}{n_k} \norm{  \mathbf{y}_k - \mathbb{X}_k \boldsymbol{\beta}_k }_2^2 + \alpha \norm{\mathbb{B}}_F^2 + \lambda  \sum_{k=1}^K \norm{\boldsymbol{\beta}_k - \bar{{\boldsymbol{\beta}}} }_2^2 \\
	\mbox{s.t. }~~& z_{j} \in \{0,1 \} ~~ \forall ~j \in [p], \sum_{j=1}^p z_j\leq s,~~\notag \\
	&\beta_{k,j}(1 - z_j) = 0~ \forall ~j \in [p], k \in [K]. \notag 
\end{align}
In problem~\eqref{suppCommon}, binary variables $\{z_j\}_1^p$ encode the common support across $K$ tasks. The constraint $\sum_j z_j\leq s$ ensures that the common support is sparse (with at most $s$ nonzero entries). 
The optimization variables ${z}_j$'s do not appear in the objective, and enforce sparsity through the constraints: for covariate $j$, if $z_j = 0$, then $\beta_{k,j} = 0$ for all $k \in [K]$. At optimality $\bar{\boldsymbol{\beta}}=\frac{1}{K}\sum_{k=1}^K \boldsymbol{\beta}_k$ is the average of the regression coefficients. As a result, the penalty $\sum_{k=1}^K \norm{\boldsymbol{\beta}_k - \bar{{\boldsymbol{\beta}}} }_2^2$ encourages regression coefficients to share strength across tasks. This penalty has precedents in \cite{bbar_cite} \cite{rashid2020modeling} and \cite{trippa}. Estimator \eqref{suppCommon} is a generalization of the group $\ell_0$ penalized methods (e.g.,  see~\cite{groupedL0} and references therein). Removing the constraint that task models share the same support, we get the support heterogeneous (HET) estimator:
\begin{align} \label{suppHet}
\beforetext{(HET)}	\underset{\B{Z},\mathbb{B},\bar{\boldsymbol{\beta}},\bar{\boldsymbol{z}} }{\mbox{min }}~& \sum_{k=1}^K \frac{1}{n_k} \norm{\mathbf{y}_k - \mathbb{X}_k \boldsymbol{\beta}_k }_2^2 + \alpha \norm{\mathbb{B}}_F^2  + \lambda  \sum_{k=1}^K \norm{\boldsymbol{\beta}_k - \bar{{\boldsymbol{\beta}}} }_2^2  + \delta \sum_{k=1}^K\norm{\boldsymbol{z}_k - \bar{\boldsymbol{z}}}_2^2\\
	~~~~ \mbox{s.t.}~& z_{k,j} \in \{0,1 \},~~\sum_{j=1}^p z_{k,j}\leq s~\forall k\in[K],\nonumber\\
	&      \beta_{k,j}(1 - z_{k,j}) = 0~ \forall ~j \in [p], k \in [K]. \notag
\end{align}
Similar to $\bar{\boldsymbol{\beta}}$ in the common support case, at optimality, $\bar{\boldsymbol{z}}$ is the average of the $\boldsymbol{z}_k$'s. As a result, 
the penalty $\sum_{k=1}^K\norm{\boldsymbol{z}_k - \bar{\boldsymbol{z}}}_2^2$ encourages the supports of different tasks to be similar, without forcing them to be identical. Note that the common support problem can be obtained as special case of problem~\eqref{suppHet} when $\delta \to \infty$ (in terms of algorithms however, we use a specialized algorithm for the limiting case, i.e., problem~\eqref{suppCommon}).
In practice (cf Section~\ref{data_apps}), we include intercepts, $\beta_{k,0}$, for problems~\eqref{suppCommon} and~\eqref{suppHet}, that are neither penalized nor subject to sparsity constraints, but we omit this for notational conciseness.

The penalty term, $\sum_{k=1}^K \norm{\boldsymbol{z}_k - \bar{\boldsymbol{z}}}_2^2$ denoted as \textit{Zbar},
encourages \textit{support heterogeneity regularization}. Formulations~\eqref{suppCommon} and~\eqref{suppHet} also include a Ridge penalty. In practice, we found that a small Ridge penalty, (e.g., $\alpha < 10^{-6}$), was useful for some of the estimators discussed above, but the solutions were not sensitive to the exact hyperparameter value, provided the penalty was sufficiently small (see Section~\ref{cv} for hyperparameter tuning details). We either set $\lambda >0$ or $\alpha>0$, but never both.
Problem~\eqref{suppHet} is our general estimator and includes as special cases important estimators that use subsets of the three penalties (i.e., we set at least one of $\{\lambda, \alpha, \delta\}$ to zero). 
Table \ref{table:losses_table} presents a summary of the methods considered based upon two criteria: 1) whether the supports are allowed to vary across the tasks' regression coefficients, and 2) which model parameter penalties are included. By combining different penalties, we define a larger set of methods that allow us to understand which properties of our estimators are associated with improvements in prediction and support recovery accuracy.  
If the term ``CS'' (common support) is not used in a method name, this indicates that the supports of the regression coefficients are free to vary across tasks. 

\begin{table*}[!h]
	\scriptsize
	\centering
	\begin{tabular}{c | ccccc}
		\hline
		\hline
		\multicolumn{1}{c}{\bfseries Method} \vline & \multicolumn{1}{c}{\bfseries Support}  &  \multicolumn{1}{c}{\bfseries MTL Squared Error Loss } 
		&  \multicolumn{1}{c}{\bfseries L2 } 
		&  \multicolumn{1}{c}{\bfseries Bbar } 
		&  \multicolumn{1}{c}{\bfseries Zbar }\\
		\midrule
		\midrule
		$L_0 L_2$  & HET & $\sum_{k=1}^K \frac{1}{n_k} \norm{  \mathbf{y}_k - \mathbb{X}_k \boldsymbol{\beta}_k }_2^2$ & $+ \alpha \norm{\mathbb{B}}_F^2 $ & & \\
		Bbar  & HET & $\sum_{k=1}^K \frac{1}{n_k} \norm{  \mathbf{y}_k - \mathbb{X}_k \boldsymbol{\beta}_k }_2^2$ & & $+ \lambda \sum_{k=1}^K \norm{\boldsymbol{\beta}_k - \bar{{\boldsymbol{\beta}}} }_2^2 $ &
		\\ 
		Zbar+L2  &  HET & $\sum_{k=1}^K \frac{1}{n_k} \norm{  \mathbf{y}_k - \mathbb{X}_k \boldsymbol{\beta}_k }_2^2$ & $+ \alpha \norm{\mathbb{B}}_F^2 $ & &  $+ \delta \norm{\boldsymbol{z}_k - \bar{\boldsymbol{z}}}_2^2 $ \\
		Zbar+Bbar  &  HET & $\sum_{k=1}^K \frac{1}{n_k} \norm{  \mathbf{y}_k - \mathbb{X}_k \boldsymbol{\beta}_k }_2^2$ & &  $+ \lambda \sum_{k=1}^K \norm{\boldsymbol{\beta}_k - \bar{{\boldsymbol{\beta}}} }_2^2$ & $+ \delta \norm{\boldsymbol{z}_k - \bar{\boldsymbol{z}}}_2^2  $\\
		CS+L2  & CS & $\sum_{k=1}^K \frac{1}{n_k} \norm{  \mathbf{y}_k - \mathbb{X}_k \boldsymbol{\beta}_k }_2^2$ & $+  \alpha \norm{\mathbb{B}}_F^2$ & & \\
		CS+Bbar  &  CS & $\sum_{k=1}^K \frac{1}{n_k} \norm{  \mathbf{y}_k - \mathbb{X}_k \boldsymbol{\beta}_k }_2^2$ & & $+ \lambda \sum_{k=1}^K \norm{\boldsymbol{\beta}_k - \bar{{\boldsymbol{\beta}}} }_2^2 $ & \\
		\hline
	\end{tabular}
	\caption{\footnotesize Method names for the losses with different combinations of support constraints and penalties. Heterogeneous (HET) and common (CS) support estimators are special cases of problems~\eqref{suppHet}, and~\eqref{suppCommon}, respectively.}
	\label{table:losses_table}
\end{table*}

Formulation~\eqref{suppHet} includes a number of special cases that serve as benchmarks in our numerical experiments. A collection of independent $\ell_0$ constrained regressions with a Ridge penalty arises when $\lambda = \delta = 0$. This provides a gauge for the performance of sparse regressions without any information shared across tasks. We refer to these task-specific sparse regressions as ``$L_0 L_2$''. The ``Bbar'' penalty provides a way of sharing information on the regression coefficient supports through the $\boldsymbol{\beta}_k$ values. This is because for each covariate $j$, the Bbar \textit{penalty} shrinks all $K$ task-specific coefficients in the vector $\boldsymbol{\beta}^{(j)}$ together, even for $\beta_{k,j} = 0$. However, unlike the Group Lasso, which also shares information through the $\boldsymbol{\beta}_k$'s, the Bbar penalty on its own may not result in a common support. 

Problems~\eqref{suppCommon} and~\eqref{suppHet} are MIPs that can be solved to optimality (or with certificates of optimality) for small to moderate instances with commercial solvers such as Gurobi (see Section~\ref{optSec} for more details on MIP formulations and solvers). To address large-scale instances occurring in biomedical applications, it is useful to have algorithms that can deliver good solutions quickly. This is useful for wider adoption since models must be solved many times during, for example, hyperparameter tuning. Thus, in addition to developing custom MIP approaches, we also propose 
a scalable framework
for obtaining high-quality solutions to our estimators based on first-order optimization \citep{beck2009fast}, and local combinatorial search methods extending the work of~\citet{l0learn}.

\begin{remark}\label{remarkconvexrelax}
	We can obtain a convex relaxation of the Zbar penalty by relaxing the binary variables to their interval constraints $z_{k,j}\in[0,1]$ for all $k,j$ --- this does not 
	result in the Bbar penalty. We explore this formally in Supplement~\ref{app:convexrelax}. This  shows that the Bbar and Zbar penalties are fundamentally different. To our knowledge, the convex relaxation of Zbar has not been studied in prior work, further motivating our study of the Zbar penalty.
\end{remark}

\section{Statistical Theory}\label{sec:theory}
In our numerical experiments (see Section \ref{sec:simulations}) we observed that the Zbar methods tend to outperform other related methods in prediction performance and variable selection accuracy. Since estimators of this nature, to our knowledge, have not been proposed before, we explored the statistical properties associated with our Zbar (and Bbar) methods. 

\noindent \textbf{Model Setup:} We assume for every task $k\in[K]$, each row of $\bX_k$ is drawn independently as $\mathbf{x}_{k,i} \sim \cN_p(\B{0},\B{\Sigma}^{(k)})$ where $\B{\Sigma}^{(k)}\in\R^{p\times p}$ is a positive definite matrix. Although our framework is based on linear models, we do not assume the underlying model is sparse or linear allowing for model misspecification. In particular, we assume the observations are $\by_k = \by_k^*(\bX_k) + \B{\epsilon}_k$ where $\by_k^*\in\R^n$ are noiseless observations and the noise vector $\B{\epsilon}_k\in\R^{n_k}$ follows $\B{\epsilon}_k\sim \cN_n(\B{0},\sigma_k^2\B{I})$ and is independent of $\bX_k$. We define the best $s$-sparse linear approximation to $\by_k^*$ as 
\begin{equation}
	\B{\beta}^*_k\in \argmin_{\B{\beta}\in\R^p} \|\by_k^*-\bX_k\B{\beta}\|_2~~\text{s.t.}~~ \|\B{\beta}\|_0\leq s.
\end{equation}
We assume the oracle regression coefficients $\B{\beta}_k^*\in\R^p$ are $s$-sparse, that is $\|\B{\beta}^*_k\|_0=s$. We denote the support of $\B{\beta}_k$ with the binary vector $\B{z}^*_k$. We also define the error resulting from estimating $\by_k^*$ by the oracle as
$\B{r}_k=\by_k^*-\bX_k\B{\beta}^*_k$.
We introduce some additional notation.
\begin{definition}\label{sallscommdef}
	Let $\B{z}_1,\cdots,\B{z}_K\in\{0,1\}^p$. We define 
	\begin{equation}
		\begin{aligned}
			S_{\text{all}}(\B{Z}) &= \{j\in[p]: \sum_{k=1}^K z_{k,j}\geq 1\},  \\
			S_{\text{common}}(\B{Z})& = \{j\in[p]: \sum_{k=1}^K z_{k,j} = K\}.
		\end{aligned}
	\end{equation}
\end{definition}
The set $S_{\text{common}}$ is the set of coordinates which are common amongst the supports of all tasks, while the set $S_{\text{all}}$ is the set of coordinates that appear in the support of at least one task. The set-difference operation $S_{\text{all}}\setminus S_{\text{common}}$ denotes the coordinates which appear in the support of some tasks, but are not common. Consequently, a solution for which the size of the set $S_{\text{all}}\setminus S_{\text{common}}$ is small, includes more common covariates with nonzero associated regression coefficients across tasks and can thus be more interpretable. We use the notation $\bar{\B{z}}=\frac{1}{K}\sum_{k=1}^K \B{z}_k$ and $\bar{\B{\beta}}=\frac{1}{K}\sum_{k=1}^K \B{\beta}_k$.

\begin{remark}\label{multislrremark}
	Although we are assuming each row of $\bX_k$ follows a multivariate normal distribution, we do not make any assumption on the joint distribution of $\{\bX_k\}_k$ (across the tasks).
	This differs from linear models with 
	multivariate responses, which often specify a joint distribution across the different response variables~\citep{multislr1,multislr2,multislr3}. Additionally, we do not assume any dependence structure on $\B{\epsilon}_k$ across the tasks.
\end{remark}

\subsection{A General Prediction Bound}\label{sec:theory-1}
We provide a general prediction bound for the support heterogeneous case---the estimator from problem~\eqref{suppHet}. 
\begin{theorem}\label{thm1}
	Suppose $\{\hat{\B{\beta}}_k,\hat{\B{z}}_k\}_{k=1}^K$ is an optimal solution to Problem~\eqref{suppHet}. Then, under our assumed model setup with high probability\footnote{\label{foot1}An explicit expression for the probability can be found in~\eqref{thm1-prob}.} we have
	\begin{multline}\label{thm1-inequlity}
		\sum_{k=1}^K \frac{1}{n_k} \|\bX_k(\B{\beta}_k^*-\hat{\B{\beta}}_k)\|_2^2 + {\alpha\|\hat{\mathbb{B}}\|_F^2} +
		\sum_{k=1}^K \left\{ \delta  \|\hat{\B{z}}_k-\bar{\hat{\B{z}}}\|_2^2+ \lambda  \|\hat{\B{\beta}}_k-\bar{\hat{\B{\beta}}}\|_2^2\right\} \lesssim\\ \sum_{k=1}^K\left\{ \frac{\sigma_k^2 s \log(p/s)}{n_k}+\frac{1}{n_k}\|\B{r}_k\|_2^2\right\}+
		\sum_{k=1}^K\left\{\delta  \|{\B{z}}^*_k-\bar{\B{z}^*}\|_2^2+\lambda  \|{\B{\beta}}^*_k-\bar{\B{\beta}^*}\|_2^2\right\}+\alpha\|\bB^*\|_F^2.
	\end{multline}
\end{theorem}
Theorem~\ref{thm1} presents a bound on the prediction error of the estimator from problem~\eqref{suppHet} over all tasks, captured by  $ \sum_{k=1}^K  \|\bX_k(\B{\beta}_k^*-\hat{\B{\beta}}_k)\|_2^2/n_k$ in addition to the penalty terms 
$\sum_{k=1}^K \|\hat{\B{z}}_k-\bar{\hat{\B{z}}}\|_2^2$ which captures the support heterogeneity of the solution, and $\sum_{k=1}^K \|\hat{\B{\beta}}_k-\bar{\hat{\B{\beta}}}\|_2^2$ which captures the coefficient value heterogeneity of the solution. By bounding the support heterogeneity of the solution with the support heterogeneity of the oracle, this theorem 
quantifies how the Zbar penalty borrows information across the supports of the tasks. For example, it shows that if the oracle has similar or common support, the penalty term $\sum_{k=1}^K \|\hat{\B{z}}_k-\bar{\hat{\B{z}}}\|_2^2$ can be small (zero), which forces the supports of the different tasks to be similar (the same). Similarly, Theorem~\ref{thm1} shows that the Bbar penalty can shrink the regression coefficients of different tasks together. Finally, we see that if $\alpha>0$, Theorem~\ref{thm1} provides an upper bound on the squared norm $\|\hat{\mathbb{B}}\|_F^2$, which quantifies the effect of the ridge penalty term in Problem~\eqref{suppHet}. In practice, we use a small value of $\alpha$. To provide further insight into Theorem~\ref{thm1}, we next present results for a particular choice of $\alpha,\delta,\lambda$ in problem~\eqref{suppHet}.
\begin{corollary}\label{cor-general}
	Under the assumptions of Theorem~\ref{thm1}, suppose $\delta,\lambda,\alpha$ are taken such that 
	\begin{equation}\label{cor1-paramchoice}
		\begin{aligned}
			0\leq\delta & \lesssim  \frac{\sigma_k^2 s\log(p/s)}{n_k\|\B{z}_k^*-\bar{\B z^*}\|_2^2}~~\forall k\in[K] \\
			0\leq  \lambda & \lesssim  \frac{\sigma_k^2 s\log(p/s)}{n_k\|\B{\beta}_k^*-\bar{\B \beta^*}\|_2^2}~~\forall k\in[K]\\
			0\leq   \alpha & \lesssim \frac{1}{\|\bB^*\|_F^2}\sum_{k=1}^K\frac{\sigma_k^2s\log(p/s)}{n_k}.
		\end{aligned}
	\end{equation}
	Then, w.h.p.\footref{foot1} we have the following error bound
	\begin{equation*}
		\sum_{k=1}^K \frac{1}{n_k} \|\bX_k(\B{\beta}_k^*-\hat{\B{\beta}}_k)\|_2^2\lesssim  \sum_{k=1}^K\left\{ \frac{\sigma_k^2 s \log(p/s)}{n_k}+\frac{1}{n_k}\|\B{r}_k\|_2^2\right\}.
\end{equation*}
\end{corollary}
\begin{remark}\label{new-remark-1}
	As seen in Corollary~\ref{cor-general}, when the regularization parameters $\lambda,\delta,\alpha$ are chosen as in~\eqref{cor1-paramchoice} and the oracle error $\sum_{k=1}^K \|\B{r}_k\|_2^2/n_k$ is comparatively small, the estimator from problem~\eqref{suppHet} is able to achieve a prediction error rate $\sum_{k=1}^K {(\sigma_k^2 s/n_k)} \log(p/s)$ even if the underlying model does not have common support. Note that a sparse linear regression estimator on study $k$ results in prediction error of order ${(\sigma_k^2 s/n_k)} \log(p/s)$ (for example, see~\citet[Ch. 7]{wainwright2019high})---hence, the overall prediction error rate for $K$ separate sparse linear regressions is $\sum_{k=1}^K {(\sigma_k^2 s/n_k)} \log(p/s)$, which is the same as that from problem~\eqref{suppHet}. However, Theorem~\ref{thm1} additionally provides an upper bound on the support and coefficient heterogeneity of the solutions to problem~\eqref{suppHet}, given by $ \sum_{k=1}^K \left\{ \delta  \|\hat{\B{z}}_k-\bar{\hat{\B{z}}}\|_2^2+ \lambda  \|\hat{\B{\beta}}_k-\bar{\hat{\B{\beta}}}\|_2^2\right\}$ on the lhs of~\eqref{thm1-inequlity}. Such additional guarantees might not be available from learning $K$ separate sparse regression models.

\end{remark}

Next,  we consider two special cases that provide further insight into Theorem~\ref{thm1}.
\begin{corollary}\label{cor1}
	Suppose the underlying oracle model follows the common support model, that is $\B{z}^*_1=\cdots=\B{z}^*_K$. Let $\{\hat{\B{\beta}}_k,\hat{\B{z}}_k\}_{k=1}^K$ be the optimal solution to problem~\eqref{suppHet} with $\lambda$ and $\alpha$ chosen as in~\eqref{cor1-paramchoice}. Then w.h.p.\footref{foot1} 
	$$ \sum_{k=1}^K \frac{1}{n_k} \|\bX_k(\B{\beta}_k^*-\hat{\B{\beta}}_k)\|_2^2 \lesssim \sum_{k=1}^K\left\{ \frac{\sigma_k^2 s \log(p/s)}{n_k}+\frac{1}{n_k}\|\B{r}_k\|_2^2\right\}.$$
	Moreover, if $\delta\gtrsim K\sum_{k=1}^K {([\sigma_k^2 s\log(p/s)+\|\B{r}_k\|_2^2]/n_k)} $ is sufficiently large, we have $\hat{\B{z}}_1=\cdots=\hat{\B{z}}_K$.
\end{corollary}
As seen in Corollary~\ref{cor1}, when the underlying oracle has a common support, problem~\eqref{suppHet} (with sufficiently small $\lambda,\alpha$) is able to achieve a prediction error rate $\sum_{k=1}^K {(\sigma_k^2 s/n_k)} \log(p/s)$ regardless of the value of $\delta$. As we discussed in Remark~\ref{new-remark-1}, this
is the same rate we would achieve if we
were to fit $K$ separate sparse linear regressions. However, Corollary~\ref{cor1} additionally provides an upper bound on the support heterogeneity of the solutions to problem~\eqref{suppHet}. This even guarantees that the estimator yields solutions with common support if $\delta$ is sufficiently large.

In many MTL settings, analysts might acknowledge that the underlying model of the tasks have supports that are close but not exactly identical. If one were to use a common support estimator, theoretical guarantees on the performance of the estimator may not be satisfactory (i.e., the oracle error term in Corollary~\ref{cor1} becomes too large) if the underlying model does not truly have a common support. However, our framework, and in particular the Zbar penalty, allows for different sparsity patterns while still borrowing strength across tasks during variable selection. In what follows, we show that this behavior of our estimator leads to solutions that enjoy prediction performance guarantees, and controls the support heterogeneity of the solution.

To better understand and quantify the performance of our method in the support heterogeneity case, we consider a model where the supports of tasks are mostly common but not necessarily identical. In particular, we assume that if a feature is not common to all tasks, it is unique to a single task. We call such models \textit{regular}.

\begin{definition}[Regular Support]\label{def-reg}
	We call $\B{z}_1,\cdots,\B{z}_K\in\{0,1\}^p$ regular if for each $j
	\in[p]$, we have one of the following: 1) $\sum_{k=1}^K z_{k,j}=0$, 2) $\sum_{k=1}^K z_{k,j} = K$, or 3) $\sum_{k=1}^K z_{k,j}=1$.
\end{definition}

\begin{corollary}\label{cor2}
	Suppose $\B{z}_1^*,\cdots,\B{z}_K^*$ is regular and let $\{\hat{\B{\beta}}_k,\hat{\B{z}}_k\}_{k=1}^K$ be the optimal solution to problem~\eqref{suppHet}. Then, for a suitably chosen value\footnote{See~\eqref{delta-cor2} for an expression for $\delta$.} of $\delta$, with $\lambda$ and $\alpha$ chosen as in~\eqref{cor1-paramchoice}, we have w.h.p\footref{foot1}
	\begin{equation*}
		\sum_{k=1}^K \frac{1}{n_k} \|\bX_k(\B{\beta}_k^*-\hat{\B{\beta}}_k)\|_2^2 \lesssim  \sum_{k=1}^K \left\{\frac{\sigma_k^2 s \log(p/s)}{n_k}+\frac{1}{n_k}\|\B{r}_k\|_2^2\right\}
	\end{equation*}
	and
	\begin{equation*}
		|\hat{S}_{\text{all}}\setminus \hat{S}_{\text{common}}|\lesssim (K-1)|S^*_{\text{all}}\setminus S^*_{\text{common}}|
	\end{equation*}
	where $\hat{S}_{\text{all}}=S_{\text{all}}(\hat{\B{Z}}),  ~\hat{S}_{\text{common}}= S_{\text{common}}(\hat{\B{Z}})$, \\ 
	${S}^*_{\text{all}}=S_{\text{all}}({\B{Z}}^*),~  {S}^*_{\text{common}}= S_{\text{common}}({\B{Z}}^*)$ with $S_{\text{all}},~S_{\text{common}}$ defined in Definition~\ref{sallscommdef}.
\end{corollary}
Corollary~\ref{cor2} states that, under the regular (but not necessarily common) support assumption, choosing a suitable value of $\delta$ leads to an optimal prediction error bound, similar to the results in Corollary~\ref{cor1}. Moreover, Corollary~\ref{cor2} provides the guarantee that the degree of support heterogeneity of the solution is bounded above by the support heterogeneity of the true model (up to a multiplicative factor of $K$). 
As discussed above, such guarantees are not available for separate sparse linear regressions without additional assumptions.

\subsection{A Deeper Investigation of Support Recovery} \label{sec:theory-2}
In this section, we further investigate guarantees on support recovery of our methods under the regularity assumption (cf Definition~\ref{def-reg}). For simplicity, we consider the well-specified case where the underlying model is linear and sparse (i.e., $\B{r}_1=\cdots=\B{r}_K=\B{0}$). We study the Zbar penalty (i.e. $\lambda=\alpha=0$) as similar regularizers do not appear to have been explored earlier, whereas approaches similar to the Bbar penalty have been explored empirically in prior work~\citep{trippa, bbar_cite}.

\begin{asu*}
	We assume the following:
	\begin{compactenum}
		\item \label{as-rip}\label{as-condition}~ For $k\in[K]$, we have $$
		\phi_k := \min_{\substack{S\subseteq [p] \\ |S|\leq 2s  }} \lambda_{\min}\left({\B{\Sigma}}^{(k)}_{S,S}\right)>0$$ and  $\|\B{\Sigma}^{(k)}\|_{\text{op}}\leq 1$.
		\item \label{cJdef}~ We have $|\mathcal{J}|\geq 1$ where $$\mathcal{J}=\left\{k\in[K]: n_k\geq c s\log p/\phi_k^2\right\}$$ for some universal constant $c>0$ that is sufficiently large.
		\item \label{as-beta-min}~ For $k\in\mathcal{J}$ and $j\in[p]$, every nonzero coefficient $\beta^*_{k,j}\neq 0$, is bounded away from zero $|\beta^*_{k,j}|\geq\beta_{\min,k}$ where
		$$ \beta_{\min,k} = \sqrt{\frac{\eta_k \log p}{\phi_k n_k}}$$
		for some sufficiently large $\eta_k>0$.
		\item \label{as-cn}~ There exists an absolute constant $c_n\leq 1$ such that ${\bar{n}}/{\barbelow{n}}= {1}/{c_n}$
		where $\bar{n}=\max_{k\in[K]} n_k$ and $\barbelow{n}=\min_{k\in[K]} n_k$.
	\end{compactenum}
\end{asu*}

Assumption~\ref{as-rip} states that $\B{\Sigma}^{(k)}$ is well-conditioned and does not have very small or large eigenvalues. Assumption~\ref{cJdef} ensures that there exists tasks that have sufficiently large numbers of samples. Assumption~\ref{as-beta-min} is a minimum signal requirement for model identifiability. These three are common assumptions in the sparse linear regression literature. Finally, Assumption~\ref{as-cn} ensures that the number of samples does not differ too much across tasks.

\begin{theorem}\label{supthm}
	Suppose Assumptions~\ref{as-rip} to~\ref{as-cn} hold and the underlying model is regular as in Definition~\ref{def-reg}. Let $\{\hat{\B{\beta}}_k,\hat{\B{z}}_k\}_{k=1}^K$ be the optimal solution to problem~\eqref{suppHet}. Then, w.h.p.\footnote{An explicit expression for the probability can be found in~\eqref{supthmprob}\label{foot2}.} we have
	\begin{equation}\label{supthm-result}
		|S^*_{\text{all}}\setminus S^*_{\text{common}}| \geq  \frac{|\hat{S}_{\text{all}}\setminus \hat{S}_{\text{common}}|}{K}+\frac{\log p}{\delta\barbelow{n}}\sum_{k=1}^K|\tilde{S}_k|\left[0.4c_n\eta_k \1(k\in\cJ)-c_t\sigma_k^2\right]
	\end{equation}
	where $\hat{S}_{\text{all}}=S_{\text{all}}(\hat{\B{Z}}),~  \hat{S}_{\text{common}}= S_{\text{common}}(\hat{\B{Z}})$, \\ 
	${S}^*_{\text{all}}=S_{\text{all}}({\B{Z}}^*),~  {S}^*_{\text{common}}= S_{\text{common}}({\B{Z}}^*)$ with $S_{\text{all}},~S_{\text{common}}$ defined in Definition~\ref{sallscommdef}, $\tilde{S}_k=\{j:z_{k,j}^*=1,\hat{z}_{k,j}=0\}$ is the set of mistakes in the support of task $k$, and $c_t$ is an absolute constant.
\end{theorem}

Theorem~\ref{supthm} presents support recovery guarantees for problem~\eqref{suppHet}. Particularly, Theorem~\ref{supthm} limits the total number of mistakes in the support of the estimates (given by $|\tilde{S}_k|$ on the rhs in~\eqref{supthm-result}). 
Moreover, Theorem~\ref{supthm} presents a bound on support heterogeneity, i.e., 
the number of nonzero indices outside the common support $|\hat{S}_{\text{all}}\setminus\hat{S}_{\text{common}}|$. 
Similar to the previous section, we next consider some special cases that help to interpret the result of Theorem~\ref{supthm}. 

\begin{corollary}\label{cor3}
	Suppose the underlying model follows the common support model, that is $\B{z}^*_1=\cdots=\B{z}^*_K$. Let $\{\hat{\B{\beta}}_k,\hat{\B{z}}_k\}_{k=1}^K$ be an optimal solution to problem~\eqref{suppHet} with $\delta$ chosen as in Corollary~\ref{cor1}. If 
	\begin{equation}\label{supreccond}
		\sum_{k\in\mathcal{J}}\eta_k\gtrsim \sum_{k=1}^K \sigma_k^2
	\end{equation}
	where $\mathcal{J}$ is defined in Assumption~\ref{cJdef}, then under the assumptions of Theorem~\ref{supthm} we have
	$\hat{\B{z}}_1=\cdots=\hat{\B{z}}_K=\B{z}^*_1=\cdots=\B{z}^*_K$
	with high probability.\footref{foot2}
\end{corollary}
Based on Corollary~\ref{cor3}, when the true models of the $K$ tasks have a common support, then the estimator from Problem~\eqref{suppHet} can recover the support of the underlying model correctly under condition~\eqref{supreccond}. Condition~\eqref{supreccond} states that the average signal level in tasks with a sufficiently large number of samples (i.e., tasks in $\mathcal{J}$) is large compared to the noise. In fact, there may be some tasks that do not have large enough sample sizes, or some tasks that do not have signals that are high enough for support recovery, but as long as the average signal is high enough among tasks in $\mathcal{J}$, problem~\eqref{suppHet} recovers the support of every task correctly. 

\begin{remark}
	Under the setup we have considered here, $n_k \gtrsim s\log p$ samples are required to recover the support of task $k$ correctly if we were to fit a separate sparse regression on each task~\citep{minimaxlinear}. Although Corollary~\ref{cor3} requires some tasks to have at least $n_k \gtrsim s\log p$ samples, some tasks can have fewer samples when we use our methods. Moreover, some $\eta_k$'s can be small enough such that a sparse linear model fit to that task alone would be unable to recover the support, but problem~\eqref{suppHet} can use the information from every task to estimate the common support correctly. This shows that sharing information across tasks under a common support model can improve support recovery.
\end{remark}
Next we consider the regular case.
\begin{corollary}\label{cor4}
	Let $\{\hat{\B{\beta}}_k,\hat{\B{z}}_k\}_{k=1}^K$ be the optimal solution to Problem~\eqref{suppHet} with a suitably chosen value\footnote{See~\eqref{cor4-delta} for an expression for $\delta$.} of $\delta$.
	If 
	$\sum_{k\in\mathcal{J}}\eta_k\gtrsim s\sum_{k=1}^K \sigma_k^2$
	where $\mathcal{J}$ is defined in Assumption~\ref{cJdef}, then under the assumptions of Theorem~\ref{supthm} w.h.p,\footref{foot2}
	\begin{equation*}
		\begin{aligned}
			S^*_{\text{common}}&\subseteq \bigcup_{k\in\mathcal{J}}\{j:\hat{z}_{k,j}=1\} \\
			|\hat{S}_{\text{all}}\setminus \hat{S}_{\text{common}}| &\leq |{S}^*_{\text{all}}\setminus {S}^*_{\text{common}}|.
		\end{aligned}
	\end{equation*}
\end{corollary}

Based on Corollary~\ref{cor4}, every feature that appears in 
the support common to the true regression coefficients across tasks also appears in the estimated support of some task. 
In other words, our proposed estimator from Problem~\eqref{suppHet}
is able to identify nonzero coefficients that are common across all tasks. Moreover, Corollary~\ref{cor4} limits the number of nonzero indices outside the common support. Such joint guarantees on recovery and support heterogeneity might not be available if we were to use independent sparse linear regressions for different tasks.

\subsection{Statistical Properties of Approximate Solutions} \label{stats_approx}
An appealing aspect of our MIP-based global optimization framework is its ability to
deliver a certificate of how close the current solution is to the optimum. In this section we show the statistical properties of approximate solutions to MIP in problem~\eqref{suppHet}. Loosely speaking, if the objective value of an approximate solution is sufficiently close to the optimal objective, the statistical properties of the approximate solution are quite similar to that of the optimal ones.  These guarantees can inform the statistical properties of a solution obtained from a MIP solver terminated early due to computation budget constraints or an approximate solution available from heuristics.

Formally, let $(\check{\mathbb{B}},\check{\B{Z}})$ be a feasible solution\footnote{This can be obtained, for example, from a MIP solver such as the one from Section~\ref{app:mip-solver} or an approximate solution.} to Problem~\eqref{suppHet}. For simplicity, let us study the case with $\alpha=\lambda=0$. Our MIP framework  
returns a dual (aka lower) bound on the optimal value as 
\begin{equation}\label{mip-lowerbound}
	\sum_{k=1}^K \frac{1}{n_k} \norm{  \mathbf{y}_k - \mathbb{X}_k \check{\boldsymbol{\beta}}_k }_2^2  + \delta \sum_{k=1}^K\norm{\check{\boldsymbol{z}}_k - \bar{\check{\boldsymbol{z}}}}_2^2\leq \sum_{k=1}^K \frac{1}{n_k} \norm{  \mathbf{y}_k - \mathbb{X}_k \hat{\boldsymbol{\beta}}_k }_2^2  + \delta \sum_{k=1}^K\norm{\hat{\boldsymbol{z}}_k - \bar{\hat{\boldsymbol{z}}}}_2^2+\tau
\end{equation}
where $(\hat{\mathbb{B}},\hat{\B{Z}})$ is an optimal solution to Problem~\eqref{suppHet} and $\tau\geq 0$ is the optimality gap. Below, we discuss the statistical properties of the approximate solution $(\check{\mathbb{B}},\check{\B{Z}})$. Indeed, we prove a more general version of Proposition~\ref{app-prop1} in Supplement~\ref{app:prop1_proof}.\\

\begin{proposition}\label{app-prop1}
	Suppose $\{\check{\B{\beta}}_k,\check{\B{z}}_k\}_{k=1}^K$ is as defined in~\eqref{mip-lowerbound}. Then, under our assumed model setup w.h.p.\footref{foot1} we have
	\begin{equation*}
		\sum_{k=1}^K \left\{\frac{1}{n_k} \|\bX_k(\B{\beta}_k^*-\check{\B{\beta}}_k)\|_2^2+ \delta  \|\check{\B{z}}_k-\bar{\check{\B{z}}}\|_2^2\right\} \lesssim  \sum_{k=1}^K\left\{ \frac{\sigma_k^2 s \log(p/s)}{n_k}+ \delta  \|{\B{z}}^*_k-\bar{\B{z}^*}\|_2^2+\frac{1}{n_k}\|\B{r}_k\|_2^2\right\}+\tau.
	\end{equation*}
\end{proposition}
Proposition~\ref{app-prop1} shows that as long as $\tau\lesssim \sum_{k=1}^K {(\sigma_k^2 s)\log(p/s)/n_k }$, the prediction error rate from the optimal solution $(\hat{\mathbb{B}},\hat{\B Z})$ and the approximate solution $(\check{\mathbb{B}},\check{\B Z})$ are the same. Hence, such approximate solutions can still have strong statistical guarantees. 

We next study the variable selection properties of approximate solutions. 
\begin{proposition}\label{app-prop2}
	Under the assumptions and notation of Theorem~\ref{supthm}, let $\{\check{\B{\beta}}_k,\check{\B{z}}_k\}_{k=1}^K$ be as defined in~\eqref{mip-lowerbound}. Then, w.h.p.\footref{foot2} we have
	\begin{equation}
		|S^*_{\text{all}}\setminus S^*_{\text{common}}| \geq  \frac{|\check{S}_{\text{all}}\setminus \check{S}_{\text{common}}|}{K}+\frac{\log p}{\delta\barbelow{n}}\sum_{k=1}^K|\bar{S}_k|\left[0.4c_n\eta_k \1(k\in\cJ)-c_t\sigma_k^2\right] - \frac{\tau}{\delta}
	\end{equation}
	where $\check{S}_{\text{all}}=S_{\text{all}}(\check{\B{Z}}),~ \check{S}_{\text{common}}= S_{\text{common}}(\check{\B{Z}})$, and \\ $\bar{S}_k=\{j:z_{k,j}^*=1,\check{z}_{k,j}=0\}$ is the set of mistakes in the support of task $k$.
\end{proposition}
From Proposition~\ref{app-prop2} we see that if $\tau < \delta/K$, the quantity $|\check{S}_{\text{all}}\setminus \check{S}_{\text{common}}|$ can increase at most by one, compared to the case where $\tau=0$ (see Theorem~\ref{supthm} and related discussions for corresponding properties of an optimal solution).

\section{Exact Optimization Algorithms}\label{optSec}

Here we discuss exact optimization approaches including our custom algorithms designed to compute our estimators introduced in Section~\ref{sec:method}. Such exact algorithms can be used to obtain globally optimal solutions to our estimator for moderate problem sizes\footnote{We can generally handle problems where number of features $p$ is in thousands. Please refer to Section~\ref{sec:scalability} for numerical demonstrations and specific problem sizes.}.

\subsection{Mixed Integer Program (MIP) Formulations}\label{mip-app}
We first present MIP formulations for problems \eqref{suppCommon} and \eqref{suppHet} that can be solved to optimality with off-the-shelf solvers such as Gurobi, Mosek, etc.  
Problems~\eqref{suppCommon_mip} and~\eqref{suppHet_mip} below present equivalent MIP reformulations for  problems~\eqref{suppCommon} and~\eqref{suppHet}, respectively.

\begin{align} \label{suppCommon_mip}    
	\min_{\boldsymbol{z},\mathbb{B},\bar{\boldsymbol{\beta}} }&~~ \sum_{k=1}^K \frac{1}{n_k} \norm{  \mathbf{y}_k - \mathbb{X}_k \boldsymbol{\beta}_k }_2^2 + \lambda  \sum_{k=1}^K \norm{\boldsymbol{\beta}_k - \bar{{\boldsymbol{\beta}}} }_2^2+\alpha\|\bB\|_F^2 \\
	\text{s.t.}&~~ z_{j} \in \{0,1 \} ~~ \forall ~j \in [p]  \notag \\
	& ~~ -M z_{j} \leq \beta_{k,j} \leq M z_j ~~\forall~ k \in [K], ~ \forall ~j \in [p]  \notag \\
	& ~~ \sum_{j=1}^p z_j\leq s.   \notag 
\end{align}

\begin{align} \label{suppHet_mip}
	\min_{\B{Z},\mathbb{B},\bar{\boldsymbol{\beta}},\bar{\boldsymbol{z}} }& ~~ \sum_{k=1}^K \frac{1}{n_k} \norm{  \mathbf{y}_k - \mathbb{X}_k \boldsymbol{\beta}_k }_2^2 +\alpha\|\bB\|_F^2   + \lambda  \sum_{k=1}^K \norm{\boldsymbol{\beta}_k - \bar{{\boldsymbol{\beta}}} }_2^2 + \delta \sum_{k=1}^K\norm{\boldsymbol{z}_k - \bar{\boldsymbol{z}}}_2^2\\
	\text{s.t.}&~~ z_{k,j} \in \{0,1 \} ~~ \forall ~j \in [p] \notag \\
	& ~~ -M z_{k,j} \leq \beta_{k,j} \leq M z_{k,j} ~~\forall~ k \in [K], ~ \forall ~j \in [p],  \notag \\
	& ~~ \sum_{j=1}^p z_{k,j}\leq s ~~\forall~ k \in [K]. \notag 
\end{align}
Above, $M>0$ is a pre-specified large positive constant often known as the Big-M parameter (see \cite{bestSubset} for further details pertaining to the best subset selection problem). 

\subsection{Our Custom Exact Solver}\label{app:mip-solver}
Off-the-shelf commercial solvers such as Gurobi or Mosek can obtain globally optimal (or near-optimal) solutions to Problem~\eqref{suppHet_mip} for moderately large instances (with $pK$ in hundreds). Such solvers however may face challenges for larger instances typical in biomedical applications. This is because Problem~\eqref{suppHet_mip} has $O(pK)$ binary variables and $O(pK)$ continuous variables. Here, we present a custom global optimization-based algorithm for  Problem~\eqref{suppHet_mip}
that can scale to larger problem instances. 
For simplicity, we present the $\lambda=0$ case here (i.e., the Zbar + L2 estimator from Table~\ref{table:losses_table}). We discuss the general case ($\lambda  \geq 0$) in Supplement~\ref{app:oa-detail}. We begin by reformulating~\eqref{suppHet_mip} into the form
\begin{align}\label{generalmi}
	\min_{\B{Z}, \bar{\B{z}}}& ~~~   \sum_{k=1}^K F_k(\B z_k)  +\delta \sum_{k=1}^K \|\B{z}_k-\bar{\B z}\|_2^2\\  \text{s.t.}& ~~  \B z_{k}\in\{0,1\}^p;~~   \sum_{j=1}^p z_{k,j} \leq s~\forall k\in[K],\nonumber
\end{align}
where, as we show subsequently, $F_k:[0,1]^p\to\R$ are convex functions for $k\in[K]$. 
Problem~\eqref{generalmi} is now an optimization problem with $pK$ binary variables and $p$ continuous variables, unlike~\eqref{suppHet_mip} which is a function of $O(pK)$ continuous and $O(pK)$ binary variables. To optimize~\eqref{generalmi}, we employ a convex outer approximation algorithm~\citep{duran1986outer,bart} as discussed below. The following section shows how to reformulate Problem~\eqref{suppHet_mip} into form~\eqref{generalmi}. 

‌\subsubsection{A reformulation of Problem~\eqref{suppHet_mip}}\label{a reform sec}
For $k\in[K]$ and $\B z_k\in[0,1]^p$, let us define the following function:
\begin{align}\label{Freg2}
	F_k(\B z_k)=\min_{\B\beta_k,\B\xi_k} &\quad   \frac{1}{n_k} \left\Vert \B\xi_k \right\Vert_2^2 +\alpha \|\B\beta_k\|_2^2\\
	\text{s.t.} & \quad   |\beta_{k,j}|\leq Mz_{k,j}~~ \forall j\in[p]\nonumber \\
	& \quad\B\xi_k = \mathbf{y}_k-\bX_k\B\beta_k.\nonumber 
\end{align}
Then, we have the following result.
\begin{proposition}\label{prop5-main}
	Problem~\eqref{suppHet_mip} for $\lambda=0$ is equivalent to solving~\eqref{generalmi} where $F_k(\B z_k)$ for $k \in [K]$ is implicitly described via display~\eqref{Freg2}.
\end{proposition}
Functions $F_k(\cdot)$ are implicitly defined, and have several desirable properties. As we discuss in Supplement~\ref{app:oa-zbar} and Proposition~\ref{subgradpers}, these functions are convex and sub-differentiable. Moreover, for every $k\in[K]$,
we can compute subgradients of function $F_k(\cdot)$ by solving optimization Problem~\eqref{Freg2} using first-order methods. For each $k\in[K]$, Problem~\eqref{Freg2} effectively has only $s$ nonzero variables due to sparsity in binary variables---these optimization problems and subgradients can be computed efficiently since we consider sparse regimes with $s \ll p$ in practice.
Next, we discuss our custom MIP algorithm for Problem~\eqref{generalmi}.

\subsubsection{Outer Approximation Algorithm}\label{outersec}
We present an outer approximation (or cutting plane) algorithm~\citep{duran1986outer} to solve Problem~\eqref{generalmi}.
Our algorithm requires access to oracles that can compute the tuple 
$(F_k(\B z_k), \B g_k^{(\B z_k)})$
where, $\B g_k^{(\B z_k)} \in \partial F_k(\B z_k)$ is a subgradient of $F_k$ at an integral $\B z_k$. We refer the reader to Proposition~\ref{subgradpers} for computation of $(F_k(\B z_k), \B g_k^{(\B z_k)})$.
As $F_k$ is convex, 
$F_k(\B x)\geq F_k(\B z_k) +(\B x-\B z_k)^T\B g_k^{(\B z_k)}$ for all $\B x\in [0,1]^p$. Therefore, for $\B z_k^0,\cdots,\B z_k^{t-1}\in[0,1]^p$, we have:
\begin{equation}\label{lowerbound}
	\begin{aligned}
		F_{k,LB}(\B z_k):=&\max\left\{F_k(\B z_k^{a})+(\B z_k-\B z_k^{a})^T\B g_k^{(\B z_k^{a})}\right\}_{a=1}^{t-1} \\
		\leq& F_k(\B z_k)
	\end{aligned}
\end{equation}
where, $\B z_k \mapsto F_{k,LB}(\B z_k)$ is a piecewise linear convex lower bound to the map $\B z_k \mapsto F_k(\B z_k)$ on $\B z_k \in [0,1]^p$.
At iteration $t\geq 1$, the outer approximation algorithm replaces each $F_k$ for $k\in[K]$ in~\eqref{generalmi} with the lower bound $F_{k,LB}$. As each $F_{k,LB}$ is a piece-wise linear function, this results in a Mixed Integer Quadratic Program (MIQP):
\begin{align}\label{outerMILP}
	 \left((\B z_k^t,\eta_k^t)_{k=1}^K,\bar{\B z} \right)\in & \argmin_{\B{Z},\bar{\B z},\eta_k} ~~   \sum_{k=1}^K\eta_k + \delta \sum_{k=1}^K \|\B z_k - \bar{\B z}\|_2^2\\
	& \text{s.t.} ~~  \B z_k\in\{0,1\}^p,~\sum_{i=1}^p z_{k,i} \leq s~~\forall k\in[K]\nonumber\\
	& ~~~~~~\eta_k \geq F_k(\B z_k^i)+(\B z_k-\B z_k^i)^T\B g_k^{(\B z_k^{i})},  i\leq t-1,~k\in[K]\nonumber .
\end{align}
As the feasible set of Problem \eqref{generalmi} contains finitely many elements, an optimal solution is found after finitely many iterations\footnote{This is true as the lower bound obtained by the outer approximation in each iteration removes the current solution from the feasible set, unless it is optimal---see~\citet{duran1986outer} for details.} where iterations are indexed by $t$.  In addition, the objective in~\eqref{lowerbound}
$\sum_{k=1}^K\eta_k^t+\delta \sum_{k=1}^K \|\B z_k-\bar{\B z}\|_2^2$ 
is a lower bound of the optimal objective value in~\eqref{generalmi}. Recall that any feasible solution $\{\B z_k^{t}\}$ leads to an upper bound for Problem~\eqref{generalmi}.
Consequently, the optimality gap of the outer approximation algorithm can be calculated as $\text{OG}=(\text{UB}-\text{LB})/\text{UB}$ where $\text{LB}$ is the current (and the best) lower bound achieved by the piecewise approximation 
in~\eqref{outerMILP}, and $\text{UB}$ is the best current upper bound. 

We note that the MIQP in~\eqref{outerMILP} can be solved by standard off-the-shelf solvers such as Mosek and Gurobi. Importantly, 
Problem~\eqref{outerMILP} involves fewer continuous variables than~\eqref{suppHet_mip} ($\mathcal{O}(p)$ vs $\mathcal{O}(pK)$), and is generally faster to solve.

\begin{remark}
	Outer approximation algorithms have been recently used to successfully 
	solve 
	best-subset selection-type problems where we minimize 
	a least squares loss with an $\ell_0$-regularizer---see for example~\citep{bart,behdin2021sparse} and references therein. 
	The structure of the problem we consider here is, however, different from earlier works. 
\end{remark}

\section{Approximate Algorithms}\label{section:approx-algs}
We now discuss approximate algorithms for our estimators. Such algorithms enable us to obtain high-quality feasible solutions quickly. This is important when, for example, the runtime is of essence or when tuning hyper-parameters. High-quality approximate solutions are also helpful for warm-starting our exact solver. These approximate methods do not, however, deliver optimality certificates---these certificates are available from the MIP-based global optimization framework discussed in Section~\ref{app:mip-solver}. 

First-order optimization methods are often used to obtain good solutions to discrete optimization problems in sparse learning~\citep{bestSubset,beckeldar,l0learn}. However, the structure of the optimization problem in~\eqref{suppHet} is different from related estimators, such as best subset selection. Due to the specific structure of Problem~\eqref{suppHet}, we use the block Coordinate Descent (CD) procedure where we partially minimize the objective w.r.t. the blocks of variables $(\{\B\beta_k,\B z_k\},\bar{\B\beta},\bar{\B z})$. 
The block CD procedure is quite efficient but can get stuck in sub-optimal solutions, so we present a local combinatorial search procedure that can further improve the quality of the solutions~\citep{l0learn,groupedL0,beckeldar}.

\subsection{Block Coordinate Descent} \label{section:blockcd}
We describe our block CD procedure for problem~\eqref{suppHet}. 
We write the objective function of problem~\eqref{suppHet} as $\sum_{k=1}^K g_k(\B{\beta}_k,\B{\bar{\beta}}) +\sum_{k=1}^K h_k(\B{z}_k,\B{\bar{z}})$, 
where 
\begin{align}
	g_k(\B{\beta}_k,\B{\bar{\beta}}) &=  \frac{1}{n_k} \norm{  \mathbf{y}_k - \mathbb{X}_k \boldsymbol{\beta}_k }_2^2 + \lambda   \norm{\boldsymbol{\beta}_k - \bar{{\boldsymbol{\beta}}} }_2^2 + \alpha \norm{\B{\beta}_k}_2^2, \nonumber\\
	h_k(\B{z}_k,\B{\bar{z}}) &= \delta \norm{\boldsymbol{z}_k - \bar{\boldsymbol{z}}}_2^2. \label{ghdef}
\end{align}
In our block CD algorithm, we first optimize problem~\eqref{suppHet} over decision variables of task $k$, $(\B{\beta}_k, \B{z}_k)$, holding all other variables fixed. 
To update $\B{w}_k$, we use a proximal gradient update~\citep{beck2009fast}. 
Specifically, we update $(\B{\beta}_k, \B{z}_k)$ by finding a minimizer of the following optimization problem (w.r.t. $\B{\beta}_k, \B{z}_k$):
\begin{equation}\label{IHT1}
	\begin{aligned}
		\min~ & \frac{L_k}{2} \left\Vert\B{\beta}_k - \left[\B{\beta}_k^+-\frac{1}{L_k}\nabla_{\B{\beta}_k}g_k(\B{\beta}_k^+,\bar{\B{\beta}})\right]\right\Vert_2^2 + h_k(\B{z}_k,\bar{\B{z}})\\
		\text{s.t.}~&~ z_{k,j}\in\{0,1\}~\forall j\in[p], k\in[K],\\ 
		&\beta_{k,j}(1-z_{k,j})=0~\forall j\in[p], k\in[K], \\
		&\sum_{j=1}^p z_{k,j}\leq s ~~ \forall k\in[K]  
	\end{aligned}
\end{equation}
where $L_k$ is the Lipschitz constant of the gradient 
$\B{\beta}_k \mapsto \nabla_{\B{\beta}_k}g_k(\B{\beta}_k,\bar{\B{\beta}})$, and $\boldsymbol{\beta}_k^{+}$ is the current value. 
{Once we update the $\B{w}_k$'s, we update $\B{\bar{z}}$ as $ \B{\bar{z}} = \frac{1}{K} \sum_{k=1}^K \B{z}_k$; we then update $\B{\bar{\beta}}$ as $\B{\bar{\beta}} = \frac{1}{K}\sum_{k=1}^K \B{\beta}_k$.}
By construction, the update rules for $\B{\bar{z}}$, $\B{\bar{\beta}}$, and $\B{w}$ result in a descent algorithm (i.e., one that does not increase the objective value in~\eqref{suppHet}), while maintaining feasibility. We note that problem~\eqref{IHT1} is itself a mixed-integer problem and is considerably more challenging to solve compared to the proximal sub-problems appearing in best-subset selection. However, as we discuss in Supplement~\ref{app:blockcd}, this can be solved in closed form efficiently. Furthermore, to make the block CD algorithm faster, we use active set updates, where only a subset of variables are updated in each iteration (see Supplement~\ref{app:activeset} for details). 

\subsection{Local Combinatorial Optimization}\label{sec:localsearch}
Once we obtain a good solution to problem~\eqref{suppHet} using the block CD algorithm (discussed above), we apply a local search method to further improve the quality of the solution. Let $\tilde{\mathbb{B}},\tilde{\B{Z}}$ be a solution from the block CD method and fix $k\in[K]$. The main idea behind the local search algorithm is to swap a coordinate inside the support of $\tilde{\boldsymbol{\beta}}_{k}$ with a coordinate outside the support (by setting the coordinate inside to zero and letting the coordinate outside become nonzero) and checking if optimizing over the new coordinate in the support leads to an improvement in the objective value. Mathematically, let

\begin{equation}
	\begin{aligned}
		{g}(\mathbb{B})& = \sum_{k=1}^K g_k(\B{\beta}_k,\sum_{k=1}^K \B{\beta}_k/K),\\  
		{h}(\B{Z}) &=\sum_{k=1}^K h_k(\B{z}_k,\sum_{k=1}^K \B{z}_k/K)
	\end{aligned}
\end{equation}
where $g_k,h_k$ are defined in~\eqref{ghdef}. Next, for a fixed $k_0\in[K]$ we consider
\begin{equation}
	\min_{\substack{j_1: \tilde{\beta}_{k_0,j_1}\neq 0 \\ j_2: \tilde{\beta}_{k_0,j_2}= 0 }}  \min_{b\in\R} \bigg\{ {g}(\tilde{\mathbb{B}}-\tilde{\beta}_{k_0,j_1}\boldsymbol{E}_{k_0,j_1}+b\boldsymbol{E}_{k_0,j_2}) + h(\tilde{\B{Z}}-\boldsymbol{E}_{k_0,j_1}+\boldsymbol{E}_{k_0,j_2})\bigg\} \label{localsearch}
\end{equation}
where $\B{E}_{j_1,j_2}$ is the matrix with all coordinates set to zero except coordinate $(j_1,j_2)$ set to one.
In words, Problem~\eqref{localsearch} identifies a swap between coordinates inside and outside of the support that leads to the lowest objective. If the optimal solution to problem~\eqref{localsearch} has a lower objective compared to $(\tilde{\mathbb{B}},\tilde{\B{Z}})$, we update our current solution. We cycle through $k$ until no swap improves the objective. For any $j_1,j_2$, the inner optimization problem in problem~\eqref{localsearch} is a convex quadratic problem over $b$. In Supplement~\ref{app:localsearch}, we show the inner optimization in problem~\eqref{localsearch} can be solved efficiently.
Importantly, Problem~\eqref{localsearch} finds the swap that leads to the best overall objective in Problem~\eqref{suppHet}, not just the swap that improves the contribution of task $k_0$ in the objective. The local combinatorial optimization problem~\eqref{localsearch} considers a Bbar/Zbar penalty, and hence results in a solution that is different from what we would get by using local combinatorial search on each task separately.

\section{Simulations} \label{sec:simulations}

\paragraph{Experimental Setup} \label{sec:sim_design}

We simulated datasets to better understand the effect of four aspects of sparse MTL settings on the performance of our proposed methods: 1) the strength of correlation among covariates, 2) the degree of support heterogeneity across tasks, 3) the sample size of each task-specific dataset, and 4) the sparsity level, $s$, of the coefficient \textit{estimates}. 

We simulated data from the linear model $y_{k,i} = \beta_{0,k} + \mathbf{x}_{k,i}^T ({\boldsymbol{\beta}}_k \odot \boldsymbol{z}_k) + \epsilon_{k,i}$
where ${\boldsymbol{\beta}}_k \overset{\text{iid}}{\sim} \cN_{p}(\boldsymbol{\mu}_{\beta}, \sigma^2_{\beta} {\B{I}} )$, $\mathbf{x}_{k,i} \overset{\text{iid}}{\sim} \mathcal{N}_p( \boldsymbol{0}, \B{\Sigma} )$, and $\boldsymbol{z}_{k} \in \{0,1\}^{p}$. We drew the errors and coefficients independently (i.e., $\beta_{k,j} \ind \epsilon_{k,i}~\forall~ k, i, j$). We describe the distribution of $\epsilon_{k,i}$ below. The covariance matrix of the covariates, $\B{\Sigma}$, had an exponential correlation structure: $\Sigma_{l,r} = \rho^{|l-r|}$ for $l \neq r$ and $\Sigma_{r,r} = 1$. We set $\rho = 0.5$ as this level of correlation provided a good testing ground to characterize the support recovery accuracy of the methods at the sample sizes tested \citep{l0learn}. We conducted ablation studies and varied correlations levels with $\rho \in \{0.2,0.5,0.8\}$. The results from these experiments are shown in Supplement~\ref{supp_sims}. In our experiments, we let each $\B{z}_k$ to have $s^*=\norm{\boldsymbol{\beta}_k}_0 = 10$ nonzero elements. We use $s^*$ to denote the \textit{true} (simulated) sparsity level of the underlying model, whereas $s$ denotes the sparsity level specified via our estimator. We evaluated the performance of each estimator at sparsity levels $s \in \{7,10,13\}$. This allowed us to compare estimation performance when the sparsity of the solutions was below, equal to, and above the true simulated sparsity level. 

We simulated data across a range of support heterogeneity levels to empirically evaluate the statistical properties we theoretically examined in Section~\ref{sec:theory}. We simulated the support of task $k$, $\mathcal{Q}_k=\{j:z_{k,j}=1\}$, by drawing it uniformly at random from the set $\tilde{\mathcal{Q}}=\{1,3,\cdots,2q-1\}$ for some pre-specified $q \geq s^*$. The choice of odd integers as the elements of the set $\tilde{\mathcal{Q}}$ is motivated by the correlation structure of the covariates. The high multicollinearity of adjacent features induced by the exponential correlation structure can make learning the support of adjacent elements of $\B{\beta}_{k}$ difficult. We thus set $\tilde{\mathcal{Q}}$ to have an alternating structure so that $\boldsymbol{z}_k$ never had successive nonzero elements. We show results for simulations in which we set $\tilde{\mathcal{Q}}$ to have a non-alternating structure (i.e., $\boldsymbol{z}_k$ \textit{can} have successive nonzero elements) in Supplement~\ref{sec:non_alt_sims}. We varied $q = |\tilde{\mathcal{Q}}|$, to adjust the level of support heterogeneity and we use the metric $s^*/q$ as a measure of support homogeneity. When  $q=s^*$, there is no support heterogeneity across tasks as $\tilde{\mathcal{Q}} = {\mathcal{Q}}_k ~ \forall ~k$ w.p.1. Conversely, smaller values of $s^* / q$ indicate higher values of support heterogeneity on average. Both the supports, $\B{z}_k$, and values of the nonzero model coefficients, $\boldsymbol{\beta}_k$, were allowed to (independently) vary across tasks. These were used to simulate two forms of between-task model heterogeneity that are characteristic of many MTL settings. 

To simulate heterogeneity in the SNR across tasks, we drew $\epsilon_{k,i} | \sigma^2_k \overset{\text{iid}}{\sim} \mathcal{N}(0, \sigma^2_k$), and $\sigma^2_k \overset{\text{iid}}{\sim} \mbox{Unif}(1/2, 2)$ for each simulation replicate and task. We varied $n_k \in \{50, 100\}$ and set $p=250$. We drew the means of the nonzero coefficients $\mu_{\beta_{k,j}} \overset{\text{iid}}{\sim} \mbox{Unif}(-0.5, -0.2) \cup \mbox{Unif}(0.2, 0.5)$ to ensure they were bounded away from zero. We set $\sigma^2_{\beta} = 50$ to simulate high heterogeneity in model coefficient values across tasks. This value was selected based on the average sample variance of the regression coefficient estimates in our data applications (see Supplement~\ref{sim_scheme} for more detail). 

We simulated 100 replications for each set of simulation parameters. In each replicate, we simulated $K$ tasks that each had $n_k$ training observations and $n_k$ test set observations. On each replicate, we fit models and calculated performance measures for each task and method. We averaged the metrics over tasks and show the distribution of the average performance measures across the 100 replicates. We present out-of-sample prediction performance (RMSE) and the F1 score of the supports to compare support recovery (i.e., to compare $\hat{\B{z}}_k$ and $\B{z}_k$). We provide a detailed description of performance metrics in Supplement~\ref{performance}.

\begin{table*}[!h]
	\scriptsize
	\label{losses_table}
	\centering
	\begin{tabular}{c | ccc}
		\hline
		\hline
		\multicolumn{1}{c}{\bfseries Method} \vline & \multicolumn{1}{c}{\bfseries Abbreviation}  &  \multicolumn{1}{c}{\bfseries MTL Squared Error Loss } 
		&  \multicolumn{1}{c}{\bfseries Penalty } \\
		\midrule
		\midrule
		Group Lasso  & GL & $\sum_{k=1}^K \frac{1}{n_k} \norm{  \mathbf{y}_k - \mathbb{X}_k \boldsymbol{\beta}_k }_2^2$ & $+~ \lambda \sum_{j=1}^p \norm{\boldsymbol{\beta}^{(j)}}_2$ \\
		Sparse Group Lasso  & SGL & $\sum_{k=1}^K \frac{1}{n_k} \norm{  \mathbf{y}_k - \mathbb{X}_k \boldsymbol{\beta}_k }_2^2$ & $~+(1-\alpha) \lambda \sum_{j=1}^p \norm{\boldsymbol{\beta}^{(j)}}_2 + \alpha \lambda \sum_{j=1}^p \norm{\boldsymbol{\beta}^{(j)}}_1$ \\
		Group Exponential Lasso  &  gel & $\sum_{k=1}^K \frac{1}{n_k} \norm{  \mathbf{y}_k - \mathbb{X}_k \boldsymbol{\beta}_k }_2^2$ & $+~\lambda^2 / \tau \sum_{j=1}^p \left ( 1 - exp \left [ -\tau / \lambda \norm{\boldsymbol{\beta}^{(j)}}_1 \right ] \right )$ \\
		Composite MCP  &  cMCP & $\sum_{k=1}^K \frac{1}{n_k} \norm{  \mathbf{y}_k - \mathbb{X}_k \boldsymbol{\beta}_k }_2^2$ & $+\sum_{j=1}^p \mbox{MCP}_{\lambda,\gamma_1} \left [\sum_{k=1}^K \mbox{MCP}_{\lambda,\gamma_2} \left ( | \beta_{j,k} | \right ) \right ] $ \\
		Group MCP  & grMCP & $\sum_{k=1}^K \frac{1}{n_k} \norm{  \mathbf{y}_k - \mathbb{X}_k \boldsymbol{\beta}_k }_2^2$ & $+\sum_{j=1}^p \mbox{MCP}_{\lambda,\gamma_1} \left ( \norm{\boldsymbol{\beta}^{(j)}}_2 \right )$ \\
		\hline
	\end{tabular}
	\caption{\footnotesize Benchmark method names and loss functions. $\mbox{MCP}_{\lambda,\gamma}$ denotes the minimax concave penalty $\rho(x; \lambda, \gamma) = \lambda \int_0^{|x|} (1 - t/(\gamma \lambda))_+dt, \gamma > 1$ \citep{groupSel_review}. Defaults $\tau = 1/3$ for gel, and $\gamma = 3$ for MCP were used.}
	\label{table:losses_table2}
\end{table*}

\paragraph{Simulation Modeling and Tuning} \label{cv}

We tuned models with a 10-fold cross-validation procedure in which we set aside a validation set for each task and fold. For each fold, we averaged across the $K$ task-specific cross-validation errors and selected the hyperparameters associated with the lowest average error. When tuning our $\ell_0$-constrained models, for each $s$, we fit a path of solutions similar to the approach taken for a path of Lasso \citep{glmnet} or $\ell_0$-penalized regression solutions \citep{l0learn}. To avoid tuning over a 2-dimensional grid, we tuned the $\alpha$ for the $L_0 L_2$ method, and used that tuned $\alpha / 2$ as a fixed hyperparameter when tuning and fitting the final Zbar+L2 models.
In practice, we found the $\alpha$ values selected by cross-validation were extremely small (typically less than $10^{-6}$) and usually had no measurable effect on prediction performance. For the small values of $s$ typical in $\ell_0$-constrained models, we expect that tuning $\alpha$ is unnecessary provided it is sufficiently small. For the CS+L2 models, we nevertheless tuned $\alpha$ since it is the only hyperparameter we tuned for that method.

We compared the performance of our methods with a wide range of existing sparse MTL methods as shown in Table~\ref{table:losses_table2}. 
The Group Lasso (GL) tends to set coefficients to zero for all tasks or none of the tasks. The Sparse Group Lasso (SGL) yields solutions for which coefficients can be zero for some tasks and nonzero for others. GL and SGL are convex methods. To overcome the downsides of excess shrinkage arising from $\ell_1$ penalties, non-convex alternatives have been proposed. Similar to GL, the group MCP (grMCP) encourages coefficients to be all zero or nonzero across tasks, but applies the non-convex MCP penalty. The Group Exponential Lasso (gel) and composite MCP (cMCP) apply non-convex penalties that, like SGL, allow for sparsity pattern heterogeneity. 
We note that for non-convex penalties, the quality of the solution can depend on specific choices of the optimization algorithm\footnote{This should be contrasted with our approach which is based on MIP and global optimization.}.
Importantly, our methods control the sparsity level (through the hyperparameter $s$), the degree of sparsity pattern heterogeneity (through $\delta$), and information sharing on the coefficient \textit{values} (through $\lambda$) in a more direct and transparent fashion.

All methods had unpenalized task-specific intercepts. To tune the hyperparameters of both our methods and these benchmark methods (e.g., Sparse Group Lasso) across a wide range of values, we tuned hyperparameters in two steps. This allowed us to identify a neighborhood of hyperparameter values that yielded good prediction performance and had comparable support sizes. 
We used the $\texttt{sparsegl}$ package \citep{sparsegl_package} for the sparse group lasso (SGL), the $\texttt{grpreg}$ package for the group exponential lasso (gel), group MCP (gMCP), and composite MCP (cMCP) \citep{groupreg1, groupreg2, groupreg3}. We tuned over hyperparameter grids with a total of 100 values that complied with the sparsity constraint, for both $\ell_0$ and benchmark methods. We provide additional details in Supplement~\ref{sims_tuning}.

\paragraph{Simulation Results} \label{sims}

In Figure \ref{fig:sims_f1_rmse} we compare the prediction performance and support recovery of the methods at different sample sizes, $n_k$, and levels of support heterogeneity in the underlying model. Given the number of methods compared, we present figures from the methods that exhibited the best performance for the remainder of the manuscript. Supplemental Figures~\ref{fig:sims_rmse_supp}–\ref{fig:sims_f1_supp} show, however, the results from a wider range of simulation parameters and include performance summaries from all methods described in Tables \ref{table:losses_table} and \ref{table:losses_table2}. 

These results show that, across different levels of $s$, the Zbar+L2 and Bbar methods substantially outperform the group penalties when support heterogeneity is higher (i.e., $s^*/q$ is low). There were, however, meaningful differences in relative performance across the different sparsity levels. For example, when $s = 7 < s^* = 10$, the prediction performance of the Zbar+L2 and Bbar methods were substantially better than the group penalties (e.g., SGL, gel) across the range of support heterogeneity levels tested. The support recovery, as measured by F1 score, of the Zbar+L2 and Bbar methods were also superior except in the common support case, $s^*/q = 1$, where they were comparable to CS+L2 and grMCP. As expected, performance of the common support method, CS+L2, degrades as support heterogeneity grows. Interestingly, at low sample sizes ($n_k=50$), most methods that yield heterogeneous supports (e.g., Zbar+L2, gel) outperform methods that yield common supports (e.g., CS+L2, grMCP) even when $s^*/q = 1$. When $s=s^*=10$, the Zbar+L2 and Bbar methods vastly outperformed the other methods in RMSE and F1 score except in the common support case, in which the CS+L2 and grMCP performed comparably. Interestingly, the Zbar+L2 outperformed the Bbar method for higher levels of support homogeneity. Finally, when $s= 13 > s^*=10$, the Zbar+L2 and Bbar methods strongly outperformed the other methods in RMSE and F1 score at high levels of support heterogeneity. For low levels of support heterogeneity, the performance of group penalties was comparable to, or even slightly better than the Zbar and Bbar methods. Importantly, the Zbar+L2 outperformed the Bbar at most levels of support heterogeneity. 

For most settings tested, Figure \ref{fig:sims_f1_rmse} and Supplemental Figures~\ref{fig:sims_rmse_supp}–\ref{fig:sims_f1_supp} show that the shrinkage-based group penalties (e.g., SGL, gel) appear to result in greater false positives and bias than the proposed $\ell_0$-constrained estimators. However, when both $s=13$, and $s^*/q$ is high, some shrinkage-based methods exhibit better prediction performance than our methods.

Comparing across the $\ell_0$ estimators, specifically, illustrates the advantages of methods that borrow strength across tasks through the coefficient \textit{supports} rather than through the coefficient \textit{values}. For example, the Zbar+L2 method exhibits comparable or superior prediction performance and support recovery compared to the Bbar method. Even at parameter settings for which prediction performance is comparable, the Zbar+L2 often exhibits substantially higher support recovery compared to the Bbar. The benefits of the Zbar method are more pronounced at lower sample sizes: when $n_k = 50$, the Zbar method exhibits competitive support recovery and even superior prediction performance to the common support methods. In fact, the Zbar+L2 method outperforms the CS+L2 even when the true model has a common support (i.e., $s^*/q=1$). This suggests that the Zbar method is useful both for support heterogeneous and homogeneous settings while the CS methods only perform well when the supports are nearly identical. 

Supplement~\ref{supp_sims} shows results with a wider range of simulation settings and performance metrics. 
In Supplemental Figure~\ref{fig:sims_rmse_unadjust}, we present the raw prediction performance results shown in Figure \ref{fig:sims_f1_rmse}, without adjustment by the performance of $L_0 L_2$. These results show that for most methods, prediction performance improves as support heterogeneity in the underlying model falls (as $s^*/q \rightarrow 1$). Under our simulation setting, the supports and the values of nonzero coefficients are drawn independently. Thus as the true support heterogeneity decreases, the overall between-task heterogeneity in the simulated coefficient values decreases too. As such, lower overall task heterogeneity in the true model coefficients implies that borrowing strength across tasks is expected to improve estimation, and hence overall prediction performance should increase. We note this is also expected from our statistical theory: In Theorem~\ref{thm1}, the r.h.s of~\eqref{thm1-inequlity} is smaller for given values of $\delta,\alpha,\lambda$ when the between-task heterogeneity is lower. Importantly, lower sparsity levels, $s$, do not in general result in improved prediction performance if $s$ is misspecified (i.e., $s \neq s^*$). Indeed, Supplemental Figure~\ref{fig:sims_rmse_unadjust} shows that when $s$ is misspecified, the performance of all methods degrades. 

Finally, Supplement~\ref{noCommSupp} shows results from a setting in which task coefficient supports were simulated to have no overlap (i.e., $\B{z}_k^T \B{z}_{k'} = 0 ~\forall ~k \neq k'$). In this setup, the relative method performances depended heavily on the covariate correlation ($\rho$), sparsity ($s$), and sample sizes ($n_k$). Interestingly, gel exhibited prediction performance superior to the non-MTL methods, $L_0 L_2$ and Lasso, at low sample sizes ($n_k=50$). At higher sample sizes ($n_k = 100$), the $L_0 L_2$ outperformed all MTL methods except when covariate correlation levels were low ($\rho = 0.2$); in that case, the Zbar methods slightly outperformed other methods.

In summary, the simulations illustrate that on balance the proposed Bbar and Zbar+L2 methods exhibit superior prediction and support recovery performance compared to the existing shrinkage-based group penalties across a wide range of simulation settings. Importantly, while the Zbar+L2 often performs comparably to the Bbar for high levels of support heterogeneity, it exhibits substantially better prediction and support recovery performance as support heterogeneity falls. Even in cases where prediction performance is comparable, the Zbar+L2 often exhibits far superior support recovery compared to the Bbar. Together, these results highlight the benefits of using a combinatorial penalty that encourages the coefficient \textit{supports} to be similar rather than shrinking the coefficient \textit{values} together across tasks.

\begin{figure*}
	\centering
	\begin{tabular}{cc}
		\centering
		\includegraphics[width=0.99\linewidth]{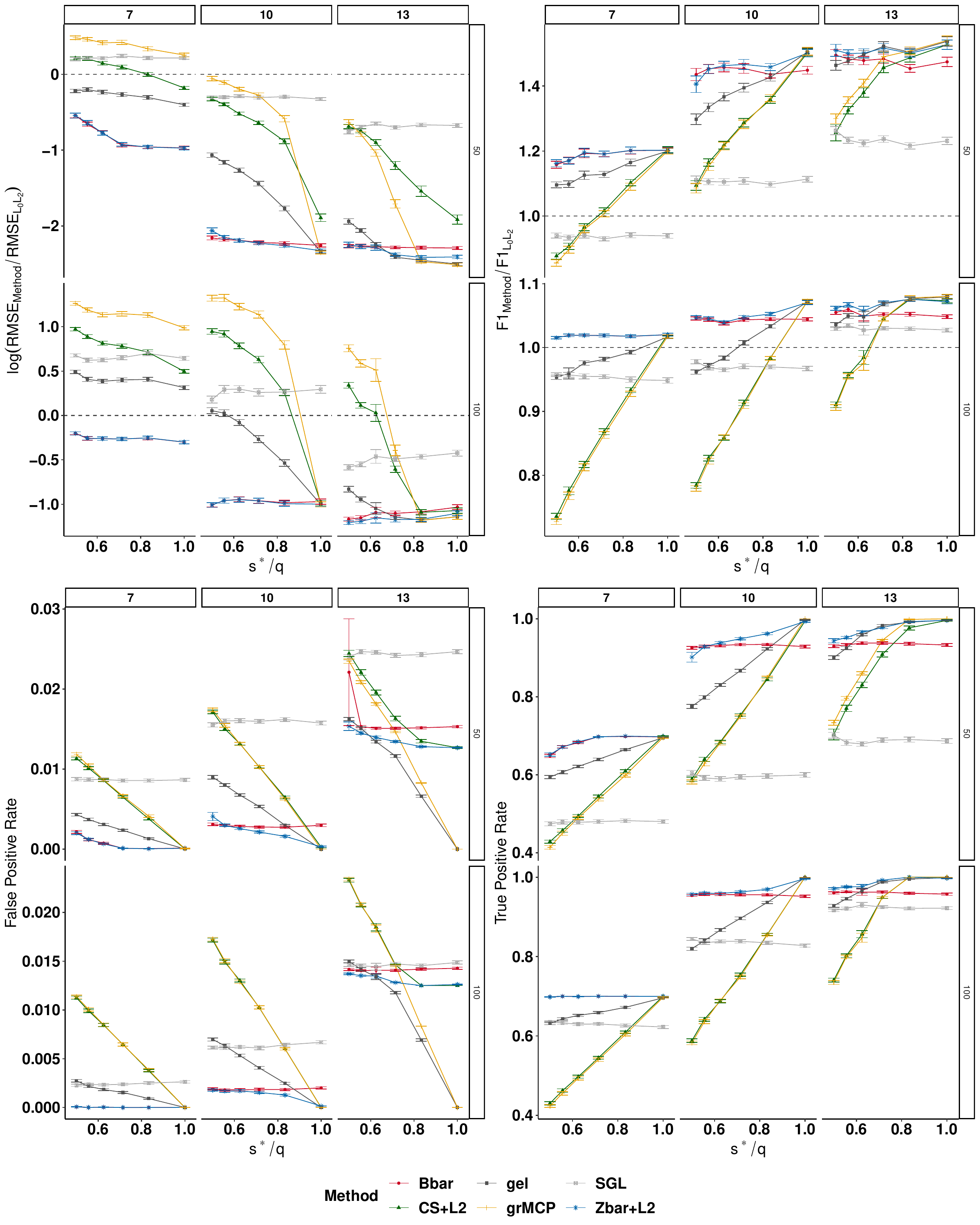} &
	\end{tabular}
	\caption{\footnotesize Prediction performance (RMSE) [Top Left], F1 score of support recovery [Top Right], false positive rate [Bottom Left], and true positive rate [Bottom Right] averaged across tasks for different $n_k$ (rows), support sizes (columns),
		and levels of support homogeneity of the true (simulated) model  ($s^*/q = 1$ indicates that all tasks had identical support). 
		Lower RMSE and higher F1 scores indicate superior (relative) performance. Data were simulated with support size  $s^*=10$, and covariate correlation parameter $\rho = 0.5$. }
	\label{fig:sims_f1_rmse}
\end{figure*}

\subsection{An Illustration of Coefficient Estimation} \label{coef_supp}

In Figure~\ref{fig:fullSupp} we include the results of a greater set of methods on the example shown above in Figure~\ref{fig:introSupp}. To visualize the effect of the different methods, we simulated a dataset where $p = 50$, $\sigma^2 = 5$, $\rho = 0.5$, $\beta_{k,j} \in \{-1, 0, 1\}$, $s^* = 7$, and $n_k = 25$ for $k \in \{1,2\}$. We then tuned and fit models with the cardinality of the solution set to $s=7$ for all $\ell_0$ methods explored to ensure no differences in model performance were due to sparsity level tuning. All other simulation parameters were the same as those described in Section~\ref{sec:sim_design}. We compared regression coefficient estimates against their true values. The $L_0 L_2$ method only partially recovers the support, and the nonzero estimates are overly shrunk, likely due to the unfavorable aspect ratio, $p/n_k = 2$, and the fact that no information is shared across tasks. The CS+L2 approach exhibits poor estimation accuracy likely because the common support constraint is misspecified in this example. The Bbar method shrinks the $\boldsymbol{\beta}_k$ values towards each other and therefore performs poorly in this example where 1) the supports vary across tasks, and 2) the coefficient values that are nonzero in both tasks, for fixed $j$, have opposite signs (i.e., $\beta_{1,j} = -\beta_{2,j}$). While this example is simulated in a way that is especially challenging for methods that borrow strength across the $\boldsymbol{\beta}_k$ values, the example provides a useful illustration of when the Zbar methods would be expected to outperform other methods. We show in a neuroscience application that the support heterogeneity and sign changes in regression coefficients across tasks simulated here do arise in practice.

\begin{figure*}[t] 
	\centering
	\begin{subfigure}[t]{0.4\textwidth}
		\centering
		\includegraphics[width=0.8\linewidth]{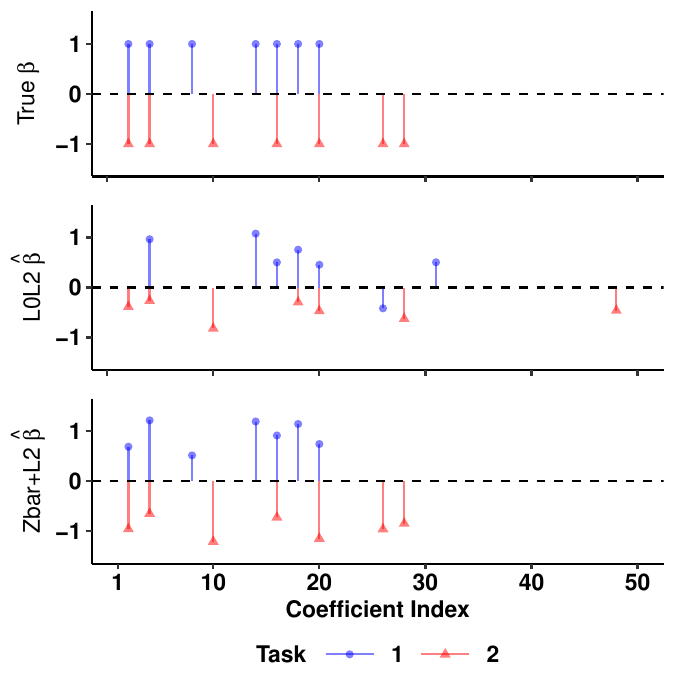}
		\label{fig:SHR_c}
	\end{subfigure}
	\begin{subfigure}[t]{0.4\textwidth}
		\centering
		\centering
		\includegraphics[width=0.8\linewidth]{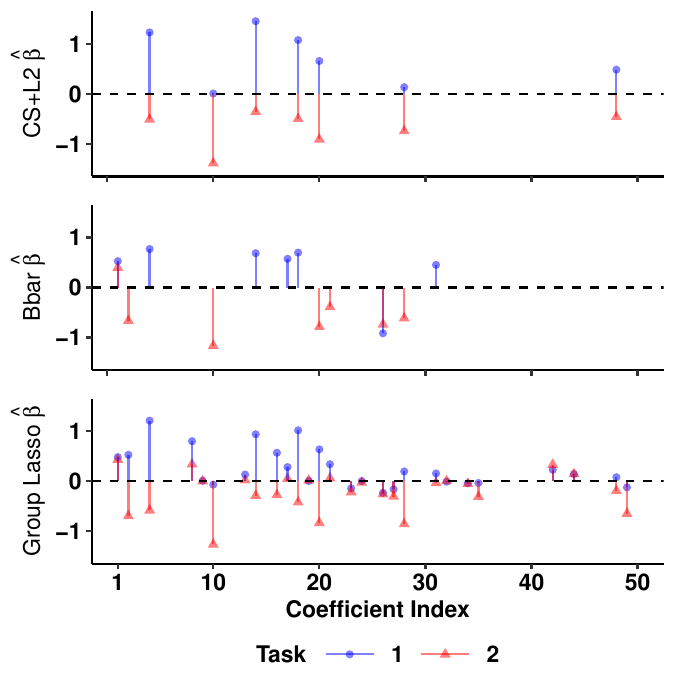}
		\label{fig:SHR_d}
	\end{subfigure}
	\caption{\footnotesize The Zbar+L2 method improves support recovery and coefficient estimate accuracy. The top left panel shows the true $\boldsymbol{\beta}$. Other panels show coefficient estimates. Color and symbol indicate task index.}
	\label{fig:fullSupp}
\end{figure*}

\subsection{Optimization Algorithms Simulations}\label{sec:opt-sim}
\subsubsection{Experimental Setup}
Most of our setup is similar to the one in Section~\ref{sec:sim_design}. We set $K=5$ and for $i\in[n],k\in[K]$, we set $y_{k,i}=\mathbf{x}_{k,i}^T\B\beta_k + \epsilon_{k,i}$. For $k\in[K]$, We draw $\mathbf{x}_{k,i}\overset{\text{iid}}{\sim} \mathcal{N}_p(\B 0,\B\Sigma)$ where $\B \Sigma$ follows an exponential correlation form with correlation $\rho$. 
We also let $\epsilon_{k,i}\sim\mathcal{N}(0,\sigma_k^2)$ where we set $\sigma_k^2$ to match the desired signal-to-noise ratio (SNR) that we take to be the same for all tasks, $SNR=\|\bX_k\B{\beta}_k\|_2/\|\B\epsilon_k\|_2$.
Next, we explain how we draw $\B\beta_k$. To this end, we fix the support size to $s^*$ and we consider a setup where the coordinates $\{1,\cdots,s^*-1\}$ are common among the supports of all the $\B\beta_k$'s, simulating a setup with approximately common support. We take the last coordinate in the support of $\B\beta_k$ to be $s^*+k$.
We draw the nonzero $\beta_{k,i}$'s independently from $\text{Unif}(0,1)$. Finally, we normalize $\B\beta_k$ to have $\|\B\beta_k\|_2=1$.  We study the case with $\lambda=0$, i.e., we only consider the Zbar penalty in~\eqref{suppHet}. Experiments in this section were done on a machine equipped with Intel Xeon 8260 CPU with 32 GB of RAM.

\subsubsection{Scalability of the exact solver}\label{sec:scalability}
To illustrate the performance of our custom exact solver (from Section~\ref{app:mip-solver}) for Problem~\eqref{suppHet_mip}, we compare it to Gurobi, a state-of-the-art commercial MIP solver. In particular, we consider Problem~\eqref{suppHet_mip} with Zbar and L2 penalties, with $\lambda=0$. We set $n=300,~SNR=100$ and $\rho=0.4$. We used the solutions from the approximate methods (applying both block CD and local search) as a warm start for both solvers. 

We consider two values of $s^* \in \{5,6\}$ and report the runtime of exact solvers in Table~\ref{table-mip}. Particularly, we report the runtime to reach less than optimality gap $\text{OG}=5\%$ (see Section~\ref{outersec}). We report the final optimality gap if the method does not reach a $5\%$ gap in 400 seconds.

\begin{table*}[t!]
	\small
	\centering
	\begin{tabular}{ c|cc|cc }
		\midrule
		\midrule
		& \multicolumn{2}{c|}{$s=5$} & \multicolumn{2}{c}{$s=6$}\\
		& Outer Approximation  & Gurobi & Outer Approximation   & Gurobi \\
		& (Our custom Algorithm) & &  (Our custom Algorithm) & \\
		\midrule
		\midrule
		$p=100$& $2.4\pm 0.4$ & $28.8\pm14.1$  & $8.1\pm 3.0$  & $32.2\pm14.9$ \\
		$p=200$& $4.5\pm 1.2$ &  $164\pm 64$ & $23.4\pm 3.8$ & $(6.9\%)$ \\
		$p=500$& $20.4\pm 5.9$&  - & $51.3\pm 13.8$ & -\\
		$p=1000$& $121\pm 31$ &  - & $184\pm 15$ & - \\
		$p=1500$& $146 \pm  46$ &  -  & $326\pm56$  & - \\
		$p=2000$& $198\pm76$ &  -  & ($20.6\%$) & - \\
		$p=2500$& $254\pm 83$ &  -  & ($33.2\%$) & - \\
		\hline
	\end{tabular}
	
	\caption{\footnotesize Comparison of the average runtime (in seconds) across 10 replicates ($\pm$ standard error) of our tailored outer approximation method with an off-the-shelf solver (Gurobi) from Section~\ref{sec:scalability}.  The numbers in parenthesis show the optimality gap after 400 seconds, if the optimality gap is not less than 5$\%$. A `--' (dash) indicates no valid lower bound was returned by Gurobi after 400 seconds. }
	\label{table-mip}
\end{table*}

Table~\ref{table-mip} shows that Gurobi struggles to obtain optimality certificates for problems with $p\approx 200$. On the other hand, our custom algorithm is able to obtain and certify optimal solutions for problems with $p=2500$ in minutes. We note that the feasible set of Problem~\eqref{suppHet_mip} is of size $\approx {p \choose s}^K$. This implies that a simple exhaustive search over all feasible solutions for largest examples we discuss would require visiting ${2500 \choose 5}^5\approx 10^{74}$ instances. Interestingly, our 
proposed specialized algorithm can deliver optimality certificates for such large problems in minutes. This demonstrates the usefulness of our exact solver. 

Importantly, these experiments also demonstrate the quality of our approximate algorithms. Indeed, in all replicates in these experiments, the solutions returned from the exact solver and our approximate algorithms were the same. This shows that our approximate methods are able to deliver high-quality (near)-optimal solutions although they are unable to certify the optimality of the solution (via dual bounds). 
This observation, together with our theory for approximate solutions (cf Section~\ref{stats_approx}), supports the strong statistical performance of our approximate algorithms in practice.

Our approximate algorithm is quite efficient in addressing problems with $pK$ in tens of thousands.
The algorithm (i.e., both block CD and local search together) took 3 (for $p=100$) to 10 seconds (for $p=2500$) to converge when fit with tuned values of hyperparameters. Our algorithms are 
therefore applicable in settings 
where it is useful to compute solutions fast, and optimality certificates (via dual bounds) are not of primary need.

\subsubsection{Effectiveness of local search}\label{app:local-experiments}
Next, we study how our local search can improve the quality of the solutions from our block CD method. To this end, we run our approximate framework, with and without local search. Then, we compare the quality of the solutions, as well as the runtime of methods. Here, we set $p=250,s^*=10$ and $SNR=10$. In this section, we study Zbar+L2.

We report the results for these experiments in Figure~\ref{fig:ls-new}. In particular, we report the RMSE on a test set of size $n=300$, which is drawn independently from the same distribution as the training data. The value of RMSE is normalized by the RMSE of the $L_0 L_2$ estimator. We also report the runtime normalized by the runtime of $L_0 L_2$. We consider performance with and without local search (i.e., just block CD method). 

Figure~\ref{fig:ls-new} shows that in all cases, using local search reduces the test RMSE, compared to not using it. In fact, for large $n$, the Zbar solutions without local search can perform worse than the $L_0 L_2$ with local search. However, when applying local search, our Zbar estimator always outperforms the $L_0 L_2$. Importantly, local search does not increase total runtime by much. Together, these results show that local search quickly improves the quality of the solutions using our approximate algorithms.

\begin{figure*}[t!]
	\centering
	\begin{tabular}{cccc}
		\multicolumn{2}{c}{$\log(\text{RMSE}/\text{RMSE}_{L_0L_2})$} & \multicolumn{2}{c}{$\log(\text{Time}/\text{Time}_{L_0L_2})$}\\
		$\rho = 0.4$ & $\rho=0.5$ & $\rho = 0.4$ & $\rho=0.5$ \\
		\includegraphics[width=0.23\linewidth]{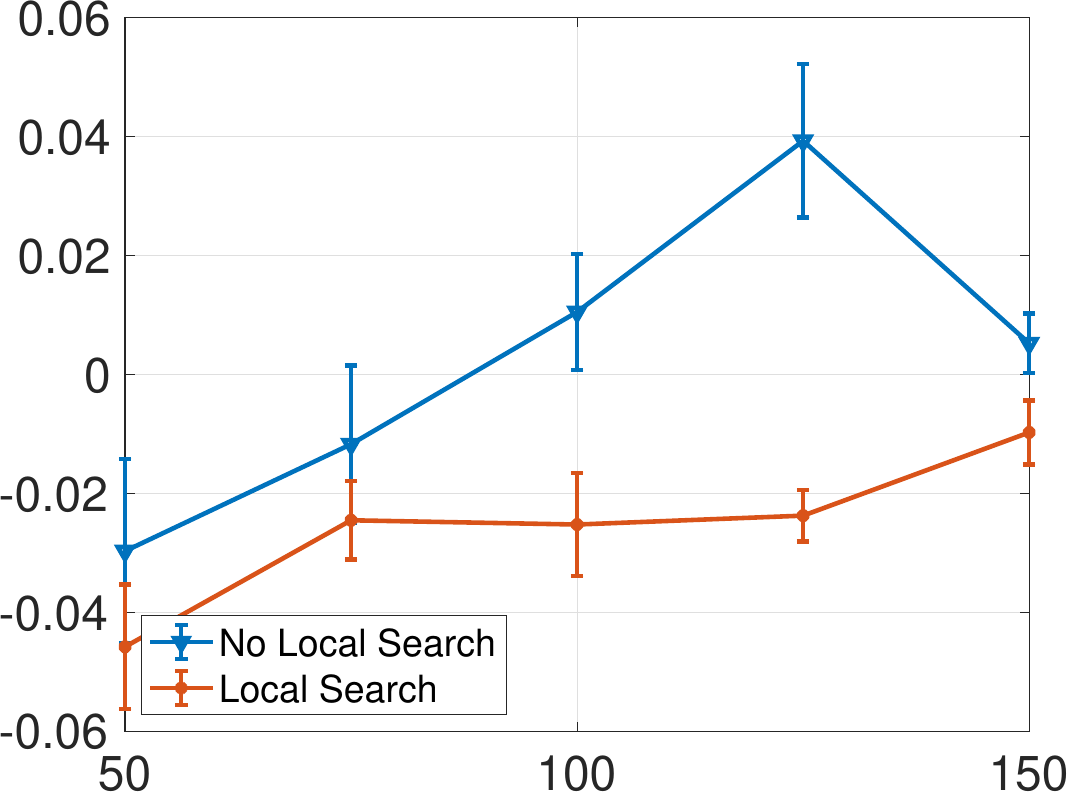}& 
		\includegraphics[width=0.23\linewidth]{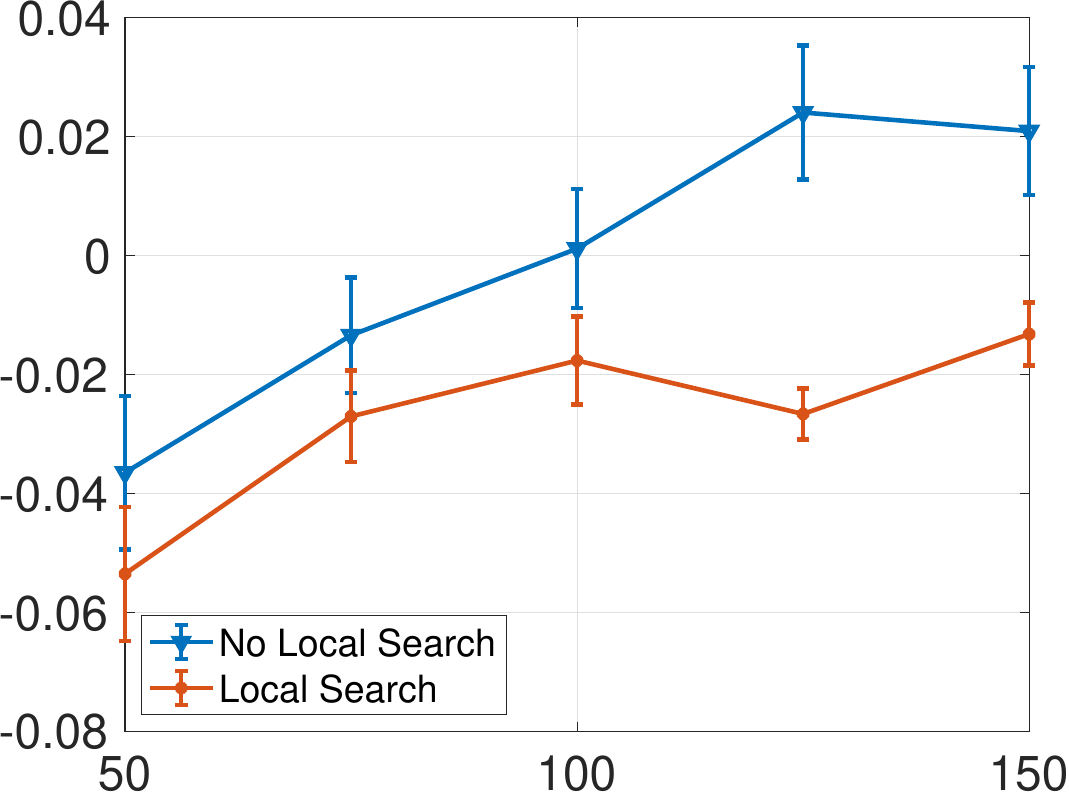} & \includegraphics[width=0.23\linewidth]{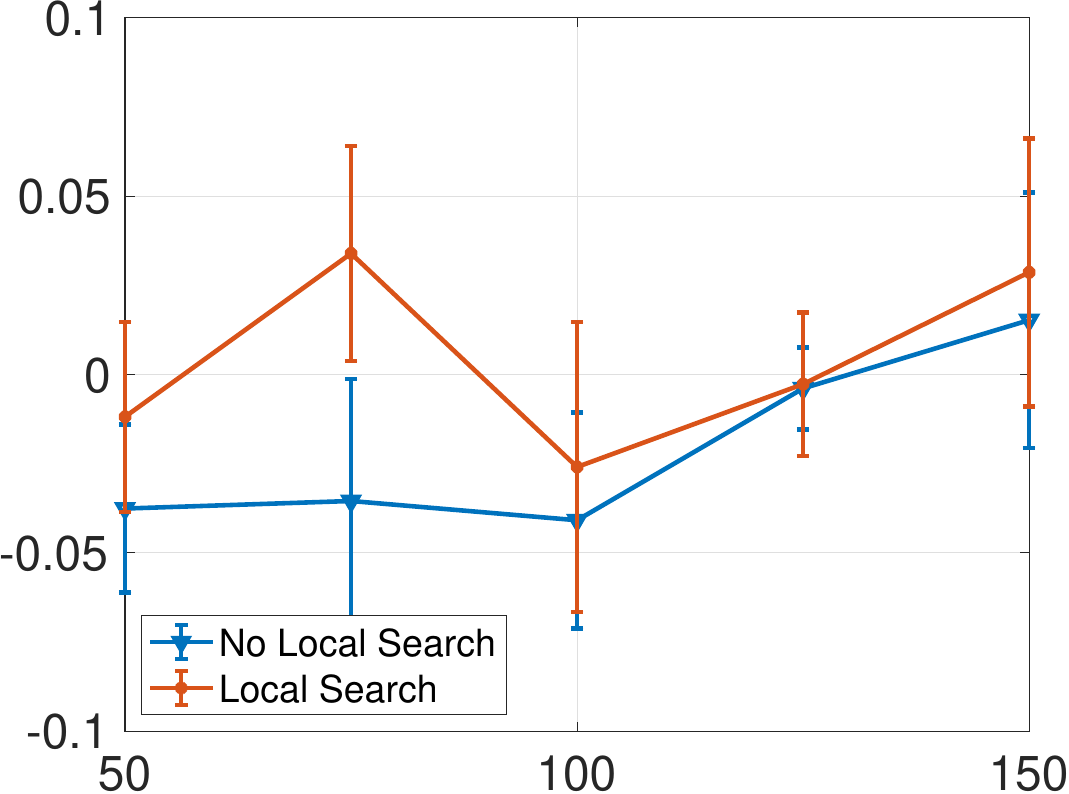}& 
		\includegraphics[width=0.23\linewidth]{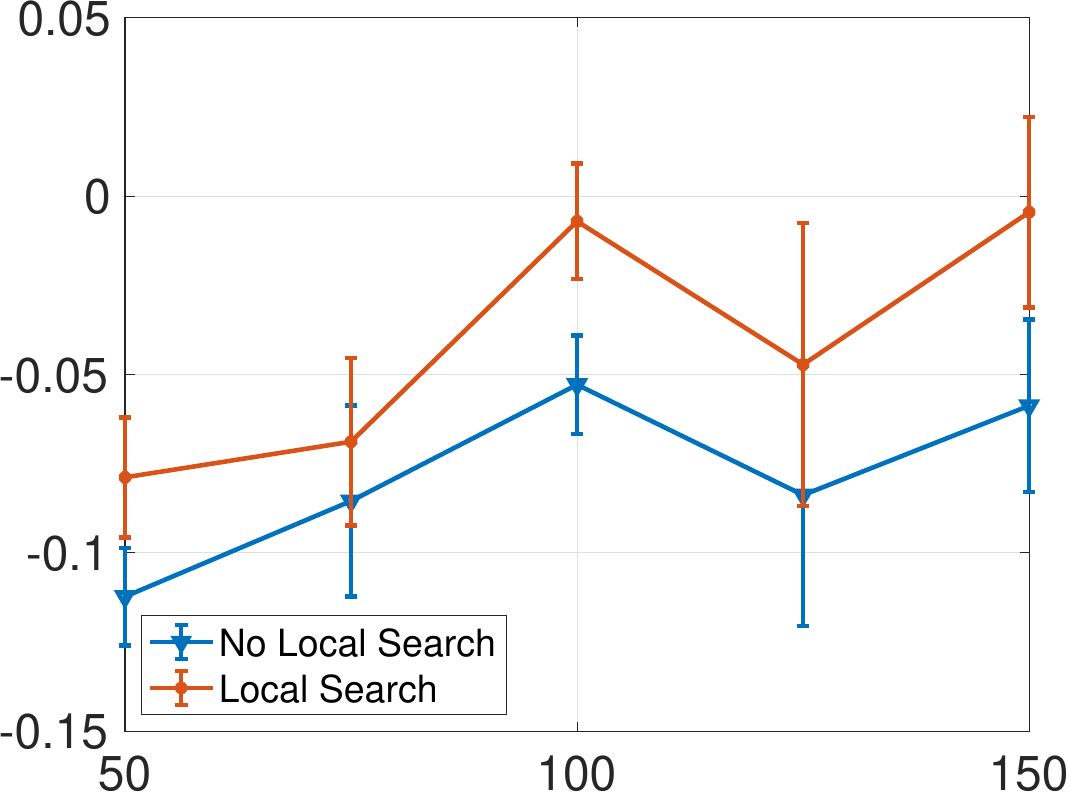} \\
		$n$ & $n$ & $n$ & $n$
	\end{tabular}
	
	\caption{\footnotesize Local search improves solution quality without substantially increasing algorithm runtime. Two left panels compare out-of-sample prediction performance (RMSE) of solutions delivered with (local search and block CD) and without (block CD only) local search at two levels of covariate correlation levels. These show that adding local search reduces the RMSE. Two right panels compare the runtime of the algorithm with and without local search. Local search does not substantially increase runtime.}
	\label{fig:ls-new}
\end{figure*}

\section{Real Data Applications} \label{data_apps}

We explore the performance of our methods in two multi-task biology settings to showcase different properties of the proposed estimators. First we explore performance in a chemometrics neuroscience application to demonstrate the method's suitability to compression settings. This application includes a collection of datasets, where each dataset is treated as a separate task. Each task was collected with a different recording electrode. Exploratory analyses suggested that the support heterogeneity regularization provided by the Zbar penalty could improve prediction performance in these data because regression coefficient estimates fit separately to each task differed substantially in both their support and values. In the second application, we explore a ``multi-label'' setting where the design matrix is fixed but the outcome differs across tasks (i.e., a multivariate outcome). Interpretability of model coefficients is of key importance in this application because the covariates are expression levels of genes in a network.
\subsection{Neuroscience Application}\label{fscv_modeling}
\paragraph{Dataset Background}
Studying neurotransmitters (e.g., dopamine), that serve as chemical messengers between brain cells (neurons) is critical for developing treatments for neurological diseases. Recently, the application of Fast Scan Cyclic Voltammetry (FSCV) has been used to study neurotransmitter levels in humans \citep{Kishida, Kishida2011PLoSone}. The implementation of FSCV in humans relies on prediction methods to estimate neurotransmitter concentration based upon the raw electrical current measurements recorded by electrodes (a high dimensional time series signal). This vector of current measurements at a given time point can be used as covariates to model the concentration of a neurotransmitter. In vitro datasets are generated to serve as training sets because the true concentrations, the outcome, are known (i.e., the data are labelled). The trained models are then used to make predictions of neurochemical concentration in the brain. In practice, each in vitro dataset is generated on a different electrode, which we treat as a task here because signals of each electrode differ in the marginal distribution of the covariates and in the conditional distribution of the outcome given the covariates \citep{Loewinger, Bang, Kishida, Moran}. Given the high dimensional nature of the recordings, researchers typically apply regularized linear models \citep{Kishida}. Importantly, estimates from sparse linear models fit on each task separately ($L_0 L_2$) exhibit considerable heterogeneity in both the coefficient values and support as can be seen in Figure \ref{fig:suppHet_fig}. For these reasons, we hypothesized that multi-task methods that employ regularizers that encourage models to share information through the $\boldsymbol{\beta}_k$ values (e.g., the Bbar method, group penalties) would perform worse than methods that share information through the supports, $\B{z}_k$ (i.e., the Zbar methods). We provide a more detailed description of this application in Supplement~\ref{neuro_background}. 

We explored the performance of our methods at different values of $n_k$ and $p$. This is motivated by the fact that FSCV labs have reported fitting models with only a subset of the covariates \citep{randomBurst} as the covariate values are on the same scale and units, and contain redundant information. This allows us to characterize the performance of our methods at different aspect ratios of $n_k/p$.
\paragraph{Modeling}\label{fscv_modeling_paragraph} 
In order to assess out-of-sample prediction performance as a function of $n_k$ and $p$, we repeated the following 100 times. For each replicate, we randomly selected with uniform probability $K = 4$ out of the total 15 datasets (tasks) to train models on. For each set of tasks, we split the datasets into training and testing sets and fit models on each task's training data. We used the fit for task $k$, $\hat{\boldsymbol{\beta}}_k$, to make predictions on the test set of task $k$ and estimated performance as described above. 
Since we intend to obtain sparse solutions from our methods, we consider solutions from benchmark and $\ell_0$ methods having similar sparsity levels. Specifically, we tuned the hyperparameters of the benchmark group penalty methods (e.g., gel, cMCP, SGL) to the best cross-validated values that produced solutions of cardinality no greater than $s$. This is because if no cardinality constraints were imposed on the benchmark methods, cross-validation tended to select regularization hyperparameter values that resulted in coefficient estimates with large support sizes. As expected, across all methods, dense solutions tended to yield better prediction performance than sparse estimates. 
We provide additional details about dataset pre-processing in Supplement~\ref{neuro_modeling}. 
\paragraph{Results}
We first evaluated the hypothesis that sparse regression models fit separately to each task exhibited support heterogeneity. To characterize this observed pattern, we tuned and fit an $L_0 L_2$ on four of the 15 available tasks (in vitro datasets) with sparsity $s = 50$ and plotted the coefficient estimates in Figure \ref{fig:suppHet_fig}. We only inspected one set of tasks to plot in order to avoid biases. The support of the estimates differed markedly between tasks and for some covariates, coefficient estimates were positive for some tasks and negative for others. The values of the nonzero coefficients varied so greatly across tasks that we had to use a log-transformation to visualize them all on the same scale. We also tuned and fit Zbar+L2, gel and Bbar models to visualize the impact of the different methods on coefficient estimate values and supports (Figure \ref{fig:suppHet_fig}). These methods estimated coefficients with very different levels of support heterogeneity. To summarize this, we calculated a measure of support homogeneity, which we defined as $\binom{K}{2}^{-1} p^{-1} \sum_{j=1}^p \sum_{k=1}^K \sum_{l < k} \hat{z}_{k,j} \hat{z}_{l,j}$ (more details are provided in Supplement~\ref{performance}). The support homogeneity measures for the Zbar+L2, gel, and Bbar were over 3.5, 21, and 4 times that of the $L_0 L_2$, respectively (higher values indicate greater support \textit{homogeneity}).  We include a figure that shows the support across the full vector of covariates in Supplemental Figure~\ref{fig:SHR_supp}. 

\begin{figure*}[!t]\centerline{\includegraphics[width=0.65\linewidth]{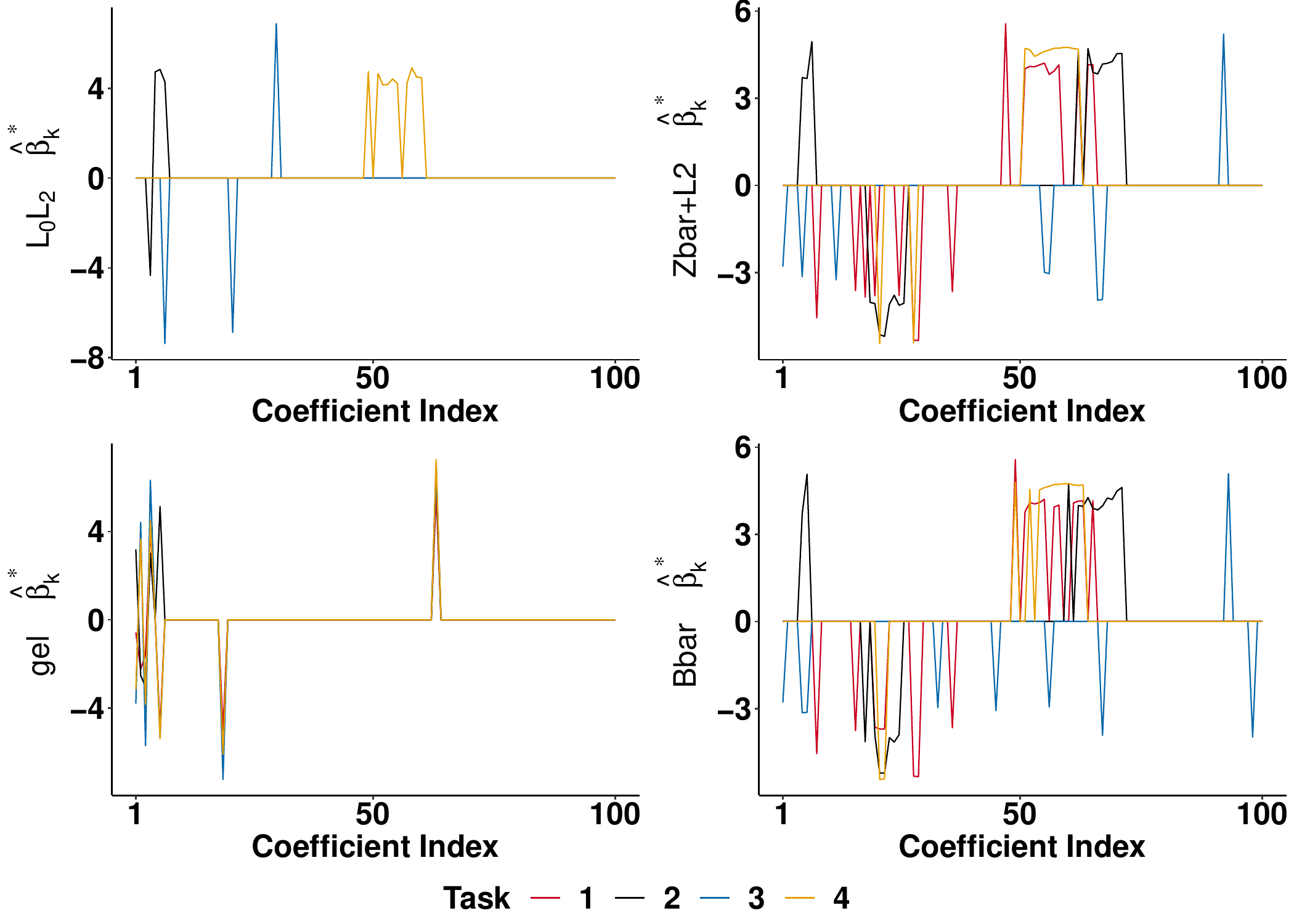}}
	\caption{\footnotesize Coefficient estimates with $L_0 L_2$ [Top Left], Zbar+L2 [Top Right], gel [Bottom Left] and Bbar [Bottom Right]. To show the solutions on the same scale across tasks, we plot the nonzero elements of $\hat{\boldsymbol{\beta}}_k^* = sgn(\hat{\boldsymbol{\beta}}_k) \odot \log(|\hat{\boldsymbol{\beta}}_k|)$.}
	\label{fig:suppHet_fig}
\end{figure*}

\begin{figure*}[!t]
	\centering
	\begin{subfigure}[t]{0.99\textwidth}
		\centering
		\includegraphics[width=0.8 \linewidth]{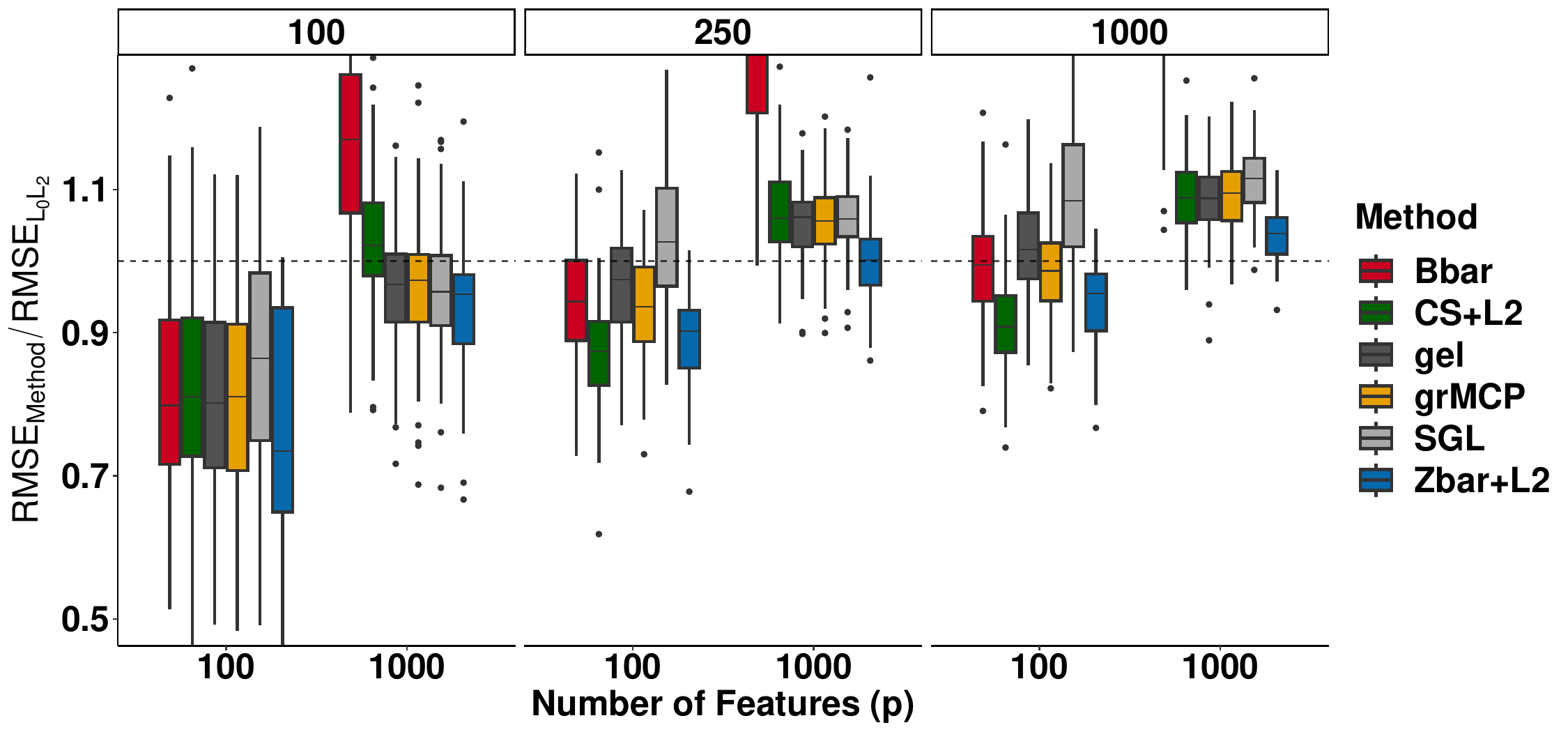}
	\end{subfigure}
	\caption{\footnotesize {Neuroscience application results.} Out-of-sample prediction performance (RMSE) averaged across random sets of tasks ($K=4$) for different $p$ (displayed on the x-axis) and $n_k$ (displayed on the panels).} 
	\label{fig:fscvRMSE}
\end{figure*}

We next study the out-of-sample prediction performance of our methods as a function of $n_k$ and $p$ averaged across tasks. We show boxplots in Figure \ref{fig:fscvRMSE} for a subset of methods at $s=25$ and results for other methods and sparsity levels in Supplemental Figures~\ref{fig:fscvRMSE_supp1} and \ref{fig:fscvRMSE_supp2}. The relative benefit of applying MTL methods, over separate $L_0 L_2$ regressions, tended to be greater at lower sample sizes, $n_k$, and higher sparsity values, $s$. 
The relative performances of the MTL methods (i.e., relative to each other) appear to be fairly consistent across the different sparsity levels. 

The results in Figure \ref{fig:fscvRMSE} illustrate the benefits of sharing information through the supports. Indeed, the Zbar+L2 remained competitive with or outperformed all methods that borrow strength directly through the $\boldsymbol{\beta}_k$ values (i.e., Bbar, gel, SGL). For example, the Bbar and group penalty methods performed well only when both $p$ and $n_k$ were low, but their relative performance degraded as those parameters increased. In fact, when $p$ was high, the Bbar method performed substantially worse than all other methods, even the $L_0 L_2$. Although the CS+L2 remained competitive with the Zbar+L2 for low $p$, the Zbar+L2 outperformed all other $\ell_0$ methods in most settings, emphasizing that the gains from our estimator were not only a function of the $\ell_0$-constraint in our proposed estimator. In sum, these results suggests that in settings with high between-task heterogeneity in coefficient values and supports (as suggested by Figure \ref{fig:suppHet_fig}) prediction performance may be improved more by borrowing strength across tasks through the coefficient \textit{supports}, rather than through the coefficient \textit{values}.

\subsection{Cancer Genomics Application}
\paragraph{Dataset Background}
We next explore the performance of our method in a cancer biology application to highlight our proposed method's strength in high-dimensional genomics settings where interpretability is critical. This also serves as an important point of comparison because unlike the neuroscience dataset, one might expect common support methods to perform well in this application given that the tasks are defined by different gene expression levels from the same network of highly correlated genes. We build off of work by \cite{deVito} focused on a multi-study dataset of breast cancer gene expression measurements. In cancer biology, the elucidation of regulatory gene network structure is critical for the characterization of the biology of disease subtypes and potentially the development of targeted therapeutics. As such, we sought to jointly predict the expression levels of a group of highly co-expressed genes (i.e., as the outcome variables) using the expression levels of all other genes in the dataset as covariates. In view of the importance of the Estrogen Receptor (ESR1) in breast cancer development and therapy, we selected the expression levels of ESR1 and co-regulated genes FOXA1, TBC1D9, GATA3, MLPH, CA12, XBP1 and KRT18 to serve as the outcome for each task as these were found to be highly inter-connected in \cite{deVito}.
Specifically, we specified the genes that form our collection of task-specific outcome variables based upon their central location in clusters identified in an estimated gene co-expression network \citep{deVito}. 
We selected the final set of genes before fitting any models to avoid biases. 
\paragraph{Modeling}
The dataset is comprised of data from 18 studies. In order to assess out-of-sample prediction performance, we implemented a hold-one-study-out testing procedure, whereby we iteratively fit models for each of the tasks on 17 of the datasets and tested prediction accuracy of those tasks jointly on the held-out study. To explore prediction performance as a function of $K$ and $p$, we ran 100 replicates of the following experiments: we randomly selected, with uniform probability, 1) $K$ of the eight possible outcome variables to serve as tasks, 2) $p$ of the total covariates, and 3) one of the studies to serve as a held-out test set. For each replicate, we fit all models and compared performance. We considered all combinations of $K \in \{2,4,6\}$ and $p \in \{100,200,1000\}$. We constrained all methods to provide estimates with sparsity level $s \leq 50$ and tuned $\ell_0$ methods with local search. 
All other hyperparameter selection details are described in Section~\ref{fscv_modeling_paragraph}. We show additional results for other methods and sparsity levels in Supplemental Figures~\ref{fig:cancerRMSE_supp} and \ref{fig:cancerRMSE_supp1}. As the original dataset had over 10,000 covariates, we first reduced the dimension of the covariates to less than 2,000 through using only those covariates associated with coefficients selected in a Group Lasso model fit to the full set of covariates. We then fit models on a set of covariates randomly selected from this smaller dataset.

\paragraph{Results}
\begin{figure*}[t!]
	\centerline{
		\includegraphics[width=0.8 \linewidth]{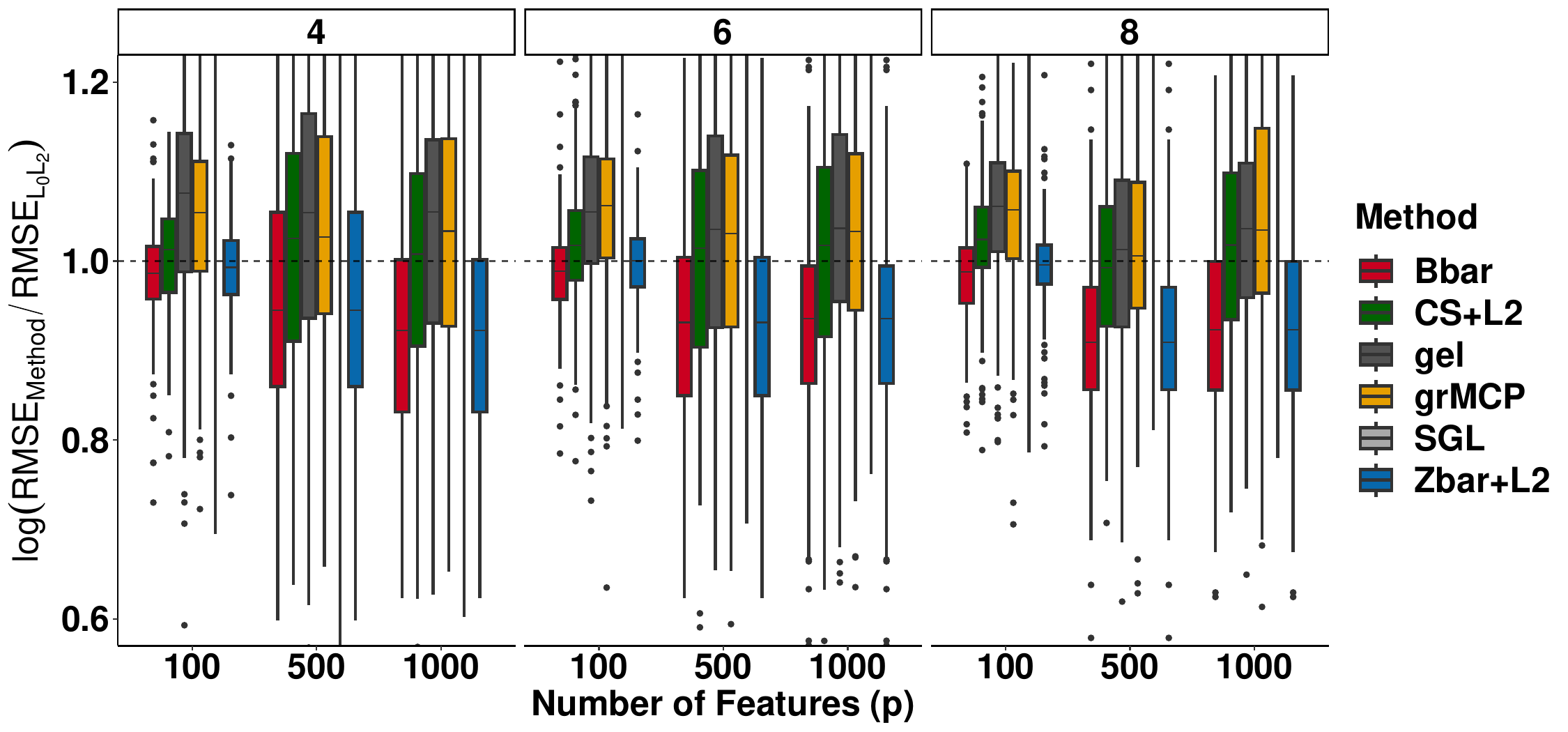}
	}
	\caption{ \footnotesize {Cancer application results.} Hold-one-study-out prediction performance of $\ell_0$ methods averaged across tasks for different $K$. The performance (RMSE) is presented relative to the performance of the $L_0 L_2$ method.} 
	\label{fig:cancerRMSE}
\end{figure*}

We show boxplots of relative prediction performance for a subset of methods in Figure \ref{fig:cancerRMSE} and for additional methods in Supplemental Figure \ref{fig:cancerRMSE_supp2}. Similar to the neuroscience data, the relative benefit of applying MTL methods, over separate $L_0 L_2$ regressions, tended to be greater at higher sparsity values, $s$. The relative performances of the MTL methods (i.e., relative to each other) were, however, fairly consistent across the different sparsity levels. We observed that the Zbar+L2 and Bbar methods consistently outperformed the benchmark methods (e.g., gel, SGL), as well as the common support $\ell_0$ methods, especially for higher $K$ and $p$. However, unlike in the neuroscience application, the Bbar method performed comparably to the Zbar methods. Interestingly, the relative performance of the Zbar+L2 and Bbar methods \textit{improved} as a function of~$p$. Conversely, the prediction accuracy (RMSE) of both the common support and benchmark group penalties \textit{decreased} with higher $p$.

\subsection{Application Comparison}

The comparison between the Zbar and the Bbar methods in these two applications sheds light on when a continuous penalty that shrinks the $\boldsymbol{\beta}_k$ \textit{values} together can perform as well as a combinatorial penalty that shrinks the $\boldsymbol{\beta}_k$ \textit{supports} together. While the two datasets differ in many ways, inspection of simple measures of the distribution of the $\hat{{\beta}}_{k,j}$ across tasks, $k$, reveal important insights. For each application, we fit task-specific Ridge regressions and compared rug plots of the $\hat{\beta}_{k,j}$ across tasks for the 10 most ``important'' covariates (Figure~\ref{fig:rugPlot}). Specifically, we selected the 10 $j$'s associated with the greatest average magnitude of coefficient estimates ($\frac{1}{K} \sum_{k=1}^K | \hat{\beta}_{k,j}|$). Importantly, the estimates in the cancer plots are all positive, have similar values and are relatively far from zero, potentially reflecting low heterogeneity in the coefficient signs, values and supports. The Bbar penalty might therefore be expected to perform well in this application. The neuroscience dataset estimates, on the other hand, exhibit considerable heterogeneity in coefficient values and signs. Moreover, for some covariates, a subset of the task estimates has large magnitudes, while the other task estimates are close to 0, potentially reflecting support heterogeneity. In cases like these, a penalty that encourages information sharing through coefficient supports may be preferable to penalties that share information by shrinking the $\boldsymbol{\beta}_k$ values together. These patterns may partially explain why the Zbar+L2 method strongly outperformed the Bbar method in the neuroscience application, but performed similarly to the Bbar method in the cancer genomics application.

Interestingly, the neuroscience coefficient estimates also tended to fall into clusters of similar patterns (e.g., coefficients 1, 5, and 10 comprise one cluster) on the rug plot. This may be because the covariates in the FSCV dataset exhibit a structure characteristic of functional covariates such as a natural ordering and similar levels of correlation with the outcome across adjacent covariates (see~\cite{Loewinger}). More extensive rug plots in Supplement~\ref{rugPlots} confirm that the patterns described are representative across a large set of covariates. 

Finally, we observed that the relative prediction performances were more variable in the cancer application than in the neuroscience application. This may be because each replicate displayed in Figure \ref{fig:cancerRMSE} was generated from fitting models using a different random subset of genes as task outcomes (i.e., a multi-label problem), whereas each replicate in the neuroscience data (Figure \ref{fig:fscvRMSE}) was calculated from fitting models on a different subset of electrodes as tasks, but using data with the same outcome type (i.e., dopamine concentration). The relative performance of any given multi-task method may be more variable across different \textit{subsets} of genes because each \textit{individual} gene may be better predicted by different methods. The variability in relative performance of most methods may be lower across different electrode subsets in the neuroscience application because the outcome was the same (dopamine concentration) in all tasks.

\begin{figure*}[t] 
	\centering
	\begin{subfigure}[t]{0.49\textwidth}
		\centering		\includegraphics[width=0.9\linewidth]{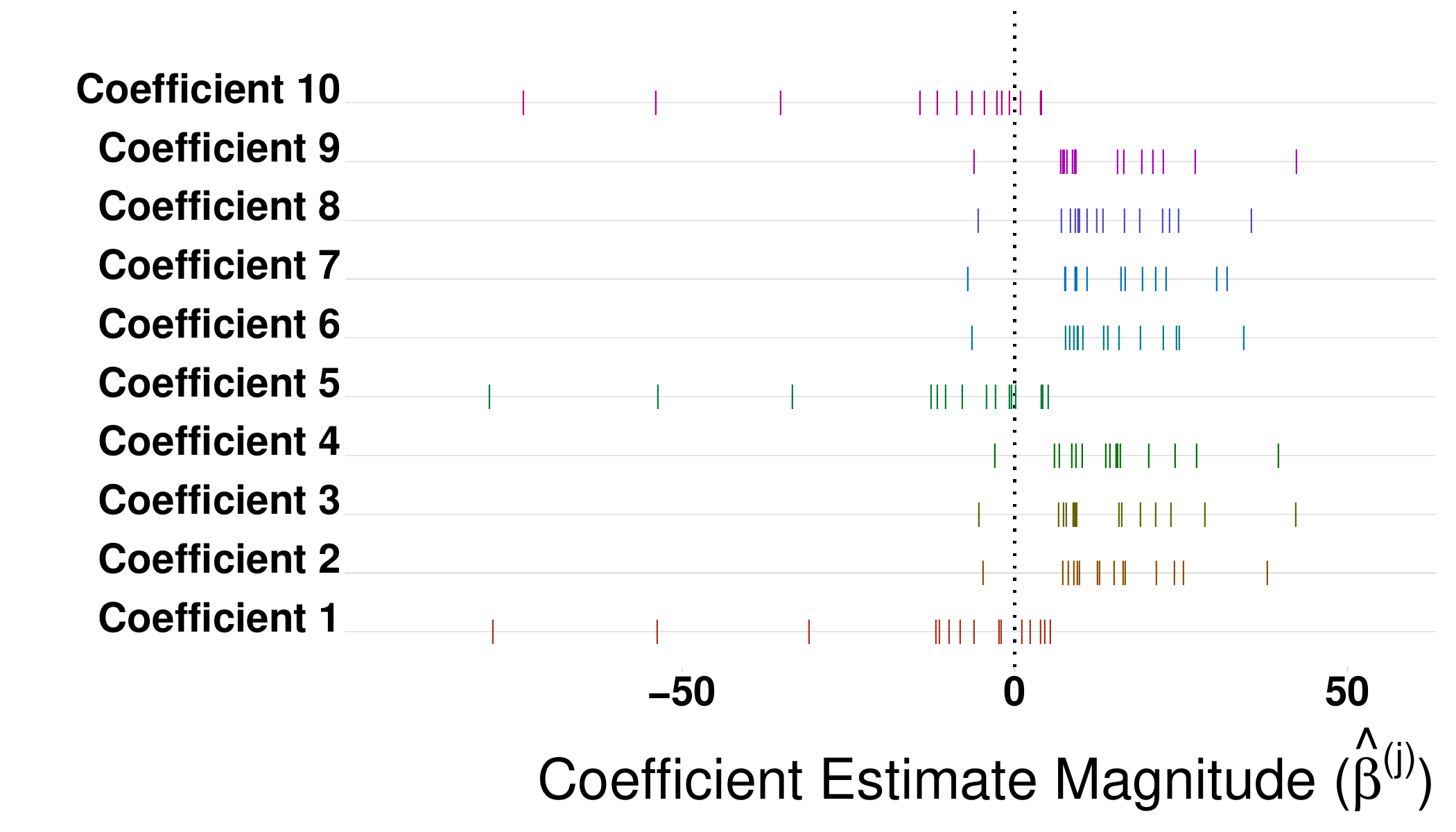}
	\end{subfigure}
	\hfill
	\begin{subfigure}[t]{0.49\textwidth}
		\centering
		\includegraphics[width=0.9\linewidth]{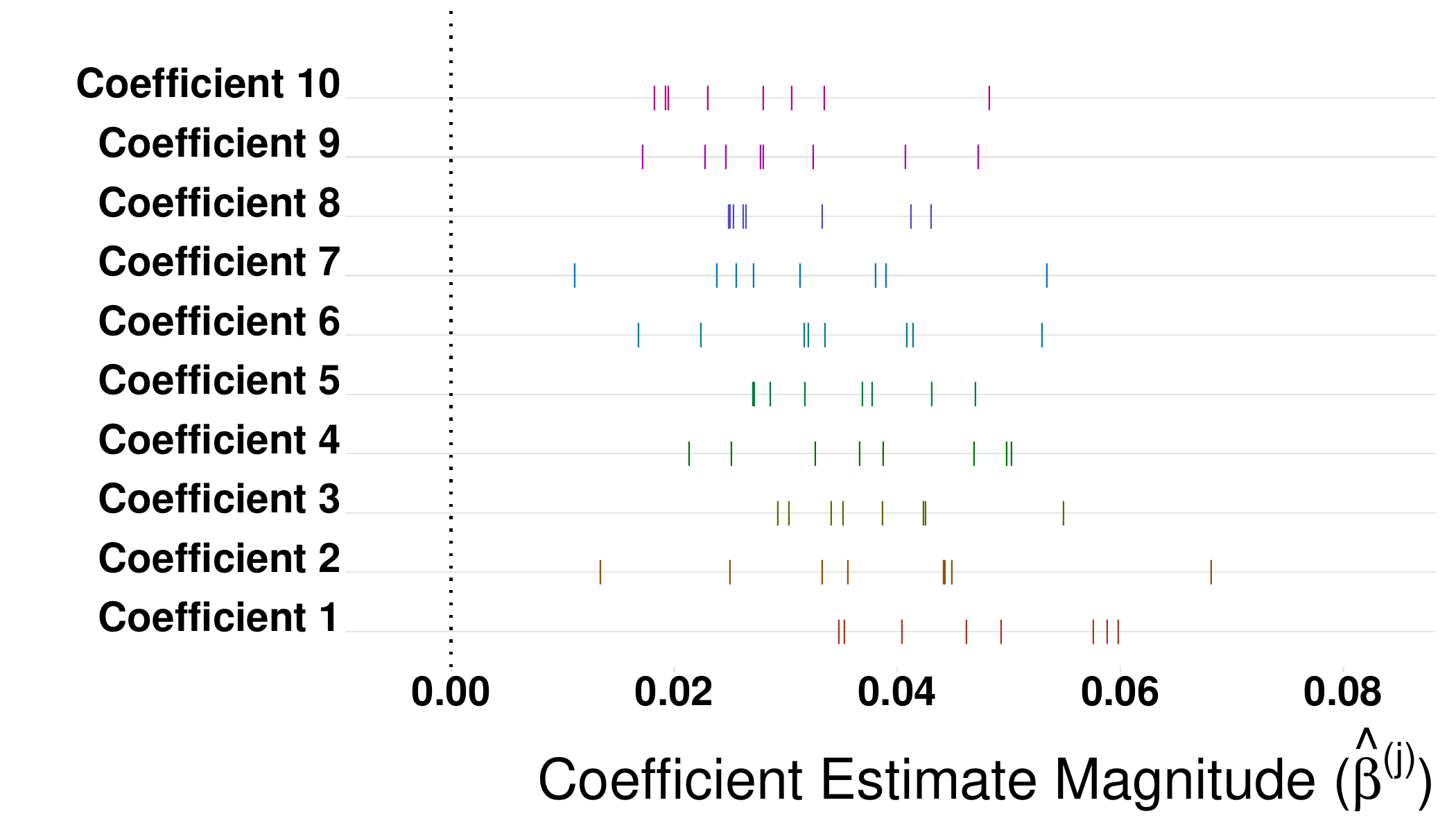}
		\label{fig:breastCancer_rugplot}
	\end{subfigure}
	\caption{\footnotesize Rug plots displaying $\hat{\beta}_{k,j}$ for the 10 $j$'s with the greatest average magnitude across the $K$ task-specific Ridge regressions ($\frac{1}{K} \sum_{k=1}^K | \hat{\beta}_{k,j}|$). For Coefficient $j$, each mark on the horizontal line is one of $K$ task-specific empirical estimates of the $\hat{\boldsymbol{\beta}}_{k,j}$. [Left] Neuroscience dataset. [Right] Breast Cancer dataset.}.
	\label{fig:rugPlot}
\end{figure*}

\section{Discussion}

We propose a class of estimators that penalize the heterogeneity across tasks in the support vectors $\boldsymbol{z}_k$. These methods allow for information to be shared directly through the support of regression parameters without relying on methods that shrink coefficient values together, an approach that can perform poorly when the coefficient values differ substantially between tasks. Our theoretical analysis shows that the Zbar penalty shrinks the supports together across tasks, improves support recovery accuracy when the true regression coefficients have (near) common support, and maintains statistically optimal prediction performance. We developed algorithms based on block CD and local search to obtain high-quality solutions to our estimator. We also developed a tailored exact solver for our estimator based on outer approximation. Our experiments on synthetic and real data show the usefulness of our method. Based on these numerical results, we recommend the Zbar+L2 over the Zbar+Bbar. Indeed, the Zbar+L2 can encourage models to borrow strength through the coefficient supports, and can improve prediction performance regardless of the distribution of the nonzero $\B{\beta}_k$ \textit{values}. Formally, as seen in Theorem~\ref{thm1} and Corollary~\ref{cor-general}, if the true regression coefficient values differ substantially between tasks, using a large Bbar penalty (i.e. large $\lambda$) can introduce additional prediction error. Nevertheless, our framework is flexible and allows users to combine multiple penalties in cases where borrowing strength across both coefficient \textit{supports} and \textit{values} is useful.
Taken together, the present work promises to be a useful methodological contribution to the MTL literature.

\section{Software and Reproducibility}
\label{sec5}

Code and instructions to reproduce analyses, figures and tables are available at: \url{https://github.com/gloewing/sMTL}. This contains code to tune and fit penalized linear regression problems with all the methods assessed here. Block CD and local search algorithms are coded in $\texttt{Julia}$ and are called from $\texttt{R}$ where simulations and data analyses are implemented. We will release the $\texttt{R}$ package, $\texttt{sMTL}$, that implements all sparse regression methods implemented here for multi-task learning, domain generalization, multi-label learning as well as standard sparse regression for single dataset/task settings.

\section*{Acknowledgments}
GCL conducted the majority of the work for this paper while a PhD student at the Harvard University Department of Biostatistics.
GCL was supported by the NIH-NIDA F31DA052153. GP was supported by NSF-DMS grants 1810829 and 2113707. RM acknowledges partial research support from NSF-IIS-1718258, and ONR N000142112841. KTK acknowledges support from NIH-NIMH R01MH121099, 
NIH-NIDA R01DA048096, 
NIH-NIMH R01MH124115, 
NIH 5KL2TR001420.
{\it Conflict of Interest}: None declared.

\bibliographystyle{biorefs}
\singlespacing
% \bibliography{refs}
{\footnotesize \bibliography{refs}}
%\spacingset{1.9}

\newpage

\appendix

\numberwithin{equation}{section}
\numberwithin{lemma}{section}
\numberwithin{proposition}{section}
\numberwithin{figure}{section}
\numberwithin{table}{section}

\section{Exploring the convex relaxations of Zbar}\label{app:convexrelax}
As we discussed in Section~\ref{mip-app}, our estimator in~\eqref{suppHet} can be written as a MIP, as in Problem~\eqref{suppHet_mip}. We can obtain a convex (interval) relaxation of Problem~\eqref{suppHet_mip} by relaxing all binary variables $z_{k,j}$ to take values in the interval $[0,1]$, instead of taking binary values.
An interesting question is that how a convex relaxation of an estimator with Zbar penalty compares to the other estimators, specifically, with convex penalties. Particularly, one might ask if a convex relaxation of the Zbar penalty results in a penalty that resembles the Bbar penalty. To answer this question, let us study a special case of our estimator~\eqref{suppHet_mip} to gain more insights. Particularly, let us assume $n=p=1$ and $K=2$ and $\alpha=\lambda=0$ for simplicity. Also, to make the presentation of the results easier, we consider an $\ell_0$-penalized version of~\eqref{suppHet_mip}, instead of a cardinality constrained version:
\begin{align}\label{convex-1}
	\min_{\B\beta,\B z}&~~~ (y_1-\beta_1x_1)^2 + (y_2-\beta_2x_2)^2 + \frac{\delta}{2} (z_1-z_2)^2 + \lambda_0(z_1+z_2)\\
	\text{s.t.} & ~~~ \B z\in\{0,1\}^2, |\B\beta|\leq M\B z\nonumber
\end{align}
for some $\lambda_0>0$. Note that problem~\eqref{convex-1} can be written as
\begin{align}
	\min_{\B\beta} (y_1-\beta_1x_1)^2 + (y_2-\beta_2x_2)^2 + \psi_0(\B\beta)
\end{align}
where
\begin{align}\label{psi0}
	\psi_0(\B\beta) = \min_{\B z}  \frac{\delta}{2} (z_1-z_2)^2 + \lambda_0(z_1+z_2)~~
	\text{s.t.}  ~~ \B z\in\{0,1\}^2, |\B\beta|\leq M\B z.
\end{align}
Note that a convex (interval) relaxation of~\eqref{psi0} is given by relaxing the binary constraints $\B z\in\{0,1\}^2$ into interval constraints, $\B z\in[0,1]^2$. This results in
\begin{align}\label{psi1}
	\psi_1(\B\beta) = \min_{\B z}  \frac{\delta}{2} (z_1-z_2)^2 + \lambda_0(z_1+z_2)~~
	\text{s.t.}  ~~ \B z\in[0,1]^2, |\B\beta|\leq M\B z.
\end{align}
Suppose $M\geq \beta_1,\beta_2\geq 0$. Looking at~\eqref{psi1}, one might guess that at optimality of~\eqref{psi1}, we have $\B z\overset{?}{=}\B\beta/M$, which implies $\psi_1(\B\beta)\overset{?}{=}\phi(\B\beta)$ where
\begin{align}\label{phi-def}
	\phi(\B\beta) =  \frac{\delta}{2M^2} (\beta_1-\beta_2)^2 + \frac{\lambda_0}{M}(\beta_1+\beta_2).
\end{align}
If this hypothesis holds, then a convex relaxation of Zbar leads to the Bbar penalty, which is given by $\delta(\beta_1-\beta_2)^2/(2M^2)$ in~\eqref{phi-def}. However, below we show that this is not necessarily the case. Consider Proposition~\ref{prop-convex-relax}.
\begin{proposition}\label{prop-convex-relax}
	Suppose $M\geq\beta_1\geq\beta_2\geq 0$. Then, at optimality of~\eqref{psi1} we have
	$$z^*_1=\frac{\beta_1}{M},~~z^*_2=\max\left\{\frac{\beta_1}{M}-\frac{\lambda_0}{\delta},\frac{\beta_2}{M}\right\}.$$
	Moreover, if $0\leq \lambda_0\leq \delta(\beta_1-\beta_2)/M$,
	$$\psi_1(\B\beta)=-\frac{\lambda_0^2}{2\delta} + \frac{2\lambda_0 \beta_1}{M}.$$
\end{proposition}

Proof of Proposition~\ref{prop-convex-relax} is presented at the end of this section. We see that Proposition~\ref{prop-convex-relax} shows that in general, $\psi_1(\B\beta)\neq \phi(\B\beta)$. As an example, if $\beta_2=0$ and $0\leq \lambda_0\leq \delta\beta_1/M$, 
$$\psi_1(\B\beta) =   \frac{2\lambda_0 \beta_1}{M}-\frac{\lambda_0^2}{2\delta},~~\phi(\B\beta)=\frac{\delta\beta_1^2}{2M^2}+\frac{\lambda_0\beta_1}{M}.$$
As an example, let $\beta_1=M=1,\beta_2=0,\lambda_0=\delta/2$ which implies
$$\psi_1(\B\beta)=\frac{7\delta}{8}\neq \delta = \phi(\B\beta).$$
This shows that fundamentally, Zbar and Bbar penalties are different, and we are not aware of any work that discusses the convex relaxation of the Zbar penalty.

\begin{proof}[\textbf{Proof of Proposition~\ref{prop-convex-relax}}]
	First, we show that $z_2^*\leq z_1^*= \beta_1/M$. Suppose $\B z$ is feasible for~\eqref{psi1}, with $z_2>z_1$. Then, $(z_1,z_1)$ is also feasible for~\eqref{psi1} and we have 
	$$\frac{\delta}{2} (z_1-z_2)^2 + \lambda_0(z_1+z_2)\geq \frac{\delta}{2} (z_1-z_1)^2 + \lambda_0(z_1+z_1) = 2\lambda_0 z_1$$
	showing we have $z_2^*\leq z_1^*$. Next, Suppose $\B z$ is feasible for~\eqref{psi1} with $z_1>\beta_1/M$ and $z_1\geq z_2$. Then,  $(\max(\beta_1/M,z_2),z_2)$ is feasible for~\eqref{psi1} and 
	$$\frac{\delta}{2} (z_1-z_2)^2 + \lambda_0(z_1+z_2)> \frac{\delta}{2} (\max(\beta_1/M,z_2)-z_2)^2 + \lambda_0(\max(\beta_1/M,z_2)+z_2).$$
	Therefore, we must have $z_1^*=\beta_1/M$, which implies 
	\begin{equation}\label{prop-a-1-helper-1}
		z_2^*\in\argmin_{z_2} \frac{\delta}{2}(z_2-\beta_1/M)^2 + \lambda_0 z_2~~\text{s.t.}~~1\geq z_2\geq \beta_2/M.
	\end{equation}
	Setting the derivative of the objective in~\eqref{prop-a-1-helper-1} equal to zero, we get
	$$z_2=\frac{\beta_1}{M}-\frac{\lambda_0}{\delta}$$
	which results in $z_2^*$ as given in the proposition.

\end{proof}

\section{Outer Approximation Details}\label{app:oa-detail}
Here we discuss supplementary technical details pertaining to using outer approximation to solve Problem~\eqref{suppHet_mip}. We first consider the special case 
of $\lambda=0$ (see Section~\ref{app:oa-zbar}) which is the case presented in the main paper. We then present the general case $\lambda \geq 0$ (Section~\ref{app:oa-general}).

\subsection{Special Case of Zbar + L2}\label{app:oa-zbar}
Here, we discuss some properties of the functions $F_k(\cdot)$ defined in~\eqref{Freg2}. In particular, we show how subgradients of these functions can be calculated efficiently for sparse and binary $\B z_k$.
Proposition~\ref{subgradpers} shows that $F_k$ is convex and characterizes its subgradient. Before stating the proposition, for $k\in[K]$ and $\B z_k\in[0,1]^p$, let us define
\begin{align}\label{linregperstosolve}
	\bar{\B\beta}_{k}\in\argmin_{\B\beta_k} & ~ \frac{1}{n_k} \left\Vert \mathbf{y}_k-\bX_k\B\beta_k \right\Vert_2^2+\alpha\|\B\beta_k\|_2^2 \\  \text{s.t.}& ~~|\beta_{k,j}|\leq Mz_{k,j} ~\forall j\in[p].\nonumber
\end{align}
Moreover, let
\begin{equation}\label{alphak-def}
	\B\zeta_k=\frac{2}{n_k}\bX_k^T(\mathbf{y}_k-\bX_k\B\bar{\B\beta}_k)-2\alpha\bar{\B\beta}_k\in\R^p.
\end{equation}

\begin{proposition}\label{subgradpers}
	Let $F_{k}$ be as defined in~\eqref{Freg2}. The followings hold true:
	\begin{enumerate}
		\item (\textit{Convexity}) The function $\B z_k \mapsto F_k(\B z_k)$ on 
		$\B z_k \in [0,1]^p$ is convex.  
		\item (\textit{Subgradient}) 
		Let $\B\zeta_k$ be as defined in~\eqref{alphak-def}.
		The vector $\B g_k\in\mathbb{R}^p$ with $i$-th coordinate given by
		$$g_{k,i}= -M|\zeta_{k,i}|$$
		for $i\in[p]$, is a subgradient of $\B z_k \mapsto F_k(\B z_k)$ for $\B z_k\in\{0,1\}^p$.
	\end{enumerate}
\end{proposition}

The proof of Proposition~\ref{subgradpers} is presented at the end of this section. Note that $F_{k}$ in~\eqref{Freg2} is implicitly defined via the solution of a quadratic program (QP). For a feasible $\B z_k$ that has at most $s$ nonzero values,~\eqref{Freg2} is a QP with at most $s$-many variables. Therefore, calculating all values of $F_k$ for all tasks $k\in[K]$ requires solving $K$-many QPs with $s$-many variables.
This makes calculating $F_k$'s for (sparse) feasible binary $\B z_k$'s substantially faster compared to general dense $\B z_k$'s, which would require solving $K$-many QPs each with $p$-many variables. Moreover, we do not have a closed form expression for a subgradient of $F_{k}(\B z_k)$ (cf Proposition~\ref{subgradpers}), but the subgradient can be computed for $\B z_k \in \{0,1\}^p$ as a by-product of solving the QP in~\eqref{Freg2}.

\begin{proof}[\textbf{Proof of Proposition~\ref{subgradpers}}]
	{\bf Part 1)} The map appearing in the cost function of~\eqref{Freg2}, that is:
	\begin{equation}\label{Gconvex}
		(\B z_k,\B\xi_k,\B\beta_k) \mapsto H_k(\B z_k,\B\xi_k,\B\beta_k):=\frac{1}{n_k}\left\Vert \B\xi_{k} \right\Vert_2^2+\alpha \|\B\beta_k\|_2^2
	\end{equation}
	is jointly convex in $\B z_k,\B\xi_k,\B\beta_k$. 
	Let 
	\begin{equation*}
		\mathbb{S} = \bigg\{ (\B z_k,\B\xi_k,\B\beta_k): |\beta_{k,j}|\leq Mz_{k,j}~ \forall j\in[p],~
		\B\xi_{k} = \by_k-\bX_k\B\beta_k,~~\B z_z\in[0,1]^p\bigg\}\subseteq\R^{p}\times \R^{n}\times \R^{p}.
	\end{equation*}
	Note that $\mathbb{S}$ defined above is also convex. Let $\1_C(\cdot)$ denote the characteristic function of a set $C$, 
	$$\1_C(x) = \begin{cases} 0 &\text{if}~ x\in C \\
		\infty & \text{if}~ x\notin C. \end{cases} $$
	Note that if $C$ is a convex set, $\1_C(\cdot)$ is a convex function. With this notation in place, Problem~\eqref{Freg2} can be written as 
	\begin{equation}\label{Freg_chara}
		F_k(\B z_k)=\min_{\B\beta_k,\B\xi_k} ~~ H_k(\B z_k,\B\xi_k,\B\beta_k) +\1_{\mathbb{S}}(\B z_k,\B\xi_k,\B\beta_k)
	\end{equation}
	where $H_k$ is defined in~\eqref{Gconvex}. Based on our discussion above, the  function $(\B z_k,\B\xi_k,\B\beta_k) \mapsto H_k(\B z_k,\B\xi_k,\B\beta_k) +\1_{\mathbb{S}}(\B z_k,\B\xi_k,\B\beta_k)$ 
	is convex. As $F_k(\B z_k)$ is obtained after a marginal minimization of a jointly convex function over a convex set, the map $\B z_k \mapsto F_k(\B z_k)$ is convex on $\B z_k \in [0,1]^p$ \citep[Chapter 3]{boyd2004convex}. 
	
	\smallskip
	
	\noindent {\bf Part 2)} Note that as the objective function of \eqref{Freg2} is convex [see Part 1] and all its constraints are affine, strong duality holds for \eqref{Freg2} by  enhanced  Slater's condition (see \citet{boyd2004convex}, Section 5.2.3). 
	Next, we derive the dual of Problem~\eqref{Freg2}. 
	Considering Lagrange multipliers $\B\Lambda^{+},\B\Lambda^{-}\in\R^{p}$, $\B\gamma\in\R^{n_k}$   for problem constraints, the Lagrangian for this problem
	${\mathcal L}(\B\beta_k,\B\xi_k,\B\Lambda^{+},\B\Lambda^{-},\B\gamma)$ or ${\mathcal L}$ (for short),
	can be written as
	\begin{equation}\label{persconds1}
		\begin{aligned}
			\mathcal{L} = &   \frac{1}{n_k}\left\Vert\B\xi_k \right\Vert_2^2  +\alpha \|\B\beta_k\|_2^2+   \langle\B\Lambda^{+},\B\beta_k - M\B z_k\rangle -\langle\B\Lambda^{-},\B\beta_k + M\B z_k\rangle + \langle\B\gamma,\by_k - \bX_k\B\beta_k-\B\xi_k\rangle   \\
			= & \frac{1}{n_k}\|\B\xi_k\|_2^2 +\alpha \|\B\beta_k\|_2^2+ \langle\B\beta_k,-\bX_k^T\B\gamma+\B\Lambda^{+}-\B\Lambda^{-}\rangle -M\langle \B z_k,\B\Lambda^++\B\Lambda^-\rangle
			+\langle\B\gamma,\by_k -\B\xi_k \rangle.
		\end{aligned}
	\end{equation}
	By setting the gradient of the Lagrangian with respect to $\B\beta_k$, $\B\xi_k$ equal to zero, we get
	\begin{equation}\label{kkt}
		\begin{aligned}
			\frac{2}{n_k}\B\xi_k&=\B\gamma, \\
			2\alpha\B\beta_k+   \B\Lambda^+-\B\Lambda^- &= \bX_k^T\B\gamma.
		\end{aligned}
	\end{equation}
	Therefore, the dual of Problem~\eqref{Freg2} is given as
	\begin{equation}\label{freg-dual}
		F_k(\B{z}_k)=\max_{\B\Lambda\geq 0,\B\gamma} -\frac{n_k}{4}\|\B\gamma\|_2^2 -\frac{1}{4\alpha}\|\bX_k^T\B\gamma+\B\Lambda^--\B\Lambda^+\|_2^2 + \B\gamma^T\by_k - M \langle \B z_k,\B\Lambda^++\B\Lambda^-\rangle.
	\end{equation}
	For the rest of proof, we assume $\B z_k\in\{0,1\}^p$ as we focus on the subgradient for feasible binary solutions. First, by~\eqref{kkt}, we can see that at optimality of~\eqref{freg-dual},
	$$\B\Lambda^+-\B\Lambda^- := \B\zeta_k = \frac{2}{n_k}\bX_k^T(\by_k-\bX_k\bar{\B\beta}_k)-2\alpha\bar{\B\beta}_k$$
	where $\bar{\B\beta}_k$ is defined in~\eqref{linregperstosolve}.
	Upon inspection of~\eqref{freg-dual}, we can see that at optimality of~\eqref{freg-dual}, for $j\in[p]$ a possible choice of $\B\Lambda$ is 
	$${\Lambda}^+_{j} = \begin{cases}
		\zeta_{k,j}  &  \zeta_{k,j}\geq 0 \\
		0 & \text{otherwise}
	\end{cases},{\Lambda}^-_{j} = \begin{cases}
		-\zeta_{k,j}  &  \zeta_{k,j}< 0 \\
		0 & \text{otherwise},
	\end{cases}$$
	implying that with this choice of $\B\Lambda^+,\B\Lambda^-$ at the optimality of~\eqref{freg-dual}, 
	\begin{equation}\label{gradient-z_k}
		\Lambda^+_j+\Lambda^-_j=|\zeta_{k,j}|.
	\end{equation}
	In \eqref{freg-dual}, $F_k(\B z_k)$ is written as the maximum of a linear function of $\B{z}_k$, therefore, the gradient of the dual cost function w.r.t $\B z_k$ at an optimal dual solution, is a subgradient of $F_k(\B{z}_k)$ by Danskin's Theorem \citep{bertsekas1997nonlinear}. Finally, the gradient of cost in~\eqref{freg-dual} with respect to $z_{j,k}$ is given as $-M(\Lambda^+_j+\Lambda^-_j)$.
	
\end{proof}

\subsection{The General Case}\label{app:oa-general}

Next, we present our outer approximation solver for the general case of our estimator given by Problem~\eqref{suppHet_mip} where $\alpha,\lambda,\delta$ can be nonzero. The basic idea is similar to the previous case but the algebra gets more involved. We start by presenting a reformulation of Problem~\eqref{suppHet_mip} that takes the form:
\begin{align}\label{generalmi-general}
	\min_{\B{Z}, \bar{\B{z}}}& ~~~   F(\B Z)  +\delta \sum_{k=1}^K \|\B{z}_k-\bar{\B z}\|_2^2\\  \text{s.t.}& ~~  \B z_{k}\in\{0,1\}^p;~~   \sum_{j=1}^p z_{k,j} \leq s~\forall k\in[K]\nonumber
\end{align}
where, for any $\B Z\in[0,1]^{p\times K}$, we define the function $F(\B Z)$ as follows:
\begin{align}\label{Freg2-general}
	F(\B Z)=\min_{\mathbb{B},\bar{\B\beta}} &\quad \sum_{k=1}^K  \frac{1}{n_k} \left\Vert \mathbf{y}_k-\bX_k\B\beta_k\right\Vert_2^2 + \lambda \sum_{k=1}^K \|\B\beta_k-\bar{\B\beta}\|_2^2 +\alpha \|\mathbb{B}\|_F^2 \\
	\text{s.t.} & \quad   |\beta_{k,j}|\leq Mz_{k,j}~~ \forall j\in[p], k\in[K].\nonumber 
\end{align}
Note that if $\lambda=0$, we have that $F(\B Z)=\sum_{k=1}^K F_k(\B z_k)$ where $F_k(\B z_k)$ is defined in~\eqref{Freg2}. Under this definition, we can generalize Proposition~\ref{prop5-main} to the case where all regularization coefficients $\alpha,\lambda,\delta$ can be nonzero.

\begin{proposition}\label{prop5-supp}
	Problem~\eqref{suppHet_mip} is equivalent to solving the optimization problem~\eqref{generalmi-general} where $F(\B Z)$ is implicitly described via display~\eqref{Freg2-general}.
\end{proposition}
Similar to functions $F_k(\B z_k)$, the function $\B{Z} \mapsto F(\B Z)$ has several desirable properties. In particular, $F(\B Z)$ is convex and its subgradient can be calculated as a by-product of solving the optimization problem in~\eqref{Freg2-general} (this is discussed later in Proposition~\ref{app:subgrad-prop}). Similar to the case of $F_k(\cdot)$, we have a pointwise linear lower bound for $F$:
$$F(\B X)\geq F(\B Z) +(\B X-\B Z)^T\B g^{(\B Z)}~~\forall \B X\in [0,1]^{p\times K}$$ 
where $\B g^{(\B Z)} \in \partial F(\B Z)$ is a subgradient of $F$ at $\B Z\in[0,1]^{p\times K}$. As a result, we can apply the same outer approximation idea from Section~\ref{outersec} to 
to Problem~\eqref{generalmi-general}. In particular, at iteration $t>0$ of our outer approximation solver, we substitute $F(\cdot)$ in~\eqref{generalmi-general} with a piece-wise linear lower bound, resulting in the MIQP
\begin{align}\label{outerMILP-general}
	\left(\B Z^t,\eta^t,\bar{\B z} \right)\in  \argmin_{\B{Z},\bar{\B z},\eta} &~~   \eta + \delta \sum_{k=1}^K \|\B z_k - \bar{\B z}\|_2^2 \\
	\text{s.t.} & ~~ \B Z\in\{0,1\}^{p\times K},~\sum_{i=1}^p z_{k,i} \leq s~~\forall k\in[K]\nonumber\\
	& ~~\eta \geq F(\B Z^i)+(\B Z-\B Z^i)^T\B g^{(\B Z^i)},  i\leq t-1\nonumber .
\end{align}
Similar to the case of $\lambda=0$, after finitely-many iteration of solving Problem~\eqref{outerMILP-general}, we obtain an optimal solution to Problem~\eqref{generalmi-general}, as the feasible set of Problem~\eqref{generalmi-general} is finite. Moreover, at each iteration, the optimal objective value of Problem~\eqref{outerMILP-general} is a lower bound for the optimal objective value of Problem~\eqref{generalmi-general}, similar to the case discussed in Section~\ref{outersec}.

Next, we discuss how to compute subgradients of the function $F(\cdot)$. Let
\begin{align}\label{underlinebarbeta}
	((\underline{\B\beta}_k)_{k=1}^K,\underline{\bar{\B\beta}})\in\argmin_{\mathbb{B},\bar{\B\beta}} &\quad \sum_{k=1}^K  \frac{1}{n_k} \left\Vert \mathbf{y}_k-\bX_k\B\beta_k\right\Vert_2^2 + \lambda \sum_{k=1}^K \|\B\beta_k-\bar{\B\beta}\|_2^2 +\alpha \|\mathbb{B}\|_F^2 \\
	\text{s.t.} & \quad   |\beta_{k,j}|\leq Mz_{k,j}~~ \forall j\in[p], k\in[K].\nonumber 
\end{align}
For $k\in[K]$ we define
\begin{equation}\label{zeta_k-general}
	\B{\zeta}_k = \frac{2}{n_k}\bX_k^T(\by_k-\bX_k\underline{\B\beta}_k)-2\alpha\underline{\B\beta}_k -2\lambda(\underline{\B\beta}_k - \underline{\bar{\B\beta}})\in\R^p.
\end{equation}
Proposition~\ref{app:subgrad-prop} presents an expression for the subgradient of $F(\cdot)$.
\begin{proposition}\label{app:subgrad-prop}
	Let $F$ be as defined in~\eqref{Freg2-general}. The following results hold true:
	\begin{enumerate}
		\item (\textit{Convexity}) The function $\B Z \mapsto F(\B Z)$ on 
		$\B Z \in [0,1]^{p\times K}$ is convex. 
		\item (\textit{Subgradient}) 
		Let $\B\zeta_k$ be as defined in~\eqref{zeta_k-general}.
		The matrix $\B G\in\mathbb{R}^{p\times K}$ with $k$-th column given by
		$$\B G_{:,k}= -M|\B\zeta_{k}|$$
		for $k\in[K]$, is a subgradient of $\B Z \mapsto F(\B Z)$ for $\B Z\in\{0,1\}^{p\times K}$.
	\end{enumerate}    
\end{proposition}
The proof of Proposition~\ref{app:subgrad-prop} is presented at the end of this section. Note that from Proposition~\ref{app:subgrad-prop}, the subgradient of $F$ can be calculated as a by-product of solving Problem~\eqref{Freg2-general}. As $\B{Z}$ is sparse with at most $sK$ nonzero coordinates, Problem~\eqref{Freg2-general} is a convex QP with $\mathcal{O}(sK)$ variables, which can be solved efficiently via an off-the-shelf solver for sufficiently sparse problems (i.e., $s$ is small).

\begin{proof}[\textbf{Proof of Proposition~\ref{app:subgrad-prop}}]
	\textbf{Part 1)} The proof is similar to the proof of the first part of Proposition~\ref{subgradpers}.
	
	\noindent\textbf{Part 2)} 
	Note that Problem~\eqref{Freg2-general} can be equivalently written as
	
	\begin{align}\label{Freg2-rewrite-general}
		\min_{\mathbb{B},\B\xi_k,\B\omega_k,\bar{\B\beta}} &\quad \sum_{k=1}^K  \frac{1}{n_k} \left\Vert \B\xi_k\right\Vert_2^2 + \lambda \sum_{k=1}^K \|\B\omega_k\|_2^2 +\alpha \|\mathbb{B}\|_F^2 \\
		\text{s.t.} & \quad   |\beta_{k,j}|\leq Mz_{k,j}~~ \forall j\in[p], k\in[K] \nonumber \\ 
		&\quad  \B\xi_k= \mathbf{y}_k-\bX_k\B\beta_k,~~\B\omega_k=
		\B\beta_k-\bar{\B\beta}~\forall k\in[K].\nonumber
	\end{align}
	We start by deriving the dual of Problem~\eqref{Freg2-rewrite-general} (note that strong duality holds similar to Proposition~\ref{subgradpers}). Considering Lagrange multipliers $\B\Lambda^{+}_k,\B\Lambda^{-}_k\in\R^{p}$, $\B\gamma_k\in\R^{n_k},\B\nu_k\in\R^p$ for $k\in[K]$ for problem constraints, the Lagrangian for this problem
	${\mathcal L}(\B\beta_k,\bar{\B\beta},\B\omega_k,\B\xi_k,\B\Lambda_k^{+},\B\Lambda_k^{-},\B\gamma_k,\B\nu_k)$ or ${\mathcal L}$ (for short),
	can be written as
	\begin{align}
		\mathcal{L} = &  \sum_{k=1}^K\Big\{ \frac{1}{n_k}\left\Vert\B\xi_k \right\Vert_2^2  +\alpha \|\B\beta_k\|_2^2+  \lambda \|\B\omega_k\|_2^2+ \langle\B\Lambda_k^{+},\B\beta_k - M\B z_k\rangle -\langle\B\Lambda_k^{-},\B\beta_k + M\B z_k\rangle\nonumber \\ 
		&+ \langle\B\gamma_k,\by_k - \bX_k\B\beta_k-\B\xi_k\rangle +  \langle\B\nu_k, \B\omega_k - \B\beta_k+\bar{\B\beta}\rangle\Big\}\nonumber\\
		= & \sum_{k=1}^K\Big\{ \frac{1}{n_k}\left\Vert\B\xi_k \right\Vert_2^2  +\alpha \|\B\beta_k\|_2^2 + \langle\B\beta_k, -\bX_k^T\B\gamma_k +\B\Lambda^+_k-\B\Lambda^-_k -\B\nu_k\rangle -\langle\B\xi_k,\B\gamma_k\rangle \nonumber \\
		& +\lambda\|\B\omega_k\|_2^2 + \langle\B\omega_k,\B\nu_k\rangle + \langle\bar{\B\beta},\B\nu_k\rangle -M\langle \B z_k,\B\Lambda_k^++\B\Lambda_k^-\rangle + \langle\B\gamma_k,\by_k\rangle\Big\}.
	\end{align}
	By the first order stationary conditions (i.e, setting the gradient of the Lagrangian with respect to $\B\xi_k$, $\B\beta_k$, $\B\omega_k$ and $\bar{\B\beta}$ equal to zero), we get for all $k$:
	\begin{equation}\label{kkt-2}
		\begin{aligned}
			\frac{2}{n_k}\B\xi_k&=\B\gamma_k, \\
			2\alpha\B\beta_k+   \B\Lambda^+_k-\B\Lambda^-_k -\B\nu_k&= \bX_k^T\B\gamma_k\\
			2\lambda\B\omega_k & = -\B\nu_k\\
			\sum_{k=1}^K \B\nu_k & = 0.
		\end{aligned}
	\end{equation}
	Therefore, the dual of Problem~\eqref{Freg2-rewrite-general} is given by:
	\begin{align}
		\max_{\B\Lambda_k\geq 0,\B\nu_k,\B\gamma_k, \forall k} & ~~ \sum_{k=1}^K \Big\{-\frac{n_k}{4}\|\B\gamma_k\|_2^2 -\frac{1}{4\alpha}\|\bX_k^T\B\gamma_k +\B\nu_k+ \B\Lambda^-_k - \B\Lambda^+_k \|_2^2 -\frac{1}{4\lambda}\|\B\nu_k\|_2^2+\B\gamma_k^T\by_k - M\langle \B{z}_k,\B\Lambda^+_k+\B\Lambda^-_k\rangle  
		\Big\}\\
		\text{s.t.} & ~~ \sum_{k=1}^K \B\nu_k  = 0. \nonumber
	\end{align}
	Additionally, note that from~\eqref{kkt-2} at optimality of Problem~\eqref{Freg2-rewrite-general}, we have that for $k\in[K]$,
	\begin{equation}
		\begin{aligned}
			\B{\nu}_k &= 2\lambda(\underline{\bar{\B\beta}} -\underline{\B\beta}_k) \\
			\B\gamma_k & = \frac{2}{n_k}(\by_k-\bX_k\underline{\B\beta}_k)\\
			\B\Lambda_k^+-\B\Lambda_k^- & = \bX_k^T\B\gamma_k -2\alpha\underline{\B\beta}_k +\B\nu_k
		\end{aligned}
	\end{equation}
	where $\underline{\B\beta}_k,\underline{\bar{\B\beta}}$ are defined in~\eqref{underlinebarbeta}. The rest of the proof follows similar to the second part of Proposition~\ref{subgradpers}.
\end{proof}

\section{Details from Approximate Algorithms Implementation}\label{app:alg-sec}
\subsection{Block CD Algorithm}\label{app:blockcd}
In this section, we discuss how problem~\eqref{IHT1} can be solved efficiently. Suppose $k\in[K]$ is given and fixed. The goal is to find the optimal values of $\B{\beta}_k,\B{z}_k$. As $\B{z}_k$ is a binary vector, for $j\in[p]$ we consider two cases $z_{k,j}=0,1$. If $z_{k,j}=1$, the optimal value of $\beta_{k,j}$ is given as 
$$\beta_{k,j}= b_{k,j}:=\beta_{k,j}^+-\frac{1}{L_k}\frac{\partial g_k(\B{\beta}_k^+,\bar{\B{\beta}})}{\partial \beta_{k,j}}.$$
Moreover, $(z_{k,j}-\bar{z}_j)^2=(1-\bar{z}_j)^2$. If $z_{k,j}=0$, then we have $\beta_{k,j}=0$ and $(z_{k,j}-\bar{z}_j)^2=\bar{z}_j^2$. Hence, the difference in the contribution of term $j$ in the objective of~\eqref{IHT1} for $z_{k,j}=1,0$ is 
$$\Delta_{k,j} = (1-\bar{z}_j)^2 - \left(\frac{L_k}{2}b_{k,j}^2 + \bar{z}_j^2\right).$$
If $\Delta_{k,j}\geq 0$, setting $z_{k,j}=0$ leads to a lower objective and therefore for any such $j$, the optimal values are given as $\beta_{k,j}=z_{k,j}=0$. For any other $j$, setting $z_{k,j}=1$ results in a lower objective, however, at most $s$ of values of $z_{k,j}$ can be set to one to ensure feasibility. As a result, we select at most $s$ values of $j$ that lead to smallest (most negative) value of $\Delta_{k,j}$. 

Before moving on, let us discuss the computational complexity of our block CD algorithm. For simplicity, assume $n_k=n$ for all $k\in[K]$. Suppose we fix $k$. 
Then, calculating $\nabla_{\B\beta_k}g_k(\B\beta_k,\bar{\B\beta})$ requires calculating 
$\bX_k^T(\by_k - \bX_k\B\beta_k)$
which can be done in $\mathcal{O}(np)$.
For fixed $k$, given $\nabla_{\B\beta}g_k(\B\beta,\bar{\B\beta})$, the computational cost of calculating and (partial) sorting $\Delta_{k,j}$'s is $\mathcal{O}(p+p\log s)$. Therefore, for a fixed $k$, the computational complexity of solving~\eqref{IHT1} is $\mathcal{O}(np)$. We also need to calculate $L_k$, which requires access to $\lambda_{\max}(\bX_k^T\bX_k)=\lambda_{\max}(\bX_k\bX_k^T)$. However, $\lambda_{\max}(\bX_k\bX_k^T)$ can be obtained in $\mathcal{O}(n^2p)$ using the power method. If we perform $T$ passes of the block CD method, considering the initial cost of  calculating $L_k$,
block CD will have the computational complexity $\mathcal{O}(Kn^2p + TKnp)$.

\subsection{Active Sets}\label{app:activeset}
To further speed-up the convergence of our algorithm, we implement an active set version of the block CD algorithm. We start by an initial active set $\mathcal{I}_{\text{active}}$. Then, we run the block CD algorithm on problem~\eqref{suppHet} with the additional constraint 
$$\beta_{k,j}=0~~\forall k\in [K], j\notin \mathcal{I}_{\text{active}}.$$
Under this constraint, the dimension of the problem is effectively reduced to $|\mathcal{I}_{\text{active}}|$, instead of $p$, which leads to faster convergence if $|\mathcal{I}_{\text{active}}|\ll p$. After the solution to the restricted problem is found, we run an iteration of block CD on the original problem (problem~\eqref{suppHet}) with $p$ features. If this iteration does not change the solution, the current solution to the restricted problem is also a solution to the original problem. Otherwise, we update the active set as
$$ \mathcal{I}_{\text{active}}\leftarrow \mathcal{I}_{\text{active}}\cup \{j:\exists k\in[K]: \beta_{k,j}\neq 0\}.$$

\subsection{Local Search Method}\label{app:localsearch}
Suppose $k_0\in[K]$ is fixed. In Proposition~\ref{lsprop} below, we give a closed form solution for the inner optimization in problem~\eqref{localsearch} (i.e., the value of $b$ for given $j_1,j_2$).
\begin{proposition}\label{lsprop}
	Fix $k_0,j_1,j_2$ in problem~\eqref{localsearch}. Let $\tilde{b}$ be the optimal solution to the inner problem in Problem~\eqref{localsearch}. Moreover, let
	\begin{equation}
		\begin{aligned}
			\mathbf{r} & = \mathbf{y}_{k_0}-\mathbb{X}_{k_0}\tilde{\boldsymbol{\beta}}_{k_0}\\
			p_1 & = \frac{\|\mathbf{x}_{k_0,j_2}\|_2^2}{n_{k_0}}+\alpha + \lambda \frac{K^2-K}{K^2}\\
			p_2 & = -\frac{2}{n_{k_0}}\mathbf{r}^T\mathbf{x}_{k_0,j_2}-\frac{2}{n_{k_0}}\mathbf{x}_{k_0,j_2}^T\mathbf{x}_{k_0,j_1}\tilde{\beta}_{k_0,j_1}-\frac{2\lambda}{K}\sum_{k=1}^K\tilde{\beta}_{k,j_2}\\
			p_3 & =\frac{\|\mathbf{x}_{k_0,j_1}\|_2^2}{n_{k_0}}\tilde{\beta}_{k_0,j_1}^2 + \frac{1}{n_{k_0}}\mathbf{r}^T \mathbf{x}_{k_0,j_1}\tilde{\beta}_{k_0,j_1} .
		\end{aligned}
	\end{equation}
	Then,
	\begin{enumerate}
		\item One has 
		$$\tilde{b} = -\frac{p_2}{2p_1}.$$
		\item Moreover, 
		\begin{multline*}
			g(\tilde{\mathbb{B}}-\tilde{\beta}_{k_0,j_1}\boldsymbol{E}_{k_0,j_1}+\tilde{b}\boldsymbol{E}_{k_0,j_2}) - g(\tilde{\mathbb{B}})=p_3 -\frac{p_2^2}{4p_1}\\ -\lambda\left[\sum_{k=1}^K\left(\tilde{\beta}_{k,j_1}-\frac{\sum_{k=1}^K \tilde{\beta}_{k,j_1}}{K}\right)^2\right]+\lambda\left[\sum_{k:k\neq k_0}^K\left(\tilde{\beta}_{k,j_1}-\frac{\sum_{k:k\neq k_0} \tilde{\beta}_{k,j_1}}{K}\right)^2\right]+\lambda \left(\frac{\sum_{k:k\neq k_0} \tilde{\beta}_{k,j_1}}{K}\right)^2
		\end{multline*}
		and 
		\begin{multline*}
			h(\tilde{\B{Z}}-\boldsymbol{E}_{k_0,j_1}+\boldsymbol{E}_{k_0,j_2}) - h(\tilde{\B{Z}}) =\\ -\delta\left[\sum_{k=1}^K\left(\tilde{z}_{k,j_1}-\frac{\sum_{k=1}^K \tilde{z}_{k,j_1}}{K}\right)^2\right]+\delta\left[\sum_{k:k\neq k_0}^K\left(\tilde{z}_{k,j_1}-\frac{\sum_{k:k\neq k_0} \tilde{z}_{k,j_1}}{K}\right)^2\right]+\delta \left(\frac{\sum_{k:k\neq k_0} \tilde{z}_{k,j_1}}{K}\right)^2.
		\end{multline*}
	\end{enumerate}
\end{proposition}
Proposition~\ref{lsprop} shows that the inner optimization in problem~\eqref{localsearch}, and therefore the entire mathematical program~\ref{localsearch}, can be solved efficiently via a closed form solution. In our implementation, we perform the calculations described in Proposition~\ref{lsprop} in matrix form. As a result, we can identify the optimal swap for each $k$ without any for loops over $j_1,j_2$. Our numerical experiments confirm the benefits of local search in practice.

\section{Proofs of Main Results}

\subsection{Proof of Theorem~\ref{thm1}}
The proof of this theorem is based on the following technical lemma.
\begin{lemma}\label{maximalimp}
	Suppose $\B{X}\in\R^{n\times p}$ is independent of $\B{\epsilon}\sim\mathcal{N}(0,\sigma^2 \B{I}_n)$. Let 
	\begin{equation}\label{event4}
		\mathcal{E}(\B{X},\B{\epsilon})=\left\{\sup_{\substack{\B{v}\in\R^p \\ \|\B{v}\|_0\leq 2s}}\left\vert\B{\epsilon}^T \frac{\B{Xv}}{\|\B{Xv}\|_2}\right\vert \leq  c_1\sigma \sqrt{s\log(p/s)}\right\}
	\end{equation}
	for some absolute constant $c_1>0$. Then 
	$$ \p(\mathcal{E}(\B{X},\B{\epsilon})) \geq 1- \exp(-10s\log(p/s)).$$
\end{lemma}
\begin{proof}
	For any $S\subseteq [p]$, let $\B{X}_S\in\R^{n\times |S|}$ be the submatrix of $\B{X}$ with columns indexed by $S$. Moreover, let $\B{\Phi}_S\in\mathbb{R}^{p\times |S|}$ be an orthonormal basis for the column span of $\B{X}_S$. One has
	\begin{align*}
		\p \left(\left.\sup_{\substack{\B{v}\in\R^p\\\|\B{v}\|_0\leq 2s}}\left\vert\B{\epsilon}^T\frac{\B{Xv}}{\|\B{Xv}\|_2}\right\vert> t\right\vert \B{X}\right) &= \p \left(\left.\max_{\substack{S\subseteq[p]\\ |S|=2s}}\sup_{\substack{\B{v}\in\R^p\\ S(\B{v})=S}}\left\vert\B{\epsilon}^T\frac{\B{Xv}}{\|\B{Xv}\|_2}\right\vert> t\right\vert \B{X}\right) \\
		& \stackrel{(a)}{=}\p \left(\left.\max_{\substack{S\subseteq[p]\\ |S|=2s}}\sup_{\substack{\B{v}\in\R^{|S|}\\ \|\B{v}\|_2=1}}(\B{\Phi}_{S}^T\B{\epsilon})^T\B{v}> t\right\vert \B{X}\right)\\
		& \stackrel{(b)}{\leq} {p \choose 2s}\p \left(\left.\sup_{\substack{\B{v}\in\R^{|S|}\\ \|\B{v}\|_2=1}}(\B{\Phi}_{S}^T\B{\epsilon})^T\B{v}> t\right\vert \B{X}\right)\\
		& \stackrel{(c)}{\leq} {p\choose 2s}\exp\left(-\frac{t^2}{8\sigma^2}+2s\log 5\right) \\
		& \stackrel{(d)}{\leq} \left(\frac{ep}{2s}\right)^{2s}\exp\left(-\frac{t^2}{8\sigma^2}+2s\log 5\right)\\
		& \leq \exp\left(-\frac{t^2}{8\sigma^2}+2s\log 5 + 2s\log(ep/2s)\right)
	\end{align*}
	where $(a)$ is due to the definition of $\B{\Phi}_S$, $(b)$ is due to union bound, $(c)$ is true as by independence of $\B{X}$ and $\B{\epsilon}$, we have $\B{\Phi}_S^T\B{\epsilon}|\B{X} \sim\mathcal{N}(\B{0},\sigma^2 \B{I}_{2s})$ and Theorem~1.19 of~\cite{rigollet2015high} and $(d)$ is by the inequality ${p\choose 2s}\leq (ep/2s)^{2s}$. Take $t^2=64\sigma^2c_1s\log(p/s)$ and $c_1$ large enough to have that $8c_1 s\log(p/s)-2s\log 5 -2s\log(ep/2s)>10s\log(p/s)$. As a result, 
	\begin{align}\label{fastbount}
		\p \left(\left.\mathcal{E}(\B{X},\B{\epsilon})\right\vert \B{X}\right)\geq 1- \exp(-10s\log(p/s)).
	\end{align}
	Finally, note that 
	\begin{align*}
		\p(\mathcal{E}(\B{X},\B{\epsilon}))=\int_{\B{X}} \p \left(\left.\mathcal{E}(\B{X},\B{\epsilon})\right\vert \B{X}\right)d\p_{\B{X}}\geq \int_{\B{X}}(1- \exp(-10s\log(p/s)))d\p_{\B{X}} = 1- \exp(-10s\log(p/s)).
	\end{align*}
\end{proof}

\begin{proof}[\textbf{Proof of Theorem~\ref{thm1}}]
	The proof of this theorem is based on the intersection of events $\mathcal{E}(\bX_1,\B{\epsilon}_1),\cdots,\mathcal{E}(\bX_K,\B{\epsilon}_K)$ which by Lemma~\ref{maximalimp} and union bound happens with probability at least 
	\begin{equation}\label{thm1-prob}
		1- K\exp(-10s\log(p/s)).
	\end{equation}
	To lighten the notation, let 
	\begin{equation}
		\begin{aligned}
			h(\B{z}_1,\cdots,\B{z}_K)&=\delta\sum_{k=1}^K \|\B{z}_k - \bar{\B{z}}\|_2^2 \\ l(\B{\beta}_1,\cdots,\B{\beta}_K)&=\lambda\sum_{k=1}^K \|\B{\beta}_k - \bar{\B{\beta}}\|_2^2
		\end{aligned}
	\end{equation}
	with $\bar{\B{z}}=\sum_{k=1}^K \B{z}_k/K$  and $\bar{\B{\beta}}=\sum_{k=1}^K \B{\beta}_k/K$. Moreover, let
	\begin{equation}
		\begin{aligned}
			\hat{h}&=h(\hat{\B{z}}_1,\cdots,\hat{\B{z}}_K)+l(\hat{\B{\beta}}_1,\cdots,\hat{\B{\beta}}_K) \\ h^*&=h(\B{z}^*_1,\cdots,\B{z}^*_K)+l(\B{\beta}^*_1,\cdots,\B{\beta}^*_K).
		\end{aligned}
	\end{equation}
	By optimality of $\hat{\B{\beta}}_k,\hat{\B{z}}_k$ and feasibility of $\B{\beta}_k^*,\B{z}_k^*$ for problem~\eqref{suppHet}, we have:
	\begin{align}
		& \sum_{k=1}^K \frac{1}{n_k}\|\by_k-\bX_k\hat{\B{\beta}}_k\|_2^2+\hat{h} +\alpha\|\hat{\bB}\|_F^2\leq \sum_{k=1}^K \frac{1}{n_k}\|\by_k-\bX_k{\B{\beta}}^*_k\|_2^2+h^*+\alpha\|\bB^*\|_F^2\nonumber \\
		\stackrel{(a)}{\Rightarrow} & \sum_{k=1}^K \frac{1}{n_k}\|\bX_k\B{\beta}_k^*+\B{r}_k+\B{\epsilon}_k-\bX_k\hat{\B{\beta}}_k\|_2^2+\hat{h}  +\alpha\|\hat{\bB}\|_F^2\leq \sum_{k=1}^K \frac{1}{n_k}\|\B{r}_k+\B{\epsilon}_k\|_2^2+ h^*+\alpha\|\bB^*\|_F^2\nonumber \\
		\Rightarrow & \sum_{k=1}^K \frac{1}{n_k} \|\bX_k(\B{\beta}_k^*-\hat{\B{\beta}}_k)\|_2^2+\hat{h}+\alpha\|\hat{\bB}\|_F^2\leq \sum_{k=1}^K \frac{-2}{n_k}(\B{\epsilon}_k+\B{r}_k)^T\bX_k(\B{\beta}_k^*-\hat{\B{\beta}}_k) +h^*+\alpha\|\bB^*\|_F^2\nonumber\\
		\Rightarrow & \sum_{k=1}^K \frac{1}{n_k} \|\bX_k(\B{\beta}_k^*-\hat{\B{\beta}}_k)\|_2^2+\hat{h}+\alpha\|\hat{\bB}\|_F^2 \leq \sum_{k=1}^K \frac{-2}{n_k}(\B{\epsilon}_k+\B{r}_k)^T\frac{\bX_k(\B{\beta}_k^*-\hat{\B{\beta}}_k)}{\|\bX_k(\B{\beta}_k^*-\hat{\B{\beta}}_k)\|_2}\|\bX_k(\B{\beta}_k^*-\hat{\B{\beta}}_k)\|_2+h^*+\alpha\|\bB^*\|_F^2\nonumber \\
		\stackrel{(b)}{\Rightarrow} & \sum_{k=1}^K \frac{1}{2n_k} \|\bX_k(\B{\beta}_k^*-\hat{\B{\beta}}_k)\|_2^2+ \hat{h}+\alpha\|\hat{\bB}\|_F^2 \leq \sum_{k=1}^K \frac{2}{n_k}\left((\B{\epsilon}_k+\B{r}_k)^T\frac{\bX_k(\B{\beta}_k^*-\hat{\B{\beta}}_k)}{\|\bX_k(\B{\beta}_k^*-\hat{\B{\beta}}_k)\|_2}\right)^2+ h^*+\alpha\|\bB^*\|_F^2\label{thm1-helper1},
	\end{align}
	where $(a)$ is achieved by substituting $\by_k=\bX_k\B{\beta}_k^*+\B{\epsilon}_k+\B{r}_k$ from the oracle, and $(b)$ is a result of the inequality $-2ab\leq 2a^2+b^2/2$. Next, note that $\B{\beta}_k^*-\hat{\B{\beta}}_k$ has at most $2s$ nonzero coordinates. As a result, from event $\mathcal{E}(\bX_k,\B{\epsilon}_k)$ write
	\begin{equation}\label{thm1-helper2}
		\frac{1}{n_k}\left(\B{\epsilon}_k^T\frac{\bX_k(\B{\beta}_k^*-\hat{\B{\beta}}_k)}{\|\bX_k(\B{\beta}_k^*-\hat{\B{\beta}}_k)\|_2}\right)^2\lesssim \frac{\sigma_k^2 s\log(p/s)}{n_k}.
	\end{equation}
	Moreover, 
	\begin{align}
		\frac{1}{n_k}\left((\B{\epsilon}_k+\B{r}_k)^T\frac{\bX_k(\B{\beta}_k^*-\hat{\B{\beta}}_k)}{\|\bX_k(\B{\beta}_k^*-\hat{\B{\beta}}_k)\|_2}\right)^2& \lesssim \frac{1}{n_k}\left(\B{r}_k^T\frac{\bX_k(\B{\beta}_k^*-\hat{\B{\beta}}_k)}{\|\bX_k(\B{\beta}_k^*-\hat{\B{\beta}}_k)\|_2}\right)^2 + \frac{1}{n_k}\left(\B{\epsilon}_k^T\frac{\bX_k(\B{\beta}_k^*-\hat{\B{\beta}}_k)}{\|\bX_k(\B{\beta}_k^*-\hat{\B{\beta}}_k)\|_2}\right)^2\nonumber \\
		& \lesssim \frac{1}{n_k}\|\B{r}_k\|_2^2+ \frac{1}{n_k}\left(\B{\epsilon}_k^T\frac{\bX_k(\B{\beta}_k^*-\hat{\B{\beta}}_k)}{\|\bX_k(\B{\beta}_k^*-\hat{\B{\beta}}_k)\|_2}\right)^2.
	\end{align}
	Consequently, from~\eqref{thm1-helper1} and~\eqref{thm1-helper2} we have
	\begin{align}
		\sum_{k=1}^K \frac{1}{n_k} \|\bX_k(\B{\beta}_k^*-\hat{\B{\beta}}_k)\|_2^2+\hat{h}+\alpha\|\hat{\bB}\|_F^2\lesssim \sum_{k=1}^K \frac{\sigma_k^2 s \log(p/s)}{n_k}+h^*+\alpha\|\bB^*\|_F^2+\sum_{k=1}^K \frac{1}{n_k}\|\B{r}_k\|_2^2.
	\end{align}
\end{proof}

\subsection{Proof of Corollary~\ref{cor1}}
The proof of this corollary is based on the following lemma.
\begin{lemma}\label{tech-lem2}
	Suppose $\B{z}_1\cdots,\B{z}_K\in\{0,1\}^p$. Then, 
	$$\sum_{k=1}^K\|\B{z}_k-\bar{\B{z}}\|_2^2\geq \frac{|S_{\text{all}}\setminus S_{\text{common}}|}{K}$$
	where $\bar{\B{z}}=\sum_{k=1}^K \B{z}_k/K$, $S_{\text{all}}=S_{\text{all}}(\B{Z})$ and $S_{\text{common}}=S_{\text{common}}(\B{Z})$.
\end{lemma}
\begin{proof}
	For $j\in S_{\text{common}}$ or $j\notin S_{\text{all}}$, 
	$$\sum_{k=1}^K (z_{k,j}-\bar{z}_{j})^2 =0.$$
	Suppose $j\in S_{\text{all}}\setminus S_{\text{common}}$. Then, there exist $k_1,k_2\in[K]$ such that $z_{k_1,j}=1$ and $z_{k_2,j}=0$, implying $1/K\leq\bar{z}_j\leq (K-1)/K$. As a result, for $k_0\in[K]$ if $z_{k_0,j}=0$,
	$$(z_{k_0,j}-\bar{z}_j)^2=\bar{z}_j^2 \geq \frac{1}{K^2}$$
	and if $z_{k_0,j}=1$,
	$$(z_{k_0,j}-\bar{z}_j)^2=(1-\bar{z}_j)^2\geq (1-(K-1)/K)^2\geq \frac{1}{K^2}.$$
	Consequently, 
	\begin{align}
		\sum_{k=1}^K\|\B{z}_k-\bar{\B{z}}\|_2^2\geq \sum_{j\in S_{\text{all}}\setminus S_{\text{common}}}\sum_{k=1}^K (z_{k,j}-\bar{z}_j)^2 \geq \sum_{j\in S_{\text{all}}\setminus S_{\text{common}}}\sum_{k=1}^K \frac{1}{K^2}=\frac{|S_{\text{all}}\setminus S_{\text{common}}|}{K}.
	\end{align}
\end{proof}
\begin{proof}[\textbf{Proof of Corollary~\ref{cor1}}]
	This proof is on the probability event considered in Theorem~\ref{thm1}. Note that as $\B{z}_1^*=\cdots=\B{z}_K$, we have that $h(\B{z}_1^*,\cdots,\B{z}_K^*)=0$. The prediction error part of the corollary follows from this fact and Theorem~\ref{thm1}. Let $\hat{S}_{\text{all}}=S_{\text{all}}(\hat{\B{Z}})$ and $\hat{S}_{\text{common}}=S_{\text{common}}(\hat{\B{Z}})$. Next,
	\begin{align}
		\delta\frac{|\hat{S}_{\text{all}}\setminus \hat{S}_{\text{common}}|}{K}& \stackrel{(a)}{\leq} h(\hat{\B{z}}_1,\cdots,\hat{\B{z}}_K) \nonumber\\
		& \leq \sum_{k=1}^K \frac{1}{n_k} \|\bX_k(\B{\beta}_k^*-\hat{\B{\beta}}_k)\|_2^2+h(\hat{\B{z}}_1,\cdots,\hat{\B{z}}_K) \nonumber \\
		& \stackrel{(b)}{\lesssim} \sum_{k=1}^K \frac{\sigma_k^2 s \log(p/s)}{n_k}+\sum_{k=1}^K\frac{1}{n_k} \|\B{r}_k\|_2^2
	\end{align}
	where $(a)$ is due to Lemma~\ref{tech-lem2} and $(b)$ is due to Theorem~\ref{thm1} by our choice of $\alpha,\lambda$. In particular, if $\delta\gtrsim K\sum_{k=1}^K {([\sigma_k^2 s\log(p/s)+\|\B{r}_k\|_2^2]/n_k)} $ is sufficiently large,
	\begin{equation}
		|\hat{S}_{\text{all}}\setminus \hat{S}_{\text{common}}| \lesssim \frac{K}{\delta}\sum_{k=1}^K\left[ \frac{\sigma_k^2 s \log(p/s)}{n_k}+\frac{1}{n_k} \|\B{r}_k\|_2^2\right]\leq c^*
	\end{equation}
	for some $c^* \in [0,1)$ and as $ |\hat{S}_{\text{all}}\setminus \hat{S}_{\text{common}}|\in\mathbb{Z}_{\geq 0}$, we have that $ |\hat{S}_{\text{all}}\setminus \hat{S}_{\text{common}}|=0$.
\end{proof}

\subsection{Proof of Corollary~\ref{cor2}}
\begin{lemma}\label{tech-lem3}
	Suppose $\B{z}_1\cdots,\B{z}_K\in\{0,1\}^p$ are regular. Then, 
	$$\sum_{k=1}^K\|\B{z}_k-\bar{\B{z}}\|_2^2= |S_{\text{all}}\setminus S_{\text{common}}|\frac{K^2-K}{K^2}$$
	where $\bar{\B{z}}=\sum_{k=1}^K \B{z}_k/K$, $S_{\text{all}}=S_{\text{all}}(\B{Z})$ and $S_{\text{common}}=S_{\text{common}}(\B{Z})$.
\end{lemma}
\begin{proof}
	For $j\in S_{\text{common}}$ or $j\notin S_{\text{all}}$, 
	$$\sum_{k=1}^K (z_{k,j}-\bar{z}_{j})^2 =0.$$
	Suppose $j\in S_{\text{all}}\setminus S_{\text{common}}$ and $k_0\in[K]$ is such that $z_{k_0,j}=1$. Then, $\bar{z}_j=1/K$ and
	\begin{align}
		\sum_{k=1}^K (z_{k,j}-\bar{z}_j)^2 = \sum_{k=1}^K (z_{k,j}-1/K)^2 =\frac{K-1}{K^2}+\frac{(K-1)^2}{K^2}=\frac{K^2-K}{K^2}.
	\end{align}
	As a result,
	\begin{align}
		\sum_{k=1}^K\|\B{z}_k-\bar{\B{z}}\|_2^2= \sum_{j\in S_{\text{all}}\setminus S_{\text{common}}}\sum_{k=1}^K (z_{k,j}-\bar{z}_j)^2 = |S_{\text{all}}\setminus S_{\text{common}}|\frac{K^2-K}{K^2}.
	\end{align}
	
\end{proof}
\begin{proof}[\textbf{Proof of Corollary~\ref{cor2}}]
	In this proof, we assume $\delta$ is taken as
	\begin{equation}\label{delta-cor2}
		\delta=c_{\delta}\frac{K}{K-1}\frac{1}{|S^*_{\text{all}}\setminus S^*_{\text{common}}|}\sum_{k=1}^K \left\{\frac{\sigma_k^2 s \log (p/s)}{n_k}+\frac{1}{n_k}\|\B{r}_k\|_2^2\right\}
	\end{equation}
	where $c_{\delta}$ is an absolute constant.
	
	This proof is on the probability event considered in Theorem~\ref{thm1}.  One has
	\begin{align}
		\sum_{k=1}^K \frac{1}{n_k} \|\bX_k(\B{\beta}_k^*-\hat{\B{\beta}}_k)\|_2^2 + \delta \frac{|\hat{S}_{\text{all}}\setminus \hat{S}_{\text{common}}|}{K} &  \stackrel{(a)}{\leq} \sum_{k=1}^K \frac{1}{n_k} \|\bX_k(\B{\beta}_k^*-\hat{\B{\beta}}_k)\|_2^2 + h(\hat{\B{z}}_1,\cdots,\hat{\B{z}}_K)\nonumber \\
		& \stackrel{(b)}{\lesssim} \sum_{k=1}^K \left\{\frac{\sigma_k^2 s \log(p/s)}{n_k}+\frac{1}{n_k}\|\B{r}_k\|_2^2\right\}+ h(\B{z}_1^*,\cdots,\B{z}_K^*)\nonumber \\
		& \stackrel{(c)}{\lesssim} \sum_{k=1}^K \left\{ \frac{\sigma_k^2 s \log(p/s)}{n_k}+\frac{1}{n_k}\|\B{r}_k\|_2^2\right\}+ \delta \frac{K^2-K}{K^2}|S^*_{\text{all}}\setminus S^*_{\text{common}}| \nonumber \\
		& \stackrel{(d)}{\lesssim} \sum_{k=1}^K \left\{\frac{\sigma_k^2 s \log(p/s)}{n_k}+\frac{1}{n_k}\|\B{r}_k\|_2^2 \right\}\label{cor2-helper}
	\end{align}
	where $(a)$ is by Lemma~\ref{tech-lem2}, $(b)$ is due to Theorem~\ref{thm1} by our choice of $\lambda,\alpha$, $(c)$ is by Lemma~\ref{tech-lem3} and $(d)$ is by the choice of $\delta$ in this corollary. Finally, from~\eqref{cor2-helper} we have that
	\begin{align*}
		|\hat{S}_{\text{all}}\setminus \hat{S}_{\text{common}}| \lesssim \frac{K}{\delta}\sum_{k=1}^K \left\{\frac{\sigma_k^2 s \log(p/s)}{n_k}+\frac{1}{n_k}\|\B{r}_k\|_2^2\right\}\lesssim \frac{K^2-K}{K}|S^*_{\text{all}}\setminus S^*_{\text{common}}|.
	\end{align*}
\end{proof}

\subsection{Proof of Theorem~\ref{supthm}}

Before proceeding with the proof of the theorem, we introduce some notation and technical results we will use. \\
\textbf{Notation. } We denote the submatrix of $\bX_k$ with columns indexed by $S\subseteq[p]$ as $\bX_{k,S}$. We denote the projection matrix onto the span of columns of $\bX_k$ indexed by the set $\CS\subseteq[p]$ as $\B{P}_{\bX_{k,\CS}}$. We denote the optimal objective to the least square problem with restricted support
\begin{equation}\label{limitedls}
	\min_{\B{\beta}_{S^c}=0}\frac{1}{n_k}\|\by_k-\bX_k\B{\beta}\|_2^2
\end{equation}
as
\begin{equation}\label{optimalobj}
	\CR_{k,S}(\Bb{y}_k)= \frac{1}{n_k}\Bb{y}_k^{T}(\B{I}_{n_k}-\B{P}_{\bX_{k,S}})\Bb{y}_k.
\end{equation}
For $S_1,S_2\subseteq [p]$, $\B{\Sigma}\in\R^{p\times p}$ positive definite and $S_0=S_2\setminus S_1$, we let
\begin{equation}
	\schur{\B{\Sigma}}{S_1}{S_2} = \B{\Sigma}_{S_0,S_0}- \B{\Sigma}_{S_0,S_1}\B{\Sigma}^{-1}_{S_1,S_1}\B{\Sigma}_{S_1,S_0}.
\end{equation}
Note that $\schur{\B{\Sigma}}{S_1}{S_2}$ is the Schur complement of the block $\B{\Sigma}_{S_1,S_1}$ of the matrix
\begin{equation}\label{gs1s2}
	\B{\Sigma}(S_1,S_2)=  \begin{bmatrix}
		\B{\Sigma}_{S_1,S_1} & \B{\Sigma}_{S_1,S_0} \\
		\B{\Sigma}_{S_0,S_1} & \B{\Sigma}_{S_0,S_0}
	\end{bmatrix}.
\end{equation}
The sample covariance matrix of study $k$ is defined as $\hat{\B{\Sigma}}^{(k)}=\bX_k^T\bX_k/n_k$. We also let $S^*_k=\{j:z_{k,j}^*=1\}$, $\hat{S}_k=\{j:\hat{z}_{k,j}=1\}$ to be the true and estimated supports, respectively. Note that we have $|\hat{S}_k|\leq s$ and $|S_k^*|=s$. Let us define the following events for $S_k\subseteq[p]$ such that $|S_k|\leq s$ and $k\in[K]$:
\begin{equation}\label{events}
	\begin{aligned}
		\mathcal{E}_1(k,S_k) &= \left\{ {(\B{\beta}^*_{k,\tilde{S}_k})}^{T} (\schur{\hat{\B{\Sigma}}^{(k)}}{\CS_k}{\CS^*_k}) \B{\beta}^*_{k,\tilde{S}_k} \geq   0.8\eta_k \frac{|\tilde{\CS}_k|\log p}{n_k}\1(k\in \mathcal{J})\right\} \\
		\mathcal{E}_2(k,S_k) &= \left\{ \frac{2}{n_k}\B{\epsilon}_k^{T} (\B{I}_{n_k}-\B{P}_{\bX_{k,\CS_k}})\bX_{k,\tilde{\CS}_k}\B{\beta}^*_{k,\tilde{\CS}_k}\geq -c_{t_1}\sigma_k \sqrt{{(\B{\beta}^*_{k,\tilde{S}_k})}^{T} (\schur{\hat{\B{\Sigma}}^{(k)}}{\CS_k}{\CS^*_k}) \B{\beta}^*_{k,\tilde{S}_k}}\sqrt{\frac{|\tilde{\CS}_k|\log p}{n_k}}\right\} \\
		\mathcal{E}_3(k,S_k) &= \left\{\B{\epsilon}_k^{T} (\B{P}_{\bX_{k,\CS_k}}-\B{P}_{\bX_{k,\CS^*_k}}) \B{\epsilon}_k\leq  c_{t_2}\sigma_k^2|\tilde{S}_k|\log p\right\} 
	\end{aligned}
\end{equation}
where $c_{t_1},c_{t_2}>0$ are absolute constants, $\B{\beta}^*_{k,S}$ is the subvector of $\B{\beta}_k^*$ restricted to $S$, and $\tilde{S}_k=S^*_k\setminus S_k$. In what follows, we show the events defined above hold with high probability.

\begin{lemma}\label{suptechlem1}
	Under the assumptions of Theorem~\ref{supthm}, we have
	\begin{equation}\label{lem8-1}
		\p\left(\bigcap_{k\in[K]}\bigcap_{\substack{S\subseteq[p] \\|S|\leq s }}\mathcal{E}_1(k,S)\right)\geq 1-Kp^{-8}.
	\end{equation}
\end{lemma}
\begin{proof}
	Let the events $\mathcal{E}_0(k,S)$ for $S\subseteq[p]$ with $|S|\leq 2s$ and $\mathcal{E}$ be defined as
	\begin{equation*}
		\begin{aligned}
			\mathcal{E}_0(k,S)&=\left\{\left\|\hat{\B{\Sigma}}^{(k)}_{S,S}-\B{\Sigma}^{(k)}_{S,S}\right\|_{\text{op}} \lesssim \sqrt{\frac{s\log p}{n_k}}\right\},\\
			\mathcal{E} &= \bigcap_{k\in \mathcal{J}}\bigcap_{\substack{S\subseteq[p] \\  |S|\leq 2s }}\mathcal{E}_0(k,S).
		\end{aligned}
	\end{equation*}
	One has (for example, by Theorem~5.7 of~\cite{rigollet2015high} with $\delta = \exp(-11s\log p)$)
	$$\p(\mathcal{E}_0(k,S))\geq 1-\exp(-11s\log p)$$
	for $k\in\mathcal{J}$ as $n_k\gtrsim s\log p$ is sufficiently large for such $k$ and by Assumption~\ref{as-condition}, $\|\B{\Sigma}^{(k)}_{S,S}\|_{\text{op}}\leq 1$. As a result, by union bound
	\begin{align*}
		\p(\mathcal{E}) &\geq 1-\sum_{k\in\mathcal{J}}\sum_{\substack{ S\subseteq[ p] \\ |S|\leq 2s}}(1-\p(\mathcal{E}_0(k,S)))\geq 1-K\sum_{t=1}^{2s}{p\choose t}\exp(-11s\log p) \\
		&\geq 1-K\sum_{t=1}^{2s} p^{2s} p^{-11s} \geq 1-Kp\times p^{-9}=1-Kp^{-8}.
	\end{align*}
	The rest of the proof is on event $\mathcal{E}$. As $|S|,|S^*|\leq s$, for $k\in\mathcal{J}$,
	\begin{align}
		\|\hat{\B{\Sigma}}^{(k)}(\CS,\CS^*_k)-\B{\Sigma}^{(k)}(\CS,\CS^*_k)\|_{\text{op}} &  \leq c_b \sqrt{\frac{s\log p}{n_k}} :=\pi_k\label{stewineq2}
	\end{align}
	for some absolute constant $c_b>0$ where $\B{\Sigma}(S_1,S_2)$ is defined in~\eqref{gs1s2}. Let for $k\in\mathcal{J}$, 
	$$n_k\gtrsim  s\log p/\phi_k^2$$
	be sufficiently large such that $\pi_k\leq \phi_k/10$. Therefore, one has 
	\begin{align*}
		\lambda_{\min}(\schur{\hat{\B{\Sigma}}^{(k)}}{\CS}{\CS^*_k})&\stackrel{(a)}{\geq} \lambda_{\min}({\hat{\B{\Sigma}}}^{(k)}({\CS},{\CS^*_k})) \\
		&\stackrel{(b)}{\geq}  \lambda_{\min}({\B{\Sigma}^{(k)}}({\CS},{\CS^*_k})) - \|{\hat{\B{\Sigma}}}^{(k)}({\CS},{\CS^*_k}) - {\B{\Sigma}}^{(k)}({\CS},{\CS^*_k})\|_{\text{op}}  \geq \phi_k - 0.1\phi_k>0.8\phi_k
	\end{align*}
	where $\B{\Sigma}(S_1,S_2)$ is defined in~\eqref{gs1s2}, $(a)$ is by Corollary 2.3 of \cite{zhang2006schur} and $(b)$ is due to Weyl's inequality.
	Finally, by setting $\tilde{S}=S_k^*\setminus S$:
	\begin{align*}
		{(\B{\beta}^*_{k,\tilde{\CS}})}^{T} (\schur{\hat{\B{\Sigma}}^{(k)}}{\CS}{\CS_j^*}) \B{\beta}^*_{k,\tilde{\CS}} \geq \lambda_{\min}(\schur{\hat{\B{\Sigma}}^{(k)}}{\CS}{\CS^*})\|\B{\beta}^*_{k,\tilde{\CS}}\|_2^2 \geq 0.8\phi_k \|\B{\beta}^*_{k,\tilde{\CS}}\|_2^2\geq 0.8\eta_k \frac{|\tilde{\CS}|\log p}{n_k}
	\end{align*}
	where the last inequality is achieved by substituting $\beta_{\min}$ condition from Assumption~\ref{as-beta-min}.
\end{proof}

\begin{lemma}\label{suptechlem2}
	One has
	\begin{equation}
		\p\left(\bigcap_{k\in[K]}\bigcap_{\substack{S_k\subseteq[p]\\ |S_k|\leq s} }\mathcal{E}_2(k,S_k)\right)\geq 1-s^2Kp^{-18}.
	\end{equation}
\end{lemma}

\begin{proof}
	The proof follows a similar path to the proof of Lemma~13 of~\cite{behdin2021sparse}. Let
	\begin{equation*}
		\B{\gamma}^{(k,\CS_k)}= (\B{I}_{n_k}-\B{P}_{\bX_{k,\CS_k}})\bX_{k,\tilde{\CS}_k}\B{\beta}^*_{k,\tilde{\CS}_k}.
	\end{equation*}
	Based on Lemma~13 of~\cite{behdin2021sparse} we achieve
	\begin{equation}
		\p\left(\frac{\B{\epsilon}_{k}^T \B{\gamma}^{(k,\CS_k)}}{n_k}<- c_{t_1}\sigma_k\sqrt{{(\B{\beta}^*_{k,\tilde{\CS}_k})}^T (\schur{\B{\hat{\Sigma}}^{(k)}}{\CS_k}{\CS_k^*}) \B{\beta}^*_{k,\tilde{\CS}_k}}\sqrt{\frac{2|\tilde{S}_k|\log p}{n_k}}\right) \leq \exp(-20|\tilde{S}_k|\log p)
	\end{equation}
	for some sufficiently large absolute constant $c_{t_1}$. Finally, we complete the proof by using union bound over all possible choices of $k,S$. As a result, the probability of the desired event in the lemma being violated is bounded as
	\begin{align*}
		\sum_{k=1}^K \sum_{|S_k|=1}^s\sum_{t=1}^s \sum_{\substack{ S_k\subseteq[p]  \\|S_k^*\setminus S_k|=t}}\exp(-20t\log p)&= \sum_{k=1}^K  \sum_{|S_k|=1}^s \sum_{t=1}^s {s\choose t}{p-s \choose {|S_k|-(s-t) }}\exp(-20t\log p) \\
		& \leq K \sum_{|S_k|=1}^s \sum_{t=1}^s p^{t}p^{|S_k|-(s-t)} \exp(-20t\log p) \\
		& \stackrel{(a)}{\leq} K \sum_{|S_k|=1}^s \sum_{t=1}^s p^{t}p^{t} \exp(-20t\log p) \\
		& \leq K \sum_{|S_k|=1}^s\sum_{t=1}^s \exp(-18t\log p)\\
		& \leq K \sum_{|S_k|=1}^s\sum_{t=1}^s \exp(-18\log p)\\
		& \leq s^2K  p^{-18} = s^2Kp^{-18}
	\end{align*}
	where $(a)$ is true as $|S_k|\leq s$, or $|S_k|-s\leq 0$.
\end{proof}

\begin{lemma}\label{suptechlem3}
	One has
	\begin{equation}
		\p\left(\bigcap_{k\in[K]}\bigcap_{\substack{S_k\subseteq[p]\\ |S_k|=s} }\mathcal{E}_3(k,S_k)\right)\geq 1-4s^2Kp^{-18}.
	\end{equation}
\end{lemma}
\begin{proof}
	The proof of this lemma follows the proof of Lemma 15 of~\cite{behdin2021sparse} and a union bound argument similar to the one in Lemma~\ref{suptechlem2}.
\end{proof}

\begin{proof}[\textbf{Proof of Theorem~\ref{supthm}}]

	The proof of this theorem on the intersection of events defined in~\eqref{events}. By Lemmas~\ref{suptechlem1} to~\ref{suptechlem3} and union bound, this happens with probability at least 
	\begin{equation}\label{supthmprob}
		1-6s^2Kp^{-8}.
	\end{equation}
	Recalling the definition of $\CR_{k,S}(\cdot)$ in~\eqref{optimalobj}, one has (see calculations leading to~(89) of~\cite{behdin2021sparse}),
	\begin{align}\label{objexpanded}
		\CR_{k,\CS_k}(\Bb{y}_k)  ={(\B{\beta}^*_{k,\tilde{\CS}_k})}^T (\schur{\hat{\B{\Sigma}}^{(k)}}{\CS_k}{\CS_k^*}) \B{\beta}^*_{k,\tilde{\CS}_k} + \frac{2}{n_k}\B{\epsilon}_{k}^T (\B{I}_{n_k}-\B{P}_{\bX_{k,\CS_k}})\bX_{k,\tilde{\CS}_k}\B{\beta}^*_{\tilde{\CS}_k} + \frac{1}{n_k} \B{\epsilon}_{k}^T(\B{I}_{n_k}-\B{P}_{\bX_{k,\CS_k}})\B{\epsilon}_{k}
	\end{align}
	and
	\begin{align}\label{objexpandedopt}
		\CR_{k,\CS^*_k}(\Bb{y}_k)  = \frac{1}{n_k} \B{\epsilon}_{k}^T(\B{I}_{n_k}-\B{P}_{\bX_{k,\CS^*_k}})\B{\epsilon}_{k}.
	\end{align}
	As a result, one can write
	\begin{align}
		& \CR_{k,\CS_k}(\Bb{y}_k)-\CR_{k,\CS^*_k}(\Bb{y}_k) \nonumber\\
		\stackrel{}{=} &  {(\B{\beta}^*_{k,\tilde{S}_k})}^{T} (\schur{\hat{\B{\Sigma}}^{(k)}}{\CS_k}{\CS_k^*}) \B{\beta}^*_{k,\tilde{S}_k}+\frac{2}{n_k}\B{\epsilon}_k^{T} (\B{I}_{n_k}-\B{P}_{\bX_{k,\CS_k}})\bX_{k,\tilde{\CS}_k}\B{\beta}^*_{k,\tilde{\CS}_k}-\frac{1}{n_k}\B{\epsilon}_k^{T} (\B{P}_{\bX_{k,\CS_k}}-\B{P}_{\bX_{k,\CS^*_k}}) \B{\epsilon}_k \nonumber \\
		\stackrel{(a)}{\geq} & {(\B{\beta}^*_{k,\tilde{S}_k})}^{T} (\schur{\hat{\B{\Sigma}}^{(k)}}{\CS_k}{\CS_k^*}) \B{\beta}^*_{k,\tilde{S}_k}-c_{t_1}\sigma_k \sqrt{{(\B{\beta}^*_{k,\tilde{S}_k})}^{T} (\schur{\hat{\B{\Sigma}}^{(k)}}{\CS_k}{\CS_k^*}) \B{\beta}^*_{k,\tilde{S}_k}}\sqrt{\frac{|\tilde{\CS}_k|\log p}{n_k}}-c_{t_2}\sigma_k^2\frac{|\tilde{S}_k|\log p}{n_k} \nonumber \\
		\stackrel{(b)}{\geq} & \frac{1}{2}{(\B{\beta}^*_{k,\tilde{S}_k})}^{T} (\schur{\hat{\B{\Sigma}}^{(k)}}{\CS_k}{\CS^*_k}) \B{\beta}^*_{k,\tilde{S}_k}-\left(2c_{t_1}^2+c_{t_2}\right) \sigma_k^2 \frac{|\tilde{S}_k|\log p}{n_k} \nonumber \\
		\stackrel{(c)}{\geq} & 0.4\eta_k \frac{|\tilde{\CS}_k|\log p}{n_k}\1(k\in \mathcal{J})-\left(2c_{t_1}^2+c_{t_2}\right) \sigma_k^2 \frac{|\tilde{S}_k|\log p}{n_k} \nonumber \\
		\stackrel{}{\geq } & 0.4\eta_k \frac{|\tilde{\CS}_k|\log p}{\bar{n}}\1(k\in \mathcal{J})-\left(2c_{t_1}^2+c_{t_2}\right) \sigma_k^2 \frac{|\tilde{S}_k|\log p}{\barbelow{n}}\nonumber\\
		\stackrel{(d)}{\geq } & 0.4\eta_kc_n \frac{|\tilde{\CS}_k|\log p}{\barbelow{n}}\1(k\in \mathcal{J})-\left(2c_{t_1}^2+c_{t_2}\right) \sigma_k^2 \frac{|\tilde{S}_k|\log p}{\barbelow{n}}\label{supthm-helper25}
	\end{align}
	where $(a)$ is by events $\mathcal{E}_2,\mathcal{E}_3$ defined in~\eqref{events}, $(b)$ is by the inequality $-2ab\geq -2a^2-b^2/2$, $(c)$ is due to event $\mathcal{E}_1$ in~\eqref{events} and $(d)$ is by Assumption~\ref{as-cn}. By optimality of $\hat{S}_k$, we have
	\begin{equation}\label{supthmfromopt}
		\sum_{k=1}^K \CR_{k,\CS^*_k}(\Bb{y}_k)+h(\B{z}_1^*,\cdots,\B{z}_K^*)\geq  \sum_{k=1}^K\CR_{k,\hat{\CS}_k}(\Bb{y}_k)+h(\hat{\B{z}}_1,\cdots,\hat{\B{z}}_K)
	\end{equation}
	where
	$$h(\B{z}_1,\cdots,\B{z}_K)=\delta\sum_{k=1}^K \|\B{z}_k - \bar{\B{z}}\|_2^2.$$
	Consequently, 
	\begin{align}
		\delta |S^*_{\text{all}}\setminus S^*_{\text{common}}| & \stackrel{(a)}{\geq}h(\B{z}_1^*,\cdots,\B{z}_K^*) \nonumber \\
		& \stackrel{(b)}{\geq}\sum_{k=1}^K\CR_{k,\hat{\CS}_k}(\Bb{y}_k)- \sum_{k=1}^K\CR_{k,\CS^*_k}(\Bb{y}_k)+h(\hat{\B{z}}_1,\cdots,\hat{\B{z}}_K)\nonumber \\
		& \stackrel{(c)}{\geq} \sum_{k=1}^K\CR_{k,\hat{\CS}_k}(\Bb{y}_k)- \sum_{k=1}^K\CR_{k,\CS^*_k}(\Bb{y}_k)+\frac{\delta|\hat{S}_{\text{all}}\setminus \hat{S}_{\text{common}}|}{K}\nonumber \\
		& \stackrel{(d)}{\geq}\frac{\delta|\hat{S}_{\text{all}}\setminus \hat{S}_{\text{common}}|}{K}+\frac{\log p}{\barbelow{n}}\sum_{k=1}^K|\tilde{S}_k|\left[0.4c_n\eta_k \1(k\in\cJ)-(2c_{t_1}^2+c_{t_2})\sigma_k^2\right]\label{thm2-tobeused-prop}
	\end{align}
	where $(a)$ is a result of Lemma~\ref{tech-lem3} and regularity, $(b)$ is due to~\eqref{supthmfromopt}, $(c)$ is by Lemma~\ref{tech-lem2} and $(d)$ is by~\eqref{supthm-helper25}. This completes the proof.
	
\end{proof}

\subsection{Proof of Corollary~\ref{cor3}}

\begin{proof}
	Note that by the choice of $\delta$ here, the solution has common support from Corollary~\ref{cor1}. As a result,
	\begin{align}
		0  & \stackrel{(a)}{=} \delta |S^*_{\text{all}}\setminus S^*_{\text{common}}| \nonumber \\
		& \stackrel{(b)}{\geq}\frac{\delta|\hat{S}_{\text{all}}\setminus \hat{S}_{\text{common}}|}{K}+\frac{\log p}{\barbelow{n}}\sum_{k=1}^K|\tilde{S}_k|\left[0.4c_n\eta_k \1(k\in\cJ)-(2c_{t_1}^2+c_{t_2})\sigma_k^2\right] \nonumber \\
		& \stackrel{(c)}{= }|\tilde{S}_1|\frac{\log p}{\barbelow{n}}\sum_{k=1}^K\left[0.4c_n\eta_k \1(k\in\cJ)-(2c_{t_1}^2+c_{t_2})\sigma_k^2\right] \stackrel{(d)}{\geq} 0\label{cor3-helper}
	\end{align}
	where $(a)$ is due to the fact that the underlying model has common support, $(b)$ is by Theorem~\ref{supthm}, $(c)$ is because the optimal solution has common support, as described above, so $|\hat{S}_{\text{all}}\setminus\hat{S}_{\text{common}}|=0$ and $\tilde{S}_1=\cdots=\tilde{S}_K$. $(d)$ is by taking $\eta_k$ large enough such that
	$$\sum_{k=1}^K\left[0.4c_n\eta_k \1(k\in\cJ)-(2c_{t_1}^2+c_{t_2})\sigma_k^2\right]>0,$$
	as in Condition~\eqref{supreccond}. From~\eqref{cor3-helper}, we must have $|\tilde{S}_1|=0$ which completes the proof.
\end{proof}

\subsection{Proof of Corollary~\ref{cor4}}
\begin{proof}
	Here, we choose
	\begin{equation}\label{cor4-delta}
		\delta=\frac{s\log p}{\barbelow{n}|S^*_{\text{all}}\setminus S^*_{\text{common}}|}\sum_{k=1}^K \sigma_k^2 .
	\end{equation}
	From Theorem~\ref{supthm} we have that with high probability,
	\begin{align}
		\frac{s\log p}{\barbelow{n}}\sum_{k=1}^K \sigma_k^2 & \stackrel{(a)}{\gtrsim}  \delta |S^*_{\text{all}}\setminus S^*_{\text{common}}| +  \frac{\log p}{\barbelow{n}}\sum_{k=1}^K (2c_{t_1}^2+c_{t_2})|\tilde{S}_k|\sigma_k^2 \nonumber \\
		&\stackrel{(b)}{\geq}\frac{\delta|\hat{S}_{\text{all}}\setminus \hat{S}_{\text{common}}|}{K}+ \frac{\log p}{\barbelow{n}}\sum_{k=1}^K0.4c_n\eta_k \1(k\in\cJ)|\tilde{S}_k|\label{cor4-helper}
	\end{align}
	where $(a)$ is by choice of $\delta$ in this corollary and the fact $|\tilde{S}_k|\leq s$, and $(b)$ is by Theorem~\ref{supthm}. From~\eqref{cor4-helper}, we have that 
	$$|\hat{S}_{\text{all}}\setminus \hat{S}_{\text{common}}|\lesssim K |S^*_{\text{all}}\setminus S^*_{\text{common}}|$$
	which completes the first part of the proof. Next, assume there is some $j\in S^*_{\text{common}}$ such that $j\notin \hat{S}_k$ for $k\in\mathcal{J}$. As a result, for $k\in\mathcal{J}$, we have that $|\tilde{S}_k|\geq 1$. Consequently, from~\eqref{cor4-helper} we have
	$$s\sum_{k=1}^K\sigma_k^2\gtrsim \sum_{k\in\mathcal{J}}\eta_k$$
	which is a contradiction. As a result, $j\in \cup_{k\in\mathcal{J}}\hat{S}_k$.
\end{proof}

\subsection{Proof of Proposition~\ref{app-prop1}} \label{app:prop1_proof}

Here, we prove a more general version of Proposition~\ref{app-prop1} where we allow all regularization parameters $\lambda,\delta,\alpha$ to be possibly nonzero. In particular, for a given set of regularization parameters $\lambda,\delta,\alpha$, we assume the approximate solution $(\check{\mathbb{B}},\check{\B{Z}})$ satisfies
\begin{align}\label{mip-lowerbound-general}
	&   \sum_{k=1}^K \frac{1}{n_k} \norm{  \mathbf{y}_k - \mathbb{X}_k \check{\boldsymbol{\beta}}_k }_2^2  + \delta \sum_{k=1}^K\norm{\check{\boldsymbol{z}}_k - \bar{\check{\boldsymbol{z}}}}_2^2 + \lambda \sum_{k=1}^K \|\check{\B\beta}_k-\bar{\check{\B\beta}}\|_2^2 +\alpha\|\check{\mathbb{B}}\|_F^2 \nonumber\\
	\leq & \sum_{k=1}^K \frac{1}{n_k} \norm{  \mathbf{y}_k - \mathbb{X}_k \hat{\boldsymbol{\beta}}_k }_2^2  + \delta \sum_{k=1}^K\norm{\hat{\boldsymbol{z}}_k - \bar{\hat{\boldsymbol{z}}}}_2^2+   \lambda \sum_{k=1}^K \|\hat{\B\beta}_k-\bar{\hat{\B\beta}}\|_2^2 +\alpha\|\hat{\mathbb{B}}\|_F^2+\tau
\end{align}
for some optimality gap $\tau\geq 0$ where $(\hat{\bB},\hat{\B Z})$ is an optimal solution to Problem~\eqref{suppHet} for the given values of $\lambda,\delta,\alpha$. We prove the following:

\begin{proposition}\label{app:app-prop1}
	Suppose $\{\check{\B{\beta}}_k,\check{\B{z}}_k\}_{k=1}^K$ is as defined in~\eqref{mip-lowerbound-general}. Then, under our assumed model setup w.h.p. we have
	\begin{align*}
		&    \sum_{k=1}^K \left\{\frac{1}{n_k} \|\bX_k(\B{\beta}_k^*-\check{\B{\beta}}_k)\|_2^2+ \delta  \|\check{\B{z}}_k-\bar{\check{\B{z}}}\|_2^2+\lambda \sum_{k=1}^K \|\check{\B\beta}_k-\bar{\check{\B\beta}}\|_2^2\right\}\\ 
		\lesssim & \sum_{k=1}^K\left\{ \frac{\sigma_k^2 s \log(p/s)}{n_k}+  \delta \sum_{k=1}^K\norm{\boldsymbol{z}_k^* - \bar{\boldsymbol{z}}^*}_2^2 + \lambda \sum_{k=1}^K \|\B\beta^*_k-\bar{\B\beta}^*\|_2^2+\frac{1}{n_k} \|\B{r}_k\|_2^2\right\}+\alpha\| \mathbb{B}^*\|_F^2+\tau.
	\end{align*}
\end{proposition}

\begin{proof}
	Using a notation similar to the one from the proof of Theorem~\ref{thm1}, by~\eqref{mip-lowerbound-general}
	
	\begin{align}
		&   \sum_{k=1}^K \frac{1}{n_k} \norm{  \mathbf{y}_k - \mathbb{X}_k \check{\boldsymbol{\beta}}_k }_2^2  + \delta \sum_{k=1}^K\norm{\check{\boldsymbol{z}}_k - \bar{\check{\boldsymbol{z}}}}_2^2 + \lambda \sum_{k=1}^K \|\check{\B\beta}_k-\bar{\check{\B\beta}}\|_2^2 +\alpha\|\check{\mathbb{B}}\|_F^2 \\
		\leq & \sum_{k=1}^K \frac{1}{n_k} \norm{  \mathbf{y}_k - \mathbb{X}_k \hat{\boldsymbol{\beta}}_k }_2^2  + \delta \sum_{k=1}^K\norm{\hat{\boldsymbol{z}}_k - \bar{\hat{\boldsymbol{z}}}}_2^2+   \lambda \sum_{k=1}^K \|\hat{\B\beta}_k-\bar{\hat{\B\beta}}\|_2^2 +\alpha\|\hat{\mathbb{B}}\|_F^2+\tau \nonumber \\
		\stackrel{(a)}{\leq} & \sum_{k=1}^K \frac{1}{n_k} \norm{  \mathbf{y}_k - \mathbb{X}_k \boldsymbol{\beta}^*_k }_2^2  + \delta \sum_{k=1}^K\norm{\boldsymbol{z}_k^* - \bar{\boldsymbol{z}}^*}_2^2 + \lambda \sum_{k=1}^K \|\B\beta^*_k-\bar{\B\beta}^*\|_2^2 +\alpha\| \mathbb{B}^*\|_F^2+\tau
	\end{align}
	
	where $(a)$ is due to optimality of $(\hat{\B\beta}_k,\hat{\B z}_k)$ and feasibility of $(\B\beta^*_k,\B z_k^*)$. Rest of the proof follows from the proof of Theorem~\ref{thm1}.
\end{proof}
\subsection{Proof of Proposition~\ref{app-prop2}}
\begin{proof}
	Throughout this proof, we use notations from the proof of Theorem~\ref{supthm}. Particularly, let $\check{S}_k=\{j:\check z_{k,j}=1\}$ and we have
	\begin{align}
		\sum_{k=1}^K \CR_{k,\check{\CS}_k}(\Bb{y}_k)+h(\check{\B{z}}_1,\cdots,\check{\B{z}}_K)& \stackrel{(a)}{\leq}   \sum_{k=1}^K \frac{1}{n_k} \norm{  \mathbf{y}_k - \mathbb{X}_k \check{\boldsymbol{\beta}}_k }_2^2  + \delta \sum_{k=1}^K\norm{\check{\boldsymbol{z}}_k - \bar{\check{\boldsymbol{z}}}}_2^2 \nonumber \\
		&\stackrel{(b)}{\leq} \sum_{k=1}^K\CR_{k,\hat{\CS}_k}(\Bb{y}_k)+h(\hat{\B{z}}_1,\cdots,\hat{\B{z}}_K)+\tau \nonumber \\
		& \stackrel{(c)}{\leq} \sum_{k=1}^K\CR_{k,\CS^*_k}(\Bb{y}_k)+h(\B{z}^*_1,\cdots,\B{z}^*_K)+\tau\label{propH2-1}
	\end{align}
	where $(a)$ is by definition of $\CR_{k,{\CS}_k}$ in~\eqref{optimalobj}, $(b)$ is by~\eqref{mip-lowerbound}, and $(c)$ is by optimality of $\hat{\B Z},\hat{\mathbb{B}}$. By the chain of inequalities leading to~\eqref{thm2-tobeused-prop}, we have
	\begin{align}
		\delta |S^*_{\text{all}}\setminus S^*_{\text{common}}| & \stackrel{}{\geq}h(\B{z}_1^*,\cdots,\B{z}_K^*) \nonumber \\
		& \stackrel{(a)}{\geq}\sum_{k=1}^K\CR_{k,\check{\CS}_k}(\Bb{y}_k)- \sum_{k=1}^K\CR_{k,\CS^*_k}(\Bb{y}_k)+h(\check{\B{z}}_1,\cdots,\check{\B{z}}_K)-\tau\nonumber \\
		& \stackrel{}{\geq} \sum_{k=1}^K\CR_{k,\check{\CS}_k}(\Bb{y}_k)- \sum_{k=1}^K\CR_{k,\CS^*_k}(\Bb{y}_k)+\frac{\delta|\check{S}_{\text{all}}\setminus \check{S}_{\text{common}}|}{K}-\tau\nonumber \\
		& \stackrel{(b)}{\geq}\frac{\delta|\check{S}_{\text{all}}\setminus \check{S}_{\text{common}}|}{K}+\frac{\log p}{\barbelow{n}}\sum_{k=1}^K|\bar{S}_k|\left[0.4c_n\eta_k \1(k\in\cJ)-(2c_{t_1}^2+c_{t_2})\sigma_k^2\right]-\tau
	\end{align}
	where $(a)$ is by~\eqref{propH2-1} and $(b)$ is by~\eqref{supthm-helper25}. This completes the proof.
\end{proof}

\section{Simulations} \label{main_sims}

\subsection{Simulation Scheme} \label{sim_scheme}

We set $\sigma^2_{\beta} = 50$ to simulate high heterogeneity in model coefficients across tasks.  This was selected based on the sample variances of regression coefficient estimates in our data applications. We estimated this by fitting separate Ridge regressions to each task (the features and outcomes were centered and scaled to make the estimates more comparable across applications) and then estimating the mean (across coefficients) of the sample variances (across tasks): $\frac{1}{p(K-1)} \sum_{j=1}^p \sum_{k=1}^K (\hat{\beta}_{k,j} -\frac{1}{K} \sum_{l=1}^K \hat{\beta}_{l,j})^2$. We calculated estimates of $3.49 \times 10^{-5}$ and $51.03$ in the cancer and neuroscience application, respectively, indicating that the two multi-task applications have considerably different levels of heterogeneity across tasks according to this measure and that our $\sigma^2_{\beta}$ fell within a range justified by our data.

\subsection{Estimation and Hyperparameter Tuning} \label{sims_tuning}

For our $\ell_0$-constrained methods, we used solutions (solved with block CD without local search) to a model with a support size of $4s$ as a warm start for our final model. Using this warm start, we fit final models (with support size $s$) with block CD followed by up to 50 iterations of local search. We initialized the parameter values of the warm start block CD at a matrix of zeros. 

When fitting paths of solution for hyperparameter tuning, we constructed a hyperparameter grid of length 100 and used the solution at a given hyperparameter value as a warm start for the model trained at the subsequent hyperparameter value. For a fixed $s$ value, we ordered the $\lambda$ values (associated with the $\bar{\boldsymbol{\beta}}$ penalty) from lowest to highest and $\alpha$ values (associated with the $\norm{\mathbb{B}}_F^2$ penalty) from highest to lowest. We did not fit models in which both $\alpha$ and $\lambda$ were nonzero since both can induce shrinkage. Finally, for a given $s$ and $\lambda$ (or $\alpha$) value, we ordered $\delta$ values (associated with the $\bar{\boldsymbol{z}}$ penalty) from lowest to highest.

We tuned hyperparameters of the group penalties over a grid of values that yielded coefficient solutions with sparsity level $s$ in two steps. In step one, we tuned over a grid of 5000 values by setting $\texttt{nlambda=5000}$ in the $\texttt{sparsegl}$ and $\texttt{grpreg}$ packages. We then identified the subset of tuning parameter values that produced coefficient estimates with average cardinality no greater than $s$: ${\Lambda} := \{\lambda~:~ \frac{1}{K} \sum_{k=1}^K \norm{\hat{\boldsymbol{\beta}}_k^{(\lambda)}}_0 = \bar{s} \leq s \}$, where $\hat{\boldsymbol{\beta}}_k^{(\lambda)}$ denotes the vector of coefficient estimates for task $k$, fit with hyperparameter value $\lambda$. If $| {\Lambda} | \leq 100$, we selected the hyperparameter value in $\Lambda$ that achieved the lowest cross-validated error. Otherwise, we selected the hyperparameter value that achieved the lowest cross-validated error from a random subset (drawn with uniform probability), $\tilde{\Lambda} \subseteq {\Lambda}$,  such that $| \tilde{\Lambda} | = 100$. 

Similarly, when tuning the Bbar and Zbar+L2 methods, we tuned in two steps. Since the scale of the Bbar and Zbar penalties, relative to the MTL squared error loss, can differ between datasets, the two-stage tuning procedure allowed us to quickly identify a reasonable neighborhood of hyperparameter values without requiring us to tune over a large grid. Specifically, in step one, we tuned over an initial grid of values for $\lambda$ or $\delta$, of roughly length 90. We identified the optimal hyperparameter, $\lambda^*_1$ or $\delta^*_1$ and created a small grid from the set in that neighborhood ($[\lambda^*_1 / 10, 10 *\lambda^*_1]$ or $[\delta^*_1 / 10, 10 *\delta^*_1]$) of roughly size 10. In the second stage, we selected the final hyperparameter values, $\lambda^*_2$ or $\delta^*_2$ from this smaller grid.

\subsection{Performance Metrics} \label{performance}

We denote the outcome vector and design matrix of the \textit{test} dataset for task $k$ as $\mathbf{y}^+_{k}$ and $\mathbb{X}^+_{k}$, respectively. 

We define the following estimation metrics:

\begin{itemize}
	\item \textbf{Average out of sample prediction performance (RMSE):} $\frac{1}{K \sqrt{n_k}} \sum_{k=1}^K \norm{\mathbf{y}_k^+ - \mathbb{X}_k^+ \hat{\boldsymbol{\beta}}_k}_2$ 
	
	\item \textbf{Average coefficient estimation performance (RMSE):} $\frac{1}{K \sqrt{p+1}} \sum_{k=1}^K \norm{ \boldsymbol{\beta}_k - \hat{\boldsymbol{\beta}}_k}_2$ 
	
	\item \textbf{Average True Positives} The average number of nonzero coefficients in both $\boldsymbol{\beta}_k$ and $\hat{\boldsymbol{\beta}}_k$: $\frac{1}{K}\sum_k TP_k$, where $TP_k = \boldsymbol{z}_k^T \hat{\boldsymbol{z}}_k$
	
	\item \textbf{Average False Negatives} The average number of zero coefficients in $\hat{\boldsymbol{\beta}}_k$ that are nonzero in $\boldsymbol{\beta}_k$: $\frac{1}{K}\sum_k FP_k$, where $FP_k = \boldsymbol{z}_k^T (\mathbbm{1} - \hat{\boldsymbol{z}}_k)$
	
	\item \textbf{Average Support Recovery F1 Score}: $\sum_k\frac{2P_k R_k}{K(P_k + R_k)}$,  where $R_k = TP_k / s_k$ and $P_k =  TP_k/(TP_k + FP_k)$. When the support is fully recovered, this quantity equals 1
	
	\item \textbf{Support Homogeneity}: $\binom{K}{2}^{-1} p^{-1} \sum_{j=1}^p \sum_{k=1}^K \sum_{l < k} \hat{z}_{k,j} \hat{z}_{l,j}$, where $\hat{z}_{k,j} = (\hat{\beta}_{k,j} \neq 0)$. When the supports are identical across tasks, this quantity equals 1 and if supports are totally different, this quantity equals 0
\end{itemize}

\subsection{Main Text Simulation Results} \label{supp_sims}

We provide more extensive versions of the simulation results in Section \ref{sims}. We first provide results of the methods presented in the main text across three levels of covariate correlation levels (exponential correlation $\rho$ where $\rho \in \{0.2, 0.5, 0.8\}$) and then present results from the same simulations (and same performance criteria) on a second subset of methods. We divide these into two plots for ease of comparison given the number of methods explored.

\subsubsection{Prediction Performance: RMSE}

\begin{figure}[H]
	\centerline{
		\includegraphics[width=0.85\linewidth]{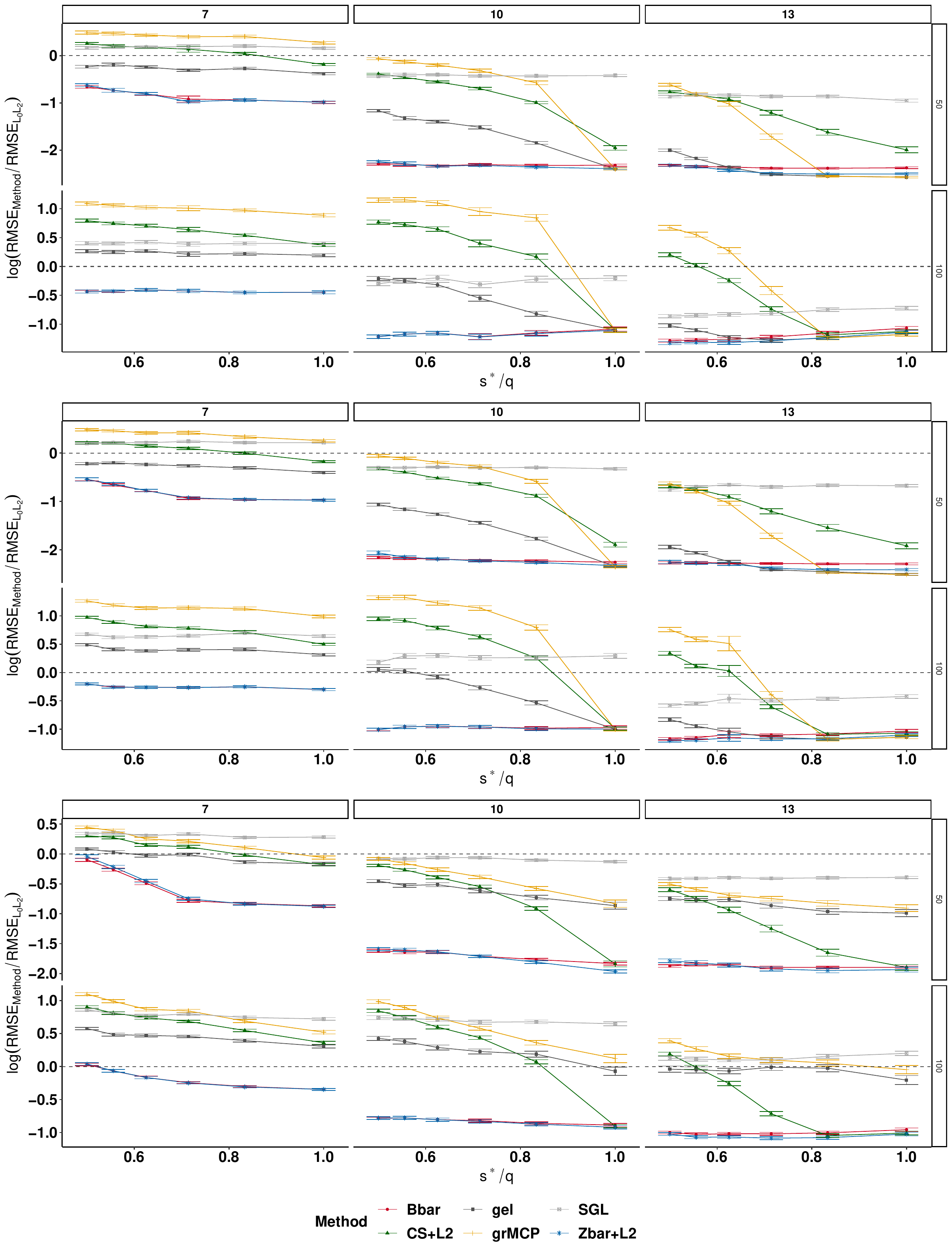}
	}
	\caption{\footnotesize  Out-of-sample prediction performance averaged across tasks for different sample sizes, $n_k$ (displayed on the horizontal panels), values of support size, $s$, for model fitting (vertical panels), degrees of support heterogeneity, $s^*/q$ (x-axis), and covariate correlation levels $\rho \in \{0.2, 0.5,0.8\}$ (from top to bottom figure). Low $s^*/q$ indicates high support heterogeneity and $s/q = 1$ indicates that all tasks had identical support. The performance of each method is presented in reference to the performance of a $L_0 L_2$ and thus lower values indicate superior (relative) prediction accuracy.} 
	\label{fig:sims_rmse_supp}
\end{figure}

\begin{figure}[H]
	\centerline{
		\includegraphics[width=0.9\linewidth]{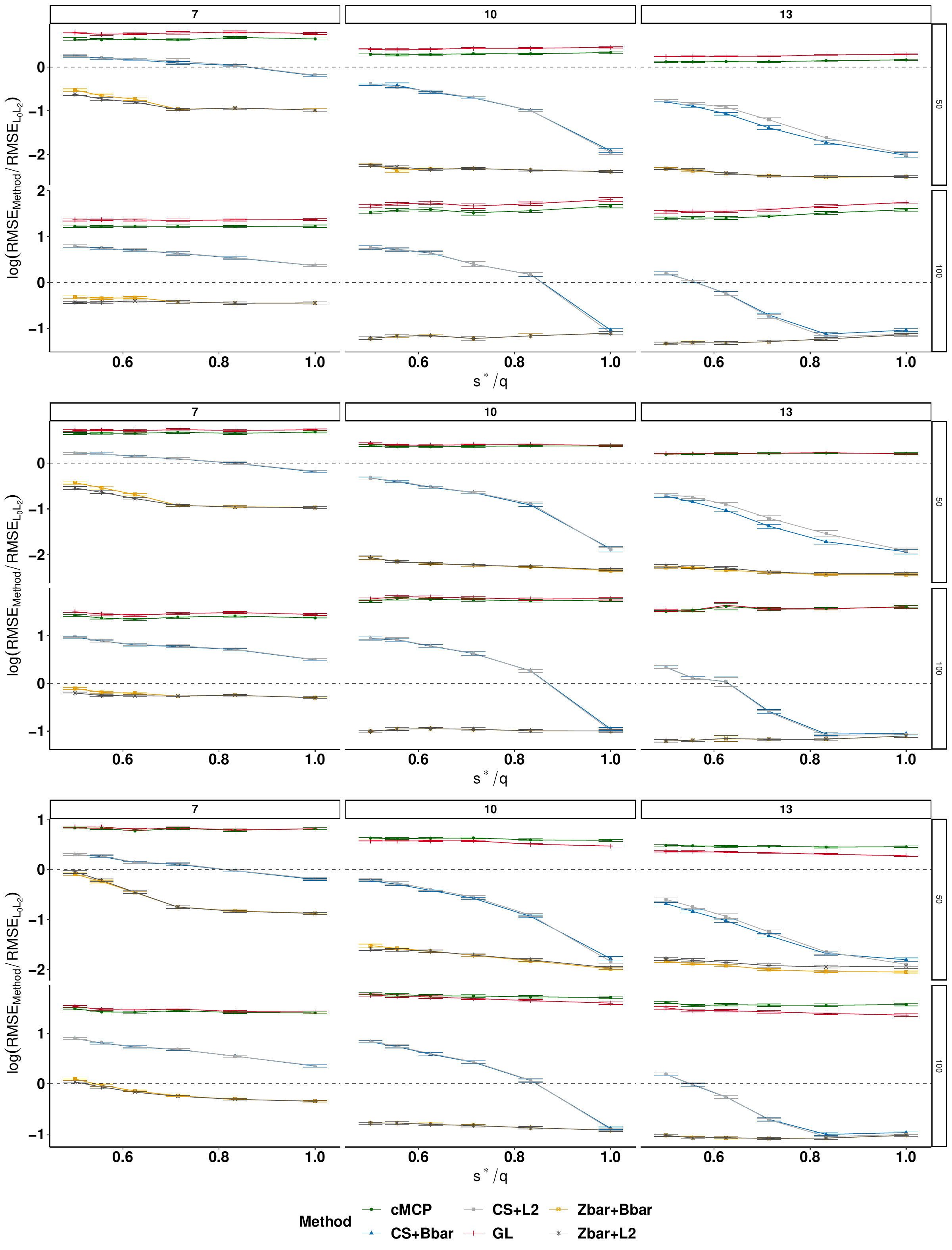}
	}
	\caption{\footnotesize Out-of-sample prediction performance averaged across tasks for different sample sizes, $n_k$ (displayed on the horizontal panels), values of support size, $s$, for model fitting (vertical panels), degrees of support heterogeneity, $s^*/q$ (x-axis), and covariate correlation levels $\rho \in \{0.2, 0.5,0.8\}$ (from top to bottom figure). Low $s^*/q$ indicates high support heterogeneity and $s/q = 1$ indicates that all tasks had identical support. The performance of each method is presented in reference to the performance of a $L_0 L_2$ and thus lower values indicate superior (relative) prediction accuracy.} 
	\label{fig:sims_rmse_supp2}
\end{figure}

\begin{figure}[H]
	\centerline{
		\includegraphics[width=0.9\linewidth]{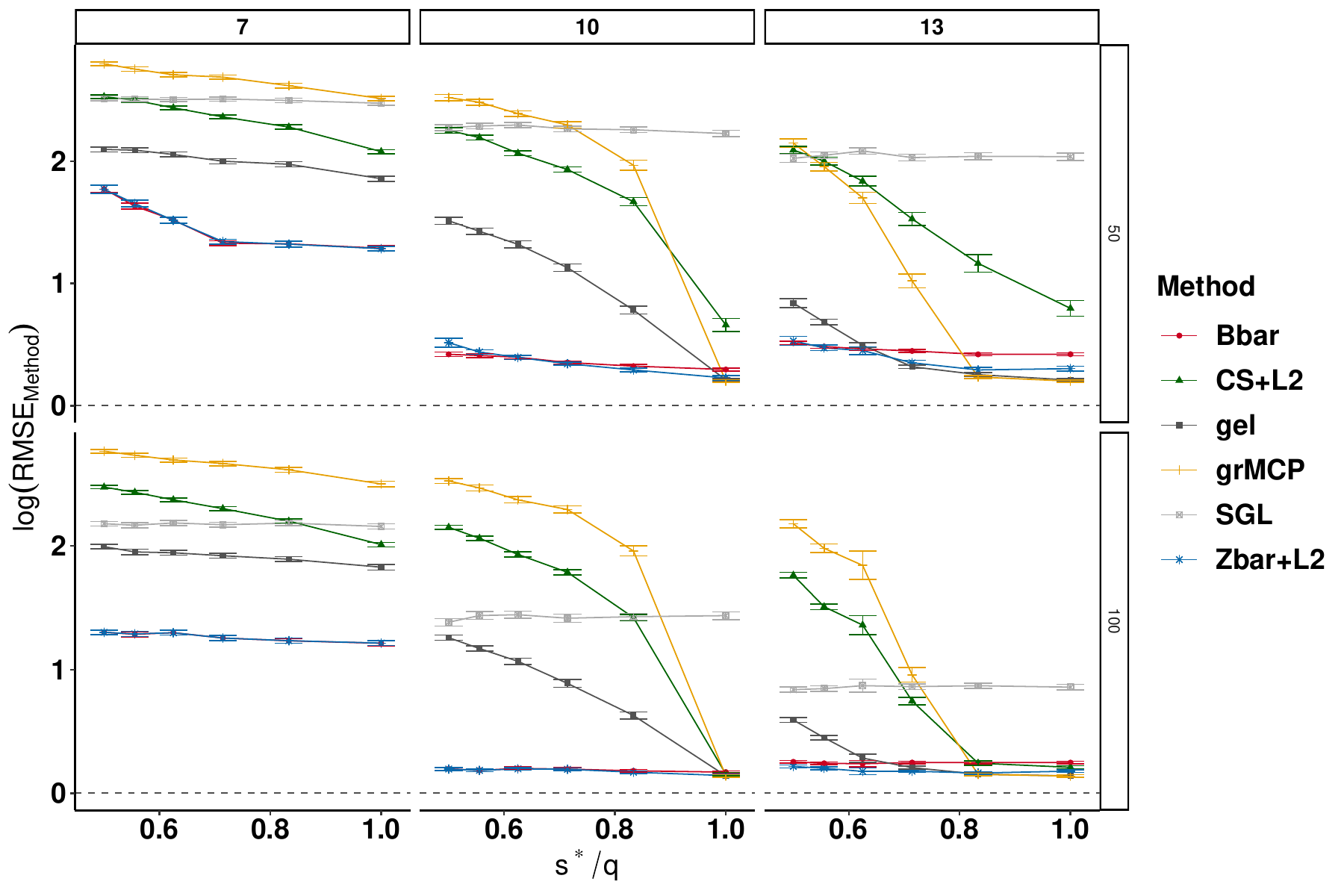}
	}
	\caption{\footnotesize  Out-of-sample prediction performance averaged across tasks for different sample sizes, $n_k$ (displayed on the horizontal panels), values of support size, $s$, for model fitting (vertical panels), degrees of support heterogeneity, $s^*/q$ (x-axis), and covariate correlation levels $\rho \in \{0.2, 0.5,0.8\}$ (from top to bottom figure). Low $s^*/q$ indicates high support heterogeneity and $s/q = 1$ indicates that all tasks had identical support. Unlike the figure in the main text, the performance of each method is presented without adjustment to the $L_0 L_2$ benchmark and is thus in terms of the raw log RMSE.} 
	\label{fig:sims_rmse_unadjust}
\end{figure}

\subsubsection{Coefficient Estimation Performance: RMSE}

\begin{figure}[H]
	\centerline{
		\includegraphics[width=0.85\linewidth]{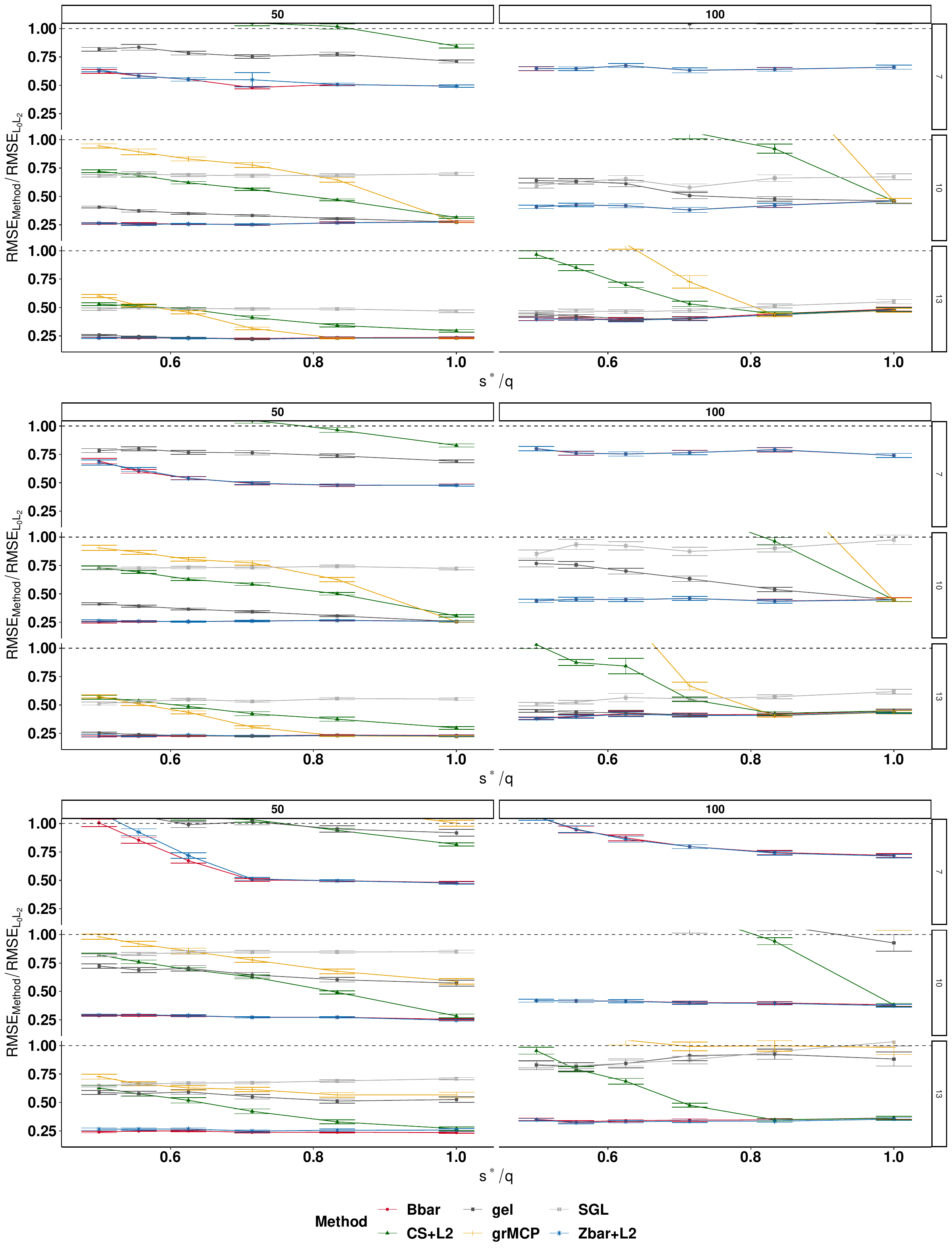}
	}
	\caption{\footnotesize  Coefficient estimation performance averaged across tasks for different sample sizes, $n_k$ (displayed on the horizontal panels), values of support size, $s$, for model fitting (vertical panels), degrees of support heterogeneity, $s^*/q$ (x-axis), and covariate correlation levels $\rho \in \{0.2, 0.5,0.8\}$ (from top to bottom figure). Low $s^*/q$ indicates high support heterogeneity and $s/q = 1$ indicates that all tasks had identical support. The performance of each method is presented in reference to the performance of a $L_0 L_2$ and thus lower values indicate superior (relative) prediction accuracy.} 
	\label{fig:sims_coef_supp}
\end{figure}

\begin{figure}[H]
	\centerline{
		\includegraphics[width=0.9\linewidth]{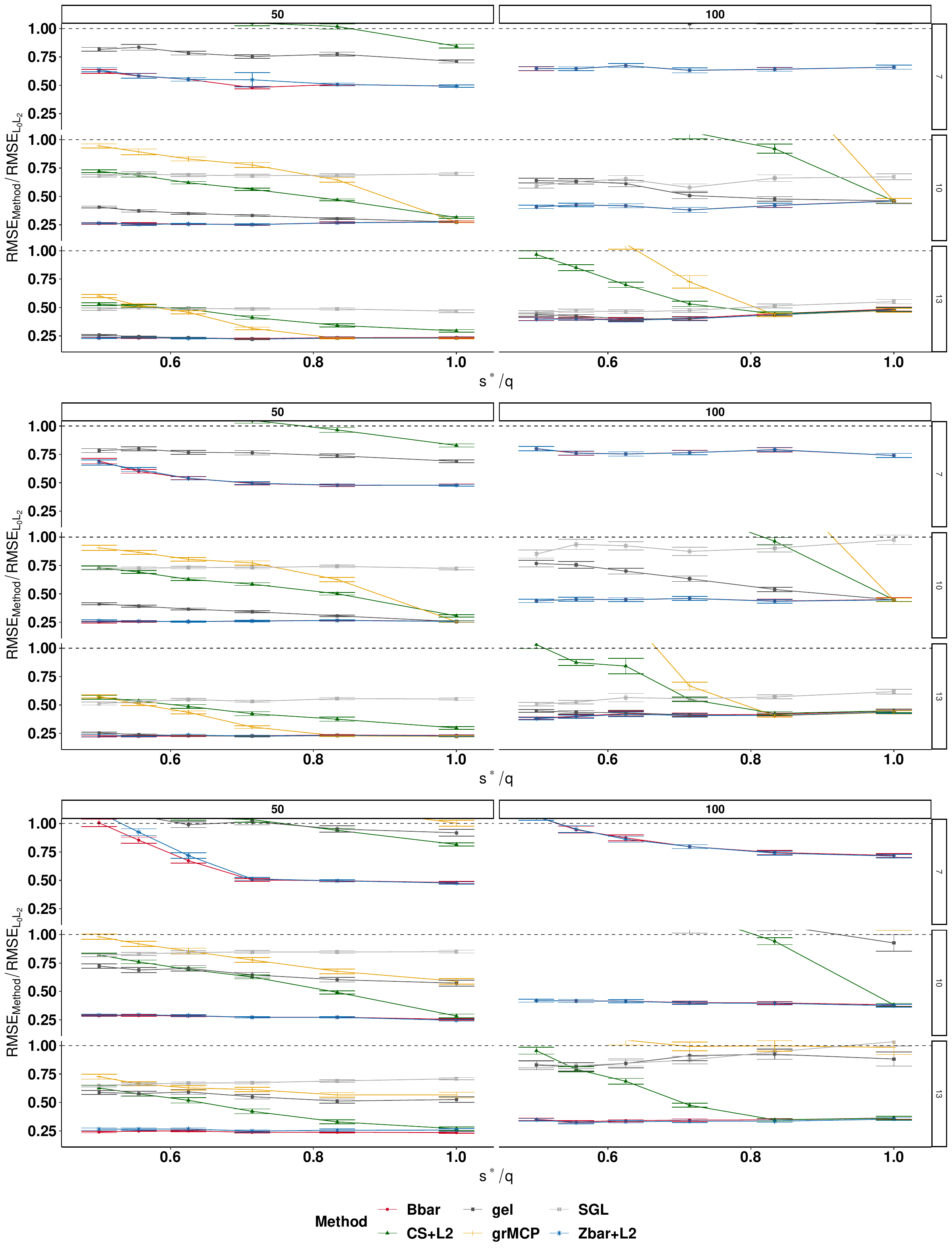}
	}
	\caption{\footnotesize  Coefficient estimation performance averaged across tasks for different sample sizes, $n_k$ (displayed on the horizontal panels), values of support size, $s$, for model fitting (vertical panels), degrees of support heterogeneity, $s^*/q$ (x-axis), and covariate correlation levels $\rho \in \{0.2, 0.5,0.8\}$ (from top to bottom figure). Low $s^*/q$ indicates high support heterogeneity and $s/q = 1$ indicates that all tasks had identical support. The performance of each method is presented in reference to the performance of a $L_0 L_2$ and thus lower values indicate superior (relative) prediction accuracy.} 
	\label{fig:sims_rmse_supp3}
\end{figure}

\subsubsection{Support Recovery: F1}
\begin{figure}[H]
	\centerline{
		\includegraphics[width=0.85\linewidth]{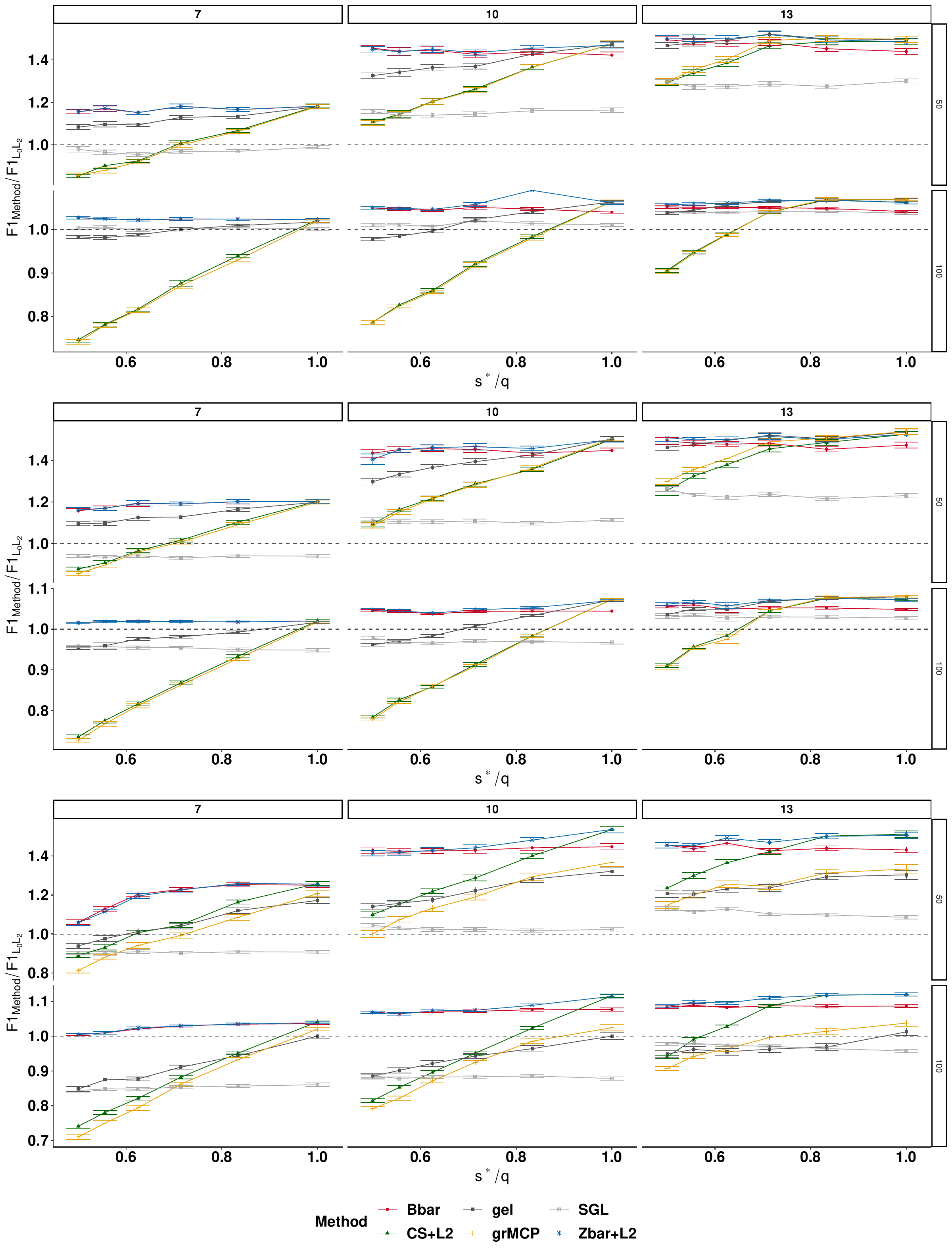}
	}
	\caption{\footnotesize Support recovery (F1 score) averaged across tasks for different sample sizes, $n_k$ (displayed on the horizontal panels), values of support size, $s$, for model fitting (vertical panels), degrees of support heterogeneity, $s^*/q$ (x-axis), and covariate correlation levels $\rho \in \{0.2, 0.5,0.8\}$ (from top to bottom figure). Low $s^*/q$ indicates high support heterogeneity and $s/q = 1$ indicates that all tasks had identical support. The performance of each method is presented in reference to the performance of a $L_0 L_2$ and thus higher values indicate superior (relative) support recovery.} 
	\label{fig:sims_f1_supp}
\end{figure}

\begin{figure}[H]
	\centerline{
		\includegraphics[width=0.9\linewidth]{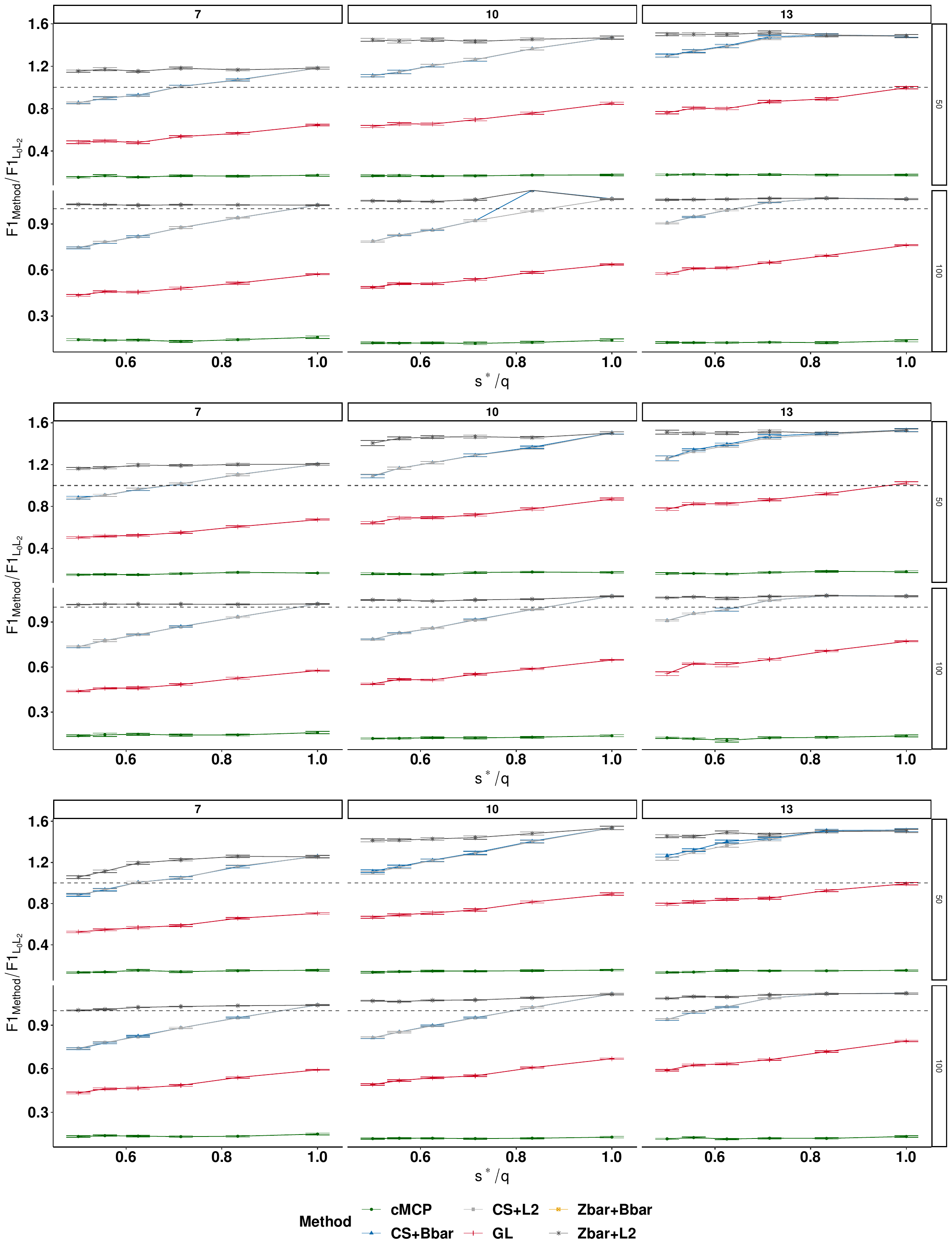}
	}
	\caption{\footnotesize  Support recovery (F1 score) averaged across tasks for different sample sizes, $n_k$ (displayed on the horizontal panels), values of support size, $s$, for model fitting (vertical panels), degrees of support heterogeneity, $s^*/q$ (x-axis), and covariate correlation levels $\rho \in \{0.2, 0.5,0.8\}$ (from top to bottom figure). Low $s^*/q$ indicates high support heterogeneity and $s/q = 1$ indicates that all tasks had identical support. The performance of each method is presented in reference to the performance of a $L_0 L_2$ and thus lower values indicate superior (relative) prediction accuracy.} 
	\label{fig:sims_f1_supp2}
\end{figure}

\subsection{No Common Support Simulations} \label{noCommSupp}
We present additional results from data simulated such that the true task coefficients had no overlap in their supports (i.e., $\B{z}_k^T \B{z}_{k'} = 0 \forall ~k \neq k'$). The data were otherwise simulated in a manner identical to the simulations presented in the main text.

\begin{figure}[t!] 
	\centering
	\begin{tabular}{cc}
		\centering
		\includegraphics[width=0.99\linewidth]{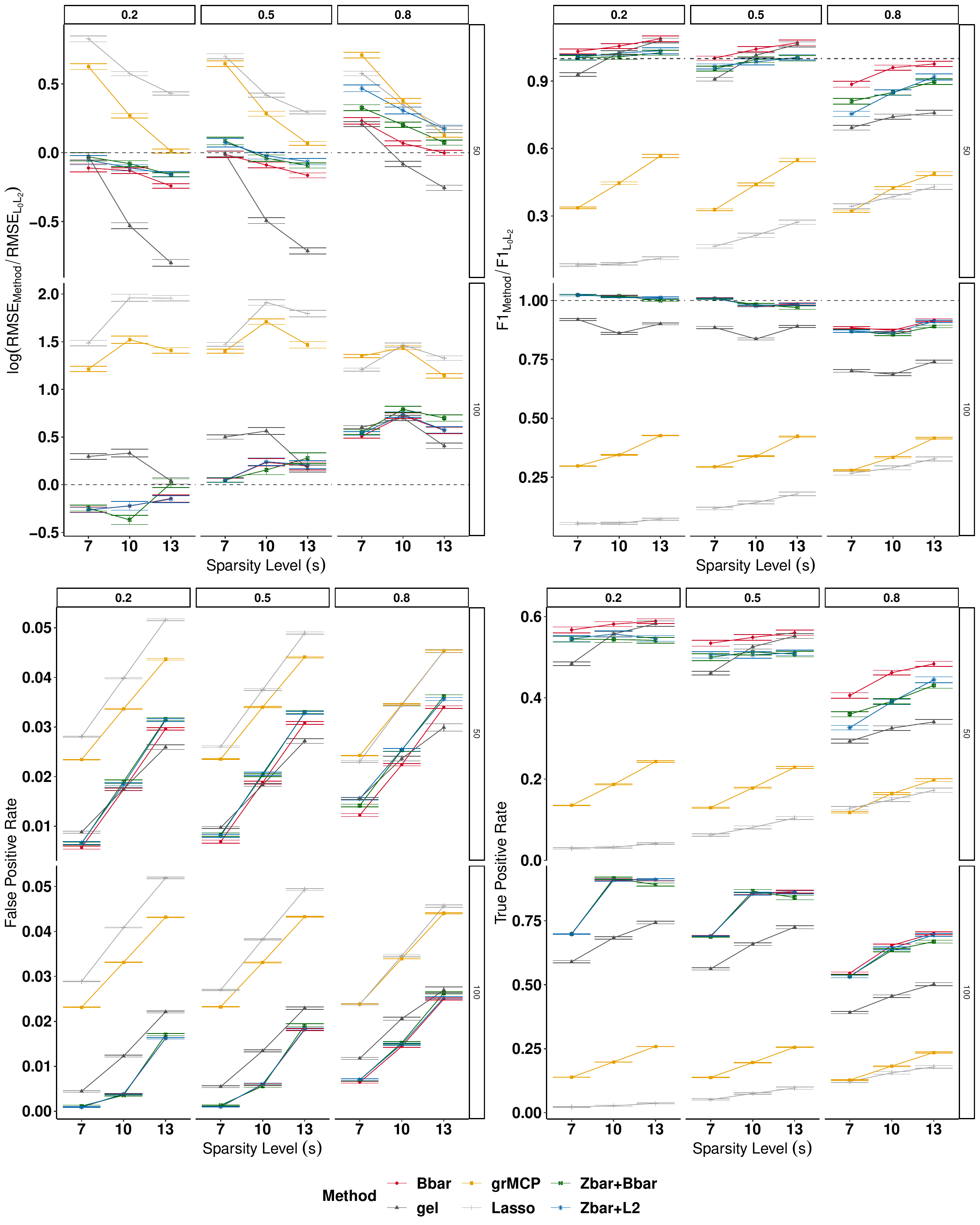} &
	\end{tabular}
	\caption{\footnotesize Prediction performance (RMSE) [Top Left], F1 score of support recovery [Top Right], false positive rate [Bottom Left], and true positive rate [Bottom Right] averaged across tasks for different $n_k$ (rows), and covariate correlation levels ($\rho$) of the true (simulated) model. Tasks were simulated to have no overlap in the supports of their true model coefficients. Lower RMSE and higher F1 scores indicate superior (relative) performance. Data were simulated with true support size $s^*=10$. We show performance of heterogeneous support methods given that data were simulated to have no common support. }
	\label{fig:sims_f1_rmse_q0}
\end{figure}

\clearpage

\subsection{Simulations with Non-Alternating Sparsity Pattern}\label{sec:non_alt_sims}

We present results from synthetic data experiments where the $\tilde{\mathcal{Q}}$ does \textit{not} have an alternating structure (i.e., $\boldsymbol{z}_k$ \textit{can} have successive nonzero elements). The simulation scheme was otherwise identical to the approach described in main text Section~\ref{sec:simulations}. The relative performance of our proposed methods in these simulations appears comparable to the simulation presented in the main text. We first present results from these simulations in a figure like that of main text Figure~\ref{fig:sims_f1_rmse}. We then present results in a series of figures like those presented in Supplement~\ref{supp_sims}.

\subsubsection*{Main Text Figure}

\begin{figure}[t!] 
	\centering
	\begin{tabular}{cc}
		\centering
		\includegraphics[width=0.99\linewidth]{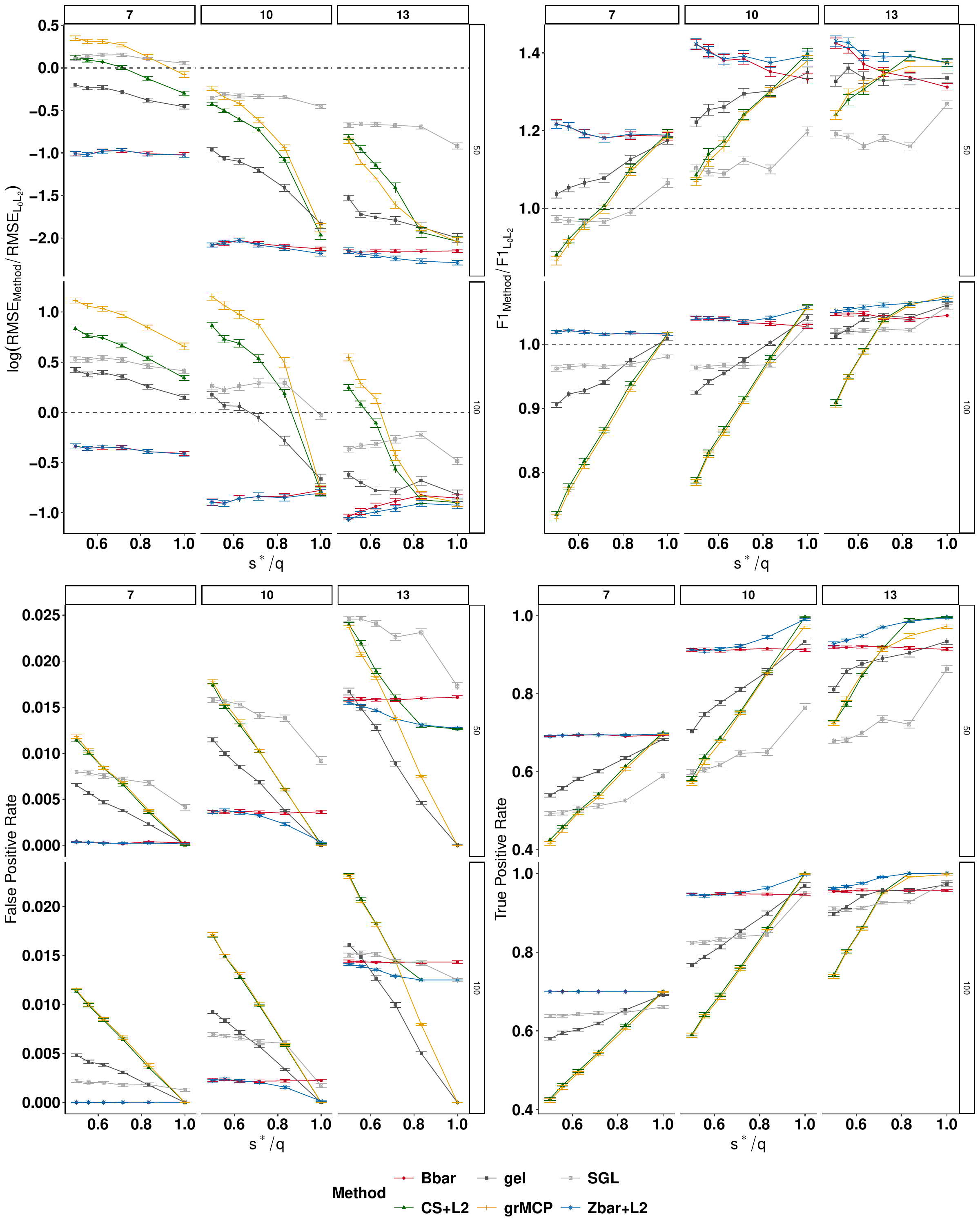} &
	\end{tabular}
	\caption{\footnotesize Prediction performance (RMSE) [Top Left], F1 score of support recovery [Top Right], false positive rate [Bottom Left], and true positive rate [Bottom Right] averaged across tasks for different $n_k$ (rows), support sizes (columns),
		and levels of support homogeneity of the true (simulated) model  ($s^*/q = 1$ indicates that all tasks had identical support). 
		Lower RMSE and higher F1 scores indicate superior (relative) performance. Data were simulated with support size  $s^*=10$, and covariate correlation parameter $\rho = 0.5$.}
	\label{fig:sims_f1_rmse_new}
\end{figure}
\clearpage

\subsubsection*{Supplemental Figures}

\paragraph{Prediction Performance: RMSE}

\clearpage
\begin{figure}[H]
	\centerline{
		\includegraphics[width=0.85\linewidth]{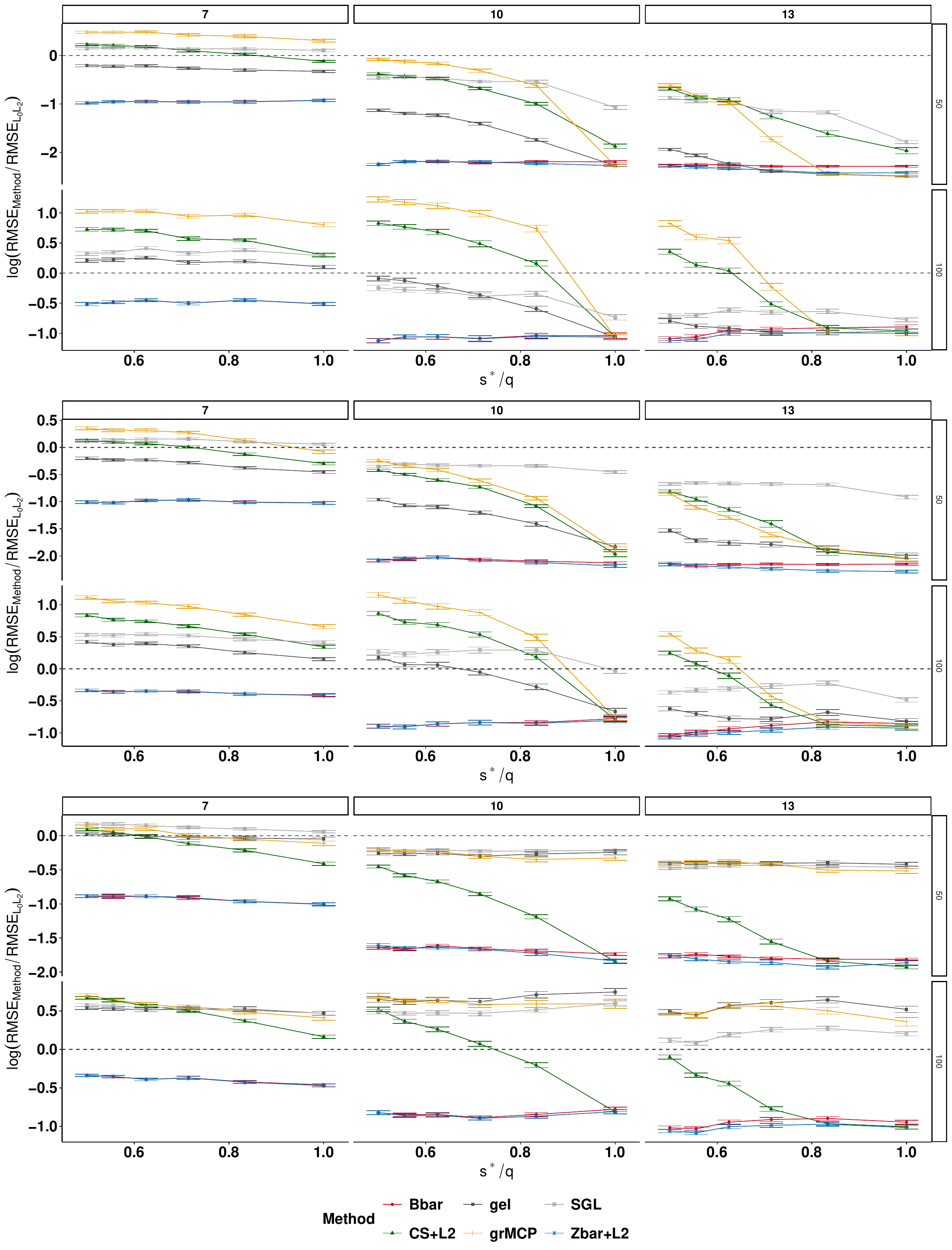}
	}
	\caption{\footnotesize  Out-of-sample prediction performance averaged across tasks for different sample sizes, $n_k$ (displayed on the horizontal panels), values of support size, $s$, for model fitting (vertical panels), degrees of support heterogeneity, $s^*/q$ (x-axis), and covariate correlation levels $\rho \in \{0.2, 0.5,0.8\}$ (from top to bottom figure). Low $s^*/q$ indicates high support heterogeneity and $s/q = 1$ indicates that all tasks had identical support. The performance of each method is presented in reference to the performance of a $L_0 L_2$ and thus lower values indicate superior (relative) prediction accuracy.} 
	\label{fig:sims_rmse_supp_new}
\end{figure}

\begin{figure}[H]
	\centerline{
		\includegraphics[width=0.9\linewidth]{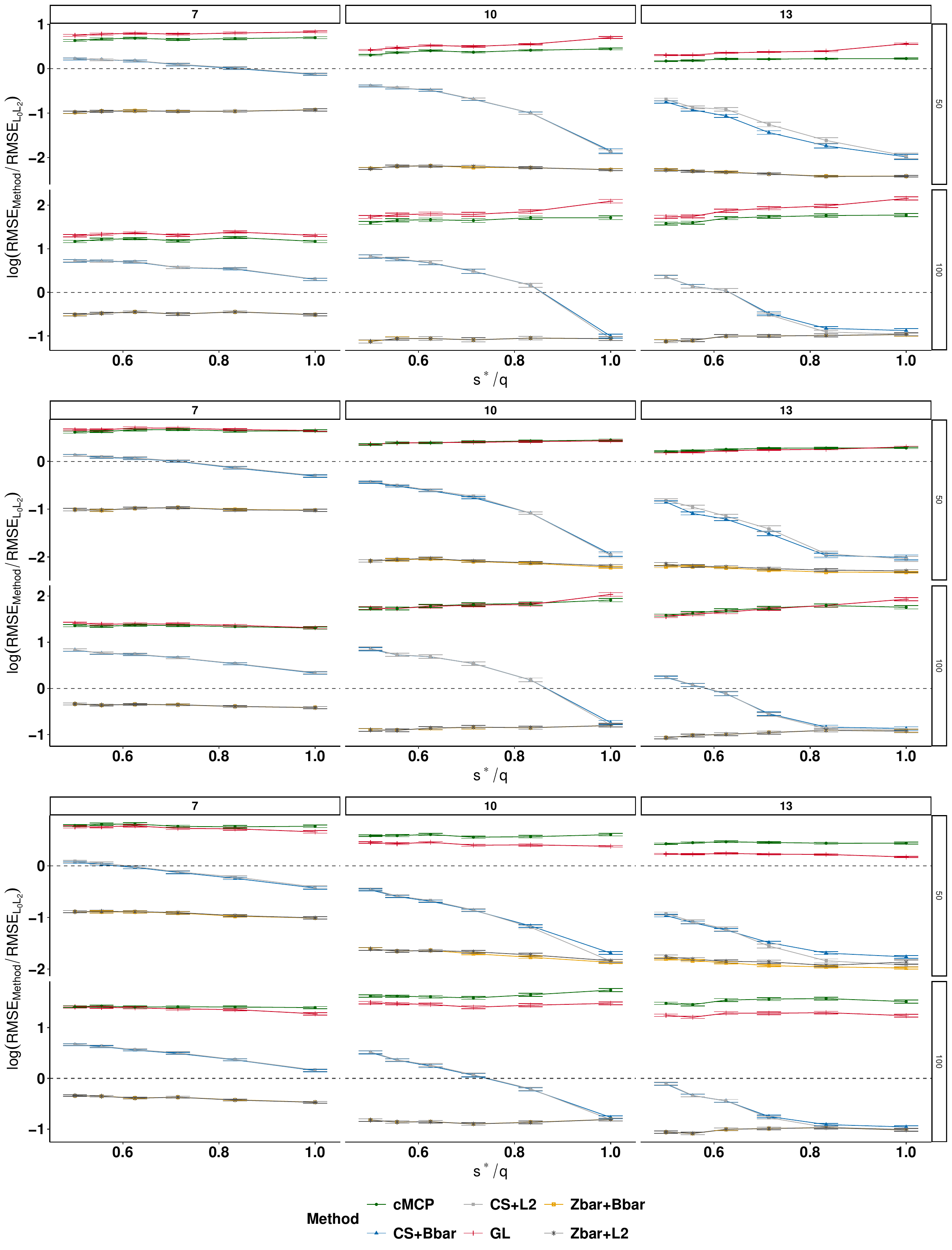}
	}
	\caption{\footnotesize Out-of-sample prediction performance averaged across tasks for different sample sizes, $n_k$ (displayed on the horizontal panels), values of support size, $s$, for model fitting (vertical panels), degrees of support heterogeneity, $s^*/q$ (x-axis), and covariate correlation levels $\rho \in \{0.2, 0.5,0.8\}$ (from top to bottom figure). Low $s^*/q$ indicates high support heterogeneity and $s/q = 1$ indicates that all tasks had identical support. The performance of each method is presented in reference to the performance of a $L_0 L_2$ and thus lower values indicate superior (relative) prediction accuracy.} 
	\label{fig:sims_rmse_supp2_new}
\end{figure}

\begin{figure}[H]
	\centerline{
		\includegraphics[width=0.85\linewidth]{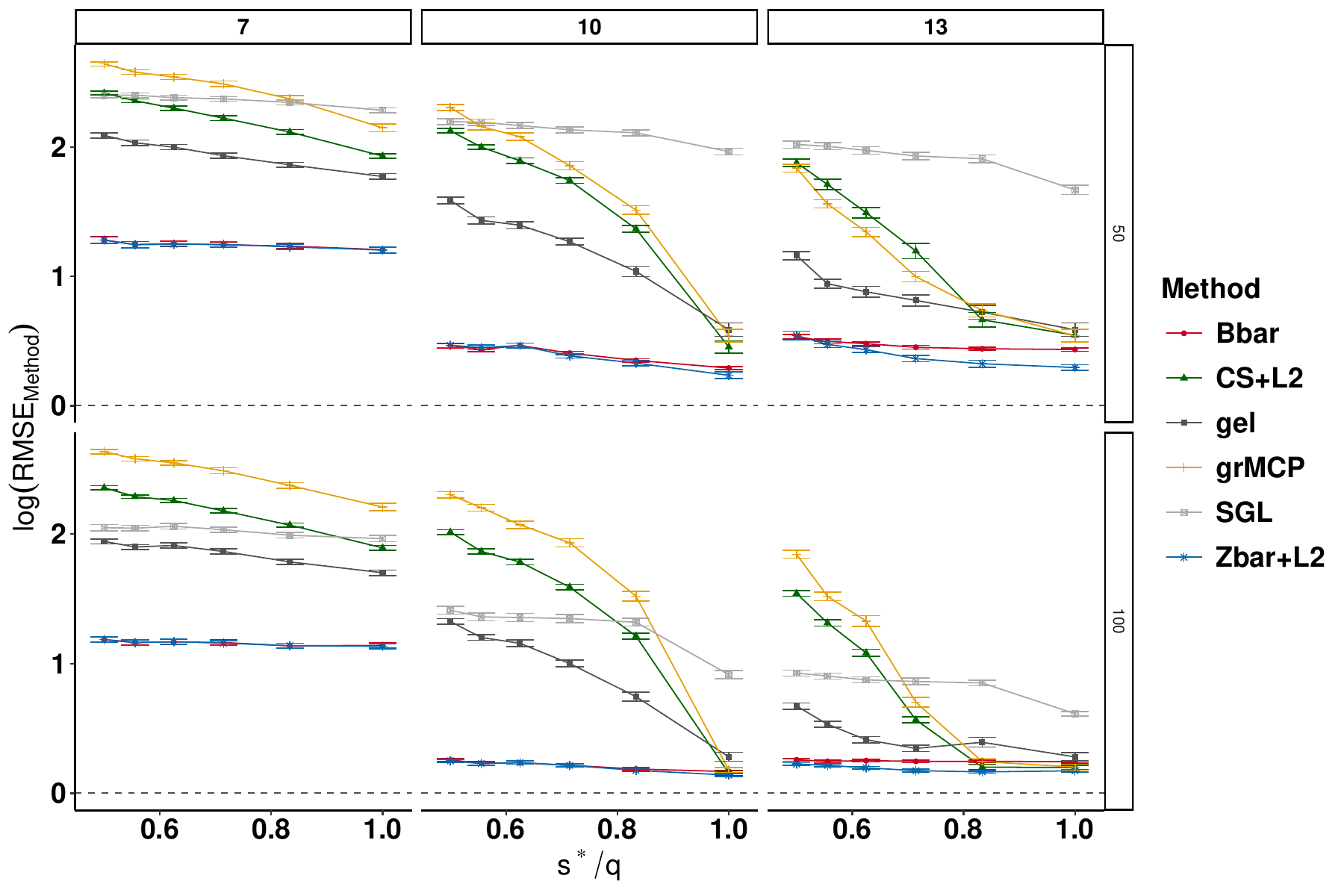}
	}
	\caption{\footnotesize  Out-of-sample prediction performance averaged across tasks for different sample sizes, $n_k$ (displayed on the horizontal panels), values of support size, $s$, for model fitting (vertical panels), degrees of support heterogeneity, $s^*/q$ (x-axis), and covariate correlation levels $\rho \in \{0.2, 0.5,0.8\}$ (from top to bottom figure). Low $s^*/q$ indicates high support heterogeneity and $s/q = 1$ indicates that all tasks had identical support. Unlike the figure in the main text, the performance of each method is presented without adjustment to the $L_0 L_2$ benchmark and is thus in terms of the raw log RMSE.} 
	\label{fig:sims_rmse_unadjust_new}
\end{figure}
\clearpage

\paragraph{Coefficient Estimation Performance: RMSE}
\clearpage
\begin{figure}[H]
	\centerline{
		\includegraphics[width=0.85\linewidth]{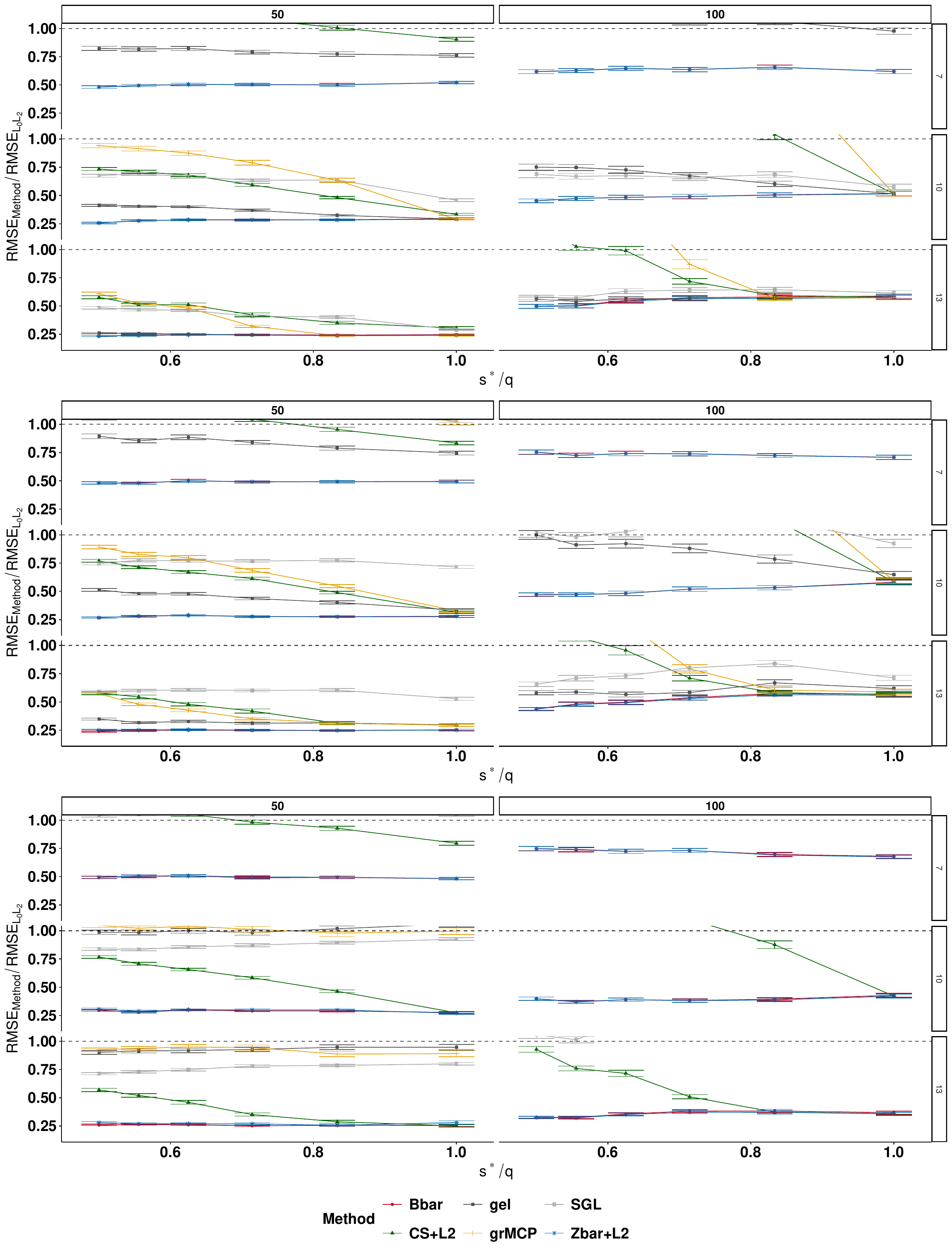}
	}
	\caption{\footnotesize  Coefficient estimation performance averaged across tasks for different sample sizes, $n_k$ (displayed on the horizontal panels), values of support size, $s$, for model fitting (vertical panels), degrees of support heterogeneity, $s^*/q$ (x-axis), and covariate correlation levels $\rho \in \{0.2, 0.5,0.8\}$ (from top to bottom figure). Low $s^*/q$ indicates high support heterogeneity and $s/q = 1$ indicates that all tasks had identical support. The performance of each method is presented in reference to the performance of a $L_0 L_2$ and thus lower values indicate superior (relative) prediction accuracy.} 
	\label{fig:sims_coef_supp_new}
\end{figure}

\begin{figure}[H]
	\centerline{
		\includegraphics[width=0.9\linewidth]{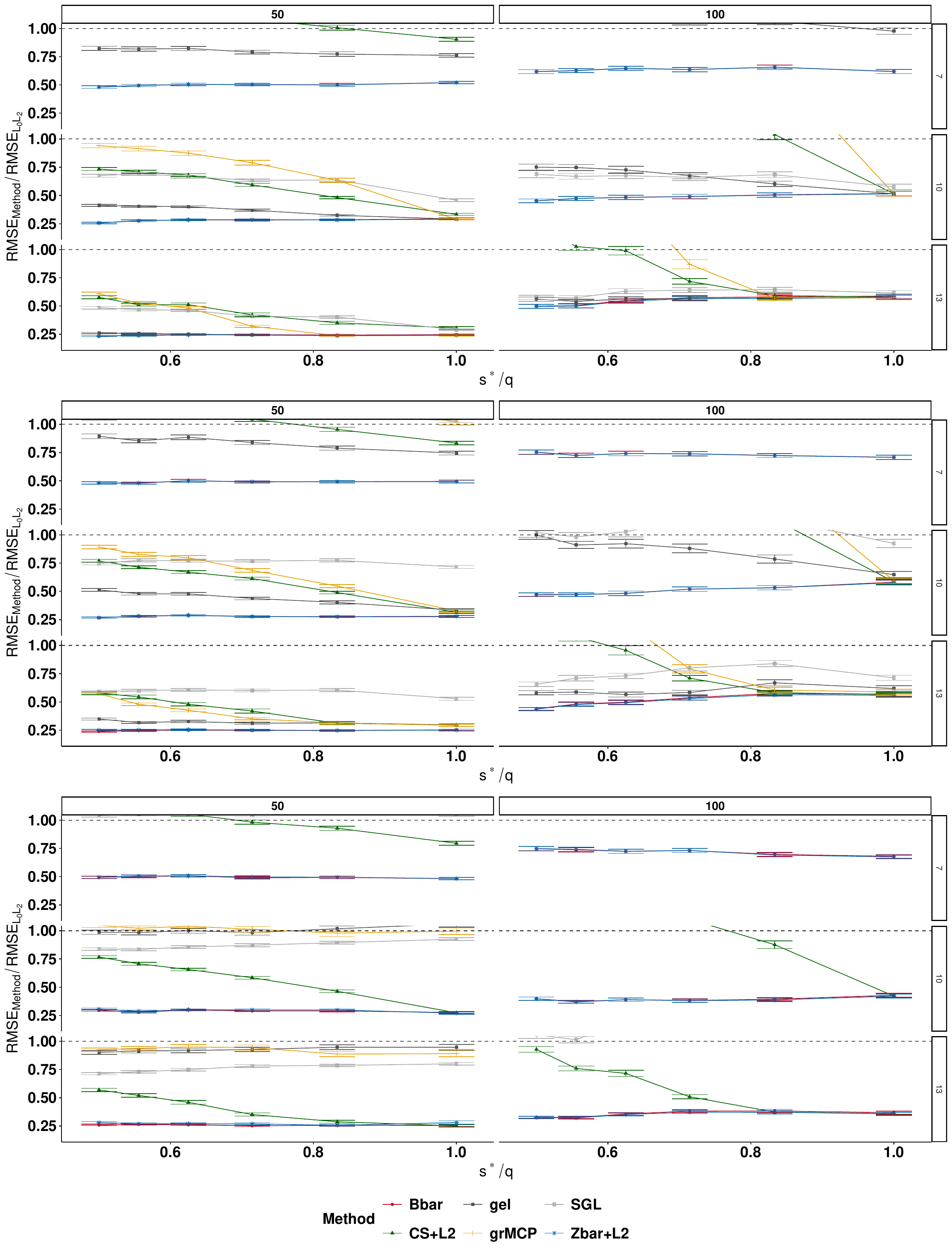}
	}
	\caption{\footnotesize  Coefficient estimation performance averaged across tasks for different sample sizes, $n_k$ (displayed on the horizontal panels), values of support size, $s$, for model fitting (vertical panels), degrees of support heterogeneity, $s^*/q$ (x-axis), and covariate correlation levels $\rho \in \{0.2, 0.5,0.8\}$ (from top to bottom figure). Low $s^*/q$ indicates high support heterogeneity and $s/q = 1$ indicates that all tasks had identical support. The performance of each method is presented in reference to the performance of a $L_0 L_2$ and thus lower values indicate superior (relative) prediction accuracy.} 
	\label{fig:sims_rmse_supp3_new}
\end{figure}
\clearpage

\paragraph{Support Recovery: F1}
\clearpage
\begin{figure}[H]
	\centerline{
		\includegraphics[width=0.85\linewidth]{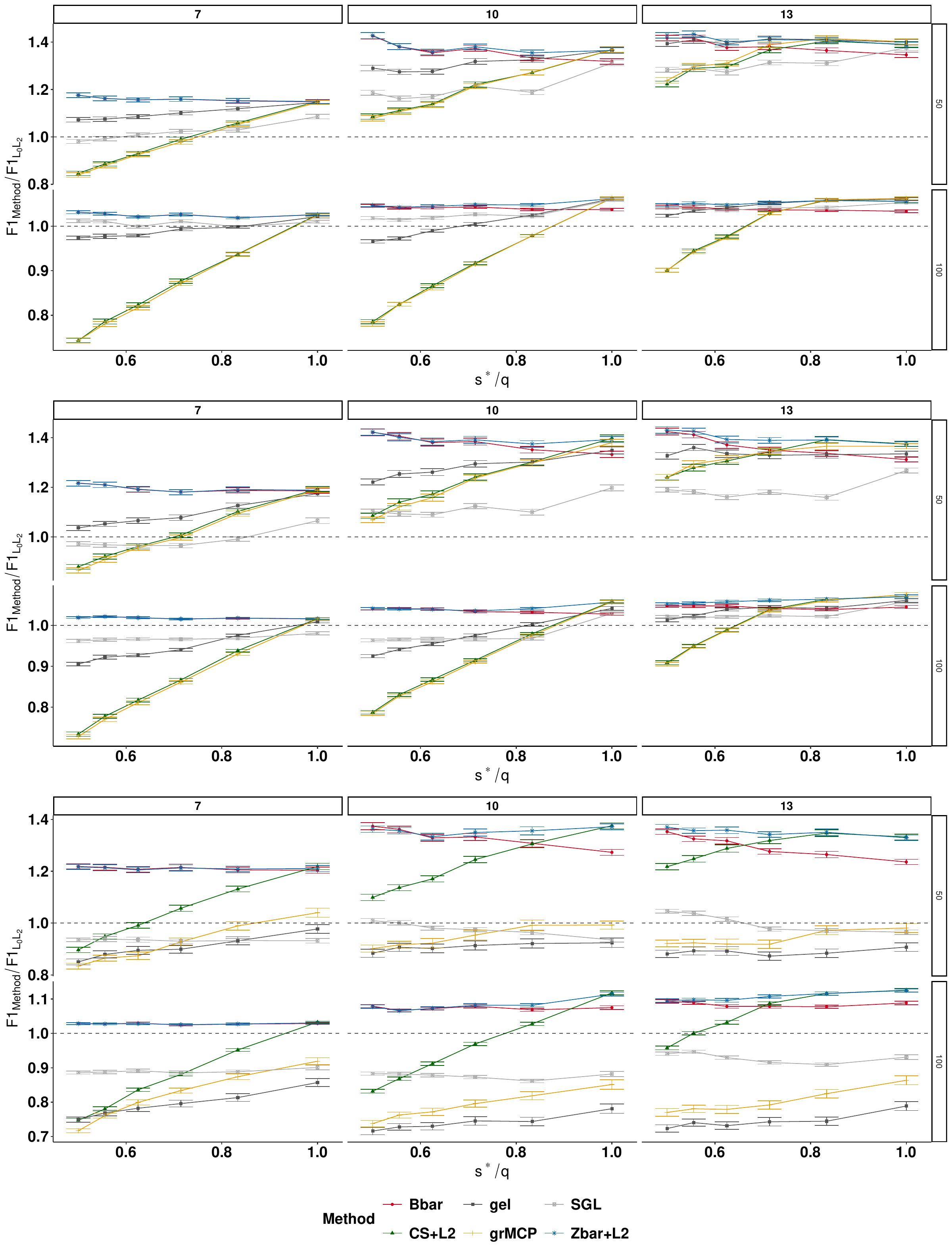}
	}
	\caption{\footnotesize Support recovery (F1 score) averaged across tasks for different sample sizes, $n_k$ (displayed on the horizontal panels), values of support size, $s$, for model fitting (vertical panels), degrees of support heterogeneity, $s^*/q$ (x-axis), and covariate correlation levels $\rho \in \{0.2, 0.5,0.8\}$ (from top to bottom figure). Low $s^*/q$ indicates high support heterogeneity and $s/q = 1$ indicates that all tasks had identical support. The performance of each method is presented in reference to the performance of a $L_0 L_2$ and thus higher values indicate superior (relative) support recovery.} 
	\label{fig:sims_f1_supp_new}
\end{figure}

\begin{figure}[H]
	\centerline{
		\includegraphics[width=0.9\linewidth]{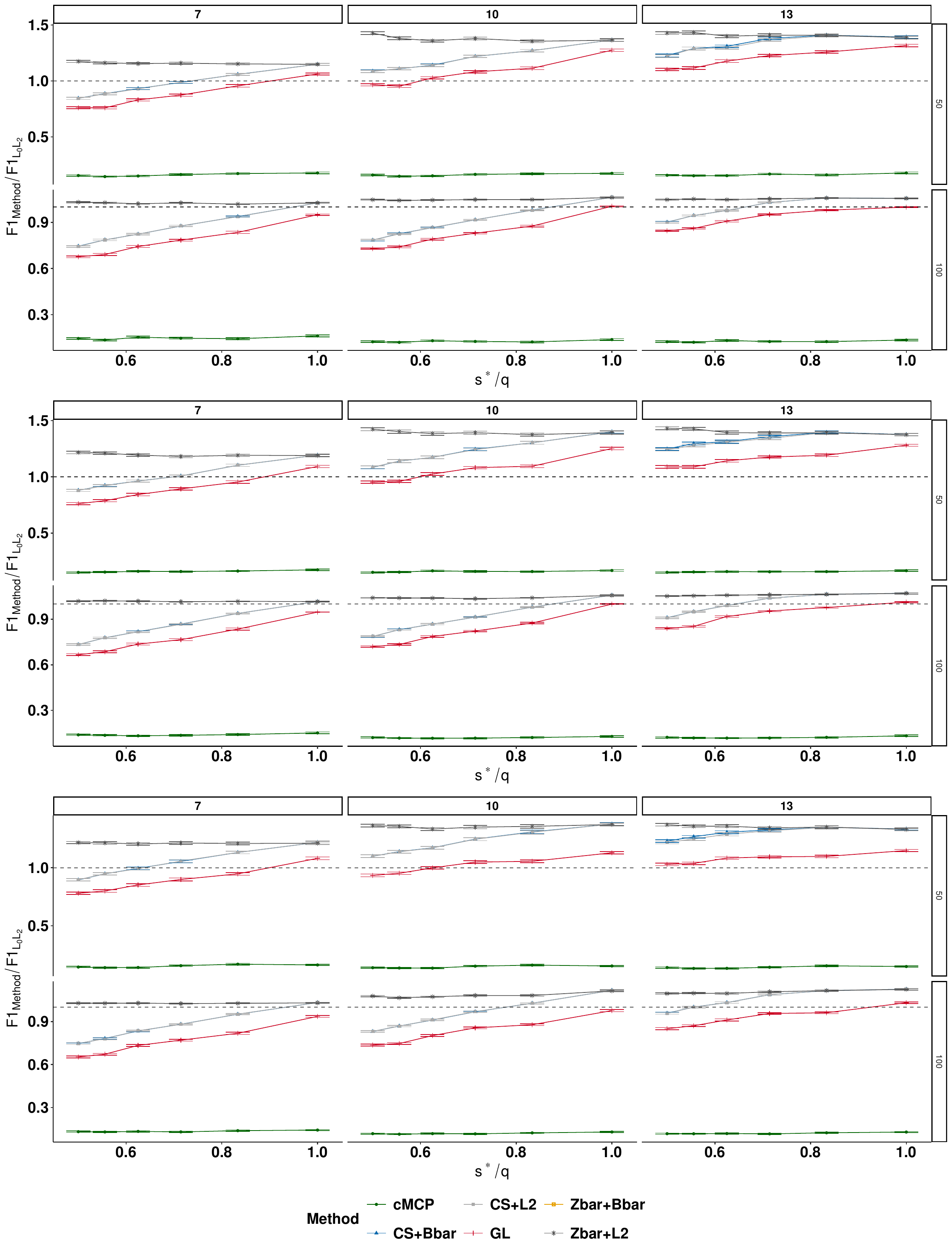}
	}
	\caption{\footnotesize  Support recovery (F1 score) averaged across tasks for different sample sizes, $n_k$ (displayed on the horizontal panels), values of support size, $s$, for model fitting (vertical panels), degrees of support heterogeneity, $s^*/q$ (x-axis), and covariate correlation levels $\rho \in \{0.2, 0.5,0.8\}$ (from top to bottom figure). Low $s^*/q$ indicates high support heterogeneity and $s/q = 1$ indicates that all tasks had identical support. The performance of each method is presented in reference to the performance of a $L_0 L_2$ and thus lower values indicate superior (relative) prediction accuracy.} 
	\label{fig:sims_f1_supp2_new}
\end{figure}

\clearpage

\subsection{Timing Experiments}\label{app:approx-time}
We present runtimes experiment results from 50 simulation replicates. We simulated data as in the main text with $s^* = 10$, $q=16$, $p \in \{250,500,1000\}$, $n_k=2p$ and $\rho = 0.5$. Since $L_0$ methods constrain the sparsity of the solution to a specified level, $s$, and benchmark methods (e.g., gel, SGL) do not, it is difficult to compare the timing of the methods; tuning parameter values influence the speed of convergence for the proposed and benchmark methods. For that reason, we present the timing of all methods averaged across a grid of 50 tuning parameter values in Table~\ref{app:table-approxCV-time}. For example, for the $L_0$ methods, this grid included sparsity levels $s \in \{5,7,10,13,15\}$ and 10 values of $\lambda$ or $\delta$ depending on the penalty. We present the average timing (i.e., total time / 50) summed across 10 folds (the same 10 folds were used for all methods). We also present the timing (in milliseconds) of a single fit on the tuned hyperparameters in Table~\ref{app:table-approx-time}. These results demonstrate that our approximate solver is fast enough to be useful in practice for a range of problem sizes. In fact, our approximate methods can be faster than most of the benchmark methods, especially as the dimension of the covariates, $p,$ grows large.

\begin{table}[H]
	\resizebox{\textwidth}{!}{
		\begin{tabular}{l|cccccccccccccccccccc}
			\hline
			\hline
			p & \multicolumn{2}{c}{Bbar} & \multicolumn{2}{c}{cMCP} & \multicolumn{2}{c}{CS+Bbar} & \multicolumn{2}{c}{CS+L2} & \multicolumn{2}{c}{gel} & \multicolumn{2}{c}{GL} & \multicolumn{2}{c}{grMCP} & \multicolumn{2}{c}{SGL} & \multicolumn{2}{c}{Zbar+Bbar} & \multicolumn{2}{c}{Zbar+L2} \\ 
			\hline
			\hline
			250  & \multicolumn{1}{r@{}}{$575.5$} & \multicolumn{1}{@{ $\pm$ }l}{$19.3$} & \multicolumn{1}{r@{}}{$\phantom{0}126.9$} & \multicolumn{1}{@{ $\pm$ }l}{$\phantom{0}8.5$} & \multicolumn{1}{r@{}}{$296.0$} & \multicolumn{1}{@{ $\pm$ }l}{$\phantom{0}8.4$} & \multicolumn{1}{r@{}}{$280.8$} & \multicolumn{1}{@{ $\pm$ }l}{$\phantom{0}7.4$} & \multicolumn{1}{r@{}}{$\phantom{0}99.5$} & \multicolumn{1}{@{ $\pm$ }l}{$\phantom{0}7.7$} & \multicolumn{1}{r@{}}{$\phantom{0}294.4$} & \multicolumn{1}{@{ $\pm$ }l}{$\phantom{0}12.6$} & \multicolumn{1}{r@{}}{$\phantom{0}578.5$} & \multicolumn{1}{@{ $\pm$ }l}{$28.5$} & \multicolumn{1}{r@{}}{$\phantom{0}313.5$} & \multicolumn{1}{@{ $\pm$ }l}{$\phantom{0}11.2$} & \multicolumn{1}{r@{}}{$561.6$} & \multicolumn{1}{@{ $\pm$ }l}{$17.5$} & \multicolumn{1}{r@{}}{$569.0$} & \multicolumn{1}{@{ $\pm$ }l}{$20.8$} \\
			500  & \multicolumn{1}{r@{}}{$438.1$} & \multicolumn{1}{@{ $\pm$ }l}{$\phantom{0}6.5$} & \multicolumn{1}{r@{}}{$\phantom{0}280.8$} & \multicolumn{1}{@{ $\pm$ }l}{$\phantom{0}6.3$} & \multicolumn{1}{r@{}}{$343.7$} & \multicolumn{1}{@{ $\pm$ }l}{$\phantom{0}9.3$} & \multicolumn{1}{r@{}}{$310.0$} & \multicolumn{1}{@{ $\pm$ }l}{$\phantom{0}8.3$} & \multicolumn{1}{r@{}}{$248.9$} & \multicolumn{1}{@{ $\pm$ }l}{$\phantom{0}9.9$} & \multicolumn{1}{r@{}}{$1074.4$} & \multicolumn{1}{@{ $\pm$ }l}{$\phantom{0}13.1$} & \multicolumn{1}{r@{}}{$\phantom{0}928.8$} & \multicolumn{1}{@{ $\pm$ }l}{$15.1$} & \multicolumn{1}{r@{}}{$1102.5$} & \multicolumn{1}{@{ $\pm$ }l}{$\phantom{00}9.9$} & \multicolumn{1}{r@{}}{$478.1$} & \multicolumn{1}{@{ $\pm$ }l}{$\phantom{0}9.0$} & \multicolumn{1}{r@{}}{$550.9$} & \multicolumn{1}{@{ $\pm$ }l}{$13.3$} \\
			1000  & \multicolumn{1}{r@{}}{$958.2$} & \multicolumn{1}{@{ $\pm$ }l}{$15.0$} & \multicolumn{1}{r@{}}{$1057.5$} & \multicolumn{1}{@{ $\pm$ }l}{$31.4$} & \multicolumn{1}{r@{}}{$529.5$} & \multicolumn{1}{@{ $\pm$ }l}{$14.7$} & \multicolumn{1}{r@{}}{$489.8$} & \multicolumn{1}{@{ $\pm$ }l}{$16.0$} & \multicolumn{1}{r@{}}{$922.8$} & \multicolumn{1}{@{ $\pm$ }l}{$23.9$} & \multicolumn{1}{r@{}}{$8059.0$} & \multicolumn{1}{@{ $\pm$ }l}{$151.7$} & \multicolumn{1}{r@{}}{$2696.4$} & \multicolumn{1}{@{ $\pm$ }l}{$58.7$} & \multicolumn{1}{r@{}}{$8205.2$} & \multicolumn{1}{@{ $\pm$ }l}{$154.7$} & \multicolumn{1}{r@{}}{$946.6$} & \multicolumn{1}{@{ $\pm$ }l}{$18.2$} & \multicolumn{1}{r@{}}{$928.3$} & \multicolumn{1}{@{ $\pm$ }l}{$17.9$} \\
			\hline 
	\end{tabular}}
	\caption{\footnotesize [Runtime for approximate methods] Comparison of the average runtime (in milliseconds) of a single model fit on tuned hyperparameters. }
	\label{app:table-approx-time}
\end{table}

\begin{table}[H]
	\resizebox{\textwidth}{!}{
		\begin{tabular}{l|cccccccccccccccccccc}
			\hline
			\hline
			p & \multicolumn{2}{c}{Bbar} & \multicolumn{2}{c}{cMCP} & \multicolumn{2}{c}{CS+Bbar} & \multicolumn{2}{c}{CS+L2} & \multicolumn{2}{c}{gel} & \multicolumn{2}{c}{GL} & \multicolumn{2}{c}{grMCP} & \multicolumn{2}{c}{SGL} & \multicolumn{2}{c}{Zbar+Bbar} & \multicolumn{2}{c}{Zbar+L2} \\ 
			\hline
			\hline
			250  & \multicolumn{1}{r@{}}{$\phantom{0}81.35$} & \multicolumn{1}{@{ $\pm$ }l}{$3.30$} & \multicolumn{1}{r@{}}{$\phantom{0}14.42$} & \multicolumn{1}{@{ $\pm$ }l}{$0.36$} & \multicolumn{1}{r@{}}{$25.74$} & \multicolumn{1}{@{ $\pm$ }l}{$0.72$} & \multicolumn{1}{r@{}}{$21.84$} & \multicolumn{1}{@{ $\pm$ }l}{$0.48$} & \multicolumn{1}{r@{}}{$\phantom{0}25.57$} & \multicolumn{1}{@{ $\pm$ }l}{$0.60$} & \multicolumn{1}{r@{}}{$\phantom{0}10.02$} & \multicolumn{1}{@{ $\pm$ }l}{$0.33$} & \multicolumn{1}{r@{}}{$\phantom{0}30.58$} & \multicolumn{1}{@{ $\pm$ }l}{$0.84$} & \multicolumn{1}{r@{}}{$\phantom{0}27.72$} & \multicolumn{1}{@{ $\pm$ }l}{$0.90$} & \multicolumn{1}{r@{}}{$\phantom{0}76.71$} & \multicolumn{1}{@{ $\pm$ }l}{$2.71$} & \multicolumn{1}{r@{}}{$108.84$} & \multicolumn{1}{@{ $\pm$ }l}{$2.94$} \\
			500  & \multicolumn{1}{r@{}}{$\phantom{0}78.86$} & \multicolumn{1}{@{ $\pm$ }l}{$1.21$} & \multicolumn{1}{r@{}}{$\phantom{0}38.06$} & \multicolumn{1}{@{ $\pm$ }l}{$0.19$} & \multicolumn{1}{r@{}}{$40.74$} & \multicolumn{1}{@{ $\pm$ }l}{$0.88$} & \multicolumn{1}{r@{}}{$30.92$} & \multicolumn{1}{@{ $\pm$ }l}{$0.58$} & \multicolumn{1}{r@{}}{$\phantom{0}61.67$} & \multicolumn{1}{@{ $\pm$ }l}{$0.64$} & \multicolumn{1}{r@{}}{$\phantom{0}23.71$} & \multicolumn{1}{@{ $\pm$ }l}{$0.16$} & \multicolumn{1}{r@{}}{$\phantom{0}59.76$} & \multicolumn{1}{@{ $\pm$ }l}{$1.10$} & \multicolumn{1}{r@{}}{$\phantom{0}76.83$} & \multicolumn{1}{@{ $\pm$ }l}{$0.44$} & \multicolumn{1}{r@{}}{$\phantom{0}76.83$} & \multicolumn{1}{@{ $\pm$ }l}{$1.33$} & \multicolumn{1}{r@{}}{$119.22$} & \multicolumn{1}{@{ $\pm$ }l}{$1.76$} \\
			1000  & \multicolumn{1}{r@{}}{$149.41$} & \multicolumn{1}{@{ $\pm$ }l}{$2.16$} & \multicolumn{1}{r@{}}{$171.27$} & \multicolumn{1}{@{ $\pm$ }l}{$2.26$} & \multicolumn{1}{r@{}}{$88.55$} & \multicolumn{1}{@{ $\pm$ }l}{$3.24$} & \multicolumn{1}{r@{}}{$53.36$} & \multicolumn{1}{@{ $\pm$ }l}{$1.91$} & \multicolumn{1}{r@{}}{$303.50$} & \multicolumn{1}{@{ $\pm$ }l}{$3.70$} & \multicolumn{1}{r@{}}{$126.94$} & \multicolumn{1}{@{ $\pm$ }l}{$2.17$} & \multicolumn{1}{r@{}}{$193.51$} & \multicolumn{1}{@{ $\pm$ }l}{$4.30$} & \multicolumn{1}{r@{}}{$483.06$} & \multicolumn{1}{@{ $\pm$ }l}{$8.78$} & \multicolumn{1}{r@{}}{$142.75$} & \multicolumn{1}{@{ $\pm$ }l}{$2.39$} & \multicolumn{1}{r@{}}{$179.07$} & \multicolumn{1}{@{ $\pm$ }l}{$3.34$} \\
			\hline 
	\end{tabular}}
	\caption{\footnotesize [Runtime for approximate methods] 
		Comparison of the average runtime $\pm SEM$ (in milliseconds) of a single model fit during cross validation of 50 tuning parameters with our approximate methods. To calculate an average time per fit we divide the total tuning time by the product 50 hyperparameter values $\times 10$ CV  folds. }
	\label{app:table-approxCV-time}
\end{table}

\subsection{Additional Experiments with the Exact Solver}\label{app:exact-expts}
In this section, we present additional experiment results with our proposed exact solver. In particular, we study the effect of warm starting our exact solver with the solution from our approximate solver. We follow the same setup from Section~\ref{sec:scalability}. We use the exact solver setup from  Section~\ref{sec:scalability} as the solver that utilizes warm starts. As for the cold start solver, we initialize $\B{z}_1^0=\cdots=\B{z}_K^0$ to be the same. In particular, we set $\B{z}_1^0=\cdots=\B{z}_K^0=\B{v}$ where the vector $\B{v}\in\R^p$ follows $v_1=\cdots=v_{p-s}=0$ and the rest of the coordinates are set to one. Under this setup, the support of $\B{z}_1^0,\cdots,\B{z}_K^0$ do not intersect with the support of any of true $\B\beta_k$.

The results for these experiments are presented in Table~\ref{app:table-mip}. As we see, using warm starts (available from our approximate solvers) help to reduce the exact solver runtime. However, even without any warm start, our outer approximation based exact solver can obtain and certify optimal solutions, significantly faster than an off-the-shelf solver such as Gurobi which can only scale to $p\leq 200$ (see Table~\ref{table-mip}). This simultaneously demonstrates the strength of our exact solver which is scalable to larger instances (compared to Gurobi) even with no warm start, and the benefits of our approximate solvers that reduce the runtime of our exact solver.

\begin{table*}[t!]
	\small
	\centering
	\begin{tabular}{ c|cc|cc }
		\midrule
		\midrule
		& \multicolumn{2}{c|}{$s=5$} & \multicolumn{2}{c}{$s=6$}\\
		& Warm start  & Cold start & Warm start   & Cold start \\
		\midrule
		\midrule
		$p=100$& $2.4\pm 0.4$ & $4.9\pm1.1$  & $8.1\pm 3.0$  & $10.7\pm4.3$  \\
		$p=200$& $4.5\pm 1.2$ & $10.6\pm3.7$   & $23.4\pm 3.8$ &  $48.6\pm9.1$ \\
		$p=500$& $20.4\pm 5.9$& $56.2\pm 12.5 $  & $51.3\pm 13.6$ & $104\pm 14 $\\
		$p=1000$& $121\pm 31$ &  $197\pm 22$ & $184\pm 15$ &  $295\pm63$\\
		\hline
	\end{tabular}
	
	\caption{\footnotesize [Runtime for exact solver] 
		Comparison of the average runtime (in seconds) across 10 replicates ($\pm$ standard error) of our tailored outer approximation method with and without warm starts (Supplement~\ref{app:exact-expts}). }
	\label{app:table-mip}
	
\end{table*}

\section{Neuroscience Application} \label{neuro}

\subsection{Application Background} \label{neuro_background}

Studying the brain circuitry involved in diseases of the nervous system is critical for developing treatments. Neuroscientists are interested in measuring fast changes in the concentration of neurotransmitters (e.g., dopamine, serotonin and norepinephrine), that serve as chemical messengers between brain cells (neurons). Recently, the application of fast scan cyclic voltammetry (FSCV) has been applied to study neurotransmitter levels in humans. During surgical procedures, researchers place electrodes into the brains of patients during awake neurosurgery \citep{Kishida, Kishida2011PLoSone} and collect FSCV measurements while subjects perform decision-making tasks (e.g., stock market games) designed to test hypotheses about the role of neurotransmitters in cognitive processes. The implementation of FSCV in humans relies on prediction methods to estimate neurotransmitter concentration based upon raw electrical measurements recorded by electrodes. The technique records changes in electrical current that arise from varying the voltage potential on the measurement electrode. This results in a high dimensional time series signal, which is used as a vector of covariates to model the concentration of specific neurotransmitters. In vitro datasets are generated to serve as training datasets because the true concentrations, the outcome, are known (i.e., the data are labelled). The trained models are then used to make predictions of neurochemical concentration in the brain. 

In practice, each in vitro dataset is generated on a different electrode, which we treat as a task here because signals of each electrode differ in the marginal distribution of the covariates and in the conditional distribution of the outcome given the features \citep{Loewinger, Bang, Kishida, Moran}. An in-depth description of the data can be found in \citep{Loewinger}. Given the high dimensional nature of the covariates, researchers typically apply regularized linear models \citep{Kishida} or linear models with a dimension reduction pre-processing step, such as principal component regression \citep{Loewinger, Johnson}. Importantly, coefficient estimates from sparse linear models fit on each task separately ($L_0 L_2$) exhibit considerable heterogeneity in both their values and supports as can be seen in Figure \ref{fig:SHR_supp}. For these reasons, one might predict that multi-task methods that employ regularizers that share information through the $\boldsymbol{\beta}_k$ values, such as the Bbar method, may perform worse than methods that borrow strength across the supports, $\B{z}_k$.

\subsection{Neuroscience Data Pre-processing} \label{neuro_modeling}

Since the outcome is a measure of chemical concentration, and the experiments were run such that the marginal distribution of the outcome was uniform across a biologically-feasible range of values, we sought to determine whether transformations of the outcome would improve prediction performance. For a given outcome, $y_{k,i}$, we considered the following functions of the data: $f_1(y_{k,i}) = y_{k,i}$, $f_2(y_{k,i}) = \log(y_{k,i} + 1 \times 10^{-6})$, $f_3(y_{k,i}) = \log \left [ y_{k,i} + F^{-1}_{y_k}(0.75) / F^{-1}_{y_k}(0.25) \right ]$, $f_4(y_{k,i}) = \phi(y_{k,i}, {\zeta})$, where $\phi$ is the Yeo-Johnson Transformation \citep{yeo} and $\zeta$ was estimated through maximum likelihood estimation via the $\texttt{car}$ package \citep{car}. We finally centered and scaled the outcome with the marginal mean and standard deviation of the outcome for each task separately. We also centered and scaled the features of each task separately. The features and outcome of the test sets were centered and scaled using the statistics from the training set. In numerical experiments conducted on a different set of datasets, we inspected results of the above experiments, and concluded that $f_2(y_{k,i})$ produced the best prediction performance for the benchmark method, $L_0 L_2$, and proceeded with this transformation for all subsequent analyses. We did not inspect the performance of our methods with the other transformation on the primary datasets in order to avoid biases.

\subsection{Neuroscience Application Figures} \label{neuro_figs}

To inspect support heterogeneity in the FSCV application we fit Zbar+L2, $L_0 L_2$, gel, and Bbar methods on four tasks. To standardize across tasks, we centered and scaled the design matrices. The marginal distribution of the outcomes was the same across tasks due to the lab experimental design. We fit $L_0 L_2$ models with $s = p/20 = 50$ on each study separately with a task-specific cross-validated Ridge penalty and plotted the $\hat{\boldsymbol{\beta}}_k$ here. 

\begin{figure}[H]
	\centerline{
		\includegraphics[width=0.9\linewidth]{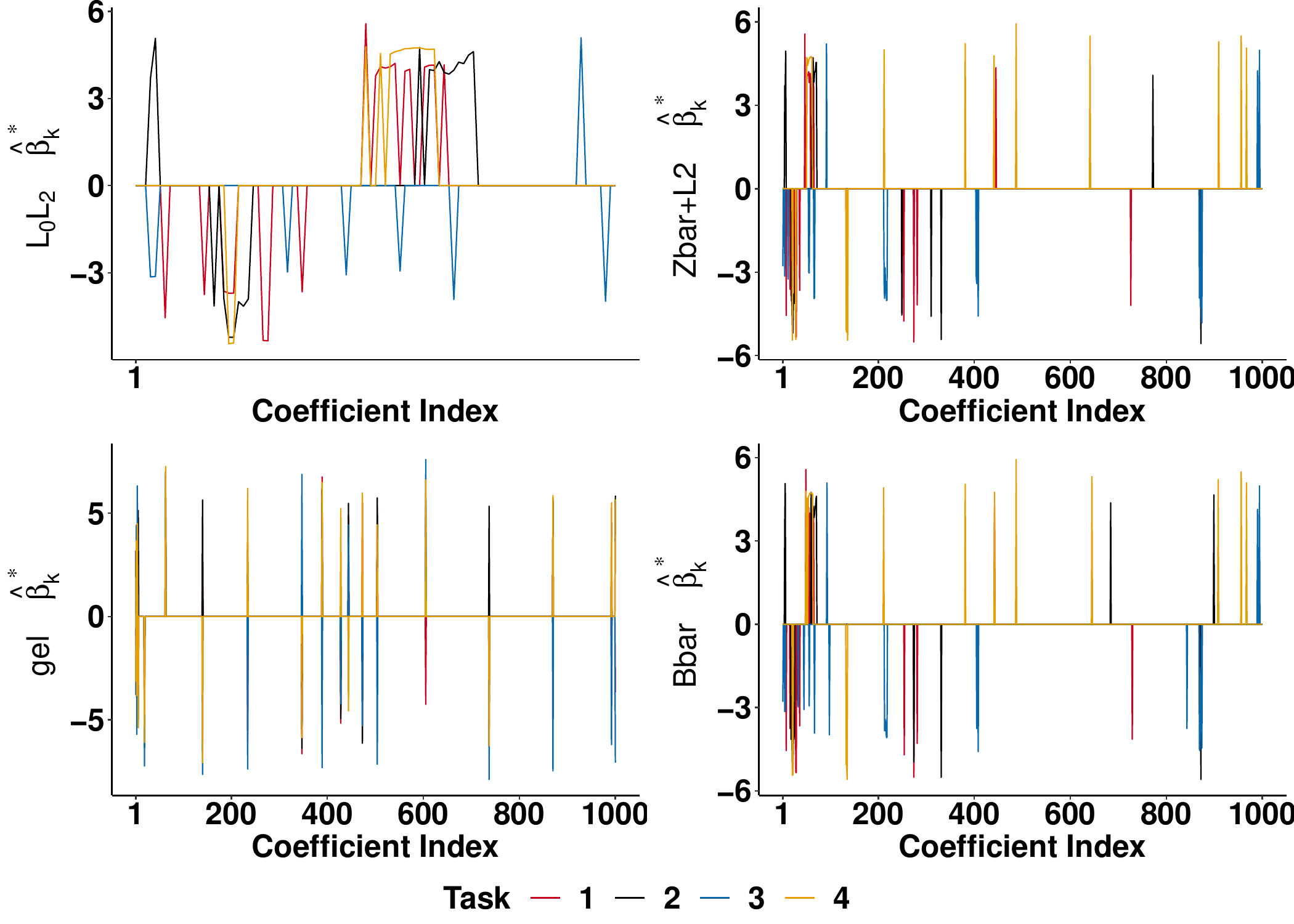}
	}
	\caption{\footnotesize Coefficient estimates with $L_0 L_2$ [Top Left], Zbar+L2 [Top Right], gel [Bottom Left] and Bbar [Bottom Right]. To show the solutions on the same scale across tasks, we plot the nonzero elements of $\hat{\boldsymbol{\beta}}_k^* = sgn(\hat{\boldsymbol{\beta}}_k) \odot \log(|\hat{\boldsymbol{\beta}}_k|)$.}
	\label{fig:SHR_supp}
\end{figure}

Next we show a more exhaustive collection of out-of-sample prediction performance figures than what was shown in the main text. As in the simulation experiments, we did not use local search during the optimization procedure when fitting models for hyperparameter tuning of the $\ell_0$ methods. Once the hyperparameters were selected, we did, however, conduct up to 50 iterations of local search when fitting the resulting models.

\begin{figure} 
	\centering
	\begin{subfigure}
		{0.7\textwidth}
		\centering
		\includegraphics[width=0.90\linewidth]{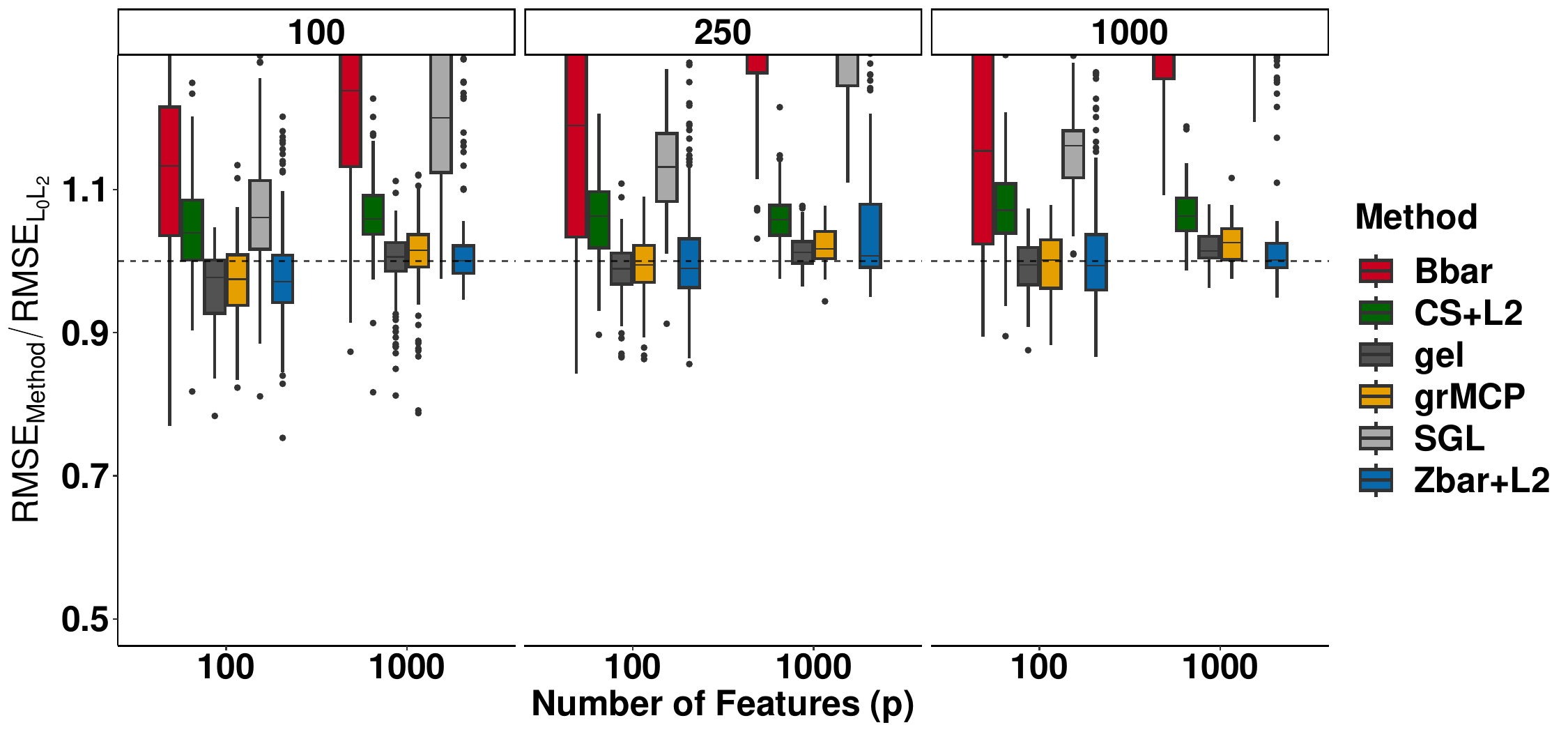}
	\end{subfigure}
	\begin{subfigure} 
		{0.7\textwidth}
		\centering
		\includegraphics[width=0.90\linewidth]{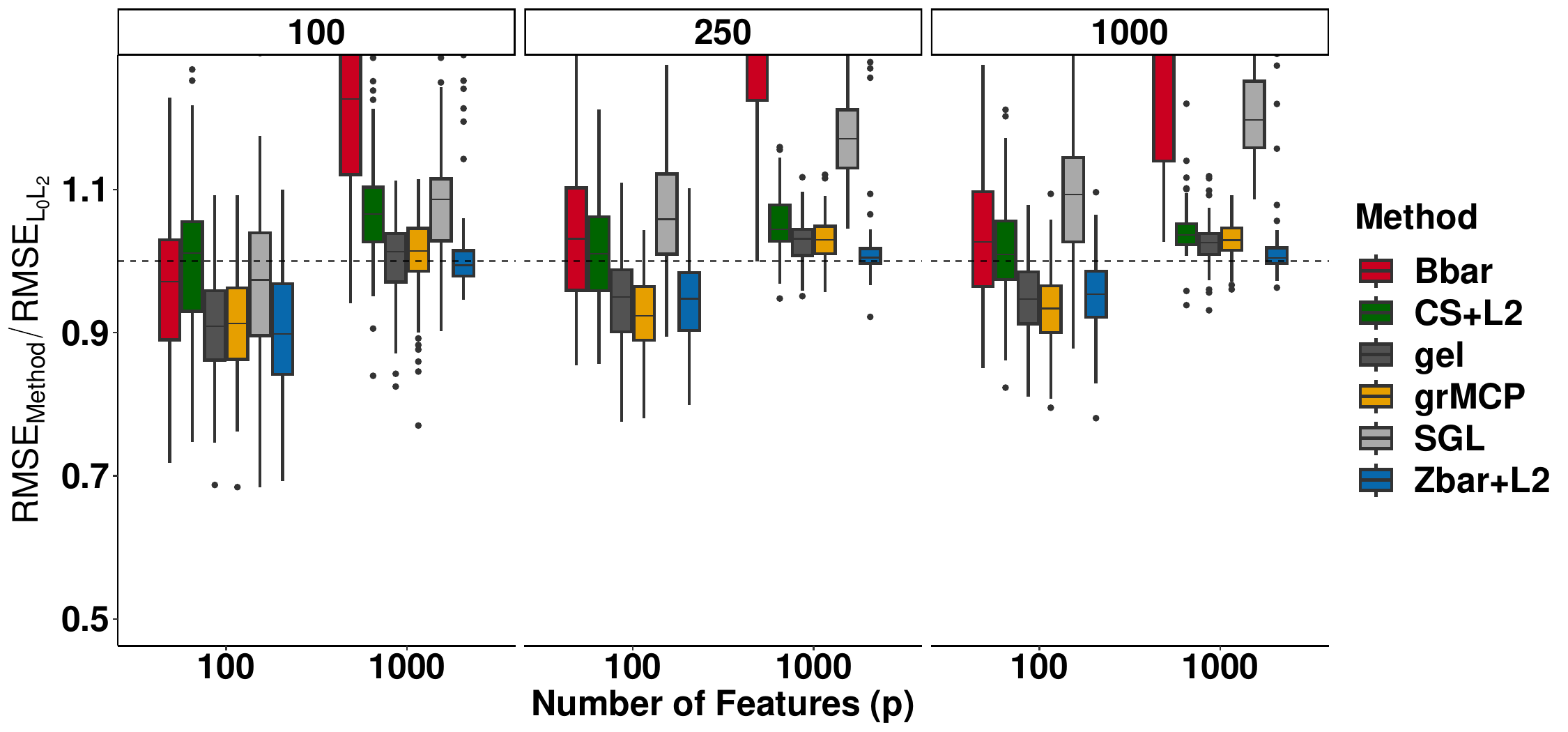}
	\end{subfigure}
	\centering
	\begin{subfigure}[t]{0.7\textwidth}
		\centering
		\includegraphics[width=0.90\linewidth]{multiTaskRMSE_fscv_L0vBenchmark_rho_25.pdf}
	\end{subfigure}
	\begin{subfigure}[t]{0.7\textwidth}
		\centering
		\includegraphics[width=0.90\linewidth]{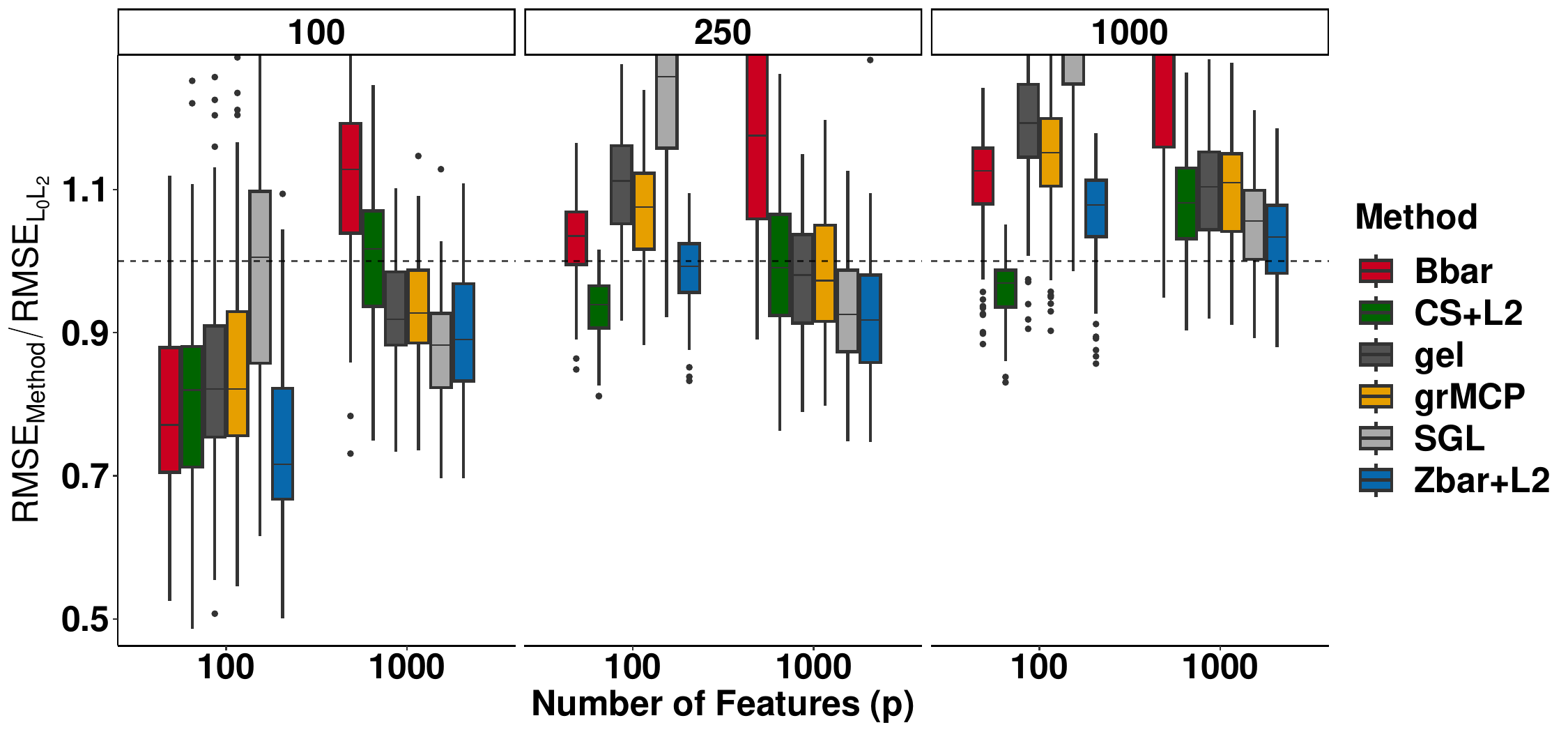}
	\end{subfigure}
	\caption{\footnotesize {Neuroscience application results.} Out-of-sample prediction performance (RMSE) averaged across random sets of tasks ($K=4$) for different $p$ (displayed on the x-axis) and $n_k$ (displayed on the panels) for methods shown in main text. Sparsity levels (from top) $s = 5, 10, 25, 50$.} 
	\label{fig:fscvRMSE_supp1}
\end{figure}

\begin{figure} 
	\centering
	\begin{subfigure} 
		{0.7\textwidth}
		\centering
		\includegraphics[width=0.90\linewidth]{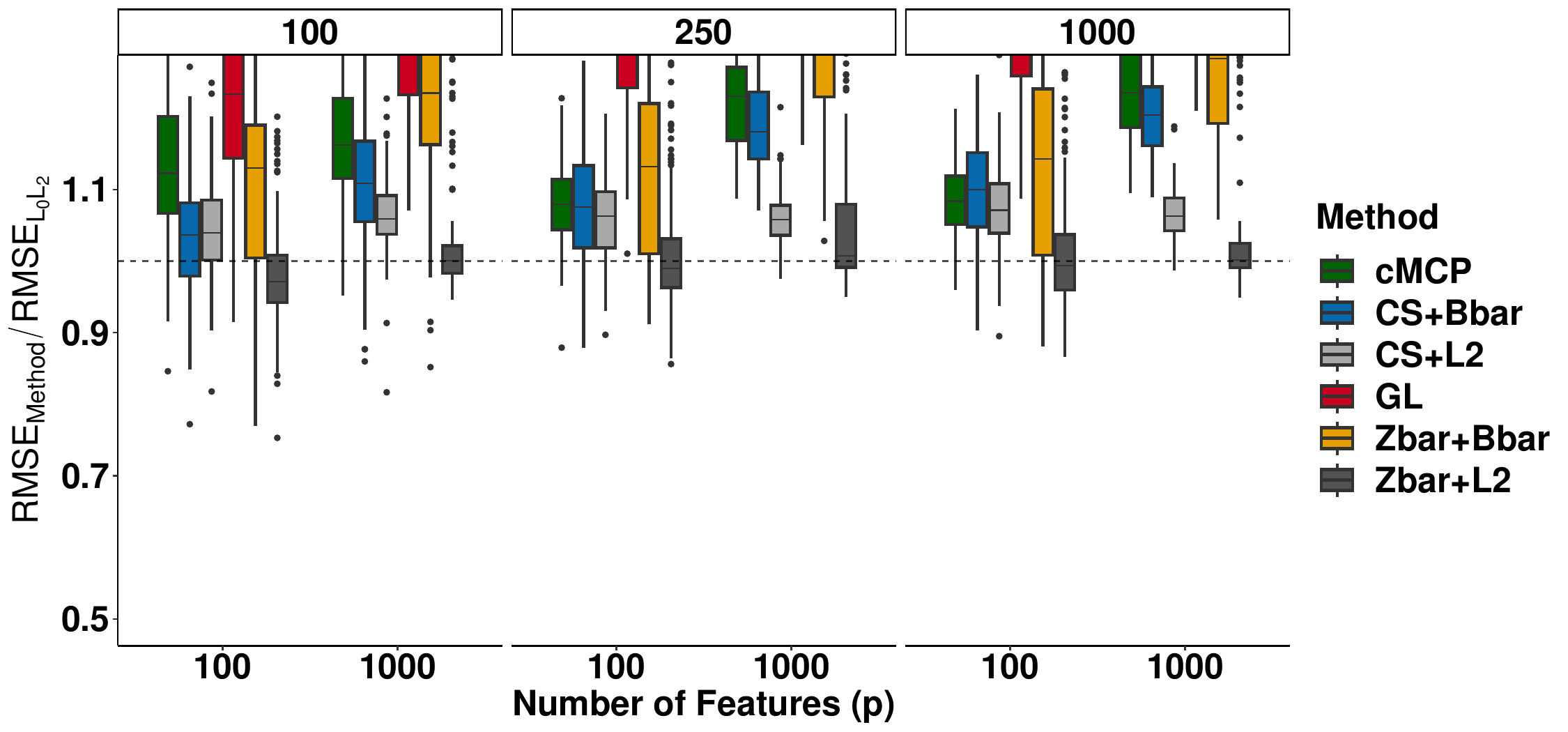}
	\end{subfigure}
	\begin{subfigure} 
		{0.7\textwidth}
		\centering
		\includegraphics[width=0.90\linewidth]{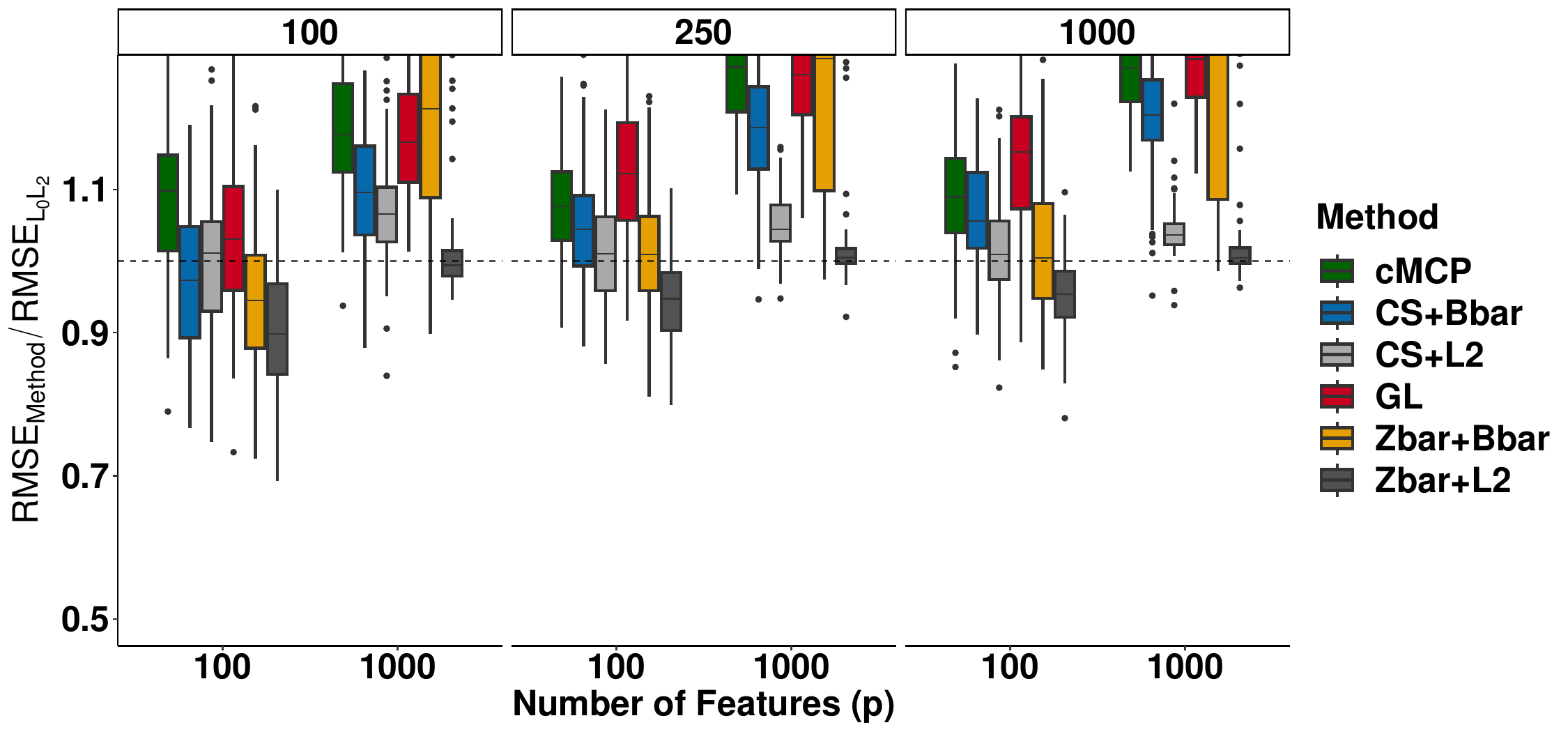}
	\end{subfigure}
	
	\centering
	\begin{subfigure}{0.7\textwidth}
		\centering
		\includegraphics[width=0.90\linewidth]{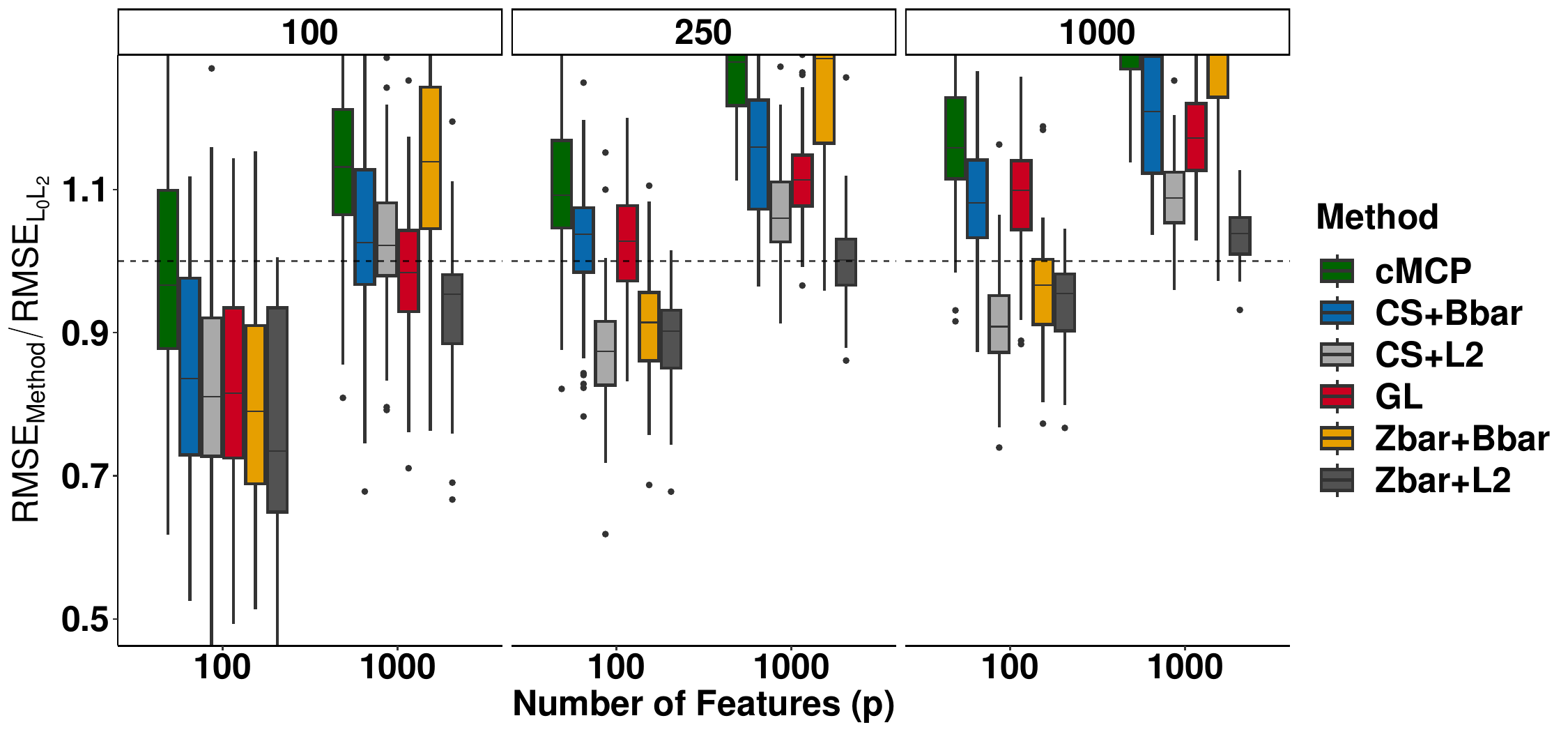}
	\end{subfigure}
	\begin{subfigure}{0.7\textwidth}
		\centering
		\includegraphics[width=0.90\linewidth]{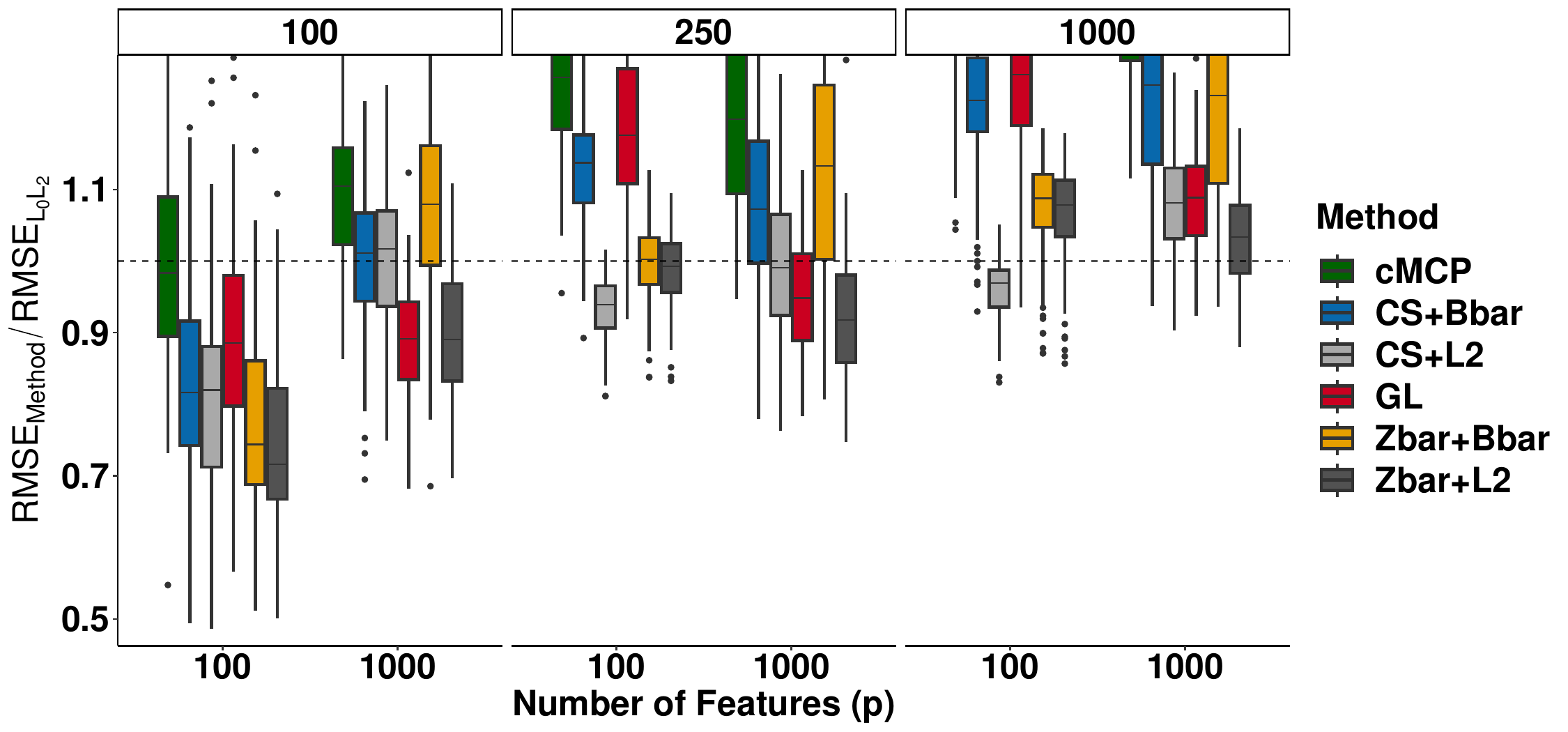}
	\end{subfigure}
	\caption{\footnotesize {Neuroscience application results.} Out-of-sample prediction performance (RMSE) averaged across random sets of tasks ($K=4$) for different $p$ (displayed on the x-axis) and $n_k$ (displayed on the panels) for additional penalties and $\ell_0$ methods not shown in main text. Sparsity levels (from top) $s = 5, 10, 25, 50$.} 	\label{fig:fscvRMSE_supp2}
\end{figure}
\section{Cancer Application} \label{cancer_supp}

\subsection{Cancer Application Figures} 

\begin{figure}[H]
	\centerline{
		\includegraphics[width=0.85 \linewidth]{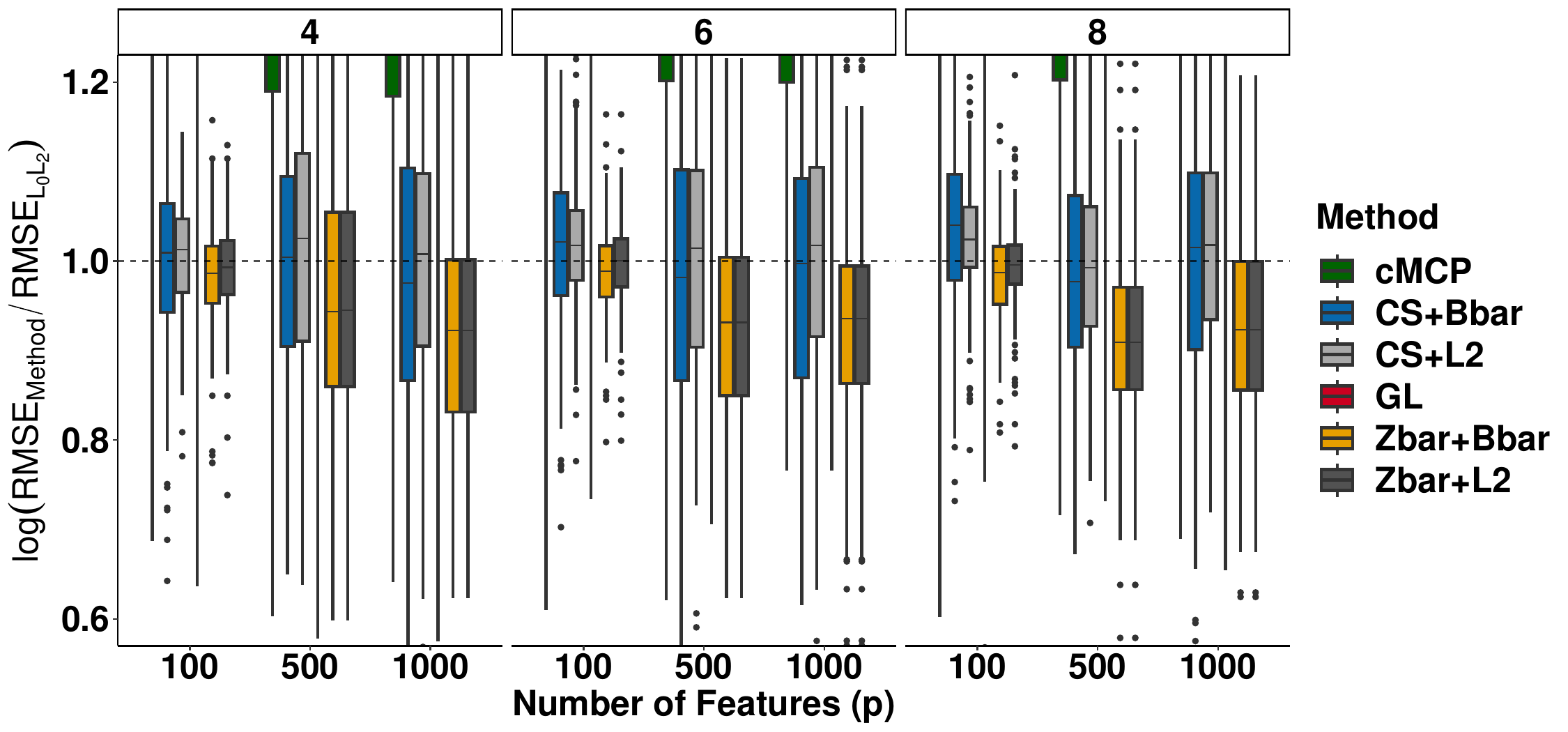}
	}
	\caption{ \footnotesize {Cancer application results.} Hold-one-study-out prediction performance of additional methods, not presented in main test, averaged across tasks for different $K$. The performance (RMSE) is presented relative to the performance of task-specific sparse regressions.} 
	\label{fig:cancerRMSE_supp}
\end{figure}

\begin{figure}[H] 
	\centering
	\begin{subfigure}[t]{0.75\textwidth}
		\centering
		\includegraphics[width=0.90\linewidth]{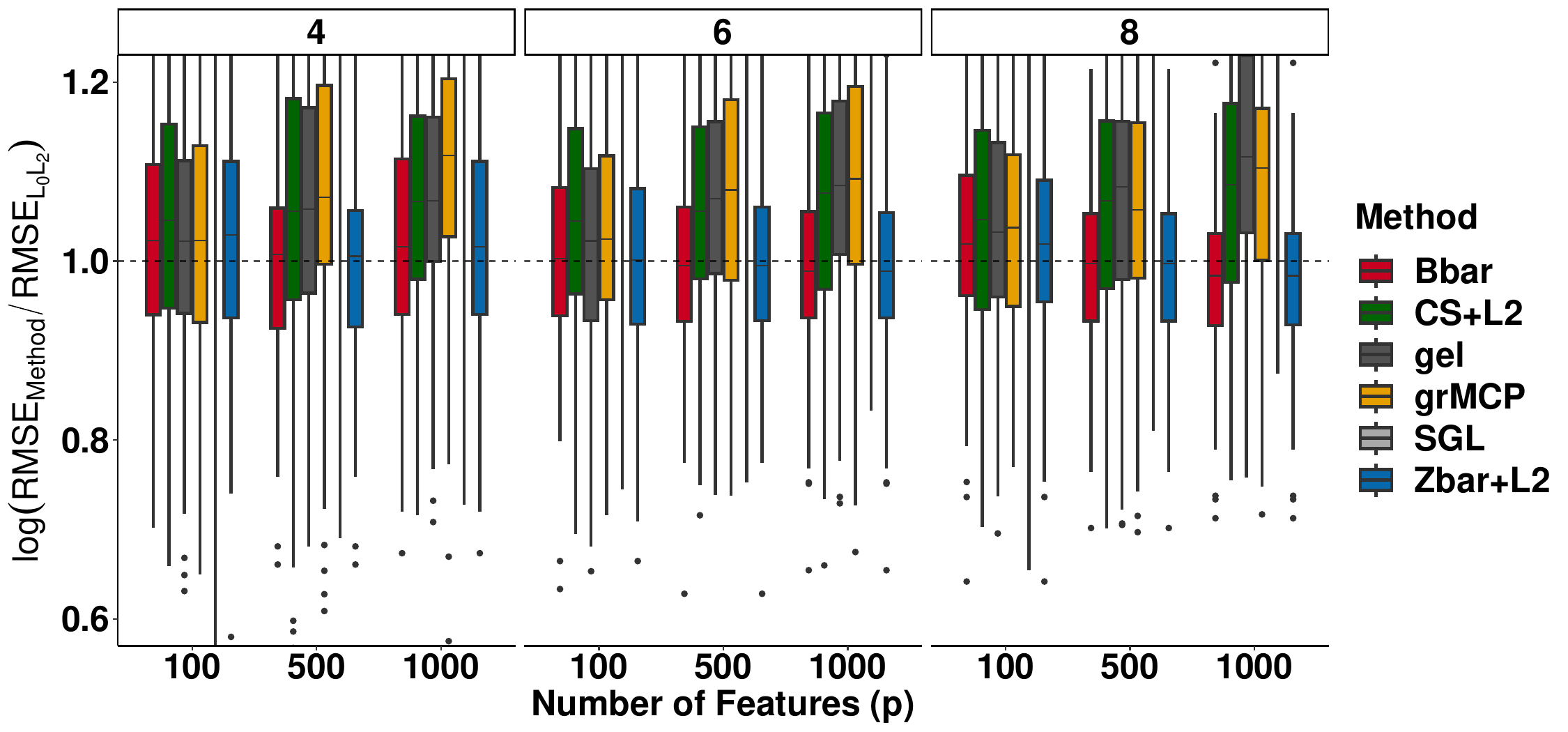}
	\end{subfigure}
	\begin{subfigure}[t]{0.75\textwidth}
		\centering
		\includegraphics[width=0.90\linewidth]{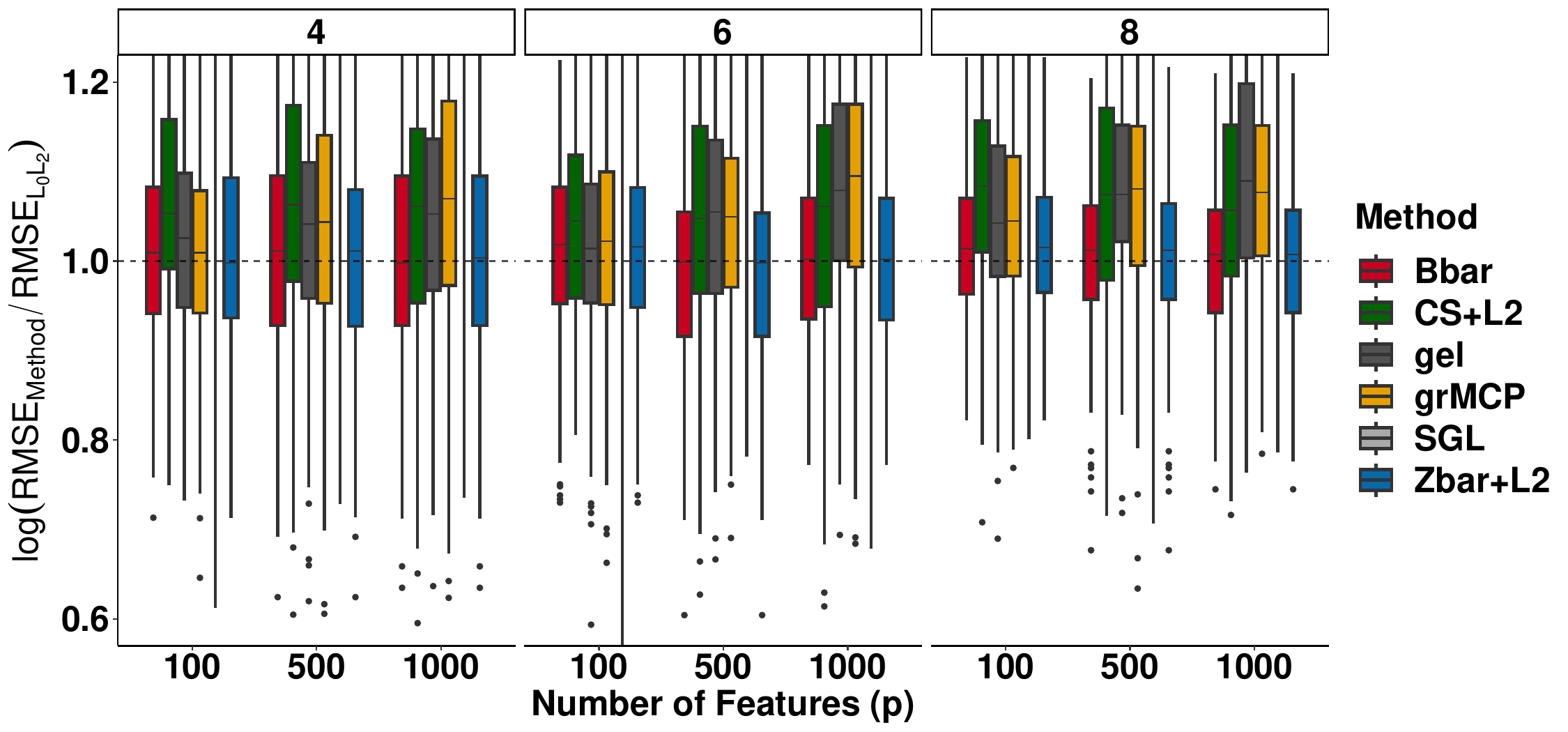}
	\end{subfigure}
	\centering
	\begin{subfigure}[t]{0.75\textwidth}
		\centering
		\includegraphics[width=0.90\linewidth]{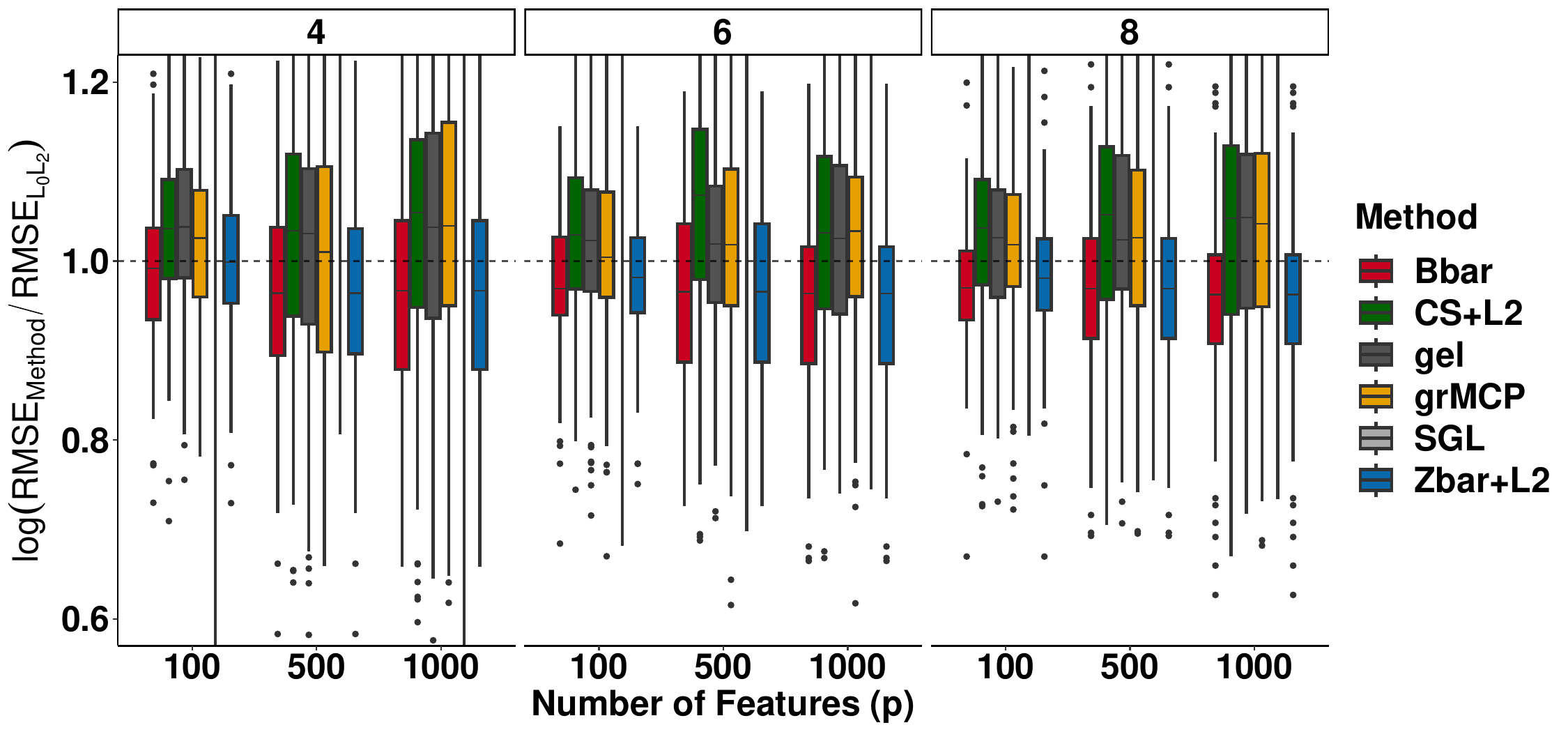}
	\end{subfigure}
	\begin{subfigure}[t]{0.75\textwidth}
		\centering
		\includegraphics[width=0.90\linewidth]{multiTaskRMSE_bCancer_L0_50.pdf}
	\end{subfigure}
	\caption{\footnotesize {Cancer application results.} Hold-one-study-out prediction performance of methods presented in main text, averaged across tasks for different $K$. Sparsity levels (from top) $s = 5, 10, 25, 50$.} 
	\label{fig:cancerRMSE_supp1}
\end{figure}

\begin{figure}[H] 
	\centering
	\begin{subfigure}[t]{0.75\textwidth}
		\centering
		\includegraphics[width=0.90\linewidth]{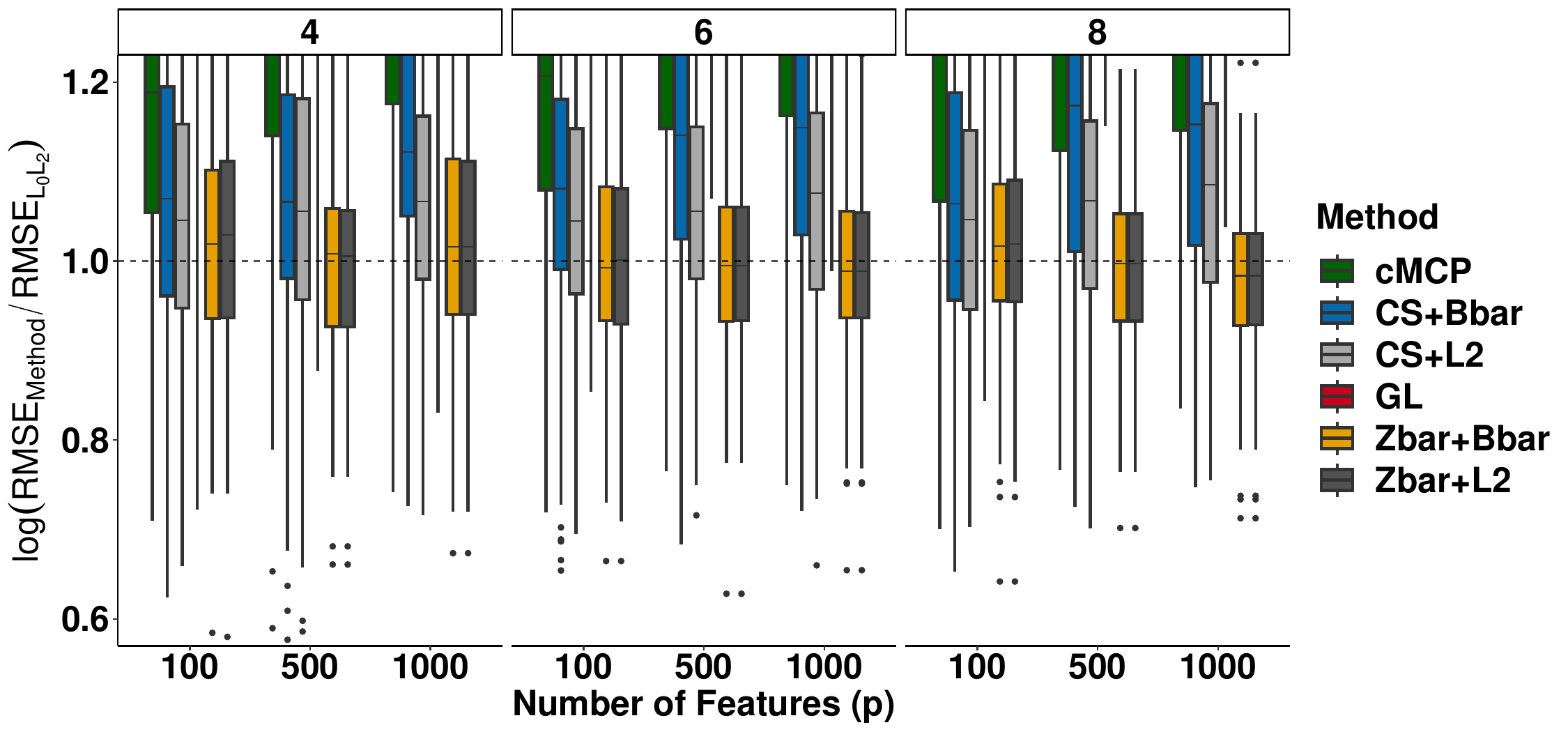}
	\end{subfigure}
	\begin{subfigure}[t]{0.75\textwidth}
		\centering
		\includegraphics[width=0.90\linewidth]{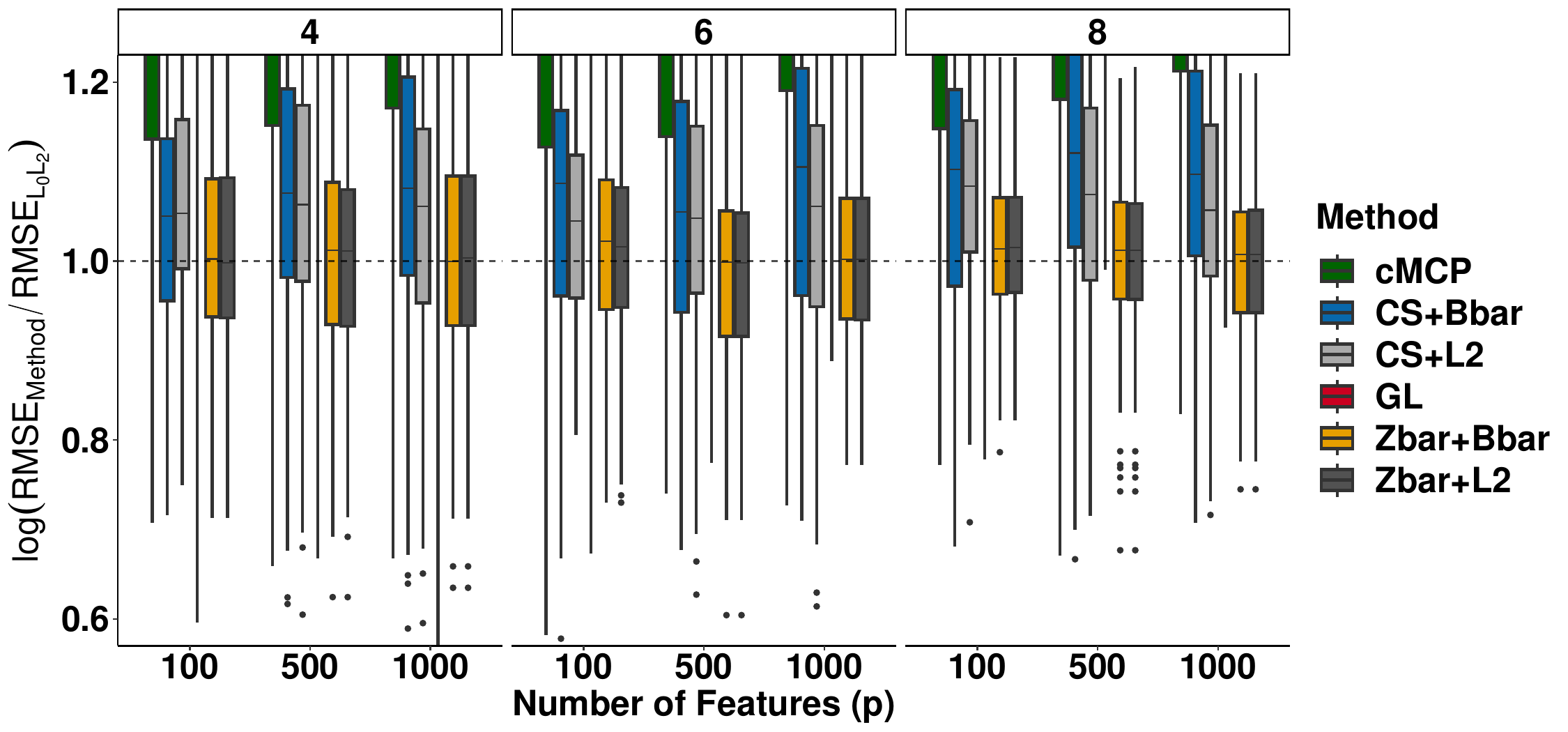}
	\end{subfigure}
	\centering
	\begin{subfigure}[t]{0.75\textwidth}
		\centering
		\includegraphics[width=0.90\linewidth]{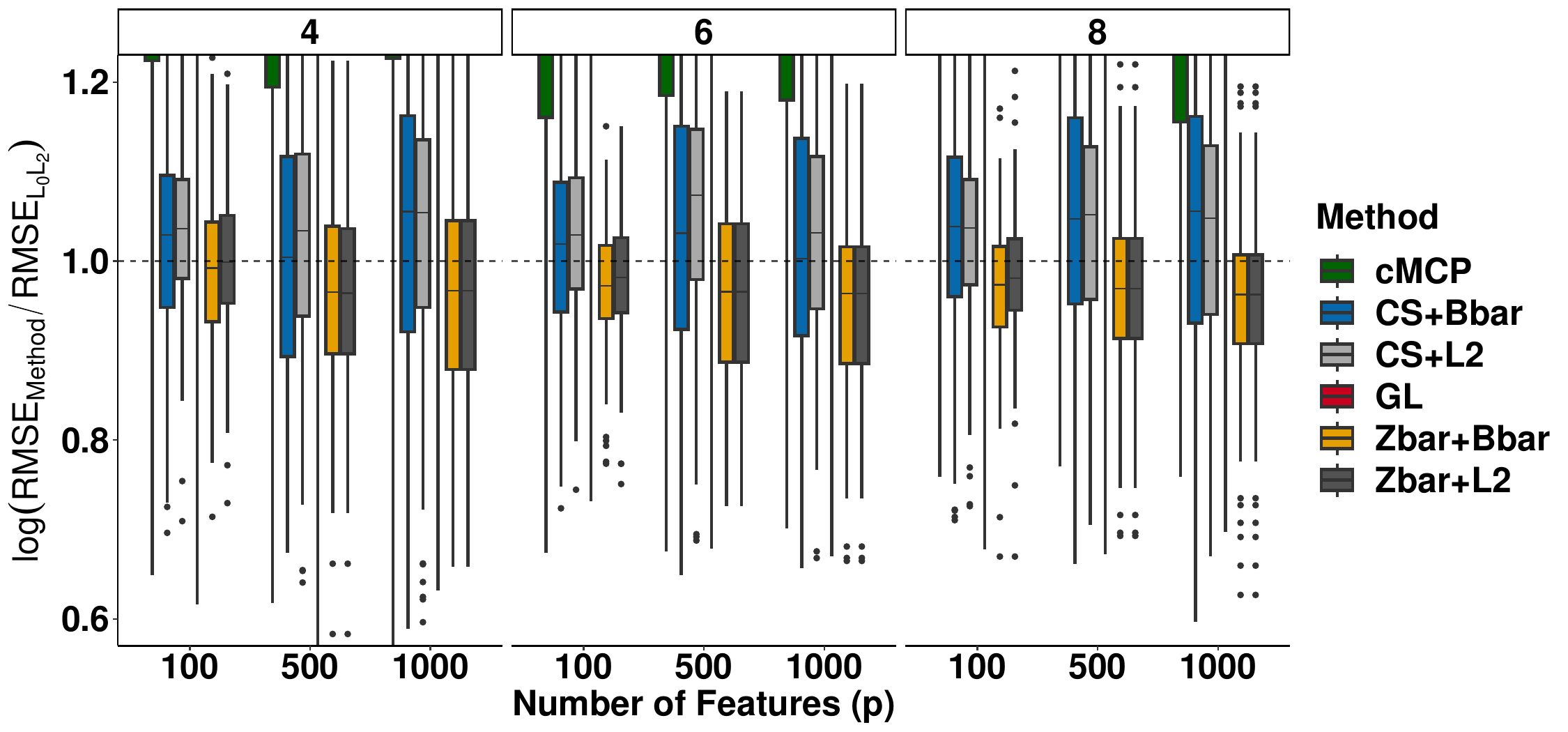}
	\end{subfigure}
	\begin{subfigure}[t]{0.75\textwidth}
		\centering
		\includegraphics[width=0.90\linewidth]{multiTaskRMSE_bCancer_all_50.pdf}
	\end{subfigure}
	\caption{\footnotesize {Cancer application results.} Hold-one-study-out prediction performance of additional methods not presented in main text, averaged across tasks for different $K$. Sparsity levels (from top) $s = 5, 10, 25, 50$.} 
	\label{fig:cancerRMSE_supp2}
\end{figure}

\section{Rug Plots} \label{rugPlots}

\begin{figure}[H]
	\centerline{
		\includegraphics[width=0.75\linewidth]{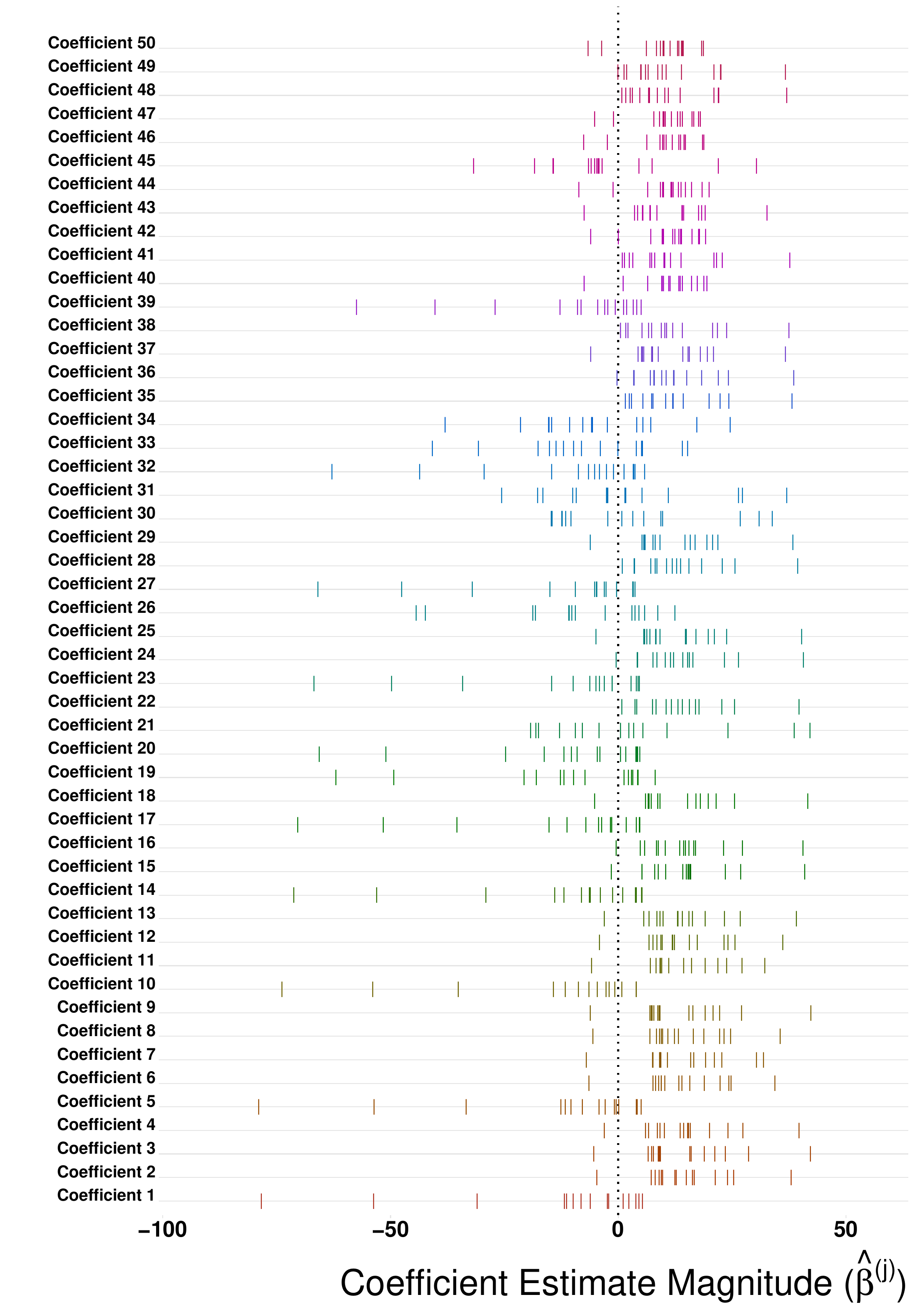}
	}
	\caption{\footnotesize FSCV application rug plots displaying $\hat{\beta}_{k,j}$ for the 50 $j$'s with the greatest average magnitude across the $K$ task-specific Ridge regressions ($\frac{1}{K} \sum_{k=1}^K | \hat{\beta}_{k,j}|$). For Coefficient $j$, each mark on the horizontal line is one of $K$ task-specific empirical estimates of the $\hat{\boldsymbol{\beta}}_{k,j}$.} 
\end{figure}

\begin{figure}[H]
	\centerline{
		\includegraphics[width=0.75\linewidth]{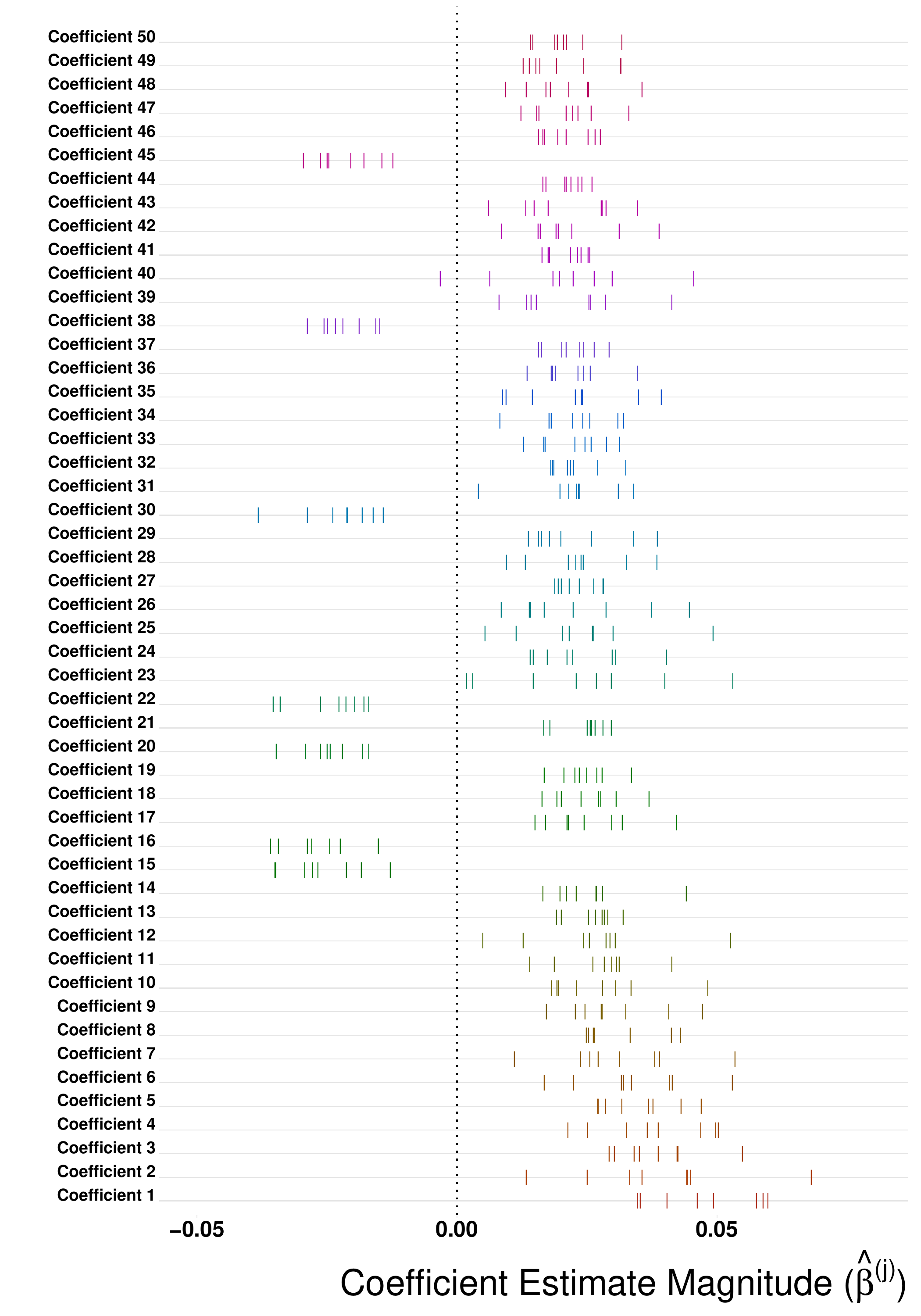}
	}
	\caption{\footnotesize Cancer genomics application rug plots displaying $\hat{\beta}_{k,j}$ for the 50 $j$'s with the greatest average magnitude across the $K$ task-specific Ridge regressions ($\frac{1}{K} \sum_{k=1}^K | \hat{\beta}_{k,j}|$). For Coefficient $j$, each mark on the horizontal line is one of $K$ task-specific empirical estimates of the $\hat{\boldsymbol{\beta}}_{k,j}$.} 
\end{figure}

\begin{figure}[H]
	\centerline{
		\includegraphics[width=0.75\linewidth]{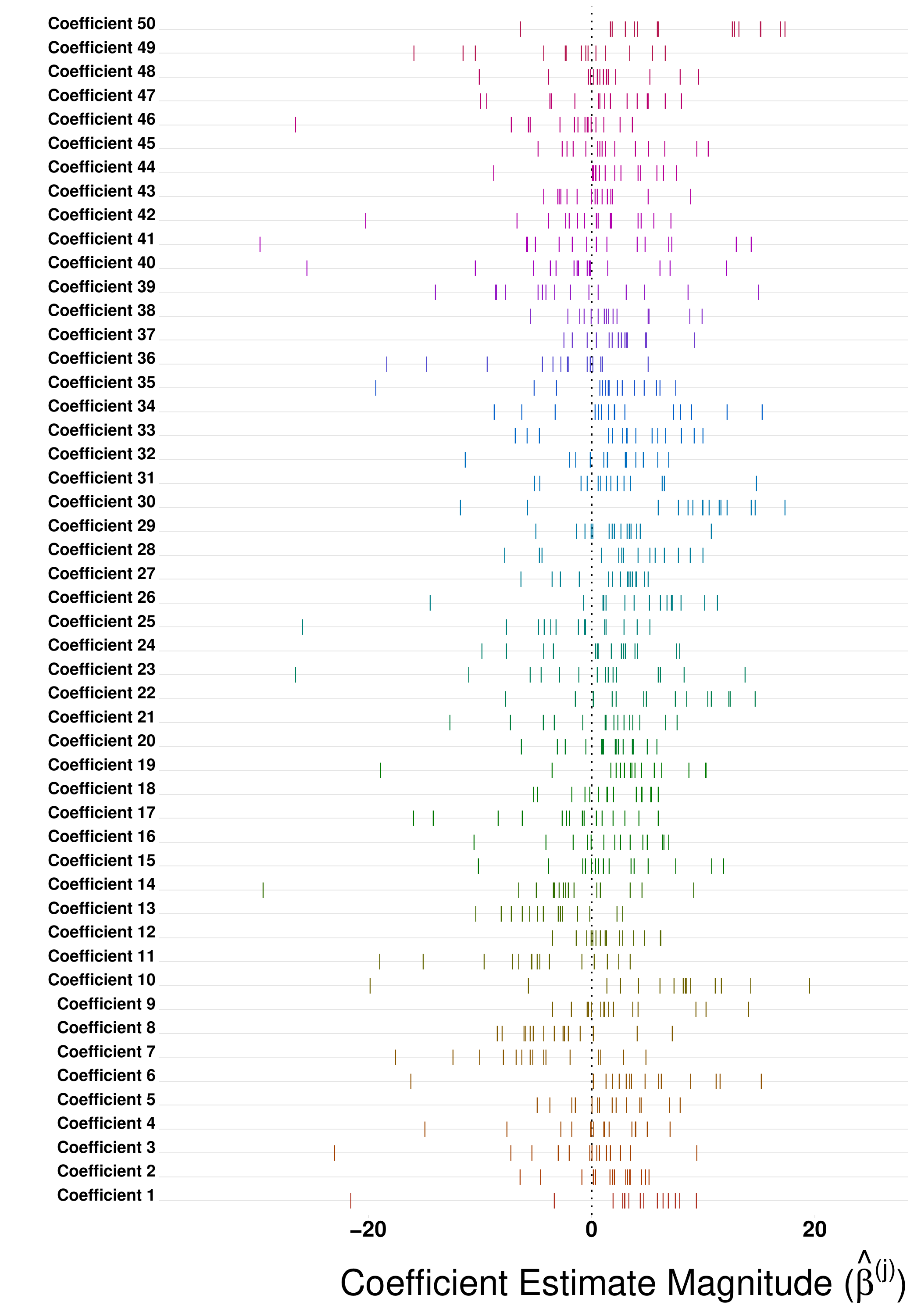}
	}
	\caption{\footnotesize FSCV application rug plots displaying $\hat{\beta}_{k,j}$ for a random set of 50 $j$'s across the $K$ task-specific Ridge regressions ($\frac{1}{K} \sum_{k=1}^K | \hat{\beta}_{k,j}|$). For Coefficient $j$, each mark on the horizontal line is one of $K$ task-specific empirical estimates of the $\hat{\boldsymbol{\beta}}_{k,j}$.} 
\end{figure}

\begin{figure}[H]
	\centerline{
		\includegraphics[width=0.75\linewidth]{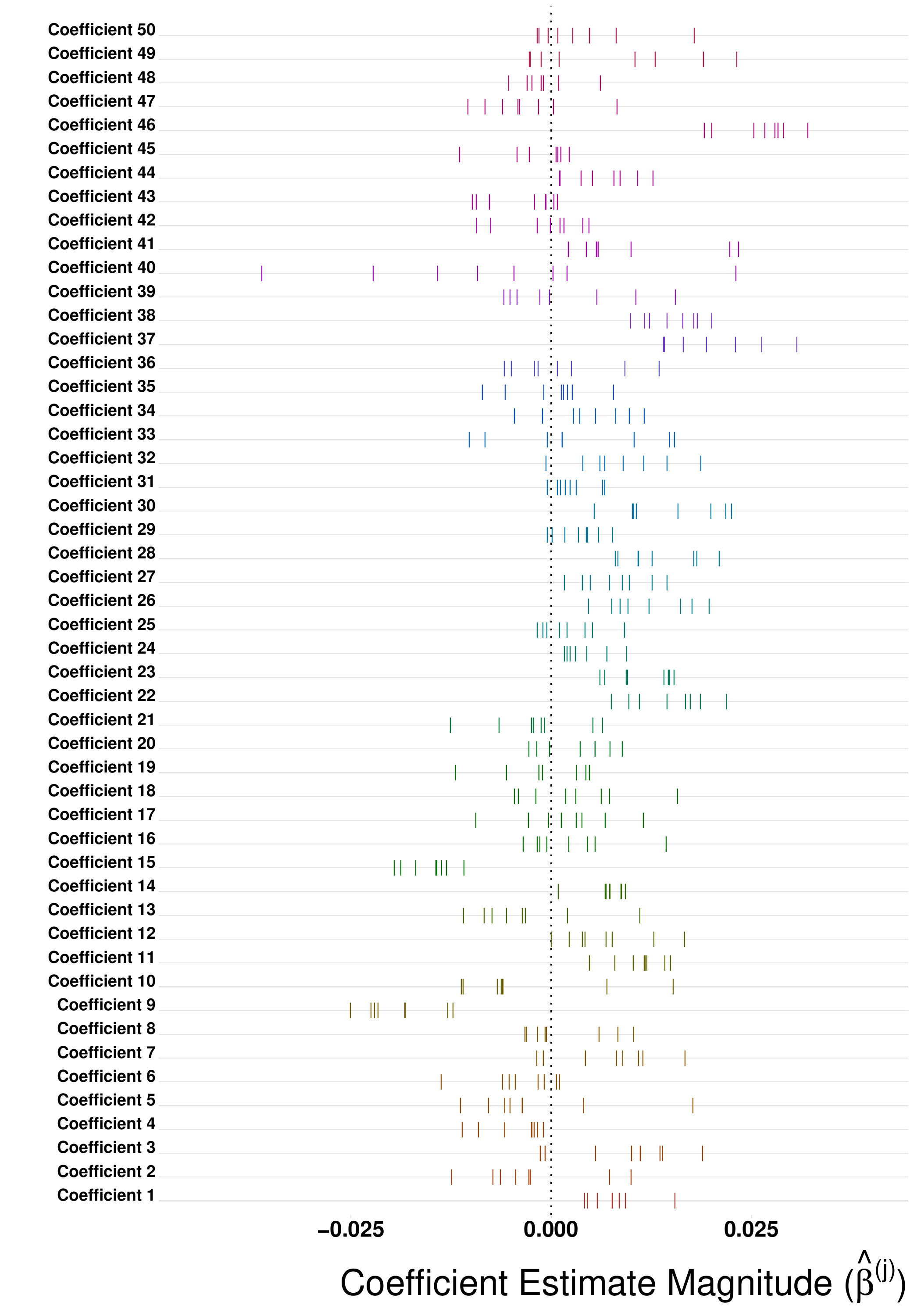}
	}
	\caption{\footnotesize Cancer genomics application rug plots displaying $\hat{\beta}_{k,j}$ for a random set of 50 $j$'s across the $K$ task-specific Ridge regressions ($\frac{1}{K} \sum_{k=1}^K | \hat{\beta}_{k,j}|$). For Coefficient $j$, each mark on the horizontal line is one of $K$ task-specific empirical estimates of the $\hat{\boldsymbol{\beta}}_{k,j}$.} 
\end{figure}

\end{document}